\documentclass[11pt,a4paper,reqno]{amsart}%
\usepackage{amsthm,amsmath,amsfonts,amssymb,amsxtra,appendix,bookmark,dsfont,latexsym,color,epsfig}

\makeatletter
\usepackage{hyperref}
\usepackage{delarray,a4,color}
\usepackage{euscript}
\usepackage[latin1]{inputenc}

\numberwithin{equation}{section}

\newtheorem{theorem}{Th\'eor\`eme}[section]
\newtheorem{lemma}[theorem]{Lemme}
\newtheorem{corollary}[theorem]{Corollaire}

\theoremstyle{definition}
\newtheorem{definition}[theorem]{D\'efinition}

\theoremstyle{remark}
\newtheorem{remark}[theorem]{Remarque}

\numberwithin{equation}{section}

\usepackage{color}

\numberwithin{equation}{section}

\newcommand{\norm}[1]{\left\lVert #1 \right\rVert}

\newcommand{\bdm}{\begin{displaymath}}
\newcommand{\edm}{\end{displaymath}}
\newcommand{\bdn}{\begin{eqnarray}}
\newcommand{\edn}{\end{eqnarray}}
\newcommand{\bay}{\begin{array}{c}}
\newcommand{\eay}{\end{array}}
\newcommand{\ben}{\begin{enumerate}}
\newcommand{\een}{\end{enumerate}}
\newcommand{\beq}{\begin{equation}}
\newcommand{\eeq}{\end{equation}}

\newcommand{\tx}{\textstyle}

\newcommand{\eps}{\varepsilon}

\newcommand{\R}{\mathbb{R}}
\newcommand{\N}{\mathbb{N}}

\newcommand{\C}{\mathbb{C}}

\newcommand{\F}{\mathcal{F}}
\newcommand{\E}{\mathcal{E}}

\newcommand{\cS}{\mathcal{S}}

\newcommand{\PP}{\mathcal{P}}

\newcommand{\one}{{\ensuremath {\mathds 1} }}

\newcommand{\al}{\alpha}

\newcommand{\ep}{\varepsilon}

\newcommand{\Om}{\Omega}
\newcommand{\om}{\omega}

\newcommand{\half}{\tx{\frac{1}{2}}}

\newcommand{\supp}{\mathrm{supp}}
\newcommand{\wto}{\rightharpoonup}

\newcommand{\bral}{\left<}
\newcommand{\brar}{\right|}
\newcommand{\ketl}{\left|}
\newcommand{\ketr}{\right>}

\newcommand{\ZN}{\mathcal{Z}_{N}}

\newcommand{\MFf}{\F ^{\rm MF}}
\newcommand{\MFe}{F ^{\rm MF}}
\newcommand{\MFmin}{\mathcal{M} ^{\rm MF}}
\newcommand{\rhoMF}{\varrho ^{\rm MF}}

\newcommand{\MFEf}{\E ^{\rm MF}}
\newcommand{\MFEe}{E ^{\rm MF}}

\newcommand{\MFfo}{\E ^{\rm MF}}
\newcommand{\MFeo}{E ^{\rm MF}}
\newcommand{\MFfal}{\E ^{\rm MF}_{\al}}
\newcommand{\MFeal}{E ^{\rm MF}_{\al}}

\newcommand{\logal}{\log_{\al}}

\DeclareMathOperator{\Tr}{Tr}
\DeclareMathOperator{\tr}{Tr}

\def\geqslant{\ge}
\def\leqslant{\le}
\def\bq{\begin{eqnarray}}
\def\eq{\end{eqnarray}}
\def\bqq{\begin{eqnarray*}}
\def\eqq{\end{eqnarray*}}
\def\nn{\nonumber}

\def\eps{\varepsilon}
\def\wto{\rightharpoonup}
\def\gB {\mathfrak{B}}

\newcommand\1{{\ensuremath {\mathds 1} }}

\newcommand{\gammaP}{\gamma_{\Psi}}
\newcommand{\dM}{{\rm d}M}

\renewcommand{\epsilon}{\varepsilon}

\def\cF {\mathcal{F}}

\def\R {\mathbb{R}}
\def\C {\mathbb{C}}

\def\cS {\mathcal{S}}

\def\E {\mathcal{E}}
\def\cE {\mathcal{E}}

\def\F {\mathcal{F}}

\def\R {\mathbb{R}}
\def\C {\mathbb{C}}

\def\gS{\mathfrak{S}}

\def\cS {\mathcal{S}}

\def\E {\mathcal{E}}
\def\cM {\mathcal{M}}

\def\d{{\rm d}}

\def\gH{\mathfrak{H}}
\newcommand\ii{{\ensuremath {\infty}}}
\newcommand\pscal[1]{{\ensuremath{\left\langle #1 \right\rangle}}}

\renewcommand{\leq}{\leqslant}
\renewcommand{\geq}{\geqslant}

\newcommand{\cEH}{\ensuremath{\cE_{\text{\textnormal{H}}}}}

\newcommand{\EH}{\E_{\rm H}}
\newcommand{\eH}{e_{\rm H}}
\newcommand{\uH}{u_{\rm H}}

\newcommand{\ENLS}{\cE_{\rm nls}}
\newcommand{\eNLS}{e_{\rm nls}}
\newcommand{\uNLS}{u_{\rm nls}}

\newcommand{\mubf}{\boldsymbol{\mu}}
\newcommand{\nubf}{\boldsymbol{\nu}}
\newcommand{\mut}{\tilde{\boldsymbol{\mu}}}
\newcommand{\Gammat}{\tilde{\Gamma}}
\newcommand{\gammat}{\tilde{\gamma}}

\newcommand{\Pp}{P_{\perp}}
\newcommand{\wep}{w_{\varepsilon}}
\newcommand{\aep}{a_{\varepsilon}}
\newcommand{\Eep}{E ^{\varepsilon}}

\newcommand{\MNLS}{\cM_{\rm nls}}

\newcommand{\Fcl}{\F_{\rm cl}}
\newcommand{\Fcle}{F_{\rm cl}}
\newcommand{\mucl}{\mu_{\rm cl}}

\numberwithin{equation}{section}

\title[Th\'eor\`emes de de Finetti et condensation de Bose-Einstein]{Théorèmes de de Finetti, limites de champ moyen et condensation de Bose-Einstein}

\author[Nicolas Rougerie]{Nicolas Rougerie\\ 
L\lowercase{aboratoire de }P\lowercase{hysique et }M\lowercase{od\'elisation des }M\lowercase{ilieux} C\lowercase{ondens\'es}, U\lowercase{niversit\'e} G\lowercase{renoble 1} \& CNRS.}
\address{Universit\'e Grenoble 1 \& CNRS, LPMMC, UMR 5493, BP 166, 38042 Grenoble, France.}
\email{nicolas.rougerie@lpmmc.cnrs.fr}

\date{Février-Mars 2014, cours Peccot, Collège de France}

\begin{document}


\renewcommand{\contentsname}{Sommaire}
\renewcommand{\refname}{R\'ef\'erences}
\renewcommand{\abstractname}{R\'esum\'e}
\renewcommand{\appendixname}{Appendice}


\maketitle

\bigskip

\begin{center}
Cours Peccot au coll\`ege de France, F\'evrier-Mars-Avril 2014. 
\end{center}

\bigskip

\begin{abstract}
Ces notes de cours traitent de l'approximation de champ moyen pour les \'etats d'\'equilibre de syst\`emes \`a $N$ corps en m\'ecanique statistique classique et quantique. Une strat\'egie g\'en\'erale pour la justification des mod\`eles effectifs bas\'es sur des hypoth\`eses d'ind\'ependance statistique des particules est pr\'esent\'ee en d\'etail. Les outils principaux sont des th\'eor\`emes de structure \`a la de Finetti qui d\'ecrivent les limites pour $N$ grand des \'etats accessibles aux syt\`emes en question, en exploitant l'indiscernabilit\'e des particules. L'accent est mis sur les aspects quantiques, notamment l'approximation de champ moyen pour le fondamental d'un grand syst\`eme bosonique, en lien avec le ph\'enom\`ene de condensation de Bose-Einstein: structure des matrices de densit\'e r\'eduites d'un grand syst\`eme bosonique, m\'ethodes de localisation dans l'espace de Fock, d\'erivation de fonctionnelles d'\'energie effectives de type Hartree ou Schr\"odinger non lin\'eaire.
\end{abstract}

\newpage

$\phantom{o}$

\bigskip
\bigskip

\bigskip
\bigskip

\bigskip
\bigskip

\bigskip
\bigskip

\bigskip
\bigskip

\bigskip
\bigskip

\noindent\emph{D\'edi\'e \`a ma fille C\'eleste dont la naissance anticip\'ee a un peu perturb\'e le d\'eroulement du cours, mais qui est si adorable qu'on ne saurait lui en vouloir. }

\newpage

$\phantom{o}$

\bigskip
\bigskip

\bigskip
\bigskip

\tableofcontents

\newpage

\section*{\textbf{Avant-propos}}\label{sec:pre intro}

Le but de ces notes de cours est de pr\'esenter de mani\`ere aussi exhaustive et p\'edagogique que possible un ensemble de r\'esultats math\'ematiques r\'ecents ayant trait au ph\'enom\`ene physique de \emph{condensation de Bose-Einstein} dans des gaz d'atomes ultra-froids. Un des nombreux probl\`emes th\'eoriques pos\'es par ces exp\'eriences consiste en la compr\'ehension du lien entre les mod\`eles effectifs, d\'ecrivant les exp\'eriences avec une pr\'ecision remarquable, et les principes de base de la m\'ecanique quantique. Le processus liant les descriptions fondamentales et effectives est souvent appell\'e une \emph{limite de champ moyen} et la th\'eorie de ces limites a motiv\'e un tr\`es grand nombre de travaux en physique th\'eorique et math\'ematique. Dans ce cours on se focalisera sur une des m\'ethodes permettant de traiter la limite de champ moyen, bas\'ee sur les \emph{th\'eor\`emes \`a la de Finetti}. On interpr\'etera l'\'emergence des mod\`eles de champ moyen comme une cons\'equence 
fondamentale de la structure des \'etats physiques en consid\'eration. 

Ce texte a pour vocation d'embrasser des aspects d'analyse, de probabilit\'e, de physique de la mati\`ere condens\'ee, de physique des atomes froids, de m\'ecanique statistique et quantique, d'information quantique. L'emphase sera mise sur la sp\'ecialit\'e de l'auteur, \`a savoir les aspects analytiques de la d\'erivation des mod\`eles de champ moyen, dans le cas de mod\`eles statiques. La pr\'esentation aura donc un aspect beaucoup plus math\'ematique que physique, mais le lecteur devrait garder \`a l'esprit le lien entre les questions soulev\'ees ici (et dans la litt\'erature cit\'ee) et la physique des atomes froids, en particulier les exp\'eriences ayant permis l'observation de condensats de Bose-Einstein en laboratoire depuis le milieu des ann\'ees 90.

\subsection*{Un mot sur les exp\'eriences}

\smallskip

Le ph\'enom\`ene de condensation de Bose-Einstein est au centre d'un vaste de champ de recherche en pleine expansion depuis le milieu des ann\'ees~90. L'extr\^eme versatilit\'e des conditions maintenant accessibles \`a l'exp\'erience rend possible l'investigation d\'etaill\'ee de nombreuses questions physiques fondamentales. Le lecteur est renvoy\'e aux textes~\cite{Aftalion-06,BloDalZwe-08,DalGioPitStr-99,DalGerJuzOhb-11,LieSeiSolYng-05,PetSmi-01,PitStr-03,Fetter-09,Cooper-08} et leurs bibliographies pour de plus amples d\'eveloppements sur la physique des atomes froids.  Des pr\'esentations tr\`es accessibles au non-sp\'ecialiste sont donn\'ees dans~\cite{Dalibard-01,CheDal-03,CohDalLal-05}. 

\medskip

Les premi\`eres observations du ph\'enom\`ene de condensation de Bose-Einstein ont eu lieu simultan\'ement au MIT et \`a Boulder, Colorado dans les groupes de W. Ketterle d'une part et E. Cornell et C. Wieman d'autre part, ce qui leur a valu le prix Nobel de physique 2001. Les possibilit\'es ouvertes par ces exp\'eriences et celles qui ont suivi en termes d'exploration de la physique quantique macroscopique en font une des pierres angulaires de la physique contemporaine.

\medskip

On entend par ``condensat de Bose-Einstein'' un objet constitu\'e par un grand nombres de particules (habituellement des atomes alcalins) dans le m\^eme \'etat quantique. La condensation n\'ecessite donc en premier lieu que les particules en consid\'eration soient des bosons, c'est-\`a-dire ne satisfassent pas le principe de Pauli qui interdit \`a deux particules d'occuper le m\^eme \'etat quantique.

Cette occupation macroscopique d'un unique \'etat quantique de plus basse \'energie ne se manifeste qu' \`a tr\`es faible temp\'erature. Concr\`etement, il existe une temp\'erature critique $T_c$ pour l'existence d'un condensat, et l'occupation macroscopique de l'\'etat de plus basse \'energie n'appara\^it que pour des temp\'eratures $T<T_c$. L'existence th\'eorique de cette temp\'erature critique remonte aux travaux de Bose et Einstein dans les ann\'ees 20~\cite{Bose-24,Einstein-24} mais des objections de taille ont \'et\'e formul\'ees \`a l'\'epoque: 
\begin{enumerate}
\item La temp\'erature critique $T_c$ est extr\^emement basse, compl\`etement inatteignable avec les moyens des ann\'ees 20.
\item A une telle temp\'erature, l'\'etat fondamental de tous les compos\'es connus est un solide, et non pas un gaz, comme suppos\'e dans la th\'eorie de Bose-Einstein. 
\item Le raisonnement propos\'e par Bose et Einstein porte sur un gaz id\'eal, sans interactions entre les particules, ce qui est assez peu r\'ealiste.
\end{enumerate}

La premi\`ere objection n'a pu \^etre lev\'ee que dans les ann\'ees 90 avec l'apparition des puissantes techniques de refroidissement par laser\footnote{Qui ont valu le prix Nobel de Physique 1997 \`a Steven Chu, William Phillips et Claude Cohen-Tannoudji.} et par \'evaporation, qui ont permis d'atteindre des temp\'eratures de l'ordre du micro-Kelvin dans des gaz quantiques pi\'eg\'es par des dispositifs magn\'eto-optiques. Quant \`a la seconde, la solution est dans la dilution des \'echantillons en question: les rencontres de trois particules ou plus n\'ecessaires pour entamer la formation de mol\'ecules puis d'un solide sont extr\^emement rares. On pourra donc observer une phase gazeuse m\'eta-stable pendant un temps suffisant pour la formation d'un condensat.

La troisi\`eme objection est de nature plus th\'eorique que pratique. L'essentiel du mat\'eriel d\'ecrit dans ces notes s'inscrit dans un programme de recherche (impliquant de nombreux auteurs, voir les r\'ef\'erences au fil du texte) visant \`a lever cette r\'eserve possible. On  aura donc amplement l'occasion de la discuter plus en d\'etail dans la suite. 

De nombreuses observations concordantes ont confirm\'e la cr\'eation de condensats de Bose-Einstein: imagerie de la r\'epartition en vitesse/\'energie des atomes pi\'eg\'es, interf\'erences de condensats, confirmation du caract\`ere superfluide des objets cr\'e\'es ... L'importance nouvelle ainsi acquise par les mod\`eles math\'ematiques utilis\'es pour la description de ce ph\'enom\`ene a motiv\'e une vaste lit\'erature consacr\'ee \`a leur analyse.

\subsection*{Quelques questions math\'ematiques pos\'ees par les exp\'eriences}

\smallskip

En admettant l'existence de la condensation, le gaz en question peut-\^etre d\'ecrit par une seule fonction d'onde $\psi:\R ^d \mapsto \C$, correspondant \`a l'\'etat quantique commun \`a toutes les particules. Une assembl\'ee de $N$ particules quantiques est normalement d\'ecrite par une fonction d'onde \`a $N$ particules $\Psi_N: \R ^{dN} \mapsto \C$, et il faut donc comprendre pourquoi et comment on peut passer d'une fonction de $N$ variables $\Psi_N$ \`a une fonction d'une seule variable d\'ecrivant un comportement collectif. L'\'etude de  la pr\'ecision de cette approximation, qui a des cons\'equences pratiques et th\'eoriques tr\`es importantes, est une t\^ache de premi\`ere importance pour le th\'eoricien et le math\'ematicien.

On peut se poser les questions suivantes:
\begin{enumerate}
\item Peut-on d\'ecrire l'\'etat fondamental (\`a temp\'erature nulle) d'un syst\`eme de bosons en interaction par une seule fonction d'onde $\psi$ ? 
\item Partant d'une description \`a une fonction d'onde et suivant la dynamique naturelle pour un syst\`eme de $N$ particules quantiques (flot de Schr\"odinger $N$ corps), la description par une seule fonction d'onde est-elle pr\'eserv\'ee par la dynamique ?
\item Peut-on prouver rigoureusement l'existence d'une temp\'erature critique $T_c$ en dessous de laquelle la description \`a une fonction d'onde est valable pour les \'etats d'\'equilibre \`a temp\'erature $T$ du syst\`eme ?
\end{enumerate}

Ces trois questions pr\'ecisent les investigations n\'ecessaires pour lever de mani\`ere satisfaisante la troisi\`eme objection soulev\'ee au paragraphe pr\'ec\'edent. On rappelle donc que le probl\`eme est bien de comprendre le ph\'enom\`ene de condensation pour un syst\`eme quantique de particules \emph{interagissantes}. Le cas de particules id\'eales est essentiellement trivial, au moins pour les questions 1 et 2. 

Il faut entendre ici, dans l'esprit de la physique statistique, que nous voulons \'etablir la validit\'e de l'approximation par une seule fonction d'onde \emph{asymptotiquement dans la limite d'un grand nombre de particules}, modulo des hypoth\`eses appropri\'ees sur le mod\`ele en question. Id\'ealement, les hypoth\`eses devraient se r\'eduire \`a celles garantissant que le mod\`ele \`a $N$ corps de d\'epart ainsi que le mod\`ele \`a un corps d'arriv\'ee soient bien d\'efinis math\'ematiquement.  Remarquons que pour des particules quantiques en interaction, le mod\`ele de d\'epart est toujours lin\'eaire alors que le mod\`ele d'arriv\'ee est toujours non-lin\'eaire.

Beaucoup des r\'esultats r\'ecents pr\'esent\'es ci-apr\`es sont le pendant naturel dans un cadre quantique de r\'esultats de m\'ecanique classique plus anciens et mieux connus, pour lesquels des question reli\'ees se posent. Pour des raisons p\'edagogiques, quelques notions sur les limites de champ moyen pour des mod\`eles de m\'ecanique classique seront donc rappell\'ees par la suite. 

\medskip

La question 1 est l'objet de ce cours, et nous verrons qu'elle nous am\`enera \`a d\'evelopper des outils d'int\'er\^et math\'ematique intrins\`eque. On peut l'attaquer par essentiellement deux approches:
\begin{itemize}
\item La premi\`ere exploite des propri\'et\'es particuli\`eres de certains mod\`eles physiques importants. Elle s'applique ainsi au cas par cas, avec des ingr\'edients diff\'erents et sous des hypoth\`eses souvent restrictives, en particulier sur la forme des interactions. Sans pr\'etendre \`a l'exhaustivit\'e, on pourra consulter~\cite{SanSer-12,SanSer-13,RouSer-13,RouSerYng-13b,Serfaty-14} et~\cite{BenLie-83,LieYau-87,Seiringer-11,GreSei-12,SeiYngZag-12,LieSeiSolYng-05,LieSei-06} pour des applications de cette approche \`a des syst\`emes respectivement classiques et quantiques.
\item La seconde, qui est l'objet de ce cours, exploite les propri\'et\'es de l'ensemble des \'etats accessibles, c'est-\`a-dire des fonctions \`a $N$ corps $\Psi_N$ admissibles pour la description d'un syst\`eme r\'eel: le fait que l'on consid\`ere des particules bosoniques se traduit par une hypoth\`ese fondamentale de sym\'etrie sur ces fonctions. 

Cette seconde approche a le m\'erite d'\^etre beaucoup plus g\'en\'erale que la premi\`ere et de permettre dans beaucoup de cas de se rapprocher de bien plus pr\`es des ``hypoth\`eses id\'ealement minimales'' pour l'\'etude de la limite de champ moyen. 

Une interpr\'etation possible est de voir la limite de champ moyen comme un r\'egime o\`u les corr\'elations entre les particules deviennent n\'egligeables. On t\^achera d'exploiter \`a fond les cons\'equences de la notion cl\'e de sym\'etrie bosonique. Nous aurons l'occasion de discuter la litt\'erature plus en d\'etail ult\'erieurement, mais citons tout de suite~\cite{MesSpo-82,CagLioMarPul-92,Kiessling-89,Kiessling-93,KieSpo-99,RouYng-14} et~\cite{FanSpoVer-80,FanVan-06,PetRagVer-89,RagWer-89,LewNamRou-13,LewNamRou-14} pour des applications de ces id\'ees en m\'ecanique classique et quantique respectivement. 
\end{itemize}
La distinction entre les deux approches est bien s\^ur un peu artificielle: on a souvent avantage \`a utiliser des id\'ees emprunt\'ees aux deux philosophies.

\medskip

Pour garder \`a ces notes une longueur raisonnable, la question 2 ci-dessus n'y est pas du tout trait\'ee bien qu'une vaste litt\'erature math\'ematique existe, voir par exemple~\cite{Hepp-74,GinVel-79,Spohn-80,BarGolMau-00,ElgErdSchYau-06,ElgSch-07,AmmNie-08,ErdSchYau-09,FroKnoSch-09,RodSch-09,KnoPic-10,Pickl-11} et r\'ef\'erences cit\'ees, ainsi que le cours~\cite{Golse-13}. Notons que les th\'eor\`emes de type ``de Finetti quantique'' qui vont nous occuper ici se sont \'egalement r\'ecemment r\'ev\'el\'es des outils tr\`es utiles dans l'analyse de la question 2, voir~\cite{AmmNie-08,AmmNie-09,AmmNie-11,CheHaiPavSei-13,CheHaiPavSei-14}. L'usage des th\'eor\`emes de de Finetti classiques dans un cadre dynamique est plus ancienne~\cite{Spohn-80,Spohn-81,Spohn-12}.

\medskip

Quant \`a la question 3, il s'agit d'un probl\`eme ouvert fameux de physique math\'ematique sur lequel bien peu de choses sont connues \`a un niveau de rigueur satisfaisant, voir cependant~\cite{SeiUel-09,BetUel-10}. Nous n'effleurerons les questions soulev\'ees par la prise en compte de la temp\'erature dans un syst\`eme de bosons en interaction qu'\`a l'Appendice~\ref{sec:large T}, et dans un cadre grandement simplifi\'e. Ce sujet sera plus amplement d\'evelopp\'e dans l'article en pr\'eparation~\cite{LewNamRou-14b}.

\subsection*{Plan du cours}\mbox{}

\smallskip

Ces notes sont organis\'ees comme suit:
\begin{itemize}
\item Une longue introduction, Chapitre~\ref{sec:intro}, rappelle les \'el\'ements de formalisme dont nous aurons besoin pour formuler pr\'ecis\'ement les probl\`emes qui nous int\'eresseront. On commencera par le formalisme de la m\'ecanique statistique classique et poursuivra avec celui de la m\'ecanique quantique\footnote{Les deux formalismes peuvent s'unifier avec le vocabulaire des alg\`ebres $C^*$: en m\'ecanique classique les observables du syst\`eme forment une alg\`ebre $C^*$ commutative, en m\'ecanique quantique une alg\`ebre $C^*$ non-commutative. Ce point de vue unifi\'e est laiss\'e de c\^ot\'e dans ces notes. Le th\'eor\`eme de Gelfand-Naimark-Segal montre que la pr\'esentation plus g\'en\'erale des deux th\'eories se r\'eduit essentiellement \`a celle que nous avons choisi de pr\'esenter.}.

La question de la justification de l'approximation de champ moyen pour les \'etats d'\'equilibre d'un Hamiltonien donn\'e est \'egalement formul\'ee dans les deux contextes. La strat\'egie de preuve qui forme le coeur du cours est d\'ecrite de mani\`ere compl\`etement formelle, pour introduire le plus rapidement possible les th\'eor\`emes de de Finetti qui sont les outils principaux de la strat\'egie. Le parall\`ele entre le cadre classique et le cadre quantique est tr\`es fort, les diff\'erences apparaissant essentiellement lorsqu'il s'agit de d\'emontrer ces th\'eor\`emes fondamentaux dans les cadres classiques et quantiques.  
\item Le Chapitre~\ref{sec:class}, essentiellement ind\'ependant de la suite des notes, contient le traitement des syt\`emes classiques : preuve du th\'eor\`eme de de Finetti classique (\'egalement appel\'e th\'eor\`eme de Hewitt-Savage), application aux \'etats d'\'equilibre d'un Hamiltonien classique. La preuve que nous donnerons du th\'eor\`eme de Hewitt-Savage, due \`a Diaconis et Freedman est purement classique et ne se g\'en\'eralise pas au cas quantique. 
\item On attaque le probl\`eme quantique au Chapitre~\ref{sec:quant}: deux versions (forte et faible) du th\'eor\`eme de de Finetti quantique sont pr\'esent\'ees sans preuve, avec leurs applications directes \`a des exemples ``relativement simples'' de syst\`emes bosoniques dans le r\'egime de champ moyen. La Section~\ref{sec:rel deF} discute les connexions entre les diff\'erentes versions du th\'eor\`eme de de Finetti quantique et d\'ecrit la strat\'egie g\'en\'erale que nous suivrons pour leur preuve dans la suite du cours.   
\item Les Chapitres~\ref{sec:deF finite dim} et~\ref{sec:locFock} contiennent les deux principales \'etapes de la preuve du th\'eor\`eme de de Finetti quantique que nous avons choisi de pr\'esenter: respectivement ``constructions et estimations explicites en dimension finie'' et ``passage \`a la dimension infinie par localisation dans l'espace de Fock''. La preuve ne devrait pas \^etre vue comme une bo\^ite noire: non seulement le r\'esultat final mais aussi les constructions interm\'ediaires seront r\'eutilis\'es dans la suite. 
\item Arm\'es des r\'esultats des deux chapitres pr\'ec\'edents, on pourra donner au Chapitre~\ref{sec:Hartree} la justification de l'approximation de champ moyen pour le fondamental d'un syst\`eme bosonique essentiellement g\'en\'erique. Contrairement aux cas trait\'es au Chapitre~\ref{sec:quant}, le th\'eor\`eme de de Finetti en lui-m\^eme ne suffira pas pour ce cas, et il faudra r\'e-invoquer certains des ingr\'edients du Chapire~\ref{sec:locFock}. 
\item La limite de champ moyen n'est pas la seule ayant un int\'er\^et physique. On \'etudiera au Chapitre~\ref{sec:NLS} un r\'egime o\`u la port\'ee des interactions tend vers $0$ quand $N\to \infty$ (gaz dilu\'e). Dans ce cas on obtient \`a la limite des fonctionnelles d'\'energie de type Gross-Pitaevskii (ou NLS), avec non-lin\'earit\'es locales. Nous pr\'esentons une strat\'egie pour la d\'erivation de ces objets \`a partir du probl\`eme \`a $N$ corps bas\'ee sur les outils du Chapitre~\ref{sec:deF finite dim}.
\end{itemize}

Le corps du texte est compl\'et\'e par deux appendices contenant deux notes non publi\'ees dues \`a Mathieu Lewin et \`a l'auteur. 

\begin{itemize}
\item  L'Appendice~\ref{sec:class quant} montre comment, dans certains cas particuliers, on peut utiliser le th\'eor\`eme de de Finetti classique pour traiter un probl\`eme quantique. Cette strat\'egie est moins naturelle (et moins performante) que celle pr\'esent\'ee aux Chapitres~\ref{sec:quant} et~\ref{sec:Hartree} mais pr\'esente un int\'er\^et conceptuel. 

\item L'Appendice~\ref{sec:large T} d\'evie du cadre trait\'e dans le reste du cours puisque les espaces de Hilbert y seront de dimension finie. On peut obtenir dans ce contexte un th\'eor\`eme de nature semi-classique qui donne des exemples de mesures de de Finetti non rencontr\'ees pr\'ec\'edemment en consid\'erant une limite de grande temp\'erature combin\'ee avec une limite de champ moyen. Ce sera l'occasion d'\'evoquer les in\'egalit\'es de Berezin-Lieb et leur lien avec les consid\'erations du Chapitre~\ref{sec:deF finite dim}.
\end{itemize}

\subsection*{Remerciements}

La motivation pour \'ecrire ces notes est venue de l'opportunit\'e qui m'a \'et\'e donn\'ee de pr\'esenter ces sujets de mani\`ere extensive dans le cadre du cours Peccot du Coll\`ege de France\footnote{On pourra trouver les vid\'eos du cours sur le site du Coll\`ege de France}. Je tiens \`a remercier le public de ces cours pour son int\'er\^et et ses remarques constructives. Je remercie \'egalement le projet Spartacus  (voir le site http://spartacus-idh.com) et en particulier Victor Rabiet pour leur proposition de publier et diffuser ces notes.  

Ce cours doit bien entendu beaucoup \`a mes collaborateurs sur les sujets trait\'es: Mathieu Lewin et Phan Th\`anh Nam pour les aspects quantiques, Sylvia Serfaty et Jakob Yngvason pour les aspects classiques. Des \'echanges avec plusieurs coll\`egues, en particulier Zied Ammari, Patrick G\'erard, Isaac Kim, Jan-Philip Solovej, J\"urg Fr\"ohlich, Elliott Lieb et Eric Carlen, ont \'egalement \'et\'e tr\`es utiles. Le soutien financier de l'ANR (Projet Mathosaq ANR-13-JS01-0005-01) est \'egalement \`a mentionner.

\newpage

\section{\textbf{Introduction: Formalisme et Position des Probl\`emes }}\label{sec:intro}

Nous passons maintenant \`a la description des math\'ematiques qui vont nous occuper dans le reste du cours. L'objet principal de notre int\'er\^et est la m\'ecanique quantique $N$ corps, mais les analogies avec certaines questions de m\'ecanique statistique classique est assez instructive pour que nous d\'ecrivions \'egalement ce formalisme. Les questions d'unit\'es et de dimension des quantit\'es sont ignor\'ees syst\`ematiquement pour all\'eger les notations.

\subsection{Formalisme de la m\'ecanique statistique et approximation de champ moyen}\label{sec:forma class}

Pour des raisons p\'edagogiques, nous rappellerons quelques notions sur les limites de champ moyen en m\'ecanique classique avant d'aborder les aspects quantiques, li\'es \`a la condensation de Bose-Einstein. Ce paragraphe a pour but de fixer les notations et rappeller quelques concepts de base sur la m\'ecanique statistique classique. On se limitera \`a la description des \'etats d'\'equilibre d'un syst\`eme classique, les aspects dynamiques \'etant volontairement ignor\'es (on pourra \`a ce propos consulter par exemple~\cite{Golse-13}).

\subsubsection*{Espace des phases.} L'\'etat d'une particule classique est enti\`erement d\'etermin\'e par la donn\'ee de sa position $x$ et de sa vitesse $v$ (ou de mani\`ere \'equivalente son moment, $p$).  Pour une particule vivant dans un domaine $\Om\subset \R ^d$ on travaille donc dans l'espace des phases $\Om \times \R ^d$, ensemble des positions et vitesses possibles. Pour un syst\`eme comprenant $N$ particules on travaillera dans $\Om ^N \times \R ^{dN}.$

\subsubsection*{Etats purs.} On appelle \'etat pur l'\'etat d'un syst\`eme o\`u les positions et moments de toutes les particules sont connus avec certitude. Les \'etats d'\'equilibre \`a temp\'erature nulle sont par exemple des \'etats purs: en m\'ecanique statistique classique, l'incertitude sur l'\'etat d'\'equilibre d'un syst\`eme n'est d\^u qu'au ``bruit thermique''. 

Pour un syst\`eme de $N$ particules, un \'etat pur correspond \`a un point 
$$(X;P) = (x_1,\ldots,x_N;p_1,\ldots,p_N) \in \Om ^N \times \R ^{dN}$$
de l'espace des phases, o\`u le couple $(x_i;p_i)$ donne la position et le moment de la particule $i$. Dans la perspective de l'introduction des \'etats mixtes au paragraphe suivant, on identifiera un \'etat pur avec une combinaison de masses de Dirac:
\begin{equation}\label{eq:def etat pur class}
\mu_{X;P} = \sum_{\sigma \in \Sigma_N} \delta_{X_{\sigma};P_{\sigma}}. 
\end{equation}
L'\'equation pr\'ec\'edente tient compte du fait que les particules sont en r\'ealit\'e \emph{indiscernables}. On ne peut donc stricto sensu pas attribuer le couple $(x_i,p_i)$ de position/moment \`a une des $N$ particules en particulier, d'o\`u la somme dans~\eqref{eq:def etat pur class} sur le groupe des permutations de $N$ \'el\'ements $\Sigma_N$. La notation adopt\'ee est
\begin{align}\label{eq:convention permutation}
X_{\sigma} &= (x_{\sigma(1)},\ldots,x_{\sigma(N)}) \nonumber
\\P_{\sigma} &= (p_{\sigma(1)},\ldots,p_{\sigma(N)})&
\end{align}
et dire que le syst\`eme est dans l'\'etat $\mu_{X;P}$ signifie que une des particules a la position $(x_i,p_i),i=1 \ldots N$, sans qu'on puisse pr\'eciser laquelle en raison de l'indiscernabilit\'e.

\subsubsection*{Etats mixtes.} En pr\'esence d'une temp\'erature non nulle, c'est-\`a-dire d'un certain bruit thermique, on ne peut d\'eterminer avec certitude l'\'etat du syt\`eme. On cherche en fait une superposition statistique d'\'etats purs, correspondant \`a sp\'ecifier la probabilit\'e que le syst\`eme soit dans un certain \'etat pur. On parle alors d'\'etats mixtes, qui sont les combinaisons convexes d'\'etats purs vus comme des masses de Dirac comme en~\eqref{eq:def etat pur class}. L'ensemble des combinaisons convexes d'\'etats purs correspond bien s\^ur \`a l'ensemble des mesures de probabilit\'e sym\'etriques sur l'espace des phases. Un \'etat mixte g\'en\'eral pour $N$ particules est donc une mesure de probabilit\'e $\mubf_N \in \PP_s (\Om ^{N}\times \R ^{dN})$ satisfaisant 
\begin{equation}\label{eq:class sym}
d\mubf_N (X;P) = d\mubf_N (X_\sigma;P_\sigma)
\end{equation}
pour toute permutation $\sigma\in \Sigma_N$. On interpr\`ete $\mubf_N (X;P)$ comme la densit\'e de probabilit\'e que la particules $i$ ait la position $x_i$ et le moment $p_i$. Les \'etats purs~\eqref{eq:def etat pur class} sont bien s\^ur des \'etats mixtes particuliers o\`u  l'incertitude statistique se r\'eduit \`a z\'ero (\`a l'indiscernabilit\'e pr\`es).

\subsubsection*{Energie libre.} On sp\'ecifie l'\'energie d'un syst\`eme classique en se donnant un Hamiltonien, une fonction sur l'espace des phases. En m\'ecanique non relativiste, l'\'energie cin\'etique d'une particule de moment $p$ est toujours $m|p|^2 / 2$, et en prenant $m=1$ pour simplifier les notations on consid\'erera une \'energie de forme  
\begin{equation}\label{eq:intro class hamil 1}
H_N (X;P) := \sum_{j=1} ^N \frac{|p_j| ^2}{2} + \sum_{j=1} ^N V(x_j) + \lambda \sum_{1\leq i<j\leq N} w (x_i-x_j) 
\end{equation}
o\`u $V$ repr\'esente un potentiel ext\'erieur (par exemple \'electrostatique) o\`u les particules sont plong\'ees et $w$ un potentiel d'interaction de paire que l'on supposera sym\'etrique
$$ w (-x) = w (x).$$ 
Le param\`etre r\'eel $\lambda$ donne la force des interactions entre particules. On pourrait bien s\^ur ajouter des interactions \`a trois corps, quatre corps, etc ... mais il est rare que la mod\'elisation le n\'ecessite, et lorsque c'est le cas il n'y a pas de difficult\'e conceptuelle suppl\'ementaire.  

L'\'energie d'un \'etat mixte $\mubf_N \in \PP_s (\Om ^{N}\times \R ^{dN})$ est alors donn\'ee par 
\begin{equation}\label{eq:def ener func class}
\E [\mubf_N] := \int_{\Om ^N \times \R ^{dN}} H_N (X;P) d\mubf_N (X;P) 
\end{equation}
ce qui se r\'eduit par sym\'etrie du Hamiltonien \`a $H_N (X;P)$ pour un \'etat pur de la forme~\eqref{eq:def etat pur class}. A temp\'erature nulle, les \'etats d'\'equilibre du syst\`eme sont donn\'es par la minimisation de la fonctionnelle d'\'energie~\eqref{eq:def ener func class}:
\begin{equation}\label{eq:def ener class}
E (N) = \inf\left\{  \E [\mubf_N], \mubf_N \in \PP_s (\Om ^{N}\times \R ^{dN})\right\rbrace 
\end{equation}
et l'infimum (l'\'energie fondamentale) est bien s\^ur \'egal au minimum du Hamiltonien $H_N$ et atteint par un \'etat pur de la forme~\eqref{eq:def etat pur class} o\`u $(X;P)$ est un point de minimum de $H_N$ (en particulier $P= (0,\ldots,0)$).

En pr\'esence d'agitation thermique, il convient de prendre en compte l'entropie 
\begin{equation}\label{eq:def class entr}
S [\mubf_N] := - \int_{\Om ^N \times \R ^{dN}} d \mubf_N (X;P) \log (\mubf_N (X;P))  
\end{equation}
qui est une mesure du degr\'e d'incertitude sur l'\'etat du syst\`eme. Notons par exemple que les \'etats purs ont l'entropie la plus basse possible, $S [\mubf_N] = -\infty$ si $\mubf_N$ est de la forme~\eqref{eq:def etat pur class}. A temp\'erature $T$, l'\'etat du syst\`eme est donn\'e par la minimisation de la fonctionnelle d'\'energie libre 
\begin{align}\label{eq:def free ener func class}
\F [\mubf_N] &= \E [\mubf_N] - T S [\mubf_N]\nonumber 
\\&= \int_{\Om ^N \times \R ^{dN}} H_N (X;P) d\mubf_N (X;P) + T \int_{\Om ^N \times \R ^{dN}} d \mubf_N (X;P) \log (\mubf_N (X;P))
\end{align}
ce qui revient \`a dire que les \'etats les plus probables pour le syst\`eme doivent trouver un \'equilibre entre avoir une faible \'energie et une forte entropie, c'est-\`a-dire un fort d\'esordre. On notera
\begin{equation}\label{eq:def free ener class}
F (N) = \inf\left\{  \F [\mubf_N], \mubf_N \in \PP_s (\Om ^{N}\times \R ^{dN})\right\rbrace
\end{equation}
sans sp\'ecifier la d\'ependance en temp\'erature. Un minimiseur se doit alors d'\^etre une mesure de probabilit\'e relativement r\'eguli\`ere de mani\`ere \`a ce que (moins) l'entropie soit finie.

\subsubsection*{Minimisation en vitesse.} En l'absence d'une relation impos\'ee entre la distibution en position et en moment d'un \'etat classique, la minimisation des fonctionnelles pr\'ec\'edentes vis \`a vis des variables de vitesse est en fait triviale. Un \'etat minimisant~\eqref{eq:def ener func class} est toujours de la forme 
$$ \mubf_N = \delta_{P=0} \otimes \sum_{\sigma \in \Sigma_N} \delta_{X=X ^0 _\sigma}$$
o\`u $X^0$ est un point de minimum de $H_N(X;0,\ldots,0)$, i.e. les particules sont toutes immobiles. Il s'agit donc de chercher les points de minimum en $X$ de $H_N (X;0,\ldots,0)$.

La minimisation de~\eqref{eq:def free ener func class} donne lieu \`a une gaussienne en vitesse multipli\'ee par un \'etat de Gibbs en les variables de position seules\footnote{Les \emph{fonctions de partition $Z_P$ et $Z_N$ normalisent l'\'etat dans $L ^1$.}}
$$ \mubf_N = \frac{1}{Z_P} \exp \left(-\frac{1}{2T} \sum_{j=1} ^N |p_j| ^2\right) \otimes \frac{1}{Z_N} \exp \left(-\frac{1}{T} H_N (X;0,\ldots, 0)\right).$$
Les variables de vitesse n'interviennent donc plus dans la minimisation des fonctionnelles donnant les \'etats d'\'equilibre et on les ignorera totalement dans la suite de ce cours. On continuera \`a utiliser les notations ci-dessus pour la minimisation en position:
\begin{align}\label{eq:intro class hamil 2}
H_N (X) &= \sum_{j=1} ^N V(x_j) + \lambda \sum_{1\leq i<j\leq N} w (x_i-x_j) \nonumber\\
\E [\mubf_N] &= \int_{\Om ^N } H_N (X) d\mubf_N (X) \nonumber\\
\F [\mubf_N] &= \int_{\Om ^N } H_N (X) d\mubf_N (X) + T \int_{\Om ^N } d \mubf_N (X) \log (\mubf_N (X))
\end{align}
o\`u $\mubf_N \in \PP_s (\Om ^N)$ est une probabilit\'e sym\'etrique des variables de position seulement.

\subsubsection*{Marginales, densit\'es r\'eduites.} Etant donn\'e un \'etat mixte \`a $N$ particules, il est tr\`es utile de consid\'erer ses marginales, ou densit\'es r\'eduites, obtenues en int\'egrant certaines variables:
\begin{equation}\label{eq:defi marginale}
\mubf_N ^{(n)} (x_1,\ldots,x_n) = \int_{\Om ^{N-n}} \mubf(x_1,\ldots,x_n,x'_{n+1},\ldots,x'_N) dx'_{n+1}\ldots dx'_N \in \PP_s (\Om ^n).
\end{equation}
La $n-$\`eme densit\'e r\'eduite $\mubf_N ^{(n)}$ s'interpr\`ete comme la densit\'e de probabilit\'e d'avoir une particule en $x_1$, une particule en $x_2$, etc... une particule en $x_n$. Vue la sym\'etrie de $\mubf_N$, le choix des $N-n$ variables sur lesquelles on int\`egre dans la definition~\eqref{eq:defi marginale} n'est pas important.

Une premi\`ere utilisation de ces marginales consiste en une r\'e\'ecriture de l'\'energie en utilisant uniquement les deux premi\`eres marginales\footnote{Plus g\'en\'eralement, une \'energie avec un potentiel \`a $n$ corps peut se r\'e\'ecrire en utilisant la $n$-i\`eme  densit\'e r\'eduite.}:
\begin{align}\label{eq:intro ener marginales}
\E [\mubf_N] &= N \int_{\Om} V(x) d\mubf_N ^{(1)} (x) + \lambda \frac{N(N-1)}{2} \iint_{\Om \times \Om} w(x-y) d\mubf_N ^{(2)}(x,y)\nonumber\\
&= \iint_{\Om\times\Om} \left( \frac{N}{2} V (x) + \frac{N}{2} V (y) + \lambda \frac{N(N-1)}{2} w(x-y) \right)d\mubf_N ^{(2)}(x,y),
\end{align}
o\`u on a utilis\'e la sym\'etrie du Hamiltonien.

\subsubsection*{Approximation de champ moyen.} R\'esoudre les probl\`emes de minimisation ci-dessus a en g\'en\'eral un co\^ut prohibitif lorsque le nombre de particules devient grand. Pour obtenir des th\'eories plus simples dont on puisse extraire plus facilement de l'information on a souvent recours \`a des approximations. La plus simple et la plus connue est l'approximation de champ moyen. On peut l'introduire de plusieures fa\c{c}ons, le but \'etant d'obtenir un probl\`eme \`a un corps auto-consistant en partant du probl\`eme \`a $N$ corps pr\'esent\'e ci-dessus. 

Nous prenons ici le point de vue ``chaos mol\'eculaire'' sur la th\'eorie de champ moyen: l'approximation consiste \`a supposer que les particules sont ind\'ependantes et identiquement distribu\'ees (iid). On prend donc un ansatz de la forme 
\begin{equation}\label{eq:intro MF ansatz}
\mubf_N (x_1,\ldots,x_N) = \rho ^{\otimes N} (x_1,\ldots,x_N) = \prod_{j=1} ^N \rho(x_j) 
\end{equation}
o\`u $\rho \in \PP (\Om)$ est une mesure de probabilit\'e \`a un corps d\'ecrivant le comportement typique d'une des particules iid que l'on consid\`ere. 

Les fonctionnelles d'\'energie et d'\'energie libre de champ moyen s'obtiennent en ins\'erant cet ansatz dans~\eqref{eq:def ener func class} ou~\eqref{eq:def free ener func class}. La fonctionnelle d'\'energie de champ moyen est donc\footnote{MF pour ``Mean-Field''.} 
\begin{equation}\label{eq:intro MF ener func class}
\MFEf [\rho] = N ^{-1} \E [\rho ^{\otimes N}] = \int_{\Om} V(x)d\rho(x) + \lambda \frac{N-1}{2} \iint_{\Om\times \Om} w(x-y) d\rho(x)d\rho(y) 
\end{equation}
dont on notera $\MFEe$ l'infimum parmi les mesures de probabilit\'e. La fonctionnelle d'\'energie libre de champ moyen s'obtient similairement:
\begin{align}\label{eq:intro MF free ener class}
\MFf [\rho] &= N ^{-1} \F [\rho ^{\otimes N}] \nonumber\\
&= \int_{\Om} V(x)d\rho(x) + \lambda \frac{N-1}{2} \iint_{\Om\times \Om} w(x-y) d\rho(x)d\rho(y) + T \int_{\Om} \rho \log \rho 
\end{align}
et on notera $\MFe$ son infimum. Le terme ``champ moyen'' provient du fait que le deuxi\`eme terme de~\eqref{eq:intro MF ener func class} correspond \`a une interaction entre la densit\'e de particules $\rho$ et le potentiel auto-consistant (dont d\'erive le champ moyen) 
$$\rho \ast w = \int_{\Om} w(.-y) d\rho(y).$$

\subsection{Formalisme de la m\'ecanique quantique et condensation de Bose-Einstein}\label{sec:forma quant}

Apr\`es ces rappels de m\'ecanique classique, nous pouvons maintenant introduire les objets quantiques qui sont l'objet principal de ce cours. Nous nous contenterons d'un survol des principes de base de la physique quantique. D'autres pr\'esentations ``mathematician-friendly'' peuvent \^etre trouv\'ees dans~\cite{LieSei-09,Solovej-notes} Notre introduction des concepts n\'ecessaires \`a la suite est par endroits volontairement simplifi\'ee.

\subsubsection*{Fonctions d'onde et \'en\'ergie cin\'etique quantique.} Un des postulats de base de la m\'ecanique quantique est l'identification des \'etats purs d'un syst\`eme avec les vecteurs normalis\'es d'un espace de Hilbert complexe $\gH$. Pour des particules vivant dans l'espace de configuration $\R ^d$, l'espace de Hilbert appropri\'e pour une particules est $L^2 (\R ^d)$, l'espace des fonctions complexes de carr\'e int\'egrable sur $\R ^d$. 

Etant donn\'ee une particule dans l'\'etat $\psi \in L ^2 (\R ^d)$, on identifie $|\psi| ^2$ avec la densit\'e de probabilit\'e de pr\'esence: $|\psi (x)| ^2$ repr\'esente la probabilit\'e que la particule soit en $x$. On impose donc la normalisation
$$ \int_{\R ^d} |\psi| ^2 =1. $$
On voit que m\^eme dans le cas d'un \'etat pur, on ne peut pas sp\'ecifier avec certitude la position de la particule. Plus pr\'ecis\'ement \emph{on ne peut pas sp\'ecifier simultan\'ement sa position et sa vitesse.} Ce \emph{principe d'incertitude} est une cons\'equence directe d'un autre postulat fondamental: $|\hat{\psi}| ^2$ donne la densit\'e de probabilit\'e en vitesse de la particule, o\`u $\hat{\psi}$ est la transform\'ee de Fourier de $\psi$. En m\'ecanique quantique l'\'energie cin\'etique d'une particule est donc donn\'ee par 
\begin{equation}\label{eq:ener cin quant}
\int_{\R ^d } \frac{|p| ^2}{2} |\hat{\psi} (p)| ^2 dp = \int_{\R ^d } \frac{1}{2} |\nabla \psi (x)| ^2 dx.  
\end{equation}
Le fait que la position et la vitesse d'une particule ne soient pas sp\'ecifiables simultan\'ement vient du fait qu'il est impossible qu'\`a la fois $|\psi| ^2$ et $|\hat{\psi}| ^2$ convergent vers une masse de Dirac. Une fa\c{c}on populaire de quantifier ce fait est le principe d'incertitude de Heisenberg: pour tout $x_0\in \R ^d$ 
$$ \left( \int_{\R ^d }|\nabla \psi (x)| ^2 dx \right)\left( \int_{\R ^d } |x-x_0| ^2 |\psi (x)| ^2 dx \right) \geq C.$$
En effet, plus la position de la particule est connue avec pr\'ecision plus le second terme du membre de gauche est petit (pour un certain $x_0$). Le premier terme du membre de gauche doit alors \^etre tr\`es grand, ce qui au vu de~\eqref{eq:ener cin quant} est incompatible avec le fait que la densit\'e de vitesse soit concentr\'ee autour d'un certain $p_0 \in \R ^d$.

Pour de nombreuses applications (voir~\cite{Lieb-76} pour une discussion de ce point), cette in\'egalit\'e est en fait insuffisante, et une meilleure mani\`ere de quantifier le principe d'incertitude est donn\'ee par l'in\'egalit\'e de Sobolev (ici dans sa version 3D):
$$\int_{\R ^3 } |\nabla \psi | ^2  \geq C \left( \int_{\R ^3} |\psi| ^6  \right) ^{1/3}.$$ 
Si la position de la particule est connue avec pr\'ecision, $|\psi| ^2$ doit approcher une masse de Dirac, auquel cas le membre de droite explose, et donc \'egalement les int\'egrales~\eqref{eq:ener cin quant}, avec la m\^eme interpr\'etation que pr\'ec\'edemment.

\subsubsection*{Bosons et Fermions.} Pour un syst\`eme de $N$ particules quantiques dans $\R ^d$, l'espace de Hilbert appropri\'e est $L ^2 (\R ^{dN}) \simeq \bigotimes ^N L ^2 (\R ^d)$. Un \'etat pur du syst\`eme est donc un certain $\Psi \in L ^2 (\R ^{dN})$ et on interpr\`ete $|\Psi (x_1,\ldots,x_N)| ^2$ comme la densit\'e de probabilit\'e pour que la particule $1$ soit en $x_1$, ..., la particule $N$ en $x_N$. Comme en m\'ecanique classique, l'indiscernabilit\'e des particules impose que 
\begin{equation}\label{eq:indis quantique moins}
|\Psi (X)| ^2 = |\Psi (X_{\sigma})| ^2 
\end{equation}
pour toute permutation $\sigma \in \Sigma_N$. Cette condition est n\'ecessaire pour d\'ecrire des particules indiscernables, mais elle n'est en fait pas suffisante, contrairement au cas de la m\'ecanique classique. Pour \'ecrire la bonne condition, on introduit l'op\'erateur unitaire $U_{\sigma}$ qui intervertit les variables suivant $\sigma \in \Sigma_N$:
$$ U_{\sigma} \: u_1 \otimes \ldots \otimes u_N = u_{\sigma (1)}\otimes \ldots \otimes u_{\sigma (N)}$$
pour tout $u_1,\ldots,u_N \in L ^2 (\R ^d)$, \'etendu par lin\'earit\'e \`a $L ^2 (\R ^{dN}) \simeq \bigotimes ^N L ^2 (\R  ^d)$ dont on peut construire une base en utilisant des vecteurs de forme $u_1 \otimes \ldots \otimes u_N$. Pour que $\Psi \in L ^2 (\R ^{dN})$ d\'ecrive des particules indiscernables il faut demander 
\begin{equation}\label{eq:indis quantique}
\left\langle \Psi, A \Psi \right\rangle_{L ^2 (\R ^{dN})} = \left\langle U_{\sigma} \Psi, A U_{\sigma} \Psi \right\rangle_{L ^2 (\R ^{dN})} 
\end{equation}
pour tout op\'erateur born\'e $A$ agissant sur $L ^2 (\R ^{dN})$. Sans rentrer dans des d\'etails qui nous emm\`eneraient trop loin, la condition~\eqref{eq:indis quantique} correspond \`a demander que toute mesure (correspondant \`a une observable $A$) effectu\'ee sur le syst\`eme soit ind\'ependante de la num\'erotation des particules. En m\'ecanique classique, les mesures possibles correspondent toutes \`a des fonctions born\'ees sur l'espace des phases, et donc~\eqref{eq:class sym} garantit l'ind\'ependance des observations vis \`a vis de la permutation des particules. En m\'ecanique quantique, les observables sont les op\'erateurs born\'es sur l'espace de Hilbert ambiant, et donc il faut imposer la condition plus forte~\eqref{eq:indis quantique}.

Une cons\'equence importante\footnote{Il s'agit d'un exercice simple mais non trivial.} de la condition de sym\'etrie~\eqref{eq:indis quantique} est qu'un syst\`eme de particules quantiques indiscernables doit satisfaire une des deux conditions suivantes, plus fortes que~\eqref{eq:indis quantique moins}: soit 
\begin{equation}\label{eq:intro bosons}
\Psi (X) = \Psi(X_\sigma) 
\end{equation}
pour tout $X\in \R ^{dN}$ et $\sigma \in \Sigma_N$, soit 
\begin{equation}\label{eq:intro fermions}
\Psi (X) = \epsilon (\sigma) \Psi(X_\sigma) 
\end{equation}
pour tout $X\in \R ^{dN}$ et $\sigma \in \Sigma_N$, o\`u $\eps (\sigma)$ est la signature de la permutation $\sigma$. On parle de bosons pour des particules d\'ecrites par des fonctions d'onde satisfaisant~\eqref{eq:intro bosons} et de fermions pour des particules d\'ecrites par une fonction d'onde satisfaisant~\eqref{eq:intro fermions}. Ces deux types de particules fondamentales ont un comportement tr\`es diff\'erent, on parle de statistique bosonique ou fermionique~\cite{Solovej-notes}. Par exemple, les fermions ob\'eissent au \emph{principe d'exclusion de Pauli} qui stipule que deux fermions ne peuvent occuper simultan\'ement le m\^eme \'etat quantique. On peut d\'ej\`a voir que~\eqref{eq:intro fermions} impose 
$$ \Psi (x_1,\ldots,x_i,\ldots,x_j,\ldots,x_N) =  -\Psi (x_1,\ldots,x_j,\ldots,x_i,\ldots,x_N)$$
et donc (formellement)
$$ \Psi (x_1,\ldots,x_i,\ldots,x_i,\ldots,x_N) = 0$$ 
ce qui implique qu'il est impossible que deux fermions occupent simultan\'ement la position~$x_i$. On pourra consulter~\cite{LieSei-09} pour un expos\'e des in\'egalit\'es de Lieb-Thirring qui sont une des cons\'equences les plus importantes du principe de Pauli.

Concr\`etement, lorsqu'on \'etudie un syst\`eme quantique, il convient de restreindre les \'etats purs accessibles \`a ceux de type bosonique, ou ceux de type fermionique. On travaillera donc 
\begin{itemize}
\item pour des bosons, dans $L^2_s (\R ^{dN}) \simeq \bigotimes_s ^N L ^2 (\R ^d)$, l'espace des fonctions sym\'etriques de carr\'e int\'egrable, identifi\'e avec le produit tensoriel sym\'etrique de $N$ copies de $L ^2 (\R ^d)$. 
\item pour des fermions, dans $L^2_{as} (\R ^{dN}) \simeq \bigotimes_{as} ^N L ^2 (\R ^d)$, l'espace des fonctions antisym\'etriques de carr\'e int\'egrable, identifi\'e avec le produit tensoriel antisym\'etrique de $N$ copies de $L ^2 (\R ^d)$.
\end{itemize}
Comme son nom l'indique, la condensation de Bose-Einstein ne peut se produire que dans un syst\`eme bosonique, et ce cours sera donc focalis\'e sur le premier cas.

\subsubsection*{Matrices de densit\'e, Etats mixtes.} On identifiera toujours un \'etat pur $\Psi \in L ^2 (\R ^{dN})$ avec la matrice de densit\'e correspondante, c'est-\`a-dire le projecteur orthogonal sur $\Psi$, not\'e $\ketl \Psi \ketr \bral \Psi \brar$. De m\^eme qu'en m\'ecanique classique, les \'etats mixtes du syst\`eme sont par d\'efinition les superpositions statistiques d'\'etats purs, c'est-\`a-dire ici les comibnaisons convexes de projecteurs orthogonaux. En utilisant le th\'eor\`eme spectral, il est clair que l'ensemble des \'etats mixtes correspond \`a l'ensemble des op\'erateurs positifs de trace $1$:
\begin{equation}\label{eq:intro etats mixtes}
\cS (L ^2 (\R ^{dN})) = \left\lbrace \Gamma \in \gS ^1 (L ^2 (\R ^{dN})), \Gamma = \Gamma ^*, \Gamma \geq 0, \tr \Gamma = 1 \right\rbrace
\end{equation}
o\`u $\gS ^1 (\gH)$ est la classe de Schatten~\cite{ReeSim4,Simon-79} des op\'erateurs \`a trace sur un espace de Hilbert $\gH$.  Pour obtenir les \'etats mixtes bosoniques (respectivement fermioniques) on consid\`ere respectivement
$$ \cS (L _{s/as} ^2 (\R ^{dN})) = \left\lbrace \Gamma \in \gS ^1 (L_{s/as} ^2 (\R ^{dN})), \Gamma = \Gamma ^*, \Gamma \geq 0, \tr \Gamma = 1 \right\rbrace.$$
Notons que dans le langage des matrices densit\'e, la sym\'etrie bosonique correspond \`a demander 
\begin{equation}\label{eq:intro bos sym mat} 
U_{\sigma} \Gamma = \Gamma 
\end{equation}
alors que la sym\'etrie fermionique correspond \`a 
$$ U_{\sigma} \Gamma = \eps (\sigma) \Gamma.$$
On consid\`ere parfois la notion de sym\'etrie plus faible
\begin{equation}\label{eq:quant sym gen}
 U_{\sigma} \Gamma U_{\sigma} ^* = \Gamma 
\end{equation}
qui est satisfaite par exemple par la superposition d'un \'etat bosonique et d'un \'etat fermionique.

\subsubsection*{Fonctionnelles d'\'energie.} La fonctionnelle d'\'energie quantique correspondant au Hamiltonien classique non relativiste~\eqref{eq:intro class hamil 1} s'obtient en op\'erant la substitution 
\begin{equation}\label{eq:quant mom}
p \leftrightarrow -i\nabla, 
\end{equation}
inspir\'ee de l'identification~\eqref{eq:ener cin quant} pour l'\'energie cin\'etique d'un \'etat quantique. Le Hamiltonien quantifi\'e devient maintenant un op\'erateur sur $L ^2 (\R ^{dN})$:
\begin{equation}\label{eq:intro quant hamil}
H_N = \sum_{j=1} ^N \left( - \frac12 \Delta_j + V (x_j) \right) + \lambda \sum_{1\leq i<j \leq N} w(x_i - x_j)  
\end{equation}
o\`u $-\Delta_j = (-i \nabla_j) ^2$ correspond au Laplacien dans la variable $x_j \in \R ^d$. L'\'energie correspondante est pour un \'etat pur $\Psi \in L ^2 (\R ^{dN})$
\begin{equation}\label{eq:intro quant ener pur}
\E [\Psi] = \langle \Psi, H_N \Psi \rangle_{L ^2 (\R ^{dN})}  
\end{equation}
ce qui se g\'en\'eralise par lin\'earit\'e au cas d'un \'etat mixte $\Gamma \in \cS (L ^2 (\R ^{dN}))$ en 
\begin{equation}\label{eq:intro quant ener mixte}
\E [\Gamma] = \tr_{L ^2 (\R ^{dN})} [ H_N \Gamma ].   
\end{equation}
A temp\'erature nulle, l'\'etat d'\'equilibre du syst\`eme s'obtient en minimisant la fonctionnelle d'\'energie ci-dessus. Vus la lin\'earit\'e de~\eqref{eq:intro quant ener mixte} en fonction de $\Gamma$ et le th\'eor\`eme spectral, il est clair qu'on peut restreindre la minimisation aux \'etats purs:
\begin{align}\label{eq:def ener quant}
E_{s/as}(N) = \inf\left\{ \E [\Gamma], \Gamma \in \cS (L ^2 _{s/as} (\R ^{dN})) \right\}\nonumber\\
&= \inf\left\{ \E [\Psi], \Psi \in L_{s/as} ^2 (\R ^{dN}), \norm{\Psi}_{L ^2 (\R ^{dN})} = 1 \right\}. 
\end{align}
Ici on utilise \`a nouveau l'indice $s$ (respectivement $as$) pour d\'esigner l'\'energie bosonique (respectivement fermionique). En l'absence d'indice il est entendu que la minimisation s'effectue sans contrainte de sym\'etrie. Vue la sym\'etrie du Hamiltonien, on pourra toutefois minimiser parmi les \'etats mixtes satisfaisant~\eqref{eq:quant sym gen}, ou parmi les fonctions d'onde satisfaisant~\eqref{eq:indis quantique moins}. 

En pr\'esence d'agitation thermique, il convient de rajouter \`a l'\'energie un terme incluant l'\'entropie de Von-Neumann:
\begin{equation}\label{eq:Von Neu entropie}
S[\Gamma] = - \tr_{L ^2 (\R ^{dN})} [\Gamma \log \Gamma] = - \sum_{j} a_j \log a_j 
\end{equation}
o\`u les $a_j$ sont les valeurs propres (r\'eelles positives) de $\Gamma$ dont l'existence est garantie par le th\'eor\`eme spectral. Similairement \`a l'entropie classique~\eqref{eq:def class entr}, l'entropie de Von-Neumann est minimis\'ee ($S[\Gamma] = 0$) par les \'etats purs, i.e. les projecteurs, qui n'ont bien s\^ur qu'une valeur propre non nulle, \'egale \`a $1$.

La fonctionnelle d'\'energie libre quantique \`a temp\'erature $T$ est alors 
\begin{equation}\label{eq:def free ener func quant}
\F [\Gamma] = \E [\Gamma] - T S [\Gamma] 
\end{equation}
et un minimiseur est en g\'en\'eral un \'etat mixte.

\subsubsection*{Autres formes pour l'\'energie cin\'etique: champs magn\'etiques et effets relativistes.} Nous n'avons introduit dans le cadre classique que la forme la plus simple possible pour l'\'energie cin\'etique d'une particule. La raison en est que pour les probl\`emes de minimisation introduits, l'\'energie cin\'etique ne joue aucun r\^ole. D'autres choix de relation entre $p$ et l'\'energie cin\'etique\footnote{D'autres \emph{relations de dispersion.}} sont toutefois possibles et ce choix s'av\`ere d\'eterminant en m\'ecanique quantique o\`u la minimisation en vitesse est non triviale. Via la relation~\eqref{eq:quant mom} on voit que dans ce contexte, les diff\'erents choix possibles changeront les espaces fonctionnels \`a utiliser. 

Outre l'\'energie cin\'etique non relativiste introduite ci-dessus, au moins deux g\'en\'eralisations sont physiquement int\'eressantes:
\begin{itemize}
\item En pr\'esence d'un champ magn\'etique $B:\R ^d \mapsto \R $ interagissant avec des particules charg\'ees, on remplace 
\begin{equation}\label{eq:quant mom magn}
p \leftrightarrow -i\nabla + A 
\end{equation}
o\`u $A:\R ^d \mapsto \R ^d$ est le \emph{potentiel vecteur}\footnote{$B$ ne d\'etermine $A$ qu'\`a un gradient pr\`es. Le choix de $A$ s'appelle \emph{choix de jauge}.}, satisfaisant 
$$ B = \mathrm{curl}\: A.$$
L'op\'erateur d'\'energie cin\'etique, prenant en compte la force de Lorentz, devient dans ce cas 
$$(p+A) ^2 = \left( -i\nabla + A \right) ^2 = -\left( \nabla - i A\right) ^2$$
appell\'e \emph{Laplacien magn\'etique}. 

Ce formalisme est \'egalement adapt\'e au cas de particules dans un rep\`ere en rotation: en nommant $x_3$ l'axe de rotation il faut alors prendre $A = \Om (-x_2,x_1,0)$ avec $\Om$ la vitesse de rotation, ce qui correspond \`a la prise en compte de la force de Coriolis. Dans ce cas, il faut aussi remplacer le potentiel $V(x)$ par $V(x) - \Om ^2 |x| ^2$ pour prendre en compte la force centrifuge. 

\smallskip

\item Lorsqu'on d\'esire prendre en compte des effets relativistes, la relation de dispersion  devient 
$$ \mbox{ Energie cin\'etique}  = c\sqrt{p ^2 + m ^2 c^2}  - m c ^2$$
avec $m$ la masse et $c$ la vitesse de la lumi\`ere dans le vide. En choisissant les unit\'es pour que $c=1$ et en rappelant~\eqref{eq:quant mom} on est donc amen\'e \`a consid\'erer l'op\'erateur d'\'energie cin\'etique 
\begin{equation}\label{eq:rel kin ener}
\sqrt{p ^2 + m ^2} - m = \sqrt{-\Delta + m ^2} - m 
\end{equation}
qui se d\'efinit ais\'ement en variables de Fourier par exemple. Dans la limite non relativiste o\`u $|p|\ll m = mc $ on retrouve formellement l'op\'erateur $-\Delta$ au premier ordre. Une certaine caricature de cet op\'erateur est parfois utilis\'ee, correspondant au cas ``relativiste extr\^eme'' $|p|\gg mc$ o\`u on prend comme op\'erateur d'\'energie cin\'etique
\begin{equation}\label{eq:rel kin ener ext}
\sqrt{p ^2} = \sqrt{-\Delta}.
\end{equation}
\item On peut bien s\^ur combiner les deux g\'en\'eralisations pour consid\'erer le cas de particules relativistes dans un champ magn\'etique en utilisant les op\'erateurs 
$$\sqrt{(-i\nabla + A) ^2 + m ^2 } - m \mbox{ et } \left|-i\nabla +A\right|$$
bas\'es sur la relation de dispersion relativiste et la correspondance~\eqref{eq:quant mom magn}.
\end{itemize}

\subsubsection*{Matrices de densit\'e r\'eduites.} De m\^eme qu'on a introduit pr\'ec\'edemment les densit\'es r\'eduites d'un \'etat classique, il sera tr\`es utile de disposer du concept correspondant en m\'ecanique quantique. Etant donn\'e un \'etat mixte $\Gamma \in \cS (L ^2 (\R ^{dN}))$ on d\'efinit sa $n$-i\`eme matrice de densit\'e r\'eduite en prenant une trace partielle sur les $N-n$ derni\`eres particules
\begin{equation}\label{eq:intro def mat red}
\Gamma ^{(n)}= \tr_{n+1 \to N} \Gamma, 
\end{equation}
ce qui veut pr\'ecis\'ement dire que pour tout op\'erateur born\'e $A_n$ sur $L ^2 (\R ^{dn})$ 
$$ \tr_{L ^2 (\R ^{dn})} [A_n \Gamma ^{(n)}] := \tr_{L ^2 (\R ^{dN})} [A_n \otimes \one^{\otimes N-n} \Gamma] $$
avec $\one$ l'identit\'e sur $L ^2 (\R ^d)$. La d\'efinition ci-dessus se g\'en\'eralise facilement au cas d'un espace de Hilbert diff\'erent de $L^2$, mais le lecteur est peut-\^etre plus familier avec la d\'efinition suivante, \'equivalente dans le cas o\`u on travaille sur $L ^2$. On identifie $\Gamma$ avec son noyau, c'est-\`a-dire la fonction $\Gamma (X;Y)$ telle que pour tout $\Psi \in L ^2 (\R ^{dN})$
$$ \Gamma\, \Psi = \int_{\R ^{dN}} \Gamma (X;Y) \Psi (Y) dY$$ 
et on peut alors \'egalement identifier $\Gamma ^{(n)}$ avec son noyau
\begin{multline*}
 \Gamma ^{(n)} (x_1,\ldots,x_n;y_1,\ldots,y_n) \\= \int_{\R ^{d(N-n)}} \Gamma (x_1,\ldots,x_n,z_{n+1},\ldots,z_{N};y_1,\ldots,y_n,z_{n+1},\ldots,z_{N})dz_{n+1}\ldots dz_N. 
\end{multline*}
Dans le cas o\`u l'\'etat $\Gamma$ a une sym\'etrie, par exemple bosonique ou fermionique, le choix des variables sur lesquelles on prend la trace partielle est arbitraire. Notons que m\^eme si on part d'un \'etat pur, les matrices de densit\'e r\'eduites sont en g\'en\'eral des \'etats mixtes. 

Similairement \`a~\eqref{eq:intro ener marginales} on peut r\'e\'ecrire l'\'energie~\eqref{eq:intro quant ener mixte} sous la forme 
\begin{align}\label{eq:quant ener red mat}
\E [\Gamma] &= N \tr_{L ^2 (\R ^d)} \left[\left( - \half \Delta + V \right) \Gamma ^{(1)} \right] + \lambda \frac{N(N-1)}{2} \tr_{L ^2 (\R ^{2d})} [w (x-y) \Gamma ^{(2)}]\nonumber \\
&= \tr_{L ^2 (\R ^{2d})} \left[ \left(\frac{N}{2} \left( - \half \Delta_x + V (x) - \half \Delta_y + V (y) \right) + \lambda \frac{N(N-1)}{2} w (x-y) \right) \Gamma ^{(2)}\right].  
\end{align}

\subsubsection*{Approximation de champ moyen.} Encore plus qu'en m\'ecanique classique, r\'esoudre en pratique le probl\`eme de minimisation~\eqref{eq:def ener quant} pour $N$ grand est bien trop co\^uteux, et il est souvent n\'ecessaire d'avoir recours \`a des approximations. Ce cours a pour but d'\'etudier la plus simple de ces approximations, qui consiste \`a imiter~\eqref{eq:intro MF ansatz} en prenant un ansatz d\'ecrivant des particules iid
\begin{equation}\label{eq:intro quant ansatz}
\Psi (x_1,\ldots,x_N) = u ^{\otimes N} (x_1,\ldots, x_N) = u(x_1)\ldots u(x_N) 
\end{equation}
pour un certain $u\in L ^2 (\R ^d).$ En ins\'erant cette forme dans la fonctionnelle d'\'energie~\eqref{eq:quant ener red mat} on obtient la fonctionnelle de Hartree 
\begin{align}\label{eq:intro def Hartree func}
\EH [u] &= N ^{-1} \E[ u ^{\otimes N}] \nonumber \\
&= \int_{\R ^d } \left( \frac{1}{2}|\nabla u| ^2 + V |u| ^2 \right) + \lambda\frac{N-1}{2} \iint_{\R ^d \times \R ^d} |u(x)| ^2 w(x-y) |u(y)| ^2 
\end{align}
et le probl\`eme de minimisation correspondant est 
\begin{equation}\label{eq:intro def Hartree ener}
\eH = \inf\left\{ \EH [u], \norm{u}_{L ^2 (\R ^d)} = 1 \right\}. 
\end{equation}
On notera que l'on passe ainsi d'un probl\`eme lin\'eaire en la fonction d'onde \`a $N$ corps (puisque la fonctionnelle d'\'energie \`a $N$ corps en quadratique en la fonction d'onde, l'\'equation variationnelle correspondante est lin\'eaire) \`a un probl\`eme cubique en la fonction d'onde \`a un corps $u$ (puisque la fonctionnelle de Hartree est quartique en la fonction d'onde, l'\'equation variationnelle correspondante est cubique). 

L'ansatz~\eqref{eq:intro quant ansatz} est une fonction d'onde sym\'etrique, qui convient pour des bosons. Elle correspond \`a chercher le fondamental sous la forme d'un condensat de Bose-Einstein o\`u toutes les particules sont dans l'\'etat $u$. A cause du principe de Pauli, des fermions ne peuvent en fait jamais \^etre compl\`etement d\'ecorr\'el\'es au sens de la forme~\eqref{eq:intro quant ansatz}. L'ansatz de champ moyen pour des fermions consiste \`a prendre 
$$ \Psi (x_1,\ldots,x_N) = \det \left(u_j (x_k)\right)_{1\leq j,k\leq N}$$
avec $u_1,\ldots,u_N$ des fonctions orthonormales (orbitales du syt\`eme) dans $L^2 (\R^{d})$. Cet ansatz m\`ene \`a la fonctionnelle de Hartree-Fock dont on ne parlera pas dans ces notes (une pr\'esentation dans le m\^eme esprit et des r\'ef\'erences se trouvent dans~\cite{Rougerie-XEDP}). 

Remarquons que dans l'approximation de champ moyen~\eqref{eq:intro MF ansatz} pour des particules classiques, on autorise un \'etat mixte $\rho$. Pour d\'ecrire un syst\`eme bosonique, on prend toujours un \'etat pur $u$ dans l'ansatz~\eqref{eq:intro quant ansatz}, ce qui appelle les commentaires suivants:
\begin{itemize}
\item Si on prend un \'etat g\'en\'eral $\gamma \in \cS (L ^2 (\R ^d))$, l'\'etat \`a $N$ corps d\'efini comme 
\begin{equation}\label{eq:ansatz MF sans sym}
 \Gamma = \gamma ^{\otimes N} 
\end{equation}
a la sym\'etrie bosonique~\eqref{eq:intro bos sym mat} si et seulement si $\gamma$ est un \'etat pur (voir par exemple~\cite{HudMoo-75}), $\gamma = | u \rangle \langle u |$, auquel cas $\Gamma = | u ^{\otimes N}\rangle \langle u  ^{\otimes N}|$ et on revient \`a l'ansatz~\eqref{eq:intro quant ansatz}.

\item Dans le cas de la fonctionnelle d'\'energie~\eqref{eq:quant ener red mat}, le probl\`eme de minimisation avec sym\'etrie bosonique impos\'ee et le probl\`eme sans sym\'etrie co\"incident (cf par exemple~\cite[Chapitre 3]{LieSei-09}). On peut donc minimiser sans contrainte, et on retombera sur l'\'energie bosonique. Ceci reste vrai dans le cas o\`u l'\'energie cin\'etique est~\eqref{eq:rel kin ener} ou~\eqref{eq:rel kin ener ext} mais est notoirement faux en pr\'esence d'un champ magn\'etique. 

\item Pour certaines fonctionnelles d'\'energie, par exemple en pr\'esence d'un champ magn\'etique ou de rotation, le minimum sans sym\'etrie est strictement inf\'erieur au minimiseur avec sym\'etrie, cf~\cite{Seiringer-03}. L'ansatz~\eqref{eq:ansatz MF sans sym} est alors appropri\'e pour approximer le minimum du probl\`eme sans sym\'etrie (au vu de la sym\'etrie du Hamiltonien, on peut toujours supposer la sym\'etrie~\eqref{eq:quant sym gen}) et on obtient alors une fonctionnelle de Hartree g\'en\'eralis\'ee pour des \'etats mixtes \`a un corps $\gamma \in \cS (L ^2 (\R ^d))$
\begin{equation}\label{eq:intro Hartree func gen}
\EH [\gamma] = \tr _{L ^2 (\R^d )} \left[\left( -\half \Delta + V \right)\gamma \right] + \lambda \frac{N-1}{2} \tr_{L ^2 (\R ^{2d})} \left[ w(x-y) \gamma ^{\otimes 2}\right].
\end{equation}
\end{itemize}

\vspace{0.5cm}

Les deux sections suivantes introduisent, respectivement dans le cas classique et le cas quantique, la question qui est l'objet principal de ces notes:
\medskip
\begin{center}
\emph{Peut-on justifier dans une certaine limite la validit\'e des ansatz de champ moyen~\eqref{eq:intro MF ansatz} et~\eqref{eq:intro quant ansatz} pour la d\'etermination des \'etats d'\'equilibre d'un syt\`eme de particules indiscernables ?} 
\end{center}

\subsection{Champ moyen et th\'eor\`eme de de Finetti classique}\label{sec:MF deF class}\mbox{}\\\vspace{-0.4cm}

Est-il l\'egitime de proc\'eder \`a la simplification~\eqref{eq:intro MF ansatz} pour d\'eterminer les \'etats d'\'equilibre d'un syst\`eme classique ? L'exp\'erience montre que c'est le cas lorsque les particules sont en grand nombre, ce qu'on mat\'erialise math\'ematiquement par l'\'etude de la limite $N\to \infty$ des probl\`emes qui nous occupent. 

Un cadre simple o\`u on peut justifier la validit\'e de l'\emph{approximation} de champ moyen est la \emph{limite} de champ moyen o\`u l'on suppose que tous les termes de l'\'energie~\eqref{eq:def ener func class} p\`esent le m\^eme poids. Au vu de~\eqref{eq:intro ener marginales}, cette condition est remplie si $\lambda$ est d'ordre $N ^{-1}$, par exemple 
\begin{equation}\label{eq:intro lambda}
\lambda = \frac{1}{N-1}, 
\end{equation}
auquel cas on peut s'attendre \`a ce que l'\'energie fondamentale par particule $N ^{-1} E(N)$ ait une limite bien d\'efinie. Le choix particulier~\eqref{eq:intro lambda} sert \`a simplifier certaines expressions, mais les consid\'erations suivantes s'appliquent tant que $\lambda$ est d'ordre $N ^{-1}$. Insistons sur le caract\`ere simplifi\'e de l'\'etude de la limite $N\to \infty$ sous l'hypoth\`ese~\eqref{eq:intro lambda}, qui est loin de recouvrir tous les cas physiquement int\'eressants. Il s'agit cependant d'un probl\`eme d\'ej\`a non trivial et tr\`es instructif. 

Le but de cette section est de discuter l'approximation de champ moyen pour les \'etats d'\'equilibre d'un syst\`eme classique \`a temp\'erature nulle. On s'int\'eresse donc au probl\`eme de minimisation 
\begin{equation}\label{eq:intro MF lim class}
E(N) = \inf \left\{\int_{\Om ^N} H_N d\mubf_N,\, \mubf_N \in \PP_s (\Om ^N)\right\} 
\end{equation}
avec
\begin{equation}\label{eq:intro MF lim hamil}
H_N (X) = \sum_{j=1} ^N V(x_j) + \frac{1}{N-1} \sum_{1\leq i < j \leq N} w(x_i-x_j).
\end{equation}
Nous allons esquisser la preuve (formelle) de la validit\'e de l'approximation de champ moyen au niveau de l'\'energie fondamentale, autrement dit de l'\'egalit\'e
\begin{equation}\label{eq:intro justif MF class}
\boxed{ \lim_{N\to \infty} \frac{E(N)}{N} = \MFEe = \inf\left\{ \MFEf [\rho], \rho \in \PP (\Om) \right\}}
\end{equation}
o\`u la fonctionnelle de champ moyen est obtenue comme en~\eqref{eq:intro MF ener func class}, en tenant compte de~\eqref{eq:intro lambda}:
\begin{equation}\label{eq:intro MF func lim}
\MFEf[\rho] = \int_{\Om} V d\rho + \frac12 \iint_{\Om \times \Om} w(x-y) d\rho(x)d\rho(y). 
\end{equation}
Obtenir la limite~\eqref{eq:intro justif MF class} pour l'\'energie fondamentale est une premi\`ere \'etape, et le sch\'ema de preuve fournit en fait des informations sur les \'etats d'\'equilibre. Pour simplifier la pr\'esentation (formelle) qui suit, nous renvoyons la discussion de ces aspects, ainsi que l'\'etude du cas \`a temp\'erature positive, au Chapitre~\ref{sec:class}.

\subsubsection*{Un passage \`a la limite formel.}  Cette section a pour but de souligner la structure ``alg\'ebrique'' du probl\`eme. La justification des manipulations auxquelles nous allons nous livrer n\'ecessite des hypoth\`eses d'analyse qui seront discut\'ees dans la suite du cours mais que nous passons sous silence pour l'instant. 

On commence par utiliser~\eqref{eq:intro ener marginales} et l'hypoth\`ese~\eqref{eq:intro lambda} pour \'ecrire 
$$ \frac{E(N)}{N} = \frac{1}{2} \inf\left\{ \iint_{\Om \times \Om} H_2 (x,y) d\mubf_N ^{(2)}(x,y), \, \mubf_N \in \PP_s (\Om ^N) \right\} $$
o\`u $H_2$ est le Hamiltonien \`a deux particules d\'efini comme en~\eqref{eq:intro MF lim hamil} et $\mubf_N ^{(2)}$ est la seconde marginale de la probabilit\'e sym\'etrique $\mubf_N$. Puisque l'\'energie ne d\'epend que de la seconde marginale, on peut voir le probl\`eme qui nous int\'eresse comme un probl\`eme d'optimisation avec contrainte pour des probabilit\'es \`a deux corps sym\'etriques:
$$ \frac{E(N)}{N} = \frac{1}{2} \inf\left\{ \iint_{\Om \times \Om} H_2 (x,y) d\mubf ^{(2)}(x,y), \mubf ^{(2)} \in \PP_N  ^{(2)} \right\} $$
avec
$$\PP_N  ^{(2)} = \left\{ \mubf ^{(2)} \in \PP_s (\Om ^2) \,|\, \exists \mubf_N \in \PP_s (\Om ^N), \mubf ^{(2)} = \mubf_N ^{(2)} \right\} $$
l'ensemble des probabilit\'es \`a deux corps qu'on peut obtenir comme les marginales d'\'etats \`a $N$ corps. En supposant que la limite existe (premier argument formel) on obtient donc 
\begin{equation}\label{eq:MF class form 1}
 \lim_{N\to \infty} \frac{E(N)}{N} = \frac{1}{2} \lim_{N\to \infty}\, \inf\left\{ \iint_{\Om \times \Om} H_2 (x,y) d\mubf ^{(2)}(x,y), \mubf ^{(2)} \in \PP_N  ^{(2)} \right\} 
\end{equation}
et il est tentant d'\'echanger la limite et l'infimum dans cette expression (deuxi\`eme argument formel) pour obtenir 
$$ \lim_{N\to \infty} \frac{E(N)}{N} = \frac{1}{2} \inf \left\{ \iint_{\Om \times \Om} H_2 (x,y) d\mubf ^{(2)}(x,y), \mubf ^{(2)} \in \lim_{N\to \infty} \PP_N  ^{(2)} \right\}. $$
On a observ\'e que la fonctionnelle d'\'energie apparaissant dans~\eqref{eq:MF class form 1} est en fait ind\'ependante de $N$. Toute la d\'ependance en $N$ se trouve dans la contrainte sur l'espace variationnel que l'on consid\`ere. Ceci sugg\`ere que le probl\`eme limite naturel consiste \`a minimiser la m\^eme fonctionnelle mais sur la limite de l'espace variationnel, comme \'ecrit ci-dessus.  

Pour donner un sens \`a la limite de $ \PP_N  ^{(2)}$, on constate que les ensembles $\PP_N  ^{(2)}$ forment une suite d\'ecroissante
$$ \PP_{N+1} ^{(2)} \subset \PP_N ^{(2)}$$
comme on s'en aper\c{c}oit facilement en notant que si $\mubf ^{(2)}\in \PP_{N+1} ^{(2)}$, alors pour un certain~$\mubf_{N+1} \in \PP_s (\Om ^{N+1})$
\begin{equation}\label{eq:MF class form obs}
\mubf ^{(2)} = \mubf_{N+1} ^{(2)} = \left(\mubf_{N+1} ^{(N)}\right) ^{(2)} 
\end{equation}
et bien s\^ur $\mubf_{N+1} ^{(N)}\in \PP_s (\Om ^N)$. On peut donc l\'egitiment faire l'identification
$$ \PP_{\infty} ^{(2)} := \lim_{N\to \infty} \PP_N  ^{(2)} = \bigcap_{N\geq 2} \PP_N ^{(2)}$$
et le probl\`eme limite naturel est alors (modulo la justification des manipulations formelles ci-dessus)
\begin{equation}\label{eq:MF class form 2}
\lim_{N\to \infty} \frac{E(N)}{N}  = E_{\infty} = \frac{1}{2} \inf \left\{ \iint_{\Om \times \Om} H_2 (x,y) d\mubf ^{(2)}(x,y), \mubf ^{(2)} \in  \PP_\infty  ^{(2)} \right\}, 
\end{equation}
soit un probl\`eme variationnel sur l'ensemble des \'etats \`a deux corps que l'on peut obtenir comme densit\'es r\'eduites d'\'etats \`a $N$ corps, pour tout $N$.

\subsubsection*{Un r\'esultat de structure.} Nous allons maintenant expliquer qu'en fait 
$$ E_{\infty}  = \MFEe, $$
ce qui d\'ecoule d'un r\'esultat fondamental sur la structure de l'espace $\PP_\infty  ^{(2)}$.  

Regardons de plus pr\`es l'espace $\PP_\infty  ^{(2)}$. Il contient bien s\^ur les \'etats produits de forme $\rho \otimes \rho,\rho \in \PP (\Om)$ puisque $\rho ^{\otimes 2}$ est la seconde marginale de $\rho ^{\otimes N}$ pour tout $N$. Par convexit\'e il contient \'egalement toutes les combinaisons convexes de tels \'etats 
\begin{equation}\label{eq:MF class form 3}
\left\{ \int_{\rho \in \PP (\Om)} \rho ^{\otimes 2}dP (\rho), P\in \PP (\PP(\Om))\right\} \subset \PP_\infty  ^{(2)}
\end{equation}
avec $\PP (\PP(\Om))$ l'ensemble des mesures de probabilit\'e sur $\PP (\Om)$. Au vu de~\eqref{eq:intro ener marginales} et~\eqref{eq:intro MF func lim} on aura justifi\'e l'approximation de champ moyen~\eqref{eq:intro justif MF class} si on peut montrer que l'infimum dans~\eqref{eq:MF class form 2} est atteint pour $\mubf ^{(2)} = \rho ^{\otimes 2}$ pour un certain $\rho \in \PP (\Om)$.

Le r\'esultat de structure qui permet d'atteindre cette conclusion est la constation qu'il y a en fait \'egalit\'e dans~\eqref{eq:MF class form 3}: 
\begin{equation}\label{eq:MF class form 4}
\left\{ \int_{\rho \in \PP (\Om)} \rho ^{\otimes 2}dP (\rho), P\in \PP (\PP(\Om))\right\} =\PP_\infty  ^{(2)}.
\end{equation}
En effet, vu la lin\'earit\'e de la fonctionnelle d'\'energie en fonction de $\mubf ^{(2)},$ on peut \'ecrire 
\begin{align*}
E_{\infty} &= \frac12 \inf \left\{ \int_{\rho \in \PP (\Om)} \iint_{\Om \times \Om} H_2 (x,y) d\rho ^{\otimes 2}(x,y)dP(\rho), P \in \PP (\PP(\Om)) \right\}\nonumber\\
&= \inf \left\{ \int_{\rho \in \PP (\Om)} \MFEf [\rho] dP(\rho), P \in \PP (\PP(\Om)) \right\}\nonumber\\
&= \MFEe
\end{align*}
puisqu'il est clair que l'infimum en $P\in \PP (\PP(\Om))$ est atteint pour $P = \delta_{\rhoMF}$, une masse de Dirac en $\rhoMF$, un minimiseur de $\MFEf$.

On voit donc que en acceptant les manipulations formelles (justifi\'ees au Chapitre~\ref{sec:class}) menant \`a~\eqref{eq:MF class form 2}, la validit\'e de l'approximation de champ moyen suit en utilisant tr\`es peu les propri\'et\'es du Hamiltonien mais beaucoup la structure des \'etats sym\'etriques \`a $N$ corps, sous la forme de l'\'egalit\'e~\eqref{eq:MF class form 4}.

Cette \'egalit\'e est une cons\'equence du th\'eor\`eme de Hewitt-Savage, ou th\'eor\`eme de Finetti classique~\cite{DeFinetti-31,DeFinetti-37,HewSav-55}, rappell\'e \`a la Section~\ref{sec:HS} et d\'emontr\'e \`a la Section~\ref{sec:DF}. Donnons quelques d\'etails pour le lecteur familiaris\'e. Il s'agit de montrer l'inclusion inverse dans~\eqref{eq:MF class form 3}. On prend donc un certain $\mubf ^{(2)}$ qui satisfait 
$$\mubf ^{(2)} = \mubf_N ^{(2)} $$
pour une certaine suite $\mubf_N \in \PP_s (\Om ^N)$. En se basant sur l'observation~\eqref{eq:MF class form obs} on peut supposer que 
$$ \mubf_{N+1} ^{(N)} = \mubf_N$$
et on est alors en pr\'esence d'une suite (hi\'erarchie) d'\'etats \`a $N$ corps consistante. Il existe alors (th\'eor\`eme d'extension de Kolmogorov) une probabilite sym\'etrique sur les suites de~$\Om$, $\mubf \in \PP_s (\Om ^{\N})$ telle que 
$$\mubf_N = \mubf ^{(N)}$$
o\`u la $N$-\`eme marginale est d\'efinie comme en~\eqref{eq:defi marginale}. Le th\'eor\`eme de Hewitt-Savage~\cite{HewSav-55} garantit alors l'existence d'une unique mesure de probabilit\'e $P\in \PP (\PP (\Om))$ telle que 
$$\mubf_N = \int_{\rho \in \PP (\Om)} \rho ^{\otimes N} dP (\rho) $$
et en prenant la seconde marginale, on obtient le r\'esultat souhait\'e.

\medskip 

Le Chapitre~\ref{sec:class} est consacr\'e entre autres aux d\'etails du sch\'ema de preuve ci-dessus. On y fera des hypoth\`eses ad\'equates permettant de mettre toutes ces consid\'erations sur une base rigoureuse. Il vaut cependant la peine de noter d\`es maintenant que ce sch\'ema (inspir\'e des r\'ef\'erences~\cite{MesSpo-82,Kiessling-89,Kiessling-93,CagLioMarPul-92,KieSpo-99}) n'utilise aucune propri\'et\'e de structure du Hamiltonien (signe des interactions, r\'epulsives ou attractives, par exemple) mais uniquement des propri\'et\'es de compacit\'e et de r\'egularit\'e. 

\subsection{Champ moyen et th\'eor\`eme de de Finetti quantique}\label{sec:MF deF quant}\mbox{}\\\vspace{-0.4cm}

Nous allons maintenant esquisser une strat\'egie pour justifier l'approximation de champ moyen au niveau de l'\'energie fondamentale d'un grand syst\`eme bosonique. La d\'emarche est la m\^eme que pour le cas classique ci-dessus. On se place dans un r\'egime de champ moyen comme pr\'ec\'edement en prenant $\lambda = (N-1) ^{-1}$ et on consid\`ere donc le Hamiltonien
\begin{equation}\label{eq:MF quant hamil}
H_N =\sum_{j=1} ^N -\left(\nabla_j + i A (x_j)\right) ^2+ V(x_j) + \frac{1}{N-1}\sum_{1\leq i<j\leq N} w(x_i-x_j)  
\end{equation}
agissant sur $L ^2 (\R^{dN})$. Par rapport \`a la situation pr\'ec\'edente nous avons rajout\'e un potentiel vecteur $A$ pour commencer de souligner la g\'en\'eralit\'e de l'approche. Il peut par exemple correspondre \`a un champ magn\'etique externe $B = \rm{curl} \,A$.

Le point de d\'epart est l'\'energie fondamentale bosonique d\'efinie pr\'ec\'edemment
\begin{align}\label{eq:MF quant E(N)}
E (N) &= \inf\left\{ \tr_{L ^2 (\R ^{dN})} \left[ H_N \Gamma_N \right], \,\Gamma_N \in \cS (L ^2_s (\R ^{dN}))  \right\} \nonumber\\
&= \inf\left\{ \left\langle \Psi_N, H_N \Psi_N\right\rangle, \, \Psi_N \in L ^2_s (\R ^{dN}), \norm{\Psi_N} _{L ^2 (\R ^{dN})}  = 1  \right\}
\end{align}
et on rappelle que l'on peut minimiser sur les \'etats purs ou les \'etats mixtes indiff\'erement, d'o\`u les deux d\'efinitions \'equivalentes.

Notre but est de montrer que pour $N$ grand l'\'energie bosonique par particule est donn\'ee par la fonctionnelle de Hartree:
\begin{equation}\label{eq:MF quant lim}
\boxed{\lim_{N\to \infty}  \frac{E(N)}{N} = \eH }
\end{equation}
o\`u 
\begin{align}\label{eq:MF quant Hartree}
\eH &= \inf \left\{  \EH[u], u \in L ^2 (\R ^d), \norm{u}_{L^2 (\R ^d )} = 1 \right\}\nonumber\\
\EH [u] &= \int_{\R ^d} \left(|(\nabla +i A) u| ^2 + V |u| ^2\right) + \frac12 \iint_{\R ^d \times \R ^d} |u(x)| ^2 w(x-y) |u(y)| ^2 dxdy.
\end{align}
Puisque l'\'energie de Hartree est obtenue en ins\'erant un ansatz $\Psi_N = u ^{\otimes N}$ dans la fonctionnelle d'\'energie \`a $N$ corps, la limite~\eqref{eq:MF quant lim} est d\'ej\`a une indication forte de l'existence de la condensation de Bose-Einstein au niveau du fondamental d'un grand syst\`eme bosonique dans le r\'egime de champ moyen. On reviendra plus bas sur les cons\'equences de~\eqref{eq:MF quant lim} sur les minimiseurs. Comme dans le cas classique, le r\'egime de champ moyen est un mod\`ele tr\`es simplifi\'e mais d\'ej\`a tr\`es instructif. On pr\'esentera au Chapitre~\ref{sec:NLS} l'analyse d'autres limites physiquement int\'eressantes.

\subsubsection*{Un passage \`a la limite formel.} La premi\`ere \'etape pour obtenir~\eqref{eq:MF quant lim} est comme dans le cas classique d'obtenir formellement un probl\`eme limite simplifi\'e. On commence par r\'e\'ecrire l'\'energie en utilisant~\eqref{eq:quant ener red mat} et l'hypoth\`ese $\lambda = (N-1) ^{-1}$
\begin{equation}\label{eq:MF quant form 1}
 \frac{E(N)}{N} = \frac12 \inf\left\{ \tr_{L ^2 (\R^{2d})} \left[ H_2 \: \Gamma ^{(2)}\right], \, \Gamma ^{(2)} \in \PP_N ^{(2)}\right\} 
\end{equation}
o\`u
$$ \PP_N ^{(2)} = \left\{ \Gamma ^{(2)} \in \cS (L_s ^2(\R ^{2d})) \,|\, \exists \, \Gamma_N \in \cS (L_s ^2(\R ^{dN})), \Gamma ^{(2)} = \tr_{3\to N}[ \Gamma_N ]\right\}$$
est l'ensemble des matrices densit\'e \`a deux corps ``$N$-repr\'esentables'', i.e. qui sont la trace partielle d'un \'etat \`a $N$ corps. 

La caract\'erisation de l'ensemble $\PP_N ^{(2)}$ est un probl\`eme ouvert fameux en m\'ecanique quantique (probl\`eme de repr\'esentabilit\'e,  voir par exemple~\cite{ColYuk-00,LieSei-09}) et la r\'e\'ecriture~\eqref{eq:MF quant form 1} n'est donc pas particuli\`erement utile \`a $N$ fix\'e. En revanche, exactement comme dans le cas classique, le probl\`eme de repr\'esentabilit\'e peut \^etre r\'esolu de mani\`ere satisfaisante dans la limite $N\to \infty$, et c'est ce fait que nous allons utiliser. On commence par noter que en prenant des traces partielles on a facilement
\begin{equation}\label{eq:MF quant form 1bis}
\PP_{N+1} ^{(2)}\subset \PP_N ^{(2)}
\end{equation}
de sorte que~\eqref{eq:MF quant form 1} peut \^etre vu comme un probl\`eme variationnel pos\'e sur un espace de plus en plus contraint lorsque $N$ augmente.

En supposant \`a nouveau qu'on puisse \'echanger un infimum et une limite (ce qui est bien s\^ur un argument formel) on obtient 
\begin{equation}\label{eq:MF quant form 2}
\lim_{N\to \infty} \frac{E(N)}{N} = E_{\infty} := \frac12 \inf\left\{ \tr_{L ^2 (\R^{2d})} \left[ H_2 \: \Gamma ^{(2)}, \Gamma ^{(2)} \in \PP_\infty ^{(2)}\right]\right\} 
\end{equation}
avec 
\begin{equation}\label{eq:MF quant form 3}
\PP_\infty ^{(2)} = \lim_{N\to \infty} \PP_N ^{(2)} = \bigcap_{N\geq 2} \PP_N ^{(2)} 
\end{equation}
l'ensemble des matrices densit\'e \`a deux corps qui sont $N$-repr\'esentables pour tout $N$. On a comme dans le cas classique utilis\'e le fait que la fonctionnelle d'\'energie elle m\^eme ne d\'epend pas de $N$ pour passer formellement \`a la limite.

\subsubsection*{Un r\'esultat de structure.} Il se trouve que la structure de l'ensemble $\PP_\infty ^{(2)}$ est enti\`erement connue et entra\^ine l'\'egalit\'e
\begin{equation}\label{eq:MF quant form 4}
E_{\infty} = \eH
\end{equation}
qui conclut la preuve de~\eqref{eq:MF quant lim} modulo la justification des manipulations formelles que nous venons d'effectuer.

La propri\'et\'e de structure garantissant la validit\'e~\eqref{eq:MF quant form 4} est une version quantique du th\'eor\`eme de de Finetti-Hewitt-Savage que nous avons mentionn\'e \`a la section pr\'ec\'edente. Soit $\Gamma ^{(2)}\in \PP_\infty ^{(2)}$. On a alors une suite $\Gamma_N \in \cS (L ^2_s (\R ^{dN}))$ telle que pour tout $N$
$$ \Gamma ^{(2)} = \Gamma_N ^{(2)}.$$
Sans perte de g\'en\'eralit\'e on peut supposer que cette suite est consistante au sens o\`u
$$\Gamma_{N+1} ^{(N)} = \tr_{N+1} [\Gamma_{N+1}] = \Gamma_N. $$
Le th\'eor\`eme de de Finetti quantique de St\o{}rmer-Hudson-Moody~\cite{Stormer-69,HudMoo-75} garantit alors l'existence d'une mesure de probabilit\'e $P \in \PP (S L ^{2} (\R ^d))$ sur la sph\`ere unit\'e de $L^2 (\R ^d)$ telle que 
$$ \Gamma_N = \int_{u\in SL ^{2} (\R ^d)} |u ^{\otimes N}\rangle \langle u ^{\otimes N} |dP (u)$$
et donc 
$$ \Gamma ^{(2)} = \int_{u\in SL ^{2} (\R ^d)} |u ^{\otimes 2}\rangle \langle u ^{\otimes 2} |dP (u).$$
On peut alors conclure que 
\begin{align*}
 E_\infty &= \inf\left\{ \frac12 \int_{u\in SL ^{2} (\R ^d)} \tr_{L ^2 (\R ^{2d})} [H_2 |u ^{\otimes 2}\rangle \langle u ^{\otimes 2} |]dP(u), P \in \PP (S L ^{2} (\R ^d)) \right\} \\ 
 &= \inf\left\{ \int_{u\in SL ^{2} (\R ^d)} \EH [u]dP(u), P \in \PP (S L ^{2} (\R ^d)) \right\} \\ 
&=\eH
 \end{align*}
o\`u la derni\`ere \'egalit\'e suit puisqu'il est clairement optimal de prendre $P= \delta_{\uH}$ avec $\uH$ un minimiseur de la fonctionnelle de Hartree.
 
On voit donc que la validit\'e de~\eqref{eq:MF quant lim} est (au moins formellement) une cons\'equence de la structure de l'ensemble des \'etats bosoniques et ne fait intervenir que marginalement les propri\'et\'es du Hamiltonien~\eqref{eq:MF quant hamil}. La justification sous diverses hypoth\`eses des manipulations formelles que nous venons d'effectuer pour obtenir~\eqref{eq:MF quant form 2} (ou une variante) ainsi que la preuve du th\'eor\`eme de de Finetti quantique (plus variantes, g\'en\'eralisations et corollaires) sont l'objet principal de ces notes.

\subsubsection*{Condensation de Bose-Einstein.} Anticipons quelque peu sur les conclusions quant aux minimiseurs que l'on peut d\'eduire de la limite~\eqref{eq:MF quant lim}. On verra dans la suite du cours que des r\'esultats du type  
\begin{equation}\label{eq:MF quant BEC 1}
\Gamma_N ^{(n)} \to \int_{u\in S L ^2 (\R^d)} |u ^{\otimes n}\rangle \langle u ^{\otimes n}| dP (u) 
\end{equation}
pour tout $n\in \N$ fix\'e quand $N\to\infty$, suivent tr\`es naturellement dans les bons cas, o\`u $\Gamma_N$ est (la matrice densit\'e d') un minimiseur de l'\'energie \`a $N$ corps et $P$ une mesure de probabilit\'e concentr\'ee sur les minimiseurs de $\EH$. La convergence aura lieu (toujours dans les ``bons cas'') en norme de trace.

Lorsque il y a unicit\'e (\`a une phase constante pr\`es) du minimiseur $\uH$ de $\EH$ on obtient donc 
\begin{equation}\label{eq:MF quant BEC 2}
\Gamma_N ^{(n)} \to |\uH ^{\otimes n}\rangle \langle \uH ^{\otimes n}|  \mbox{ quand } N\to \infty
\end{equation}
ce qui prouve l'existence de la condensation de Bose-Einstein au niveau du fondamental. En effet, la condensation de Bose-Einstein consiste \emph{par d\'efinition} (voir \cite{LieSeiSolYng-05} et r\'ef\'erences ci-incluses, en particulier \cite{PenOns-56}) en l'existence d'une valeur propre d'ordre\footnote{Rappellons que toutes les matrices densit\'e sont normalis\'ees pour avoir une trace fixe \'egale \`a $1$ dans ce cours.} $1$ dans la limite $N\to \infty$ pour $\Gamma_N ^{(1)}$, ce qui est clairement impliqu\'e par~\eqref{eq:MF quant BEC 2} qui est en fait un r\'esultat plus fort.  

On peut bien s\^ur se demander si des r\'esultats plus forts que~\eqref{eq:MF quant BEC 2} peuvent \^etre d\'emontr\'es. On pourrait penser \`a une approximation en norme du genre 
$$ \norm{ \Psi_N - \uH ^{\otimes N} }_{L ^2 (\R ^{dN})}\to 0,$$
mais il est opportun de mentionner tout de suite que des r\'esultats de ce genre sont \emph{faux} en g\'en\'eral. La bonne notion de condensation fait bien intervenir les matrices densit\'e r\'eduites, comme le montrent les deux remarques suivantes:
\begin{itemize}
\item Pensons \`a un $\Psi_N$ de forme (avec $\otimes_s$ le produit tensoriel sym\'etrique) 
 $$ \Psi_N = \uH ^{\otimes (N-1)} \otimes_s \varphi$$ 
o\`u $\varphi$ est orthogonal \`a $\uH$. Un tel \'etat est ``presque condens\'e'' puisque toutes les particules sauf une sont dans l'\'etat $\uH$. Hors, au sens $L ^2 (\R ^{dN})$ usuel, $\Psi_N$ est bien s\^ur orthogonal \`a $\uH ^{\otimes N}$ et ne peut donc pas en \^etre proche en norme, bien qu'il en soit proche au sens de la condensation faisant intervenir les matrices de densit\'e r\'eduites.
\item Dans le m\^eme ordre d'id\'ee, il est naturel de chercher des corrections au minimiseur \`a $N$ corps sous la forme 
$$ \Psi_N = \varphi_0 \, \uH ^{\otimes N} + \uH ^{\otimes (N-1)} \otimes_s \varphi_1 + \uH ^{\otimes (N-2)} \otimes_s \varphi_2 + \ldots + \varphi_N$$
avec $\varphi_0 \in \C$ et $\varphi_k \in L_s ^2 (\R ^{dk})$ pour $k=1\ldots N$. Il se trouve que l'ansatz ci-dessus est correct si la suite $(\varphi_k)_{k = 0,\ldots,N}$ est choisie via la minimisation d'un Hamiltonien effectif sur l'espace de Fock. La pr\'esence des termes non condens\'es impliquant les $\varphi_k$ pour $k\geq 1$ contribue au premier ordre \`a la norme de $\Psi_N$ mais pas aux matrices de densit\'e r\'eduites, ce qui confirme rigoureusement que la bonne notion de condensation est n\'ecessairement bas\'ee sur les matrices densit\'e r\'eduites. Cette remarque ne fait qu'effleurer la th\'eorie de Bogoliubov, qui ne sera pas trait\'ee dans ces notes et au sujet de laquelle on renvoit \`a~\cite{CorDerZin-09,LieSol-01,LieSol-04,Solovej-06,Seiringer-11,GreSei-12,LewNamSerSol-13,NamSei-14,DerNap-13} pour des r\'esultats math\'ematiques r\'ecents.
\end{itemize}

\subsubsection*{Une remarque sur la sym\'etrie.} Nous nous sommes int\'eress\'es ci-dessus au probl\`eme bosonique pour des raisons de motivation physique. Le probl\`eme fermionique n'est\footnote{Pour l'instant ?} pas couvert par ce genre de consid\'erations, mais on peut s'int\'eresser au probl\`eme sans contrainte de sym\'etrie mentionn\'e pr\'ec\'edement.

Dans le cas o\`u $A\equiv 0$ dans~\eqref{eq:MF quant hamil}, les probl\`emes avec et sans sym\'etrie bosonique impos\'ee co\"incident, mais ce n'est pas le cas en g\'en\'eral. On peut cependant suivre la m\^eme d\'emarche que ci-dessus pour l'\'etude du probl\`eme sans sym\'etrie en remarquant que vue la forme du Hamiltonien, on peut sans perte de g\'en\'eralit\'e imposer la sym\'etrie plus faible~\eqref{eq:quant sym gen} dans la minimisation. L'ensemble de matrices de densit\'e \`a deux corps qui appara\^it alors \`a la limite est \'egalement couvert par le th\'eor\`eme de St\o{}rmer-Hudson-Moody, et il s'agit alors de minimiser une \'energie sur les matrices de densit\'e \`a deux corps de la forme 
\begin{equation}\label{eq:intro deF gen} 
\Gamma ^{(2)} = \int_{\gamma \in \cS (L ^2 (\R ^d))} \gamma ^{\otimes 2} dP(\gamma) 
\end{equation}
o\`u $P$ est maintenant une mesure de probabilit\'e sur les \'etats \`a une particule, c'est-\`a-dire les op\'erateurs auto-adjoints positifs de trace $1$ sur $L ^ 2 (\R ^d)$. On obtient donc dans ce cas la convergence de l'\'energie fondamentale par particule vers le minimum de la fonctionnelle de Hartree g\'en\'eralis\'ee (cf~\eqref{eq:intro Hartree func gen}) 
$$
\EH [\gamma] = \tr _{L ^2 (\R^d )} \left[\left( -(\nabla + i A) ^2 + V \right)\gamma \right] + \lambda \frac{N-1}{2} \tr_{L ^2 (\R ^{2d})} \left[ w(x-y) \gamma ^{\otimes 2}\right]$$
et le minimum est en g\'en\'eral diff\'erent du minimum de la fonctionnelle de Hartree~\eqref{eq:MF quant Hartree} (voir par exemple~\cite{Seiringer-03}).

On a d\'ej\`a not\'e que $\gamma ^{\otimes 2}$ ne peut avoir la sym\'etrie bosonique que si $\gamma $ est un \'etat pur $\gamma = |u  \rangle \langle u|$, et on voit donc en quoi le probl\`eme sans sym\'etrie peut en g\'en\'eral donner lieu \`a un minimum strictement plus petit. Il se trouve que le minimum de~\eqref{eq:intro Hartree func gen} est toujours atteint pour un \'etat pur, ce qui met en coh\'erence les diff\'erentes observations que nous avons faites sur la sym\'etrie. Notons finalement que prendre en compte la sym\'etrie bosonique dans la limite de champ moyen quand $A\neq 0$ se fait tout naturellement gr\^ace aux diff\'erentes version du th\'eor\`eme de de Finetti quantique. 

%
%

\newpage

\section{\textbf{M\'ecanique statistique \`a l'\'equilibre}}\label{sec:class}

Avant d'entrer dans le coeur du sujet de ces notes, les limites de champ moyen quantiques, nous rappellerons dans un but p\'edagogique quelques notions sur les limites de champ moyen classiques. Dans cette section nous \'enoncerons et d\'emontrerons les th\'eor\`emes de de Finetti classiques mentionn\'es informellement auparavant. Des applications \`a des probl\`emes de m\'ecanique statistique \`a l'\'equilibre seront ensuite pr\'esent\'ees.

Pour simplifier l'exposition et se rapprocher des applications qui nous int\'eressent nous consid\'ererons dans toute cette section un domaine $\Omega \subset \R ^d$ qui pourra \'eventuellement \^etre $\R ^d$ lui-m\^eme. On notera $\Omega ^N$ et $\Omega ^{\N}$ le produit cart\'esien de $N$ copies de $\Omega$ et l'ensemble des suites de $\Omega$ respectivement. L'espace des mesures de probabilit\'e sur un ensemble $\Lambda$ sera toujours not\'e $\PP (\Lambda)$. On pourra faire l'hypoth\`ese simplificatrice que $\Om$ est compact, auxquel cas $\PP (\Omega)$ l'est \'egalement pour la convergence faible des mesures.

\subsection{Th\'eor\`eme de Hewitt-Savage}\label{sec:HS}\mbox{}\\\vspace{-0.4cm}

Nous mentionnerons seulement pour m\'emoire les premiers travaux sur ce qui est appell\'e maintenant le th\'eor\`eme de de Finetti~\cite{DeFinetti-31,DeFinetti-37,Khintchine-32,Dynkin-53}. Dans ces notes, l'histoire commencera \`a~\cite{HewSav-55} o\`u le th\'eor\`eme de de Finetti classique est d\'emontr\'e dans sa forme la plus g\'en\'erale. 

De mani\`ere informelle, le th\'eor\`eme de Hewitt-Savage~\cite{HewSav-55} \'enonce que toute mesure de probabilit\'e sym\'etrique sur $\Omega ^N$ approche une combinaison convexe de probabilit\'es produit quand $N$ est grand. Une \emph{mesure de probabilit\'e sym\'etrique} est une probabilit\'e $\mubf_N$ satisfaisant
\begin{equation}\label{eq:sym class N}
\mubf_N (A_1 \times \ldots \times A_N) = \mubf_N (A_{\sigma(1)},\ldots,A_{\sigma(N)}) 
\end{equation}
pour tout domaines bor\'eliens $A_1,\ldots,A_N \subset \Omega$ et toute permutation de $N$ indices $\sigma \in \Sigma_N$. On notera $\PP_s (\Om ^N)$ l'ensemble des mesures de probabilit\'e satisfaisant~\eqref{eq:sym class N}. Une \emph{mesure produit} construite sur $\rho \in \PP (\Om)$ est de la forme 
\begin{equation}\label{eq:proba produit}
\rho ^{\otimes N} (A_1,\ldots,A_N) = \rho (A_1) \ldots \rho (A_N)  
\end{equation}
et est bien entendu sym\'etrique. Nous cherchons donc un r\'esultat du genre
\begin{equation}\label{eq:class deF formel}
\mubf_N \approx \int_{\rho \in \PP (\Om)} \rho ^{\otimes N} dP_{\mubf_N} (\rho) \mbox{ quand } N\to \infty
\end{equation}
o\`u $P_{\mubf_N} \in \PP(\PP (\Om))$ est une mesure de probabilit\'e sur les probabilit\'es. 

\medskip

Une premi\`ere possibilit\'e pour donner un sens rigoureux \`a~\eqref{eq:class deF formel} consiste \`a prendre imm\'ediatement $N = \infty$, c'est-\`a-dire \`a consid\'erer non pas une probabilit\'e $\mubf_N$ sur $\Om ^N$ mais directement une probabilit\'e avec un nombre infini de variables, $\mubf \in \PP (\Om ^{\N})$, ce qui est le sens naturel \`a donner \`a un ``\'etat classique d'un syst\`eme \`a nombre infini de particules''. On supposera une notion de sym\'etrie h\'erit\'ee de~\eqref{eq:sym class N}: 
\begin{equation}\label{eq:sym class inf}
\mubf (A_1,A_2,\ldots) =  \mubf (A_{\sigma(1)},A_{\sigma(2)},\ldots) 
\end{equation}
pour toute suite de domaines bor\'eliens $(A_k)_{k\in \N} \subset \Omega ^{\N}$ et toute permutation d'un nombre infini d'indices $\sigma \in \Sigma_{\infty}$. On notera $\PP_s (\Om ^{\N})$ l'ensemble des probabilit\'es sur $\Om ^{\N}$ satisfaisant~\eqref{eq:sym class inf}. Le th\'eor\`eme de Hewitt-Savage est le r\'esultat suivant:

\begin{theorem}[Hewitt-Savage 1955]\label{thm:HS}\mbox{}\\
Soit $\mubf \in \PP_s (\Om ^{\N})$ satisfaisant~\eqref{eq:sym class inf}. Soit $\mubf ^{(n)}$ sa $n$-i\`eme marginale d\'efinie par 
\begin{equation}\label{eq:HS marginale n}
\mubf ^{(n)} (A_1,\ldots,A_n) = \mubf (A_1,\ldots,A_n,\Omega,\ldots,\Omega,\ldots). 
\end{equation}
Il existe une unique mesure de probabilit\'e $P_{\mubf} \in \PP (\PP (\Omega))$ telle que 
\begin{equation}\label{eq:result HS}
\mubf ^{(n)} =  \int_{\rho \in \PP (\Om)} \rho ^{\otimes n} dP_{\mubf} (\rho). 
\end{equation}
\end{theorem}

En m\'ecanique statistique, ce th\'eor\`eme est g\'en\'eralement appliqu\'e comme une version faible de l'approximation formelle~\eqref{eq:class deF formel} de la mani\`ere suivante. On part d'un \'etat classique \`a $N$ particules, une probabilit\'e $\mubf_N \in \PP (\Omega ^N)$ dont les marginales 
\begin{equation}\label{eq:class marginale N}
\mubf_N ^{(n)} (A_1,\ldots,A_n) = \mubf (A_1,\ldots,A_n,\Omega ^{N-n})
\end{equation}
convergent\footnote{On notera cette convergence $\wto_\ast$ pour la distinguer d'une convergence en norme.} dans le sens des mesures de $\PP (\Om ^n)$, \`a une sous-suite pr\`es (nous ne changerons pas de notation pour la sous-suite):
\begin{equation}\label{eq:weak proba convergence}
\mubf_N ^{(n)} \wto_{\ast} \mubf ^{(n)} \in \PP (\Om ^{(n)}), 
\end{equation}
ce qui veut pr\'ecis\'ement dire que 
\[
\mubf_N ^{(n)} (A_n) \to \mubf ^{(n)} (A_n) 
\]
pour tout bor\'elien $A_n$ de $\Omega ^n$, et donc 
\[
\int_{\Omega} f_n d\mubf_N ^{(n)} \to  \int_{\Omega} f_n d\mubf ^{(n)}
\]
pour toute fonction $f_n$ continue born\'ee de $\Om ^n$ dans $\R$. Modulo un argument d'extraction diagonale, on peut supposer que la convergence~\eqref{eq:weak proba convergence} a lieu le long de la m\^eme sous-suite pour tout $n\in \N$. En testant la convergence~\eqref{eq:weak proba convergence} sur un bor\'elien $A_n = A_m \times \Om ^{m-n}$ pour $m\leq n$ on obtient la relation de consistance
\begin{equation}\label{eq:consist class}
\left( \mubf ^{(n)} \right)  ^{(m)} = \mubf ^{(m)},\mbox{ pour tout } m \leq n
\end{equation}
qui implique que la suite $(\mubf ^{(n)})_{n\in \N}$ d\'ecrit bien un syst\`eme physique avec un nombre infini de particules. On peut ensuite voir (th\'eor\`eme d'extenstion de Kolmogorov) qu'il existe $\mubf \in \PP (\Om ^{\N})$ telle que $\mubf ^{(n)}$ soit la $n$-i\`eme marginale de $\mubf$ (d'o\`u la notation). Cette probabilit\'e satisfait~\eqref{eq:sym class inf} et on peut donc lui appliquer le Th\'eor\`eme~\ref{thm:HS}, ce qui donne 
\begin{equation}\label{eq:class deF rig}
\mubf_N ^{(n)} \wto_* \int_{\rho \in \PP (\Om)} \rho ^{\otimes n} dP_{\mubf} (\rho) \mbox{ quand } N\to \infty 
\end{equation}
o\`u $P_{\mubf} \in \PP (\PP (\Om))$, ce qui est une version affaiblie mais rigoureuse de~\eqref{eq:class deF formel}. Autrement dit \textbf{\'etant donn\'e un \'etat classique \`a $N$ particules, sa $n$-i\`eme marginale peut-\^etre approch\'ee par une combinaison convexe d'\'etats produits quand $N$ est grand et $n$ fixe.} On remarquera par ailleurs que la mesure $P_{\mubf}$ apparaissant dans~\eqref{eq:class deF rig} ne d\'epend pas de $n$. 

Bien s\^ur il faut pouvoir d'abord utiliser un argument de compacit\'e pour obtenir~\eqref{eq:weak proba convergence}, ce qui est possible par exemple si $\Om$ est compact (dans ce cas $\PP (\Om ^n)$ est compact pour la convergence au sens des mesures~\eqref{eq:weak proba convergence}). Plus g\'en\'eralement si le probl\`eme physique auquel on s'int\'eresse comporte un m\'ecanisme de confinement, on peut montrer que les marginales des \'etats minimisant l'\'energie libre forment des suites tendues, et en d\'eduire~\eqref{eq:weak proba convergence}.

\medskip

Quant \`a la preuve du Th\'eor\`eme~\ref{thm:HS}, il y a plusieurs approches possibles. Mentionnons tout de suite que l'unicit\'e est une cons\'equence assez simple d'un argument de densit\'e dans $C_b (\PP (\Om))$, les fonctions continues born\'ees \footnote{Pour simplifier on peut penser au cas o\`u $\Om$ est compact, auquel cas on remplace $C_b (\PP(\Om))$ par les fonctions continues.} de $\PP(\Om)$. L'argument semble n'avoir \'et\'e formul\'e que r\'ecemment\footnote{Les cours de Pierre-Louis Lions sont disponibles en vid\'eo sur le site du Coll\`ege de France.} par Pierre-Louis Lions~\cite{Lions-CdF}. 

\begin{proof}[Preuve du Th\'eor\`eme~\ref{thm:HS}, Unicit\'e]
On v\'erifie ais\'ement que les mon\^omes de la forme
\begin{equation}\label{eq:monomes PLL}
C_b (\PP(\Om)) \ni M_{k,\phi} (\rho) := \int \phi (x_1,\ldots,x_k) d\rho ^{\otimes k} (x_1,\ldots,x_k), k\in\N, \phi \in C_b (\Om ^k) 
\end{equation}
g\'en\`erent une sous-alg\`ebre de $C_b(\PP(\Om))$, l'espace des fonctions continues born\'ees sur $\PP(\Om)$, voir Section 1.7.3 dans~\cite{Golse-13}. Le point \`a noter est que pour tout $\mu \in \PP (\Om)$
\[
M_{k,\phi} (\mu) M_{\ell,\psi} (\mu) = M_{k+\ell,\phi\:\otimes\:\psi} (\mu). 
\]
Le fait que cette sous-alg\`ebre soit dense est une cons\'equence du th\'eor\`eme de Stone-Weierstrass.

Il suffit donc de v\'erifier que si il existe deux mesures $P_{\mubf}$ et $P'_{\mubf}$ satisfaisant~\eqref{eq:HS marginale n}, alors 
\[
\int_{\rho\in \PP (\Om)} M_{k,\phi} (\rho) d P_{\mubf} (\rho) =  \int_{\rho\in \PP (\Om)} M_{k,\phi} (\rho) d P'_{\mubf} (\rho)
\]
pour tout $k\in \N$ et $\phi \in C_b (\PP (\Om ^k))$. Mais cette derni\`ere \'egalit\'e signifie simplement
\begin{multline*}
\int_{\rho\in \PP (\Om)} \left(\int_{\Om ^k} \phi(x_1,\ldots,x_k) d\rho ^{\otimes k} (x_1,\ldots,x_k) \right) d P_{\mubf} (\rho) 
\\ =  \int_{\rho\in \PP (\Om)}  \left(\int_{\Om ^k} \phi(x_1,\ldots,x_k) d\rho ^{\otimes k} (x_1,\ldots,x_k) \right) d P'_{\mubf} (\rho) 
\end{multline*}
ce qui est \'evident puisque les deux expressions sont \'egales \`a 
\[
\int_{\Om ^k}  \phi(x_1,\ldots,x_k) d\mubf ^{(k)} (x_1,\ldots,x_k)  
\]
par hypoth\`ese.
\end{proof}

Pour l'existence de la mesure, le point le plus remarquable, nous mentionnerons trois approches possibles:
\begin{itemize}
\item La preuve originelle de Hewitt-Savage est un argument g\'eom\'etrique : l'ensemble des probabilit\'es sym\'etriques sur $\Om ^{\N}$ est bien \'evidemment convexe. Le th\'eor\`eme de Choquet-Krein-Milman indique que tout point d'un convexe est une combinaison convexe des points extr\^emaux de l'ensemble. Il suffit donc de montrer que les points extr\^emes de $\PP_s (\Om ^{\N})$ sont exactement les mesures produit, correspondant aux suites de marginales $(\rho ^{\otimes n})_{n\in \N}$. Cette approche est hautement non constructive, et la preuve que les suites $(\rho ^{\otimes n})_{n\in \N}$ sont les points extr\^emes de $\PP_s (\Om ^{\N})$ est un argument par contradiction. 
\item Une approche enti\`erement constructive est due \`a Diaconis et Freedman~\cite{DiaFre-80}. Dans cet argument probabiliste, le Th\'eor\`eme~\ref{thm:HS} devient un corollaire d'un r\'esultat d'approximation \`a $N$ fini, donnant une version quantitative de~\eqref{eq:class deF rig}.
\item Pierre-Louis Lions a d\'evelopp\'e une nouvelle approche dans le cadre de la th\'eorie des jeux de champ moyen~\cite{Lions-CdF}. Il s'agit d'un point de vue dual, o\`u l'on commence par s'int\'eresser au concept de ``fonctions continues d\'ependant faiblement d'un grand nombre de variables''. On pourra en trouver un r\'esum\'e dans les notes de cours de Fran{\c{c}}ois Golse~\cite[Section 1.7.3]{Golse-13}, et une pr\'esentation compl\`ete (avec des extensions) dans~\cite[Chapitre I]{Mischler-11}. Des d\'eveloppements et des g\'en\'eralisations importants dans cette direction sont \'egalement pr\'esent\'es dans~\cite{HauMis-14}.
\end{itemize}

Il se trouve que la preuve du th\'eor\`eme de Hewitt-Savage suivant le point de vue de Lions est en grande partie une red\'ecouverte de la m\'ethode de Diaconis et Freedman. Dans la suite nous suivrons un m\'elange des deux approches. Une pr\'esentation plus compl\`ete est donn\'ee dans~\cite{Mischler-11}.

\subsection{Th\'eor\`eme de Diaconis-Freedman}\label{sec:DF}\mbox{}\\\vspace{-0.4cm}

Comme nous venons de l'annoncer, il est en fait possible de donner une version quantitative de l'approximation~\eqref{eq:class deF rig} qui implique le Th\'eor\`eme~\ref{thm:HS}. Outre son int\'er\^et intrins\`eque, ce r\'esultat m\`ene naturellement \`a une preuve constructive qui me semble \^etre de loin la plus concr\`ete qui existe dans la litt\'erature. L'approximation se fera dans la norme naturelle pour les mesures de probabilit\'e sur un ensemble $S$, la variation totale:
\begin{equation}\label{eq:TV norm}
\norm{\mu}_{\rm TV} = \int_{S} d|\mu| = \sup_{\phi \in C_b (S)} \left| \int_{\Om} \phi \: d\mu \right|
\end{equation} 
qui co\"incide avec la norme $L^1$ pour des mesures absolument continues par rapport \`a la mesure de Lebesgue. Le r\'esultat que nous allons d\'emontrer, issu de~\cite{DiaFre-80}, est le suivant:

\begin{theorem}[Diaconis-Freedman]\label{thm:DF}\mbox{}\\
Soit $\mubf_N \in \PP_s (\Om ^N)$ une mesure de probabilit\'e sym\'etrique. Il existe $P_{\mubf_N} \in \PP(\PP (\Om))$ tel que, d\'efinissant
\begin{equation}\label{eq:Pnu}
\mut_N := \int_{\rho \in \PP (\Om)} \rho ^{\otimes N} dP_{\mubf_N} (\rho)
\end{equation}
on ait 
\begin{equation}\label{eq:DiacFreed}
\left\Vert \mubf_N ^{(n)} - \mut_N ^{(n)} \right\Vert _{\rm TV} \leq 2 \frac{n(n-1)}{N}.
\end{equation}
%
%
\end{theorem}

\begin{proof}[Preuve du Th\'eor\`eme~\ref{thm:DF}]
On fera un l\'eger abus de notation en \'ecrivant $\mubf_N(Z) dZ$ au lieu de $d\mubf_N (Z)$ pour les int\'egrales en $(z_1,\ldots,z_N) = Z \in \Om ^{N}$. En tout \'etat de cause, il est d\'ej\`a suffisament instructif de consid\'erer une mesure absolument continue par rapport \`a la mesure de Lebesgue. 

La sym\'etrie de $\mubf_N$ implique, pour $X = (x_1,\ldots,x_N)\in \Om ^N$ 
\begin{equation}\label{eq:repre P}
\mubf_N (X) = \int_{\Om ^N} \mubf_N (Z) \sum_{\sigma \in \Sigma_N} (N!) ^{-1} \delta_{X = Z_{\sigma} } dZ
\end{equation}
o\`u $Z_{\sigma}$ repr\'esente le $N$-uplet $(z_{\sigma (1)},\ldots,z_{\sigma (n)})$.
On d\'efinit
\begin{equation}\label{eq:defi Pnu}
\mut_N (X) = \int_{\Om ^N} \mubf_N (Z) \sum_{\gamma \in \Gamma_N} N ^{-N} \delta_{X = Z_{\gamma}} dZ,
\end{equation}
o\`u $\Gamma_N$ est l'ensemble de toutes les \emph{applications}\footnote{Compar\'e \`a $\Sigma_N$, $\Gamma_N$ autorise donc les indices r\'ep\'et\'es.} de $\left\{1,\ldots,N \right\}$ dans lui-m\^eme et $X_{\gamma}$ est d\'efini de la m\^eme mani\`ere que $X_{\sigma}$. Le sens pr\'ecis \`a donner \`a~\eqref{eq:defi Pnu} est 
$$ \int_{\Om ^N} \phi(X) d \mut_N (X) = \sum_{\gamma \in \Gamma_N} N ^{-N} \int_{\Om ^N} \phi (Z_{\gamma}) d\mubf_N (Z)$$
pour toute fonction r\'eguli\`ere $\phi$.

En remarquant que  
\begin{equation}\label{eq:factor urn}
\sum_{\gamma \in \Gamma_N} N ^{-N} \delta_{X = Z_{\gamma}}  = \left( N ^{-1} \sum_{j=1} ^N \delta_{z_j} \right) ^{\otimes N} \left(x_1,\ldots,x_N\right), 
\end{equation}
on peut mettre~\eqref{eq:defi Pnu} sous la forme~\eqref{eq:Pnu} en prenant 
\begin{equation}\label{eq:defi nu}
P_{\mubf_N} (\rho) = \int_{\Om ^N} \delta_{\rho = \bar{\rho}_Z} \mubf_N (Z) dZ, \quad  \bar{\rho} _Z (x) := \sum_{i=1} ^N N ^{-1} \delta_{z_j=x},
\end{equation}
et on remarquera au passage que $P_{\mubf_N}$ ne charge que des mesures empiriques (les mesures $\bar{\rho}_Z$ ci-dessus). Il s'agit maintenant de calculer la diff\'erence entre les marginales de $\mubf_N$ et $\mut_N$. Diaconis et Freedman proc\`edent comme suit: Bien s\^ur 
\[
\mubf_N ^{(n)} - \mut_N ^{(n)} =  \int_{\Om ^N}  \left( \left(\sum_{\sigma \in \Sigma_N} (N!) ^{-1} \delta_{X = Z_{\sigma}}\right) ^{(n)} -  \left(\sum_{\gamma \in \Gamma_N} N ^{-N} \delta_{X = Z_{\gamma}}\right) ^{(n)}\right) \mubf_N(Z) dZ,
\]
mais 
$$\left(\sum_{\sigma \in \Sigma_N} (N!) ^{-1} \delta_{X = Z_{\sigma}}\right) ^{(n)}$$ est la loi de probabilit\'e pour tirer $n$ boules au hasard d'une urne en contenant $N$ \footnote{Les boules sont \'etiquet\'ees $z_1,\ldots,z_N$.}, \textit{sans} remplacement alors que 
$$\left(\sum_{\gamma \in \Gamma_N} N ^{-N} \delta_{X = Z_{\gamma}}\right) ^{(n)}                                                                                                                                                                                                                                                                                                                                                                                                                                                                                                                                                   $$ 
est la loi de probabilit\'e pour tirer $n$ boules au hasard d'une urne en contenant $N$, \textit{avec} remplacement. Intuitivement il est clair que quand $n$ est petit devant $N$, le fait que l'on remplace ou pas les boules apr\`es chaque tirage n'affecte pas significativement le r\'esultat. Il n'est pas difficile d'obtenir des bornes quantitatives qui m\`enent au r\'esultat~\eqref{eq:DiacFreed}, voir par exemple~\cite{Freedman-77}.

Une autre fa\c{c}on d'obtenir~\eqref{eq:DiacFreed} est la suivante, qui semble avoir son origine dans~\cite{Grunbaum-71}: au vu de~\eqref{eq:defi Pnu} et~\eqref{eq:factor urn}, on a  
$$\mut_N  ^{(n)}(X) = \int_{\Om ^N} \mubf_N (Z) \left( N ^{-1} \sum_{j=1} ^N \delta_{z_j} \right) ^{\otimes n} \left(x_1,\ldots,x_N\right) dZ. $$
On d\'eveloppe alors le produit tensoriel et on calcule la contribution des termes o\`u tous les indices sont diff\'erents. Par sym\'etrie de $\mubf_N ^{(n)}$ on obtient 
\begin{equation}\label{eq:DF astuce}
\mut_N ^{(n)} = \frac{N (N-1) \ldots (N-n+1)}{N ^n} \mubf_N ^{(n)} + \nu_n 
\end{equation}
o\`u $\nu_n$ est une mesure positive sur $\Om ^n$ (tous les termes issus du d\'eveloppement du produit~\eqref{eq:factor urn} donnent lieu \`a des contributions positives). On a alors 
\begin{equation}\label{eq:preuve DiacFreed} 
\mubf_N  ^{(n)} - \mut_N ^{(n)} = \left(1 - \frac{N (N-1) \ldots (N-n+1)}{N ^n} \right)\mubf_N ^{(n)} - \nu_n 
\end{equation}
et comme les deux mesures du membre de gauche sont des probabilit\'es on d\'eduit 
\[
\int_{\Om ^n} d\nu_n =  \left(1 - \frac{N (N-1) \ldots (N-n+1)}{N ^n} \right).
\]
En outre, puisque le premier terme du membre de droite de~\eqref{eq:preuve DiacFreed} est positif et le deuxi\`eme n\'egatif on a par l'in\'egalit\'e triangulaire
\begin{align*}
\int_{\Om ^n} d\left|\mubf_N  ^{(n)} - \mut_N ^{(n)} \right| &\leq \left(1 - \frac{N (N-1) \ldots (N-n+1)}{N ^n} \right) + \int_{\Om ^n} d\nu_n 
\\&= 2 \left(1 - \frac{N (N-1) \ldots (N-n+1)}{N ^n} \right). 
\end{align*}
Il est ensuite facile de voir que 
\begin{multline*}
\frac{N (N-1) \ldots (N-n+1)}{N ^n} = \prod_{j=1} ^n \frac{N-j+1}{N} = \prod_{j=1} ^n \left(1 - \frac{j-1}{N}\right) 
\\ \geq  \left(1 - \frac{n-1}{N}\right) ^n \geq 1- \frac{n(n-1)}{N}
\end{multline*}
ce qui prouve~\eqref{eq:DiacFreed} avec $C=2$. Une meilleure constante $C=1$ peut \^etre obtenue, voir~\cite{DiaFre-80,Freedman-77}.
\end{proof}

On fait la remarque suivante, qui s'av\`ere utile dans les applications:

\begin{remark}[Premi\`eres marginales de la mesure de Diaconis-Freedman]\label{rem:marg DF}\mbox{}\\
On a 
\begin{equation}\label{eq:marg DF 1}
\mut_N ^{(1)} (x) = \mubf_N  ^{(1)}(x) 
\end{equation}
et 
\begin{equation}\label{eq:marg DF 2}
\mut_N ^{(2)} (x_1,x_2) = \frac{N-1}{N} \mubf_N ^{(2)} (x_1,x_2) + \frac1N \mubf_N ^{(1)} (x_1) \delta_{x_1 = x_2}.
\end{equation}
comme cons\'equence directe de la d\'efinition~\eqref{eq:factor urn}. En effet, en utilisant la sym\'etrie, 
\begin{align}\label{eq:marginals DF}
\mut_N ^{(1)} (x) &= N ^{-1} \sum_{j=1} ^N \int_{\Om ^N} \mubf_N (Z) \delta_{z_j=x} dZ = \mubf_N ^{(1)} (x)\nonumber\\
\mut_N ^{(2)} (x_1,x_2) &= N ^{-2} \int_{\Om ^N} \mubf_N (Z) \left( \sum_{j=1} ^N \delta_{z_j = x_1} \right) \left( \sum_{j=1} ^N \delta_{z_j = x_2} \right) dZ \nonumber \\
&= N ^{-2} \sum_{ 1 \leq i \neq j \leq N} \int_{\Om ^N} \mubf_N (Z) \delta_{ z_i = x_1} \delta_{z_j = x_2}  dZ + N ^{-2} \sum_{i=1} ^N \int_{\Om ^N} \mubf_N (Z) \delta_{z_i = x_1} \delta_{z_i = x_2} dZ \nonumber \\
&= \frac{N-1}{N} \mubf_N ^{(2)} (x_1,x_2) + \frac1N \mubf_N ^{(1)} (x_1) \delta_{x_1 = x_2}.
\end{align}
Les marginales sup\'erieures peuvent s'obtenir par des calculs similaires mais plus lourds. \hfill\qed
\end{remark}

En corollaire de la construction de Diaconis et Freedman nous obtenons une preuve simple de la partie existence du th\'eor\`eme de Hewitt-Savage:

\begin{proof}[Preuve du Th\'eor\`eme~\ref{thm:HS}, Existence]
Commen\c{c}ons par le cas o\`u $\Om$ est compact. On applique le Th\'eor\`eme~\ref{thm:DF} \`a $\mubf ^{(N)}$ qui est une probabilit\'e sym\'etrique sur $\Om ^N$, et on obtient donc 
\begin{equation}\label{eq:preuve HS}
\norm{\mubf ^{(n)} - \int_{\rho \in \PP (\Om)} \rho ^{\otimes n} dP_{N} (\rho)}_{\rm TV} \leq C \frac{n ^2}{N}    
\end{equation}
pour une certaine mesure $P_N \in \PP (\PP (\Om))$. Lorsque $\Om$ est compact, $\PP (\Om)$ et $\PP (\PP (\Om))$ le sont \'egalement. On peut donc (modulo une sous-suite) supposer que 
\[
 P_N \to P \in \PP (\PP (\Om))
\]
au sens des mesures et il ne reste qu'\`a passer \`a la limite $N\to \infty$ \`a $n$ fix\'e dans~\eqref{eq:preuve HS}.

Pour le cas o\`u $\Om$ n'est pas compact, on suit une id\'ee de la preuve d'existence selon Pierre-Louis Lions~\cite{Lions-CdF} (voir aussi~\cite{Golse-13}). Il s'agit de montrer que la mesure $P_N$ obtenue  en appliquant le th\'eor\`eme de Diaconis-Freedman \`a $\mubf ^{(N)}$ converge. Il suffit pour cela de la tester sur un mon\^ome de la forme~\eqref{eq:monomes PLL}:
\begin{align}\label{eq:DF PLL}
\int_{\PP(\PP(\Om))} M_{k,\phi} (\mu) dP_N (\mu) &= \int_{Z \in \Om ^N} M_{k,\phi} \left( \frac{1}{N} \sum_{j=1} ^N \delta_{z_j} \right) d\mubf ^{(N)} (Z)\nonumber\\
&= \int_{Z \in \Om ^N} \int_{X\in \Om ^k} \phi(x_1,\ldots,x_k) \prod_{j=1} ^k \left(\frac{1}{N} \sum_{j=1} ^N \delta_{z_j=x_k} \right) d\mubf ^{(N)} (Z)\nonumber\\
&= \int_{Z\in \Om ^k} \phi(z_1,\ldots,z_k) d\mubf ^{(k)} (Z) + O(N ^{-1})
\end{align}
par un calcul similaire \`a celui donnant~\eqref{eq:DF astuce}. La limite~\eqref{eq:DF PLL} existe donc pour tout mon\^ome $M_{k,\phi}$, et par densit\'e des mon\^omes pour toute fonction continue born\'ee sur $\PP(\Om)$. On d\'eduit ainsi 
\[
 P_N \wto_* P
\]
et on peut conclure comme pr\'ec\'edemment.
\end{proof}

Quelque petites remarques avant de passer aux applications des Th\'eor\`emes~\ref{thm:HS} et~\ref{thm:DF}:

\begin{remark}[Sur la construction de Diaconis-Freedman-Lions]\label{rem:DiacFreLio}\mbox{}\\
\vspace{-0.5cm}
\begin{enumerate}
\item On notera d'abord que la mesure d\'efinie par~\eqref{eq:defi nu} est celle que Lions utilise dans son approche du th\'eor\`eme de Hewitt-Savage. C'est la mani\`ere canonique de construire une mesure sur $\PP (\PP(\Om))$ \'etant donn\'ee une mesure sur $\PP_s(\Om ^N)$, en passant par les mesures empiriques. Le passage \`a la limite~\eqref{eq:DF PLL} peut remplacer l'estimation explicite~\eqref{eq:DiacFreed} si on s'int\'eresse uniquement \`a la preuve du Th\'eor\`eme~\ref{thm:HS}. Cette construction et l'astuce combinatoire~\eqref{eq:DF astuce} sont \'egalement utilis\'es par exemple dans~\cite{GolMouRic-13,MisMou-13,HauMis-14}.
\item Disposer d'une estimation explicite de la forme~\eqref{eq:DiacFreed} est \'evidemment tr\`es satisfaisant et peut s'av\'erer utile dans les applications. On peut se demander si le taux de convergence obtenu est optimal. De mani\`ere peut-\^etre surprenante c'est en fait le cas. On aurait pu s'attendre \`a une approximation utile pour $n\ll N$, mais il se trouve que $\sqrt{n}\ll N$ est optimal, voir les exemples de~\cite{DiaFre-80}.
\item Les formules simples~\eqref{eq:marginals DF} sont bien utiles en pratique (voir~\cite{RouYng-14} pour une application). Il est assez satisfaisant que $\mubf_N ^{(1)} = \mut_N ^{(1)}$ et que $\mubf_N ^{(2)}$ puisse se reconstruire en utilisant seulement $\mut_N ^{(2)}$ et~$\mut_N ^{(1)}$.
\item La construction paye sa g\'en\'eralit\'e par un comportement en r\'ealit\'e tr\`es mauvais dans de nombreux cas. Remarquons que~\eqref{eq:defi nu} ne charge que des mesures empiriques, qui ont toutes une entropie infinie. Cela pose des probl\`emes pour employer le Th\'eor\`eme~\ref{thm:DF} \`a l'\'etude d'une \'energie libre \`a temp\'erature positive. En outre, pour un probl\`eme avec interactions r\'epulsives fortes, on cherchera \`a appliquer la construction \`a une mesure satisfaisant $\mubf_N ^{(2)}(x,x) \equiv 0$ (probabilit\'e nulle d'avoir deux particules au m\^eme endroit). Dans ce cas $\mut_N ^{(2)} (x,x) = N ^{-1} \mubf_N ^{(1)} (x)$ est non nulle et donc l'\'energie de $\mut_N$ sera infinie si le potentiel d'interactions comporte une singularit\'e \`a l'origine.
\item Il serait tr\`es int\'eressant de disposer d'une construction qui satisfasse une estimation de forme~\eqref{eq:DiacFreed} tout en \'evitant les inconv\'enients mentionn\'es ci-dessus. Par exemple, est-il possible de garantir $\mut_N \in L ^1(\Om ^N)$ si $\mubf_N \in L ^1 (\Om ^N)$ ? On pourrait aussi demander que la construction laisse les mesures produits $\rho ^{\otimes N}$ invariantes, ce qui n'est pas du tout le cas de~\eqref{eq:Pnu}.
\item Si l'espace $\Om$ est remplac\'e par un ensemble fini, disons $\Om = \left\lbrace 1,\ldots, d\right\rbrace$, on peut obtenir une erreur proportionnelle \`a $dn/N$ au lieu de $n ^2/N$ en utilisant la preuve originelle de Diaconis-Freedman. On peut donc remplacer ~\eqref{eq:DiacFreed} par 
\begin{equation}\label{eq:DiacFreed 3}
\left\Vert \mubf_N ^{(n)} - \mut_N ^{(n)} \right\Vert _{\rm TV} \leq \frac{C}{N} \min\left( dn,n ^2\right). 
\end{equation}

Nous ne nous servirons pas de ce point dans la suite.
\hfill\qed
\end{enumerate}

\end{remark}

Finalement, on fait une remarque s\'epar\'ee sur une g\'en\'eralisation de l'approche pr\'esent\'ee ci-dessus:

\begin{remark}[Fonctions sym\'etriques d\'ependant faiblement de $N$ variables.]\label{rem:LioHauMis}\mbox{}\\
\vspace{-0.5cm}

Dans les applications (cf la section suivante), on cherche \`a appliquer les th\'eor\`emes \`a la de Finetti de la mani\`ere suivante: on se donne une suite $(u_N)_{N\in \N}$ de fonctions sym\'etriques de $N$ variables  et on \'etudie la quantit\'e 
$$ \int_{\Om ^N} u_N (X) d\mubf_N (X)$$
pour une probabilit\'e sym\'etrique $\mubf_N$. Les r\'esultats expos\'es dans cette section impliquent que si cette quantit\'e se trouve ne d\'ependre que d'une marginale $\mubf_N ^{(n)}$ pour $n$ fixe, un objet limite naturel \'emerge quand $N\to \infty$. En particulier, si 
\begin{equation}\label{eq:simple u_N}
 u_N (X) = {N \choose n} ^{-1} \sum_{1\leq i_1 < \ldots < i_n \leq N} \phi (x_{i_1}, \ldots, x_{i_n}),
\end{equation}
on a 
\begin{equation}\label{eq:general DiacFreed}
 \int_{\Om ^N} u_N (X) d\mubf_N (X) = \int_{\Om ^n} \phi (x_1,\ldots,x_n) d\mubf_N ^{(n)} \to \int_{\rho \in \PP (\Om)} \left( \int_{\Om ^n} \phi \: d\rho ^{\otimes n}\right) dP_{\mubf} (\rho) 
\end{equation}
pour une certaine mesure de probabilit\'e $P_{\mubf} \in \PP (\PP (\Om))$. On peut se demander si ce genre de r\'esultats est vrai pour une classe de fonctions $u_N$ d\'ependant plus subtilement de $N$. Le cadre naturel semble \^etre celui o\`u la suite $u_N$ d\'epend \emph{faiblement} des $N$ variables, au sens formalis\'e par Lions~\cite{Lions-CdF} et rappell\'e dans~\cite[Section 1.7.3]{Golse-13}. Sans rentrer dans les d\'etails, on peut facilement voir qu'\`a une telle suite correspond, modulo extraction, une unique fonction continue sur les probabilit\'es $U\in C(\PP(\Om))$: on peut d\'emontrer que, le long d'une sous-suite, 
\begin{equation}\label{eq:general Lions}
 \int_{\Om ^N} u_N (X) d\mubf_N (X) \to \int_{\rho \in \PP (\Om)} U(\rho) dP_{\mubf} (\rho). 
\end{equation}
Pensons par exemple au cas o\`u $u_N$ ne d\'epend que de la mesure empirique:
\begin{equation}\label{eq:u_N moins simple}
u_N (x_1,\ldots,x_N) = F \left( \frac{1}{N} \sum_{j=1} ^N \delta_{x_j} \right) 
\end{equation}
avec une fonction $F$ suffisament r\'eguli\`ere. Une telle fonction d\'epend faiblement des $N$ variables au sens de Lions mais ne s'exprime pas sous la forme~\eqref{eq:simple u_N}. Plus g\'en\'eralement, ce genre de consid\'erations s'applique sous des hypoth\`eses du genre 
$$ |\nabla_{x_j} u_N (x_1,\ldots,x_N)| \leq \frac{C}{N}\quad \forall j= 1\ldots N, \quad \forall (x_1,\ldots,x_N)\in \Om ^N.$$
On peut par la suite se demander si le taux de convergence dans~\eqref{eq:general Lions} peut se quantifier, comme le Th\'eor\`eme~\ref{thm:DF} permet de quantifier la convergence~\eqref{eq:general DiacFreed}. Des r\'esultats dans cette direction sont pr\'esent\'es dans~\cite{HauMis-14}.\hfill\qed
\end{remark}

\subsection{Limite de champ moyen pour une \'energie libre classique}\label{sec:appli HS}\mbox{}\\\vspace{-0.4cm}

Dans cette section nous appliquons le Th\'eor\`eme~\ref{thm:HS} \`a une fonctionnelle d'\'energie libre \`a temp\'erature positive, suivant~\cite{MesSpo-82,CagLioMarPul-92,Kiessling-93,KieSpo-99}. On consid\`ere un domaine $\Om \subset \R ^d$ et la fonctionnelle
\begin{equation}\label{eq:aHS free func}
 \F_N [\mubf] = \int_{X \in \Om ^N} H_N (X) \mubf (X) dX + T \int_{\Om ^N} \mubf(X) \log \mubf(X) dX
\end{equation}
d\'efinie pour toute mesure de probabilit\'e $\mubf \in \PP (\Om ^N)$. Ici la tempr\'erature $T$ est fixe dans la limite $N\to \infty$ et le Hamiltonien $H_N$ est choisi de type champ-moyen:
\begin{equation}\label{eq:aHS hamil}
H_N (X) = \sum_{j=1} ^N V(x_j) + \frac{1}{N-1} \sum_{1\leq i < j \leq N} w (x_i-x_j). 
\end{equation}
Ici $V$ est un potentiel semi-continu inf\'erieurement. Pour se placer dans un cadre compact, on supposera que soit $\Om$ est born\'e soit 
\begin{equation}\label{eq:aHS confine}
V(x) \to \infty \mbox{ quand } |x|\to \infty. 
\end{equation}
Le potentiel d'interaction $w$ sera born\'e inf\'erieurement et semi-continu inf\'erieurement. Pour rester concret on pourra penser \`a $w\in L ^{\infty}$, ou bien \`a un potentiel de type Coulomb r\'epulsif :
\begin{align}\label{eq:aHS Coulomb}
w(x) &= \frac{1}{|x| ^{d-2}} \mbox{ si } d = 3\\
w(x) &= - \log |x| \mbox{ si } d=2\\
w(x) &=  -| x| \mbox{ si } d=1,
\end{align}
omnipr\'esent dans les applications. On supposera \'egalement toujours
\[
w(x) = w (-x) 
\]
et si le domaine n'est pas born\'e on supposera en outre 
\begin{align}\label{eq:aHS confine 2}
w(x-y) &+ V(x) + V(y) \to \infty \mbox{ quand } |x|\to \infty \mbox{ ou } |y|\to \infty \nonumber \\
w(x-y) &+ V(x) + V(y)  \mbox{ est semi-continu inf\'erieurement.}
\end{align}
Nous nous int\'eressons \`a la limite de la mesure de Gibbs minimisant~\eqref{eq:aHS free func} dans $\PP_s (\Om ^N)$:
\begin{equation}\label{eq:aHS Gibbs}
\mubf_N (X) = \frac{1}{\ZN} \exp\left( - \frac{1}{T} H_N (X) \right) dX
\end{equation}
et \`a l'\'energie libre correspondante
\begin{equation}\label{eq:aHS free ener}
F_N = \inf_{\mubf \in \PP (\Om_s ^N)} \F_N [\mubf] = \F_N [\mubf_N] = - T \log \ZN. 
\end{equation}
La fonctionnelle d'\'energie libre se r\'e\'ecrit
\begin{align}\label{eq:aHS free func 2}
\F_N [\mubf] &= N \int_{\Om} V(x) d\mubf ^{(1)} (x) + \frac{N}{2} \iint_{\Om \times \Om} w(x-y) d\mubf ^{(2)} (x,y) + T \int_{\Om ^N} \mubf \log \mubf \nonumber\\
 &= \frac{N}{2} \iint_{\Om \times \Om} \left( w(x-y) + V(x) + V(y) \right) d\mubf ^{(2)} (x,y)  + T \int_{\Om ^N} \mubf \log \mubf. 
\end{align}
en utilisant les marginales 
\begin{equation}\label{eq:aHS marginales}
\mubf ^{(n)} (x_1,\ldots,x_n) = \int_{x_{n+1},\ldots,x_N\in \Om} d \mubf (x_1,\ldots,x_N).
\end{equation}

En ins\'erant un ansatz de forme 
\begin{equation}\label{eq:aHS ansat}
\mubf = \rho ^{\otimes N}, \rho \in \PP (\Om) 
\end{equation}
dans~\eqref{eq:aHS free func} on obtient la fonctionnelle de champ moyen
\begin{multline}\label{eq:aHS MFf}
\MFf [\rho] := N ^{-1} \F_N [\rho ^{\otimes N}] \\
= \int_{\Om} V(x) d\rho(x) + \frac{1}{2} \iint_{\Om \times \Om} w(x-y) d\rho (x)d\rho(y) + T \int_{\Om} \rho \log \rho 
\end{multline}
dont on notera $\MFe$ et $\rhoMF$ le minimum et un minimiseur parmi les mesures de probabilit\'e. Notre but est de justifier l'approximation de champ moyen en d\'emontrant le th\'eor\`eme suivant

\begin{theorem}[\textbf{Limite de champ moyen classique \`a temp\'erature fixe}]\label{thm:aHS}\mbox{}\\
On a 
\begin{equation}\label{eq:aHS result ener}
\frac{F_N}{N} \to \MFe \mbox{ quand } N \to \infty.
\end{equation}
De plus \`a extraction d'une sous-suite pr\`es, on a pour tout $n\in \N$
\begin{equation}\label{eq:aHS result marginal}
\mubf_N ^{(n)} \wto_\ast \int_{\rho \in \MFmin} \rho  ^{\otimes n} dP (\rho).  
\end{equation}
au sens des mesures, o\`u $P$ est une mesure de probabilit\'e sur $\MFmin$, l'ensemble des minimiseurs de $\MFf$.
\end{theorem}

En particulier, si $\MFf$ a un minimiseur unique on obtient 
\[
\mubf_N ^{(n)} \wto_\ast \left(\rhoMF\right) ^{\otimes n}
\]

\begin{proof}[Preuve du Th\'eor\`eme~\ref{thm:aHS}]
On suit essentiellement~\cite{MesSpo-82,CagLioMarPul-92,Kiessling-89,Kiessling-93}. Une borne sup\'erieure sur l'\'energie libre s'obtient ais\'ement en prenant une fonction test de forme $\rho ^{\otimes N}$ et on d\'eduit 
\begin{equation}\label{eq:aHS borne sup}
\frac{F_N}{N}  \leq \MFe. 
\end{equation}
Seule la borne inf\'erieure correspondante demande du travail. On commence par extraire une sous-suite le long de laquelle 
\begin{equation}\label{eq:aHS preuve 1}
\mubf_N ^{(n)} \wto_\ast \mubf ^{(n)} 
\end{equation}
pour tout $n\in \N$, avec $\mubf \in \PP_s (\Om ^{\N})$. On proc\`ede comme indiqu\'e Section~\ref{sec:HS}, en utilisant soit le fait que $\Om$ est compact soit l'hypoth\`ese~\eqref{eq:aHS confine 2} qui implique que la suite est tendue. 

Par semi-continuit\'e inf\'erieure on a imm\'ediatement
\begin{multline}\label{eq:aHS lim ener}
\liminf_{N\to \infty} \frac{1}{2} \iint_{\Om \times \Om} \left( w(x-y) + V(x) + V(y) \right) d\mubf_N ^{(2)} (x,y) 
\\ \geq \frac{1}{2} \iint_{\Om \times \Om} \left( w(x-y) + V(x) + V(y) \right) d\mubf ^{(2)} (x,y).     
\end{multline}
Pour le terme d'entropie on utilise la propri\'et\'e de sous-additivit\'e (cons\'equence de l'in\'egalit\'e de Jensen, voir~\cite{RobRue-67} ou les r\'ef\'erences cit\'ees pr\'ec\'edemment) 
\[
\int_{\Om ^N} \mubf_N \log \mubf_N \geq \left\lfloor \frac{N}{n} \right\rfloor \int_{\Om ^n} \mubf_N ^{(n)} \log \mubf_N ^{(n)} +  \int_{\Om ^{N-n\left\lfloor \frac{N}{n} \right\rfloor}} \mubf_N ^{\left(N-n\left\lfloor \frac{N}{n} \right\rfloor\right)} \log \mubf_N ^{\left(N-n\left\lfloor \frac{N}{n} \right\rfloor\right)} 
\]
o\`u $\lfloor \: . \: \rfloor$ d\'enote la partie enti\`ere. L'in\'egalit\'e de Jensen implique que pour toutes mesures de probabilit\'es $\mu$ et $\nu$, l'entropie relative de $\mu$ par rapport \`a $\nu$ est positive:
\[ 
\int \mu \log \frac{\mu}{\nu} = \int \nu \frac{\mu}{\nu} \log \frac{\mu}{\nu} \geq \left(\int \nu \frac{\mu}{\nu} \right) \log \left( \int \nu \frac{\mu}{\nu} \right) = 0.
\]
On en d\'eduit que pour tout $\nu_0 \in \PP(\Om)$
\begin{align*}
\int_{\Om ^{N-n\left\lfloor \frac{N}{n} \right\rfloor}} \mubf_N ^{\left(N-n\left\lfloor \frac{N}{n} \right\rfloor\right)} \log \mubf_N ^{\left(N-n\left\lfloor \frac{N}{n} \right\rfloor\right)} &= \int_{\Om ^{N-n\left\lfloor \frac{N}{n} \right\rfloor}} \mubf_N ^{\left(N-n\left\lfloor \frac{N}{n} \right\rfloor\right)} \log \left(\frac{\mubf_N ^{\left(N-n\left\lfloor \frac{N}{n} \right\rfloor\right)}}{\nu_0 ^{\otimes \left(N-n\left\lfloor \frac{N}{n} \right\rfloor\right)}}\right) \\
&+ \int_{\Om ^{N-n\left\lfloor \frac{N}{n} \right\rfloor}} \mubf_N ^{\left(N-n\left\lfloor \frac{N}{n} \right\rfloor\right)} \log \nu_0 ^{\otimes \left(N-n\left\lfloor \frac{N}{n} \right\rfloor\right)}\\
&\geq \left(N-n\left\lfloor \frac{N}{n} \right\rfloor\right) \int_{\Om } \mubf_N ^{(1)} \log \nu_0. 
\end{align*}
En choisissant $\nu_0\in\PP(\Om)$ de la forme $\nu_0 = c_0 \exp(-c_1 V)$ il n'est pas difficile de voir que la derni\`ere int\'egrale est born\'ee inf\'erieurement ind\'ependament de $N$ et on obtient alors pour tout $n\in \N$
\begin{equation}\label{eq:aHS lim entr}
\liminf_{N\to \infty} \frac{1}{N} \int_{\Om ^N} \mubf_N \log \mubf_N \geq \frac{1}{n} \int_{\Om ^n} \mubf ^{(n)} \log \mubf ^{(n)}
\end{equation}
par semi-continuit\'e inf\'erieure de (moins) l'entropie.

En rassemblant~\eqref{eq:aHS lim ener} et~\eqref{eq:aHS lim entr} on obtient une borne inf\'erieure en terme d'une fonctionnelle de $\mubf$:
\begin{multline}\label{eq:aHS somme}
\liminf_{N\to \infty} \frac{1}{N} \F_N [\mubf_N] \geq \F [\mubf] :=  \frac{1}{2} \iint_{\Om \times \Om} \left( w(x-y) + V(x) + V(y) \right) d\mubf ^{(2)} (x,y) 
\\ + T \limsup_{n\to \infty} \frac{1}{n} \int_{\Om ^n} \mubf ^{(n)} \log \mubf ^{(n)}.
\end{multline}
Le second terme est souvent appel\'e (moins) l'entropie moyenne de $\mubf\in \PP (\Om ^{\N})$. Il s'agit maintenant d'appliquer le th\'eor\`eme de Hewitt-Savage \`a $\mubf$. Le premier terme de $\F$ est \'evidemment affine en fonction de $\mubf^{(2)}$, ce qui est parfait pour utiliser~\eqref{eq:result HS}, mais on pourrait s'inqui\'eter \`a la vue du second terme qui semble plut\^ot convexe. En fait un argument simple de~\cite{RobRue-67} montre que cette entropie moyenne est bien affine: \'etant donn\'e $\mubf_1,\mubf_2 \in \PP (\Om ^{\N})$ on utilise d'une part la convexit\'e de $x\mapsto x\log x$ et d'autre part la croissance de $x\mapsto \log x$ pour obtenir 
\begin{align*}
\frac{1}{2}\int_{\Om ^n} \mubf_1 ^{(n)} \log \mubf_1 ^{(n)} &+ \frac{1}{2}\int_{\Om ^n} \mubf_2 ^{(n)} \log \mubf_2 ^{(n)}\geq
\int_{\Om ^n} \left( \frac{1}{2}\mubf_1 ^{(n)} + \frac{1}{2} \mubf_2 ^{(n)} \right) \log \left( \frac{1}{2}\mubf_1 ^{(n)} + \frac{1}{2} \mubf_2 ^{(n)} \right) 
\\&\geq  \frac{1}{2}\int_{\Om ^n} \mubf_1 ^{(n)} \log \mubf_1 ^{(n)} + \frac{1}{2}\int_{\Om ^n} \mubf_2 ^{(n)} \log \mubf_2 ^{(n)} 
\\&- \frac{\log (2)}{2} \left( \int_{\Om ^n} \mubf_1 ^{(n)} + \int_{\Om ^n} \mubf_2 ^{(n)}\right) 
\\&= \frac{1}{2}\int_{\Om ^n} \mubf_1 ^{(n)} \log \mubf_1 ^{(n)} + \frac{1}{2}\int_{\Om ^n} \mubf_2 ^{(n)} \log \mubf_2 ^{(n)} - \log (2).
\end{align*}
En divisant par $n$ et passant \`a la limite sup\'erieure on d\'eduit que 
\[
 \PP_s (\Om ^n) \ni \mubf \mapsto \limsup_{n\to \infty} \frac{1}{n} \int_{\Om ^n} \mubf ^{(n)} \log \mubf ^{(n)}
\]
est bien affine lin\'eare et donc que $\F[\mu]$ l'est \'egalement. Il reste \`a utiliser~\eqref{eq:result HS}, ce qui donne une probabilit\'e $P_{\mubf} \in \PP (\PP (\Om ))$ telle que  
\begin{align*}
\liminf_{N\to \infty} \frac{1}{N} \F_N [\mubf_N] &\geq \int_{\rho \in \PP (\Om)} \F [\rho ^{\otimes \infty}] dP_{\mubf}(\rho)
\\& = \int_{\rho \in \PP (\Om)} \MFf [\rho ] dP_{\mubf}(\rho) \geq \MFe.
\end{align*}
Ici on a not\'e $\rho ^{\otimes \infty}$ la probabilit\'e sur $\PP (\Om ^{\N})$ qui a pour $n$-i\`eme marginale $\rho ^{\otimes n}$ pour tout $n$ et on a utilis\'e le fait que $P_{\mubf}$ est d'int\'egrale $1$. Ceci conclut la preuve de~\eqref{eq:aHS result ener} et~\eqref{eq:aHS result marginal} se d\'eduit facilement. Il est en effet clair au vu des in\'egalit\'es pr\'ec\'edentes que $P_{\mubf}$ doit \^etre concentr\'ee sur l'ensemble des minimiseurs de l'\'energie libre de champ moyen.
\end{proof}

\subsection{Estimations quantitatives dans la limite champ moyen/temp\'erature faible}\label{sec:appli DF}\mbox{}\\ \vspace{-0.8cm}

Ici nous allons donner un exemple o\`u les pr\'ecisions apprort\'ees au th\'eor\`eme de Hewitt-Savage par la m\'ethode de Diaconis-Freedman se r\'ev\`elent utiles. Comme mentionn\'e \`a la Remarque~\ref{rem:DiacFreLio}, la construction de la Section~\ref{sec:DF} se comporte assez mal vis-\`a-vis de l'entropie, mais il existe un certain nombre de probl\`emes int\'eressants o\`u il fait sens de consid\'erer une temp\'erature petite dans la limite $N\to \infty$, auquel cas l'entropie joue un faible r\^ole.

Un exemple central est celui des syst\`emes de type ``log-gas''. Il est bien connu que la distribution des valeurs propres de certains ensembles de matrices al\'eatoires est donn\'e par la mesure de Gibbs d'un gaz classique avec interactions logarithmiques. De plus, il se trouve que la limite pertinente pour de grandes matrices est un r\'egime de champ-moyen avec temp\'erature d'ordre $N ^{-1}$. Consid\'erons le Hamiltonien suivant (les hypoth\`eses sur $V$ sont les m\^emes que pr\'ec\'edemment, avec $w = -\log |\:.\:|$):
\begin{equation}\label{eq:aDF Hamil}
H_N (X) = \sum_{j=1} ^N V(x_j) - \frac{1}{N-1} \sum_{1\leq i < j \leq N} \log |x_i-x_j|  
\end{equation}
o\`u $X = (x_1,\ldots,x_N) \in \R ^{dN}$. La mesure de Gibbs correspondante 
\begin{equation}\label{eq:aDF Gibbs}
\mubf_N (X) = \frac{1}{\ZN} \exp \left(-\beta N H_N (X) \right) dX 
\end{equation}
correspond (modulo un changement d'\'echelle d\'ependant de $\beta$) \`a la distribution des valeurs propres d'une matrice al\'eatoire dans les cas suivants:
\begin{itemize}
\item $d=1,\beta=1,2,4$ et $V(x) = \frac{|x| ^2}{2}$. On obtient respectivement les matrices gaussiennes r\'eelles sym\'etriques, complexes hermitiennes, et quatertioniques auto-duales.
\item $d=2,\beta = 2$ et $V(x) = \frac{|x| ^2}{2}$. On obtient alors les matrices gaussiennes complexes sans condition de sym\'etrie, c'est-\`a-dire l'ensemble dit de Ginibre~\cite{Ginibre-65}.
\end{itemize}
Dans le cadre de ces notes nous ne donnerons pas plus de pr\'ecisions sur l'aspect ``matrices al\'eatoires'', et nous nous contenterons de prendre les faits pr\'ec\'edents comme une motivation suffisante pour \'etudier la limite $N\to \infty$ des mesures~\eqref{eq:aDF Gibbs} avec $\beta$ fix\'e, ce qui correspond (comparer avec~\eqref{eq:aHS Gibbs}) \`a prendre $T=\beta ^{-1} N ^{-1}$, soit une temp\'erature tr\`es petite. Pour une introduction aux matrices al\'eatoires et aux log-gas nous renvoyons \`a (parmi de nombreuses sources)~\cite{AndGuiZei-10,Forrester-10,Mehta-04}. Pour des \'etudes pouss\'ees des mesures~\eqref{eq:aDF Gibbs} selon des m\'ethodes assez diff\'erentes de celles pr\'esent\'ees ici on pourra consulter entre autres~\cite{BenGui-97,BenZei-98,BouPasShc-95,BouErdYau-12,BouErdYau-14,BouErdYau-12,ChaGozZit-13,RouSer-13,SanSer-12,SanSer-13}.

Dans le cas $d=2$, les mesures de forme~\eqref{eq:aDF Gibbs} ont aussi une application naturelle \`a l'\'etude de certaines fonctions d'onde quantique apparaissant dans le cadre de l'effet Hall fractionnaire (voir~\cite{RouSerYng-13b,RouSerYng-13,RouYng-14} et les r\'ef\'erences cit\'ees). L\`a aussi il fait sens de consid\'erer $\beta$ comme \'etant fixe.

La singularit\'e en $0$ du logarithme va poser des difficult\'es dans la preuve, comme indiqu\'e \`a la Remarque~\ref{rem:DiacFreLio}, mais on peut les contourner relativement ais\'ement, au contraire de celle li\'ee \`a l'entropie. La m\'ethode que nous pr\'esenterons n'est pas limit\'ee au cas des log-gas et peut \^etre r\'eemploy\'ee dans des contextes vari\'es. 

\medskip

A nouveau,~\eqref{eq:aDF Gibbs} minimise une \'energie libre 
\begin{equation}\label{eq:aDF free func}
 \F_N [\mubf] = \int_{X \in \Om ^N} H_N (X) \mubf (X) dX + \frac{1}{\beta N} \int_{\Om ^N} \mubf(X) \log \mubf(X) dX 
\end{equation}
dont nous noterons $F_N$ le minimum. L'objet limite naturel est cette fois une \'energie libre o\`u le terme d'entropie est n\'eglig\'e\footnote{Une \'energie donc, d'o\`u la notation.}:
\begin{equation}\label{eq:aDF MFf}
\MFfo [\rho] := \int_{\R ^d} V d\rho -\frac{1}{2}\iint_{\R ^d \times \R ^d } \log|x-y| d\rho (x) d\rho (y),
\end{equation}
obtenue en ins\'erant l'ansatz $\rho ^{\otimes N}$ dans~\eqref{eq:aDF free func} et n\'egligeant le terme d'entropie qui est manifestement d'ordre inf\'erieur pour $\beta$ fixe. On notera $\MFeo$ et $\rhoMF$ respectivement l'\'energie minimum et le minimiseur (unique dans ce cas par stricte convexit\'e de la fonctionnelle). Il est bien connu (voir les r\'ef\'erences sus-cit\'ees et \'egalement~\cite{KieSpo-99}) que 
\begin{equation}\label{eq:aDF ener lim}
N ^{-1} F_N = -\frac{1}{\beta} \log \ZN \to \MFeo \mbox{ quand } N\to \infty
\end{equation}
et 
\begin{equation}\label{eq:aDF Gibbs lim}
\mubf_N ^{(n)} \wto_* \left(\rhoMF \right) ^{\otimes n}.
\end{equation}
Nous allons reprouver~\eqref{eq:aDF Gibbs lim} et donner une version quantitative de~\eqref{eq:aDF ener lim} en nous inspirant de~\cite{RouYng-14}:

\begin{theorem}[\textbf{Estimation de l'\'energie libre d'un log-gas}]\label{thm:aDF}\mbox{}\\
Pour tout $\beta \in \R$, on a 
\begin{equation}\label{eq:aDF result}
\MFeo - C N ^{-1} \left(\beta ^{-1} + \log N +1 \right) \leq  - \frac{1}{\beta} \log \ZN \leq \MFeo + C \beta ^{-1}N ^{-1}. 
\end{equation}
\end{theorem}

Des estimations fines de la fonction de partition $\ZN$ d'un log-gas semblent n'\^etre disponibles dans la litt\'erature que depuis assez r\'ecemment~\cite{RouSer-13,SanSer-12,SanSer-13}. Entre autres, les r\'ef\'erences pr\'ec\'edentes indiquent que la correction \`a $\MFeo$ est  exactement d'ordre $N ^{-1}\log N$. 

\begin{proof}[Preuve du Th\'eor\`eme~\ref{thm:aDF}]
Nous n'\'elaborerons pas sur la borne sup\'erieure, qui s'obtient en prenant l'ansatz factoris\'e habituel et en estimant son entropie. Pour la borne inf\'erieure, nous allons utiliser le Th\'eor\`eme~\ref{thm:DF}. Pour cela il nous faut d'abord borner crudement l'entropie: par positiv\'e de l'entropie relative (in\'egalit\'e de Jensen)
\[
\int_{\R ^{dN}} \mubf \log \frac{\mubf}{\nubf} \geq 0 \mbox{ pour tout } \mubf,\nubf \in \PP (\Om ^N) 
\]
on peut \'ecrire, en utilisant la mesure de probabilit\'e 
$$\nubf_N = \left(c_0 \exp\left( -V(|x|)\right)\right) ^{\otimes N},$$
la borne inf\'erieure suivante 
\begin{equation}\label{eq:aDF rel entr}
\int_{\R ^{dN}} \mubf_N \log \mubf_N \geq  \int_{\R ^{dN}} \mubf_N \log \nubf_N = -N \int_{\R ^d} V d \mubf_N ^{(1)} - N \log c_0.
\end{equation}
Pour obtenir une borne inf\'erieure sur l'\'energie il nous faut d'abord r\'egulariser le potentiel d'interaction: soit $\alpha >0$ un petit param\`etre que nous optimiserons plus tard et 
\begin{equation}\label{eq:log alpha 2}
- \logal |z| = \begin{cases}
             \displaystyle -\log\alpha+\half \left(1-\frac{|z|^2}{\alpha^2}\right) &\mbox{ si } |z|\leq \alpha \\
             - \log |z| &\mbox{ si } |z|\geq \alpha.
            \end{cases}
\end{equation}
Clairement $-\log_\alpha |z|\leq -\log |z|$ est r\'egulier en $0$. De plus, on a 
\begin{equation}\label{eq:log alpha 3}
- \frac d{d\alpha}\logal |z| = \begin{cases}
             \displaystyle -\frac 1\alpha+\frac {|z|^2}{\alpha^3} &\mbox{ si } |z|\leq \alpha\\
             0 &\mbox{ si } |z|\geq \alpha.
            \end{cases}
\end{equation}
On utilise la minoration $-\log_\alpha |z|\leq -\log |z|$ pour obtenir 
\[
\int_{X \in \Om ^N} H_N (X) \mubf (X) dX \geq  N \int_{\R ^d} V d\mubf_N ^{(1)} - \frac{N}{2} \iint_{\R ^d \times \R ^d} \logal |x-y| d\mubf_N ^{(2)} (x,y) 
\]
et nous sommes maintenant en position d'appliquer le Th\'eor\`eme~\ref{thm:DF}, et plus pr\'ecis\'ement les formules explicites~\eqref{eq:marginals DF}:   
\begin{multline}\label{eq:aDF use DF}
 \int_{X \in \Om ^N} H_N (X) \mubf_N (X) dX \geq 
 \\ N \int_{\R ^d} V d\mut_N ^{(1)} - \frac{N ^2}{2(N-1)} \iint_{\R ^d \times \R ^d} \logal |x-y| d\mubf_N ^{(2)} (x,y) + C \logal(0) 
\end{multline}
et en combinant~\eqref{eq:aDF rel entr},~\eqref{eq:Pnu} et en rappelant que la temp\'erature est \'egale \`a $\left(\beta N \right) ^{-1}$ on obtient
\begin{multline}\label{eq:aDF final}
N ^{-1} \F_N [\mubf_N] \geq \int_{\rho \in \PP(\R^d)} \MFfal [\rho] dP_{\mubf_N} (\rho) 
\\ +C N ^{-1} \left( \logal(0) - \beta ^{-1}\right) \geq \MFeal - C N ^{-1} \log (\alpha) - C (\beta N) ^{-1}  
\end{multline}
o\`u $\MFeal$ est le minimum (parmi les mesures de probabilit\'e) de la fonctionnelle modifi\'ee
\[
\MFfal [\rho] := \int_{\R ^d} V (1-\beta ^{-1} N ^{-1}) d\rho -\frac{N}{2(N-1)}\iint_{\R ^d \times \R ^d } \logal |x-y| d\rho (x) d\rho (y).
\]
En exploitant l'\'equation variationnelle associ\'ee,~\eqref{eq:log alpha 3} et le principe de Feynman-Hellmann, il n'est pas difficile de montrer que pour $\alpha$ assez petit
$$ \left| \MFeal - \MFeo \right| \leq C \alpha ^d + C N ^{-1}$$
et donc on conclut 
\[
N ^{-1} \F_N [\mubf_N] \geq \MFeo - C N ^{-1} \log (\alpha) - C \beta N ^{-1} - C N ^{-1} - C \alpha ^d - C 
\]
ce qui donne la borne inf\'erieure d\'esir\'ee en optimisant par rapport \`a $\alpha$ (prendre $\alpha = N ^{-1/d}$).
\end{proof}

\begin{remark}[Extensions possibles]\mbox{}\label{rem:log gas}\\\vspace{-0.4cm}
\begin{enumerate}
 \item On peut \'egalement prouver des version quantitatives de~\eqref{eq:aDF Gibbs lim} en suivant essentiellement la m\'ethode de preuve ci-dessus. Nous n'\'elaborerons pas plus sur ce point pour lequel nous renvoyons \`a la m\'ethode utilis\'ee dans~\cite{RouYng-14} (bas\'ee sur une impl\'ementation quantitative du principe de Feynman-Hellmann). 
 \item Un autre cas que nous pourrions traiter avec la m\'ethode ci-dessus est celui des ensembles de matrices gaussiennes unitaires, orthogonaux et symplectiques introduits par Dyson~\cite{Dyson-62a,Dyson-62b,Dyson-62c}. Dans ce cas, $\R ^d$ est remplac\'e par le cercle unit\'e, $\beta = 1,2,4$, $V\equiv 0$ dans~\eqref{eq:aDF Hamil} et le terme d'interaction n'est pas divis\'e par $N-1$. La m\'ethode s'applique malgr\'e ce dernier fait car le r\^ole du facteur $(N-1) ^{-1}$ est de rendre les deux termes de~\eqref{eq:aDF Hamil} du m\^eme ordre de grandeur. Quand $V\equiv 0$ et que les particules vivent sur un domaine fixe, on a pas ce souci et le probl\`eme limite reste bien d\'efini. \hfill \qed
\end{enumerate}
\end{remark}

\newpage

\section{\textbf{Th\'eor\`eme de de Finetti quantique et th\'eorie de Hartree}}\label{sec:quant}

Nous rentrons maintenant dans le coeur du sujet du cours, les limites de champ moyen pour de grands syst\`emes bosoniques. Nous pr\'esenterons d'abord la d\'erivation du fondamental de la th\'eorie de Hartree pour des particules confin\'ees. Dans un tel cas, correspondant par exemple \`a des particules vivant dans un domaine born\'e, le r\'esultat est alors une cons\'equence relativement imm\'ediate du th\'eor\`eme de de Finetti prouv\'e par St\o{}rmer et Hudson-Moody~\cite{Stormer-69,HudMoo-75} qui d\'ecrit toutes les limites \emph{fortes} (au sens de la norme $\gS ^1$) des matrices de densit\'e r\'eduites d'un grand syst\`eme bosonique. 

Nous passerons ensuite au cas plus complexe de syst\`emes non confin\'es. Le cas g\'en\'eral sera trait\'e plus tard et dans cette section nous supposerons que le potentiel d'interaction n'a pas d'\'etats li\'es (un exemple est un potentiel purement r\'epulsif). Il est alors suffisant de disposer d'un th\'eor\`eme \`a la de Finetti d\'ecrivant toutes les limites \emph{faibles} (au sens de la topologie faible-$\ast$ sur $\gS 
^1$) des matrices de densit\'es r\'eduites d'un grand syst\`eme bosonique. 

Le th\'eor\`eme de de Finetti quantique Finetti faible (introduit dans~\cite{LewNamRou-13}) implique le th\'eor\`eme de de Finetti quantique fort et il se trouve que les deux r\'esultats peuvent se d\'eduire d'un th\'eor\`eme encore plus g\'en\'eral apparaissant dans~\cite{Stormer-69,HudMoo-75}. Le parti-pris de ces notes est de ne pas suivre cette approche, mais plut\^ot celle (par ailleurs plus constructive) de~\cite{LewNamRou-13}. Ceci sera discut\'e plus en d\'etail \`a la Section~\ref{sec:rel deF} qui annonce aussi la d\'emarche qui nous guidera dans les chapitres suivants.

\medskip

Pour simplifier l'exposition, on se concentrera sur le cas de particules quantiques non relativistes, en l'absence de champ magn\'etique,  pour lesquelles le Hamiltonien a la forme g\'en\'erale
\begin{equation}\label{eq:quant hamil}
H_N = \sum_{j=1} ^N T_j + \frac{1}{N-1}\sum_{1\leq i < j \leq N} w(x_i-x_j), 
\end{equation}
agissant sur l'espace de Hilbert $\gH_s ^N = \bigotimes_s ^N \gH$, i.e. le produit tensoriel sym\'etrique de $N$ copies de $\gH$ o\`u $\gH$ sera l'espace $L ^2(\Om)$ pour $\Om \subset \R ^d$. L'op\'erateur $T$ est un op\'erateur de Schr\"odinger 
\begin{equation}\label{eq:op Schro}
T = - \Delta + V 
\end{equation}
avec $V:\Om \mapsto \R$ et $T_j$ est d\'efini comme agissant sur la $j$-\`eme variable: 
\[
T_j \psi_1 \otimes \ldots \otimes \psi_N = \psi_1 \otimes \ldots \otimes T_j \psi_j \otimes \ldots \otimes \psi_N. 
\]
On suppose que $T$ est auto-adjoint et born\'e inf\'erieurement, et que le potentiel d'interaction $w:\R \mapsto \R$ est relativement born\'e par rapport \`a $T$ (au sens des op\'erateurs):
\begin{equation}\label{eq:T controls w}
 -\beta_-(T_1+ T_2)-C\leq w(x_1-x_2) \leq \beta_+ (T_1+T_2) + C,
\end{equation}
sym\'etrique
\[
w(-x) = w(x), 
\]
et d\'ecroissant \`a l'infini
\begin{equation}\label{eq:decrease w}
w\in L^p (\Om) + L ^{\infty} (\Om), \max(1,d/2) < p < \infty \to 0, w(x) \to 0 \mbox{ quand }  |x|\to\infty.
\end{equation}
Ceci assure que $H_N$ est auto-adjoint et born\'e inf\'erieurement. On fera un abus de notation en notant $w$ l'op\'erateur de multiplication par $w(x_1-x_2)$ sur $L ^2 (\Om)$.

Notre but est de d\'ecrire le fondamental de~\eqref{eq:quant hamil}, c'est-\`a-dire l'\'etat r\'ealisant\footnote{Le fondamental peut ne pas exister, auquel cas on pense \`a une suite d'\'etats r\'ealisant asymptotiquement l'infimum.} 
\begin{equation}\label{eq:quant ground state}
E(N) = \inf \sigma_{\gH ^N} H_N  = \inf_{\Psi \in \gH ^N, \norm{\Psi} = 1} \left\langle \Psi, H_N \Psi \right\rangle_{\gH ^N}.
\end{equation}
Dans le r\'egime de champ moyen qui nous occupe, on s'attend \`a ce qu'un fondamental satisfasse
\begin{equation}\label{eq:quant factor form}
\Psi_N \approx u ^{\otimes N} \mbox{ quand } N\to \infty
\end{equation}
en un sens \`a pr\'eciser, ce qui m\`ene naturellement \`a la fonctionnelle de Hartree 
\begin{align}\label{eq:Hartree func}
\EH [u] &=  N ^{-1} \left\langle u ^{\otimes N}, H_N u ^{\otimes N} \right\rangle_{\gH ^N} = \langle u,T u \rangle_{\gH} +  \frac{1}{2} \langle u\otimes u,w \:u \otimes u \rangle_{\gH_s ^2} \nonumber\\
&= \int_{\Om} |\nabla u | ^2 + V |u| ^2 + \frac{1}{2} \iint_{\Om\times \Om} |u(x)| ^2 w(x-y) |u(y)| ^2 dxdy. 
\end{align}
On notera $\eH$ et $\uH$ le minimum et un minimiseur de $\EH$ respectivement. On a bien s\^ur, par le principe variationnel
\begin{equation}\label{eq:up bound Hartree}
\frac{E(N)}{N} \leq \eH 
\end{equation}
et on s'attend \`a pouvoir d\'emontrer la borne inf\'erieure correspondante pour obtenir
\begin{equation}\label{eq:Hartree lim formel}
 \frac{E(N)}{N} \to \eH \mbox{ quand } N\to \infty.
\end{equation}

\begin{remark}[G\'en\'eralisations]\label{rem:energie cin}\mbox{}\\
Toutes les id\'ees principales peuvent \^etre introduites dans le cadre pr\'ec\'edent et nous renvoyons \`a~\cite{LewNamRou-13} pour une discussion des g\'en\'eralisations possibles. On pourra m\^eme penser au cas o\`u $V$ et $w$ sont tous deux r\'eguliers \`a support compact si on d\'esire comprendre la m\'ethode dans le cas le plus simple possible. 

Une g\'en\'eralisation tr\`es int\'eressante consiste en la substitution du Laplacien dans~\eqref{eq:op Schro} par un op\'erateur d'\'energie cin\'etique relativiste et/ou incluant un champ magn\'etique, comme d\'ecrit \`a la Section~\ref{sec:forma quant}. Il faut alors adapter les hypoth\`eses~\eqref{eq:T controls w} et~\eqref{eq:decrease w} mais le message reste le m\^eme: l'approche fonctionne tant que les mod\`eles de d\'epart et d'arriv\'ee sont bien d\'efinis.

Une autre g\'en\'eralisation possible est l'inclusion d'interactions \`a plus de deux particules pour obtenir des fonctionnelles avec des non-lin\'earit\'es d'ordre plus \'elev\'e \`a la limite. Il est bien s\^ur n\'ecessaire de se placer dans une limite de champ moyen en ajoutant par exemple au Hamiltonien un potentiel de forme 
$$ \lambda_N\sum_{1\leq i < j,k \leq N} w (x_i-x_j,x_i-x_k)$$
avec $\lambda_N \propto N^{-2}$ quand $N\to \infty$. Il est aussi possible de prendre en compte une interaction plus g\'en\'erale que la multiplication par un potentiel, sous des hypoth\`eses du m\^eme genre que~\eqref{eq:T controls w}.\hfill \qed
\end{remark}

\subsection{Syst\`emes confin\'es et Th\'eor\`eme de de Finetti fort}\label{sec:Hartree deF fort}\mbox{}\\\vspace{-0.4cm}

Par ``syst\`eme confin\'e'' on entend un cadre compact. On pourra choisir l'une des deux hypoth\`eses suivantes:
\begin{equation}\label{eq:hartree confine 1}
\Om \subset \R ^d  \mbox{ est un domaine born\'e et } 
\end{equation}
ou bien
\begin{equation}\label{eq:hartree confine 2}
\Om = \R ^d  \mbox{ et } V(x)\to \infty \mbox{ quand } |x|\to \infty
\end{equation}
avec $V$ le potentiel apparaissant dans~\eqref{eq:op Schro}. On supposera en outre 
\[
 V \in L ^p_{\rm loc} (\Om), \max(1,d/2) < p \leq \infty. 
\]
Dans ces deux cas il est bien connu que 
\begin{equation}\label{eq:resol comp}
T = -\Delta + V \mbox{ est \`a r\'esolvante compact} 
\end{equation}
ce qui permet d'obtenir tr\`es facilement de la convergence forte pour les matrices de densit\'e r\'eduites d'un fondamental de~\eqref{eq:quant hamil}. On peut alors r\'eellement penser \`a l'objet limite obtenu en prenant la limite $N\to \infty$ comme un \'etat quantique \`a nombre infini de particules. En s'inspirant du cas classique discut\'e au Chapitre~\ref{sec:class} la d\'efinition naturelle est la suivante:

\begin{definition}[\textbf{Etat bosonique \`a nombre infini de particules}]\label{def:etat infini}\mbox{}\\
Soit $\gH$ un espace de Hilbert s\'eparable et pour $n\in \N$, $\gH_s ^n$ l'espace bosonique \`a $n$ particules correspondant. On appelle \emph{\'etat bosonique} \`a nombre infini de particules une suite $(\gamma ^{(n)})_{n\in \N}$ d'op\'erateurs \`a trace satisfaisant
\begin{itemize}
\item $\gamma ^{(n)}$ est un \'etat bosonique \`a $n$ particules : $\gamma ^{(n)}\in \gS ^1 (\gH_s ^n)$ est auto-adjoint positif~et 
\begin{equation}\label{eq:defi inf trace}
\tr_{\gH_s ^n}[\gamma ^{(n)}] = 1. 
\end{equation}
\item la suite $(\gamma ^{(n)})_{n\in \N}$ est consistante:
\begin{equation}\label{eq:defi inf consistance}
\tr_{n+1} [\gamma ^{(n+1)}] = \gamma ^{(n)} 
\end{equation}
o\`u $\tr_{n+1}$ est la trace partielle par rapport \`a la derni\`ere variable sur $\gH ^{n+1}$.
\end{itemize}
\hfill \qed
\end{definition}

La propri\'et\'e cl\'e ici est la consistance. C'est elle qui assure que la suite de matrices densit\'es \`a $n$-particules $(\gamma ^{(n)})_{n\in \N}$ d\'ecrit bien un \'etat physique. Notons que $\gamma ^{(0)}$ est simplement un nombre r\'eel et que la consistance assure que $\tr_{\gH ^n}[\gamma ^{(n)}] = 1$ pour tout $n$ d\`es que $\gamma ^{(0)}=1$. 

Un cas particulier d'\'etat sym\'etrique est un \'etat produit:

\begin{definition}[\textbf{Etat produit \`a nombre infini de particules}]\mbox{}\label{def:etat produit}\\
On appelle \emph{\'etat produit} \`a nombre infini de particules une suite d'op\'erateurs \`a trace $\gamma ^{(n)}\in \gS ^1 (\gH_s ^n)$ avec
\begin{equation}\label{eq:etat produit}
\gamma ^{(n)} = \gamma ^{\otimes n},  
\end{equation}
pour tout $n\geq 0$ o\`u $\gamma$ est un \'etat \`a une particule. Un \'etat produit bosonique est n\'ecessairement de la forme 
\begin{equation}\label{eq:etat produit bosonique}
\gamma ^{(n)} = |u ^{\otimes n}\rangle \langle u ^{\otimes n} | = \left(|u \rangle \langle u | \right) ^{\otimes n},  
\end{equation}
avec $u\in S\gH$.\hfill\qed
\end{definition}

Le fait que les \'etats produits bosoniques soient tous de la forme~\eqref{eq:etat produit bosonique} provient de l'observation que si $\gamma \in \gS ^1 (\gH)$ n'est pas pur (i.e. n'est pas un projecteur), $\gamma ^{\otimes 2}$ ne peut pas avoir la sym\'etrie bosonique, voir par exemple~\cite{HudMoo-75}.

Le th\'eor\`eme de de Finetti fort est l'outil appropri\'e pour d\'ecrire ces objets et sp\'ecifier le lien entre les deux d\'efinitions pr\'ec\'edentes. Sous la forme que nous citons il s'agit d'un r\'esultat de Hudson et Moody~\cite{HudMoo-75}:

\begin{theorem}[\bf De Finetti quantique fort]\label{thm:DeFinetti fort}\mbox{}\\
Soit $\gH$ un espace de Hilbert s\'eparable et $(\gamma^{(n)})_{n\in \N}$ un \'etat bosonique \`a nombre infini de particules sur $\gH$. Il existe une unique mesure de probabilit\'e $\mu \in \PP (S\gH)$ sur la sph\`ere $S\gH = \left\{ u \in\gH, \norm{u} = 1\right\}$ de $\gH$, invariante par l'action\footnote{C'est-\`a-dire par la mutliplication par une phase constante $e ^{i\theta}, \theta \in \R$.} de $S ^1$, telle que
\begin{equation}
\gamma^{(n)}=\int_{S\gH}|u^{\otimes n}\rangle\langle u^{\otimes n}| \, d\mu(u)
\label{eq:melange}
\end{equation}
pour tout $n\geq0$.
\end{theorem}

Autrement dit, \textbf{tout \'etat bosonique \`a nombre infini de particules est une combinaison convexe d'\'etats produits bosoniques}. Pour en d\'eduire la validit\'e de l'approximation de Hartree au niveau du fondamental de~\eqref{eq:quant hamil}, il suffit alors de montrer que le probl\`eme limite est pos\'e sur l'ensemble des \'etats \`a nombre infini de particules, ce qui est relativement ais\'e dans un cadre compact. Nous allons prouver le r\'esultat suivant 

\begin{theorem}[\textbf{D\'erivation de la th\'eorie de Hartree pour des bosons confin\'es}]\label{thm:confined}\mbox{}\\
Sous les hypoth\`eses pr\'ec\'edentes, en particulier~\eqref{eq:hartree confine 1} ou~\eqref{eq:hartree confine 2}
$$\lim_{N\to\ii}\frac{E(N)}{N}=\eH.$$
Soit $\Psi_N$ un fondamental de $H_N$ r\'ealisant l'infimum~\eqref{eq:quant ground state} et 
$$ \gamma_N ^{(n)} := \tr_{n+1 \to N} \left[ |\Psi_N\rangle\langle \Psi_N|\right]$$
sa $n$-i\`eme matrice de densit\'e r\'eduite. Il existe une mesure de proabilit\'e $\mu$ sur $\cM_{\rm H}$ l'ensemble des minimiseurs de $\EH$ (modulo une phase), telle que, le long d'une sous-suite et pour tout $n\in \N$ 
\begin{equation}\label{eq:def fort result}
\lim_{N\to\ii} \gamma^{(n)}_{N}=\int_{\cM_{\rm H}} d\mu(u)\;|u^{\otimes n}\rangle\langle u^{\otimes n}|
\end{equation}
fortement dans la norme de $\gS ^1 (\gH ^n).$ En particulier, si $\eH$ a un minimiseur unique (modulo une phase constante), alors pour toute la suite
\begin{equation}
\lim_{N\to\ii} \gamma^{(n)}_N=|\uH^{\otimes n}\rangle\langle \uH^{\otimes n}|.
\label{eq:BEC-confined} 
\end{equation}
\end{theorem}

Les id\'ees de la preuve sont essentiellement contenues dans~\cite{FanSpoVer-80,PetRagVer-89,RagWer-89}, appliqu\'ees \`a un contexte quelque peu diff\'erent. Nous suivrons les clarifications fournies par~\cite[Section~3]{LewNamRou-13}. 

\begin{proof}[Preuve du Th\'eor\`eme~\ref{thm:confined}]
Il nous faut prouver la borne inf\'erieure correspondant \`a~\eqref{eq:up bound Hartree}.
Comme annonc\'e plusieurs fois on commence par \'ecrire
\begin{align}\label{eq:afort ener matrices}
\frac{E(N)}{N} &=  \frac1N \left\langle \Psi_N, H_N \Psi_N \right\rangle_{\gH ^N} = \tr_{\gH} [T \gamma_N ^{(1)}] + \frac12 \tr_{\gH_s ^2} [w\gamma_N ^{(2)}]
\nonumber 
\\ &= \frac12 \tr_{\gH_s ^2} \left[ \left(T_1 + T_2 + w \right)\gamma_N ^{(2)}\right]
\end{align}
et il s'agit de d\'ecrire la limite des matrices de densit\'e r\'eduite $\gamma_N ^{(1)}$ et $\gamma_N ^{(2)}$. Vu que les suites $(\gamma ^{(n)}_N)$ sont par d\'efinition born\'ees dans $\gS ^1$, modulo un proc\'ed\'e d'extraction diagonale, on peut supposer que pour tout $n\in \N$
\[
 \gamma_N ^{(n)}\wto_\ast \gamma ^{(n)} \in \gS^1 (\gH_s ^n)
\]
faiblement-$\ast$ dans $\gS ^1 (\gH ^n)$, c'est-\`a-dire que pour tout op\'erateur compact $K_n$ sur $\gH ^n$ on a 
$$ \tr _{\gH ^n} [\gamma_N ^{(n)} K_n] \to \tr _{\gH ^n} [\gamma ^{(n)} K_n].$$ 
Nous allons montrer que la limite est en fait forte. Pour cela, il suffit (voir~\cite{dellAntonio-67,Robinson-70} ou l'Addendum H de~\cite{Simon-79}) de montrer que 
\begin{equation}\label{eq:afort masse}
\tr _{\gH_s ^n} [\gamma ^{(n)} ] = \tr _{\gH_s ^n} [\gamma_N ^{(n)} ] = 1
\end{equation}
c'est-\`a-dire qu'aucune masse n'est perdue \`a la limite. On commence par remarquer que  $\tr_{\gH} [T\gamma_N ^{(1)}]$ est uniform\'ement born\'ee et que donc, quitte \`a r\'eextraire on a 
$$ \left(T  + C_0 \right) ^{1/2} \gamma_N ^{(1)} \left(T  + C_0 \right) ^{1/2} \wto_\ast \left(T  + C_0 \right) ^{1/2} \gamma ^{(1)} \left(T  + C_0 \right) ^{1/2}$$ 
pour une certaine constante $C_0$. Cons\'equement
\begin{multline*}
1= \tr _{\gH } [\gamma_N ^{(1)} ] = \tr _{\gH } \left[ \left(T  + C_0 \right) ^{-1} \left(T  + C_0 \right) ^{1/2} \gamma_N ^{(1)}\left(T  + C_0 \right) ^{1/2}  \right] \\ \to \tr _{\gH} \left[ \left(T  + C_0 \right) ^{-1} \left(T  + C_0 \right) ^{1/2} \gamma ^{(1)}\left(T  + C_0 \right) ^{1/2}  \right] = \tr _{\gH } [\gamma ^{(1)} ]
\end{multline*}
puisque $\left(T  + C_0 \right) ^{-1}$ est par l'hypoth\`ese~\eqref{eq:resol comp} un op\'erateur compact. On obtient~\eqref{eq:afort masse} de la m\^eme mani\`ere en notant que 
\[
\tr_{\gH} [T\gamma_N ^{(1)}] = \frac{1}{n}\tr_{\gH ^n} \left[ \sum_{j=1} ^n T_j \gamma_N ^{(n)} \right] 
\]
est born\'e unfiorm\'ement en $N$ et que $\sum_{j=1} ^n T_j$ est \'egalement \`a r\'esolvante compacte, ce qui permet un raisonnement similaire.

On a donc pour tout $n\in \N$
\[
 \gamma_N ^{(n)}\to \gamma ^{(n)}
\]
fortement en norme de trace et en particulier, pour tout op\'erateur born\'e $B_n$ sur $\gH ^n$
$$ \tr _{\gH ^n} [\gamma_N ^{(n)} B_n] \to \tr _{\gH ^n} [\gamma ^{(n)} B_n].$$ 
En testant cette convergence avec $B_{n+1} = B_n \otimes \one$ on d\'eduit 
$$ \tr_{n+1} [\gamma ^{(n+1)} ] = \gamma ^{(n)}$$
et donc que la suite $(\gamma ^{(n)})_{n\in \N}$ d\'ecrit un \'etat bosonique \`a nombre infini de particules au sens de la D\'efinition~\ref{def:etat infini}. On peut lui appliquer le Th\'eor\`eme~\ref{thm:DeFinetti fort}, ce qui donne une mesure $\mu \in \PP (S\gH)$. Au vu de l'hypoth\`ese~\eqref{eq:T controls w}, l'op\'erateur $T_1+T_2 + w$ est born\'e inf\'erieurement sur $\gH ^2$, disons par $2 C_T$. Comme $\tr_{\gH ^2} \gamma ^{(2)} = 1$ on peut \'ecrire 
\begin{align*}
\liminf_{N\to \infty}\frac12 \tr_{\gH ^2} \left[ \left(T_1 + T_2 + w \right)\gamma_N ^{(2)}\right] &=  \liminf_{N\to \infty} \frac12 \tr_{\gH ^2} \left[ \left(T_1 + T_2 + w -2C_T \right)\gamma_N ^{(2)}\right] + C_T\\
& \geq \frac12 \tr_{\gH ^2} \left[ \left(T_1 + T_2 + w -2C_T \right)\gamma ^{(2)}\right] + C_T 
\\&= \frac12 \tr_{\gH ^2} \left[ \left(T_1 + T_2 + w  \right)\gamma ^{(2)}\right]
\end{align*}
en utilisant le lemme de Fatou pour les op\'erateurs positifs. On a donc, par lin\'earit\'e de l'\'energie en fonction de $\gamma ^{(2)}$ et~\eqref{eq:melange} 
\[
\liminf_{N\to \infty} \frac{E(N)}{N} \geq \int_{u\in S\gH} \frac12 \tr_{\gH ^2} \left[ \left(T_1 + T_2 + w  \right)|u ^{\otimes 2}\rangle \langle u ^{\otimes 2} |\right] d\mu(u) = \int_{u\in S\gH} \EH [u] d\mu(u) \geq \eH
\]
ce qui est la borne inf\'erieure souhait\'ee. Les autres r\'esultats du Th\'eor\`eme suivent comme d'habitude en notant qu'il doit y avoir \'egalit\'e dans toutes les in\'egalit\'es pr\'ec\'edentes et donc que $\mu$ ne charge que des minimiseurs de l'\'energie de Hartree
\end{proof}

On voit bien dans la preuve pr\'ec\'edente que c'est la structure des \'etats bosoniques \`a nombre infini de particules qui joue le r\^ole cl\'e. Le Hamiltonien lui-m\^eme pourrait \^etre choisi de mani\`ere tr\`es abstraite du moment qu'il comporte un m\'ecanisme de confinement permettant d'obtenir des limite fortes. Divers exemples sont mentionn\'es dans~\cite[Section~3]{LewNamRou-13}.

\subsection{Syst\`emes sans \'etats li\'es et Th\'eor\`eme de de Finetti faible}\label{sec:Hartree deF faible}\mbox{}\\\vspace{-0.4cm}

A la section pr\'ec\'edente nous avons utilis\'e fortement le fait que le syst\`eme \'etait confin\'e au sens de~\eqref{eq:hartree confine 1}-\eqref{eq:hartree confine 2}. Ces hypoth\`eses sont suffisantes pour comprendre de nombreux cas physiques, mais il est tr\`es d\'esirable de pouvoir s'en passer. On est alors amen\'e \`a \'etudier le cas o\`u la convergence des matrices densit\'e n'est pas mieux que faible-$\ast$, et \`a d\'ecrire le plus exhaustivement possible les diff\'erents sc\'enarii possibles, dans l'esprit du principe de concentration-compacit\'e. Une premi\`ere \'etape, avant de se poser la question de la fa\c{c}on dont la compacit\'e peut \^etre perdue consiste \`a d\'ecrire les limites faibles elles-m\^emes. Il se trouve qu'on conserve une description tr\`es satisfaisante \`a la limite. En fait, on ne pouvait esp\'erer mieux que le th\'eor\`eme suivant, d\'emontr\'e dans~\cite{LewNamRou-13}:

\begin{theorem}[\bf De Finetti quantique faible]\label{thm:DeFinetti faible}\mbox{}\\
Soit $\gH$ un espace de Hilbert s\'eparable et $(\Gamma_N)_{N\in \N}$ une suite d'\'etat bosoniques avec $\Gamma_N \in \gS ^1 (\gH_s ^N)$. On suppose que pour tout $n\in\N$ 
\begin{equation}\label{eq:def faible convergence}
\Gamma_N ^{(n)} \wto_\ast \gamma ^{(n)}
\end{equation}
dans $\gS ^1 (\gH_s ^n)$. Il existe alors une unique mesure de probabilit\'e $\mu \in \PP (B\gH)$ sur la boule unit\'e $B\gH = \left\{ u \in\gH, \norm{u} \leq 1\right\}$ de $\gH$, invariante par l'action de $S ^1$, telle que
\begin{equation}
\gamma^{(n)}=\int_{B\gH}|u^{\otimes n}\rangle\langle u^{\otimes n}| \, d\mu(u)
\label{eq:melange faible}
\end{equation}
pour tout $n\geq0$.
\end{theorem}

\begin{remark}[Sur le th\'eor\`eme de de Finetti quantique faible]\label{rem:weak deF}\mbox{}\\\vspace{-0.4cm}
\begin{enumerate}
\item L'hypoth\`ese~\eqref{eq:def faible convergence} n'en est pas vraiment une en pratique. A une extraction diagonale pr\^et, on pourra toujours supposer que la convergence a lieu le long d'une sous-suite. Ce th\'eor\`eme d\'ecrit donc bien toutes les limites possibles pour une suite d'\'etats bosoniques \`a $N$ corps quand $N\to \infty$.
\item Le fait que la mesure vive sur la boule unit\'e dans~\eqref{eq:melange faible} n'est pas une surprise puisque qu'il y a potentiellement une perte de masse dans les cas trait\'es par le th\'eor\`eme. En particulier, il est fort possible que $\gamma ^{(n)}=0$ pour tout $n$ auquel cas $\mu = \delta_0$, la masse de Dirac \`a l'origine.  
\item Le mot faible renvoit \`a ``convergence faible'' et n'indique pas une moindre port\'ee du th\'eor\`eme. Le r\'esultat est en fait plus g\'en\'eral que le th\'eor\`eme de de Finetti fort. Il suffit pour s'en rendre compte de traiter le cas sans perte de masse o\`u $\tr_{\gH ^n} [\gamma ^{(n)}] = 1$. La mesure $\mu$ doit alors bien s\^ur \^etre support\'ee sur la sph\`ere et la convergence avoir lieu en norme. Il est d'ailleurs suffisant de supposer que $\tr_{\gH ^n} [\gamma ^{(n)}] = 1$ pour un certain $n\in \N$, et la convergence est forte pour tout $n$ puisque la mesure $\mu$ ne d\'epend pas de $n$.
\item Ammari et Nier d\'emontrent des r\'esultats un peu plus g\'en\'eraux, voir~\cite{Ammari-hdr,AmmNie-08,AmmNie-09,AmmNie-11}. En particulier il n'est pas n\'ecessaire de partir d'un \'etat \`a nombre de particules fix\'e. On pourrait consid\'erer un \'etat sur l'espace de Fock du moment que des bornes convenables sur son nombre de particules (vu comme une variable al\'eatoire dans ce contexte) sont disponibles.
\item L'unicit\'e de la mesure se prouve par un argument simple, voir~\cite[Section 2]{LewNamRou-13}. Ici nous serons principalement int\'eress\'es par l'existence, qui nous suffit pour des probl\`emes statiques. Pour des probl\`eme d\'ependant du temps, l'unicit\'e est en revanche cruciale~\cite{AmmNie-08,AmmNie-09,AmmNie-11,CheHaiPavSei-13}.
\end{enumerate}

\hfill\qed
\end{remark}

Revenant \`a la d\'erivation de la th\'eorie de Hartree, rappellons que la borne sup\'erieure~\eqref{eq:up bound Hartree} est toujours vraie. Nous ne sommes donc \`a la recherche que de bornes inf\'erieures. Un cas o\`u il est suffisant de d\'ecrire la limite faible-$\ast$ des matrices densit\'es est bien s\^ur celui d'une \'energie faiblement semi-continue inf\'erieurement. Cette remarque serait d'un int\'er\^et marginal si ce cas ne recouvrait pas une classe de syst\`emes importants, les syst\`emes sans \'etats li\'es. 

Nous allons maintenant d\'emontrer la validit\'e de l'approximation de Hartree dans ce cas en nous basant sur le Th\'eor\`eme~\ref{thm:DeFinetti faible}. Nous allons travailler dans $\R ^d$ et supposer pour simplifier que le potentiel $V$ de~\eqref{eq:op Schro} ne fournit de confinement dans aucune direction:
\begin{equation}\label{eq:sans confinement}
V \in L ^p (\R ^d) + L ^{\infty} (\R ^d), \max{1,d/2}\leq p <\infty,  V(x) \to 0 \mbox{ quand } |x|\to \infty.
\end{equation}
L'hypoth\`ese d'absence d'\'etats li\'es concerne le potentiel d'interaction $w$. Elle est mat\'erialis\'ee par l'in\'egalit\'e\footnote{Ou une variante si on consid\`ere une \'energie cin\'etique diff\'erente, cf~Remarque~\ref{rem:energie cin}.}
\begin{equation}\label{eq:no bound states}
- \Delta + \frac{w}{2} \geq 0 
\end{equation}
au sens des op\'erateurs. Ceci formalise le fait que le potentiel $w$ n'est pas assez attractif pour que les particules forment des \'etats li\'es (mol\'ecules etc ...). En effet, vue l'hypoth\`ese~\eqref{eq:decrease w}, $- \Delta + \frac{w}{2}$ ne peut avoir que des valeurs propres n\'egatives, son spectre essentiel commen\c{c}ant en $0$. Vu~\eqref{eq:no bound states}, il ne peut en fait avoir de valeurs propres du tout, et donc pas de fonctions propres qui sont par d\'efinition les \'etats li\'es du potentiel. Un exemple particulier est le cas d'un potentiel purement r\'epulsif $w\geq 0$.

\medskip

Sous l'hypoth\`ese~\eqref{eq:no bound states} les particules qui s'\'echappent \'eventuellement \`a l'infini ne voient plus le potentiel $V$ et portent forc\'ement une \'energie positive qui peut \^etre n\'eglig\'ee pour une borne inf\'erieure sur l'\'energie totale. Les particules qui restent confin\'ees par le potentiel $V$ sont d\'ecrites par les limites faibles-$\ast$ des matrices densit\'e, auxquelles on peut appliquer le th\'eor\`eme de de Finetti faible. Ce raisonnement m\`ene au r\'esultat suivant, pour lequel on a besoin de la notation
\begin{equation}\label{eq:Hartree perte masse}
\eH (\lambda) := \inf_{\norm{u} ^2 = \lambda} \EH [u],\quad 0\leq \lambda \leq 1
\end{equation}
pour l'\'energie de Hartree dans le cas d'une perte de masse. Sous l'hypoth\`ese~\eqref{eq:no bound states} il n'est pas difficile de montrer que pour tout $ 0 \leq \lambda \leq 1$
\begin{equation}\label{eq:Hartree lambda}
\eH (\lambda) \geq \eH (1) = \eH 
\end{equation}
en construisant des \'etats tests constitu\'es de deux morceaux de masse bien s\'epar\'es.

\begin{theorem}[\textbf{D\'erivation de la th\'eorie de Hartree, cas sans \'etats li\'es}]\label{thm:wlsc}\mbox{}\\
Sous les hypoth\`eses pr\'ec\'edentes, en particulier~\eqref{eq:sans confinement} et~\eqref{eq:no bound states} on a
$$\lim_{N\to\ii}\frac{E(N)}{N}=\eH.$$
Soit $\Psi_N$ un fondamental de $H_N$ r\'ealisant l'infimum~\eqref{eq:quant ground state} et 
$$ \gamma_N ^{(n)} := \tr_{n+1 \to N} \left[ |\Psi_N\rangle\langle \Psi_N|\right]$$
sa $n$-i\`eme matrice de densit\'e r\'eduite. Il existe une mesure de proabilit\'e $\mu$ sur 
$$\cM_{\rm H} =\left\{ u\in B\gH, \; \EH[u] = \eH \left(\norm{u} ^2 \right) = \eH (1)\right\}$$ 
telle que, le long d'une sous-suite et pour tout $n\in \N$ 
\begin{equation}\label{eq:def faible result}
\gamma^{(n)}_{N} \wto_\ast \int_{\cM_{\rm H}} d\mu(u)\;|u^{\otimes n}\rangle\langle u^{\otimes n}|
\end{equation}
dans $\gS ^1 (\gH ^n).$ En particulier, si $\eH$ a un minimiseur unique $\uH$ avec $\norm{\uH} = 1$, alors pour toute la suite
\begin{equation}
\lim_{N\to\ii} \gamma^{(n)}_N=|\uH^{\otimes n}\rangle\langle \uH^{\otimes n}|
\label{eq:BEC-confined faible} 
\end{equation}
fortement en norme de trace.
\end{theorem}

\begin{remark}[Cas avec perte de masse]\label{rem:perte masse}\mbox{}\\
On notera qu'il est tout \`a fait possible que la convergence reste faible-$\ast$, ce qui recouvre une certaine r\'ealit\'e physique. Si le potentiel \`a un corps n'est pas assez attractif pour retenir toutes les particules on aura typiquement un sc\'enario o\`u:
\[
\begin{cases}
\eH (\lambda) = \eH (1) \mbox{ pour } \lambda_c \leq \lambda \leq 1 \\
\eH(\lambda) < \eH (1) \mbox{ pour } 0 \leq \lambda < \lambda_c 
\end{cases}
\]
avec $\lambda_c$ une masse critique qui peut \^etre li\'ee par le potentiel $V$. Dans ce cas, $\eH(\lambda)$ ne sera pas atteint pour $\lambda_c < \lambda \leq 1$ et on aura un minimiseur pour l'\'energie de Hartree uniquement pour une masse $0 \leq \lambda \leq \lambda_c$. Si par exemple le minimiseur $\uH$ \`a masse $\lambda_c$ est unique modulo une phase constante, le Th\'eor\`eme~\ref{thm:wlsc} montre que 
\[
\gamma^{(n)}_{N} \wto_\ast  |\uH^{\otimes n}\rangle\langle \uH^{\otimes n}|
\]
et on remarquera que la limite a une masse $\lambda_c ^n <1$. Ce sc\'enario se produit effectivement dans le cas d'un ``atome bosonique'', voir Section 4.2 dans~\cite{LewNamRou-13}.\hfill\qed
\end{remark}

Le Th\'eor\`eme~\ref{thm:wlsc} est prouv\'e dans~\cite{LewNamRou-13}. Pour pouvoir appliquer le Th\'eor\`eme~\ref{thm:DeFinetti faible} on commence par l'observation suivante:

\begin{lemma}[\textbf{Semi-continuit\'e inf\'erieure d'une \'energie sans \'etats li\'es}]\label{lem:wlsc}\mbox{}\\
Sous les hypoth\`eses pr\'ec\'edentes, soient deux suites  $\gamma^{(1)}_N,\gamma^{(2)}_N\geq0$ satisfaisant 
\begin{equation*}
\tr_{\gH^2}\gamma^{(2)}_N=1, \quad \gamma_N^{(1)} =\Tr_{2}\gamma^{(2)}_N
\end{equation*}
ainsi que $\gamma^{(k)}_N\wto_*\gamma^{(k)}$ faiblement-$\ast$ dans $\gS^1(\gH^k)$ pour $k=1,2$. On a
\begin{equation}
\liminf_{N\to \infty} \left( \Tr_\gH[T \gamma_N^{(1)}]+\frac{1}{2}\Tr_{\gH^2}[w \gamma_N^{(2)}] \right) \ge \Tr_\gH[T \gamma^{(1)}]+\frac{1}{2}\Tr_{\gH^2}[w \gamma^{(2)}].
\label{eq:energy wlsc} 
\end{equation}
\end{lemma}

\begin{proof}[Preuve]
On introduit deux fonctions de troncature $0\leq \chi_R,\eta_R\leq 1$ satisfaisant
\[
 \chi_R ^2 + \eta_R ^2 \equiv 1, \supp (\chi_R)\subset B(0,2R), \supp (\eta_R) \subset B(0,R) ^c , |\nabla \chi_R| + |\nabla \eta_R| \leq C R ^{-1} 
\]
et il est ais\'e de montrer (il s'agit de la formule IMS) que
\begin{multline}\label{eq:IMS}
-\Delta  = \chi_R (-\Delta) \chi_R +  \eta_R (-\Delta) \eta_R - |\nabla \chi_R | ^2 -   |\nabla \eta_R | ^2 
\\ \geq \chi_R (-\Delta) \chi_R +  \eta_R (-\Delta) \eta_R - C R ^{-2}
\end{multline}
en tant qu'op\'erateur. Pour la partie \`a un corps de l'\'energie on a donc ais\'ement
$$  \tr_{\gH} [T \gamma_N ^{(1)}] \geq \tr_{\gH} [T \chi_R\gamma_N ^{(1)} \chi_R] + \tr_{\gH} [-\Delta \eta_R \gamma_N ^{(1)} \eta_R] + r_1(N,R)$$
avec 
$$ r_1 (N,R) \geq - C  R ^{-2} + \tr_{\gH} [ \eta_R ^2 V \gamma_N ^{(1)} ]$$
et donc 
$$\liminf_{R\to \infty} \liminf_{N\to \infty} r_1(N,R) = 0$$
au vu de l'hypoth\`ese~\eqref{eq:sans confinement} qui implique
$$  \tr_{\gH} [\eta_R ^2 V \gamma_N ^{(1)}] = \int_{\R ^d} \eta_R ^2 (x) V(x) \rho_N ^{(1)} (x) dx \to 0 \mbox{ quand } R\to\infty$$
uniform\'ement en $N$. Nous avons not\'e $\rho_N ^{(1)} $ la densit\'e \`a un corps de $\gamma_N ^{(1)}$, d\'efinie formellement par 
$$\rho_N ^{(1)} (x) = \gamma_N ^{(1)}(x,x)$$
en identifiant $\gamma_N ^{(1)}$ et son noyau. Pour traiter les interactions nous introduisons 
\[
 \rho_N ^{(2)} (x,y):= \gamma_N ^{(2)} (x,y;x,y)
\]
la densit\'e \`a deux corps de $\gamma_N ^{(2)}$ (qu'on a identifi\'e avec son noyau $\gamma_N ^{(2)} (x',y';x,y)$) et \'ecrivons 
\begin{align*}
\tr_{\gH ^2} [w\gamma_N ^{(2)}] &= \iint_{\R ^d\times\R^d } w(x-y) \rho_N ^{(2)} (x,y) dxdy\\
&= \iint_{\R ^d\times\R^d } w(x-y) \chi_R ^2 (x) \rho_N ^{(2)} (x,y) \chi_R ^2 (y)  dxdy 
\\&+ \iint_{\R ^d\times\R^d } w(x-y) \eta_R ^2 (x) \rho_N ^{(2)} (x,y) \eta_R ^2 (y) dxdy 
\\&+ 2 \iint_{\R ^d\times\R^d } w(x-y) \chi_R ^2 (x)  \rho_N ^{(2)}(x,y) \eta_R ^2 (y) dxdy 
\\&= \iint_{\R ^d\times\R^d } w(x-y) \chi_R ^2 (x)\rho_N ^{(2)} (x,y)  \chi_R ^2 (y) dxdy 
\\&+ \iint_{\R ^d\times\R^d } w(x-y) \eta_R ^2 (x)  \rho_N ^{(2)} (x,y) \eta_R ^2 (y) dxdy 
\\&+ r_2(N,R)
\\&= \tr_{\gH ^2} [w\:\chi_R \otimes \chi_R \gamma_N ^{(2)} \chi_R \otimes \chi_R] + \tr_{\gH ^2} [w \:\eta_R \otimes \eta_R \gamma_N ^{(2)} \eta_R \otimes \eta_R]
\\&+ r_2(N,R)
\end{align*}
o\`u 
$$\liminf_{R\to \infty} \liminf_{N\to \infty} r_2(N,R) = 0$$
par un raisonnement de concentration-compacit\'e standard (voir les d\'etails dans la Section~4 de~\cite{LewNamRou-13}).  

On a donc \`a ce stade,
\begin{multline}\label{eq:proof wlsc}
\liminf_{N\to \infty} \Tr_\gH[T \gamma_N^{(1)}]+\frac{1}{2}\Tr_{\gH^2}[w \gamma_N^{(2)}] \geq 
\\\liminf_{R\to \infty} \liminf_{N\to \infty} \Tr_\gH[T \chi_R \gamma_N^{(1)} \chi_R] + \frac{1}{2}\Tr_{\gH^2}[w \:\chi_R \otimes \chi_R \gamma_N^{(2)}\chi_R \otimes \chi_R] 
\\+ \liminf_{R\to \infty} \liminf_{N\to \infty} \Tr_\gH[-\Delta \eta_R \gamma_N^{(1)} \eta_R]+\frac{1}{2} \Tr_{\gH^2}[w \:\eta_R \otimes \eta_R \gamma_N^{(2)}\eta_R \otimes \eta_R]
\end{multline}
Les termes de la seconde ligne donnent le membre de droite de~\eqref{eq:energy wlsc}. En effet on rappelle que $T= - \Delta + V$, et en utilisant le lemme de Fatou pour $-\Delta \geq 0$ et le fait que  
$$\chi_R \gamma_N ^{(1)}\chi_R \to \chi_R \gamma_N ^{(1)}\chi_R$$
en norme puisque $\chi_R$ est \`a support compact, on a 
$$ \liminf_{N\to \infty} \Tr_\gH[T \chi_R \gamma_N^{(1)} \chi_R] \geq \Tr_\gH[T \chi_R \gamma^{(1)} \chi_R]. $$
Il suffit alors de rappeler que $\chi_R \to 1$ ponctuellement pour conclure. Le terme d'interaction est trait\'e de la m\^eme mani\`ere en utilisant la convergence forte 
$$\chi_R \otimes \chi_R \gamma_N ^{(2)}\chi_R \otimes \chi_R  \to \chi_R \otimes \chi_R  \gamma_N ^{(2)} \chi_R \otimes \chi_R $$
puis la convergence ponctuelle $\chi_R \to 1$.

Les termes de la troisi\`eme ligne de~\eqref{eq:proof wlsc} forment une combinaison positive que l'on peut n\'egliger pour une borne inf\'erieure. Pour le voir on note que comme $0\leq \eta_R \leq 1$ on a $\eta_R \otimes \one \geq \eta_R \otimes \eta_R$ et ($\tr_2$ symbolise la trace partielle par rapport \`a la deuxi\`eme variable)
\[
 \eta_R \gamma_N ^{(1)} \eta_R \geq \tr_2 [\eta_R \otimes \eta_R \gamma_N ^{(2)} \eta_R \otimes \eta_R]
\]
ce qui donne, par sym\'etrie de $\gamma_N ^{(2)}$, 
\begin{multline*}
\Tr_\gH[-\Delta \eta_R \gamma_N^{(1)} \eta_R]+\frac{1}{2} \Tr_{\gH^2}[w \:\eta_R \otimes \eta_R \gamma_N^{(2)}\eta_R \otimes \eta_R] \\ 
\geq \frac{1}{2} \Tr_{\gH^2}\left[\left( (-\Delta) \otimes \one + \one \otimes (-\Delta) + w \right)\:\eta_R \otimes \eta_R \gamma_N^{(2)}\eta_R \otimes \eta_R\right]
\end{multline*}
et il n'est pas difficile\footnote{Il suffit de d\'ecoupler le centre de masse du mouvement relatif des deux particules, c'est \`a dire faire le changement de variable $(x_1,x_2) \mapsto (x_1 + x_2, x_1-x_2)$.} de voir que l'hypoth\`ese~\eqref{eq:no bound states} implique
\[
(-\Delta) \otimes \one + \one \otimes (-\Delta) + w \geq 0
\]
ce qui assure que la troisi\`eme ligne de~\eqref{eq:proof wlsc} est positive et conclut la preuve.
\end{proof}

Nous pouvons conclure la 

\begin{proof}[Preuve du Th\'eor\`eme~\ref{thm:wlsc}.]
Partant d'une suite $\Gamma_N = |\Psi_N\rangle \langle\Psi_N| $ d'\'etats \`a $N$ corps on extrait des sous-suites comme dans la preuve du Th\'eor\`eme~\ref{thm:confined} pour avoir 
\[
 \Gamma_N ^{(n)}\wto_\ast \gamma ^{(n)}.
\]
Gr\^ace au Lemme~\ref{lem:wlsc}, on obtient 
\begin{align*}
\liminf_{N\to \infty} N ^{-1} \tr_{\gH ^N} [H_N \Gamma_N] &= \liminf_{N\to \infty} \left(\tr_{\gH} [T \Gamma_N ^{(1)}] + \frac12 \tr_{\gH ^2} [w\Gamma_N ^{(2)}] \right)
\\&\geq \tr_{\gH} [T \gamma ^{(1)}] + \frac12 \tr_{\gH ^2} [w\gamma ^{(2)}]
\end{align*}
et il reste \`a appliquer le th\'eor\`eme~\ref{thm:DeFinetti faible} \`a la suite $(\gamma ^{(n)})_{n\in \N}$ pour obtenir
\[
 \liminf_{N\to \infty} N ^{-1} \tr_{\gH ^N} [H_N \Gamma_N] \geq \int_{B\gH} \EH[u] d\mu(u) \geq \eH
\]
en utilisant~\eqref{eq:Hartree lambda} et le fait que $\int_{B\gH} d\mu(u) = 1$. Encore une fois, les autres conclusions du th\'eor\`eme suivent ais\'ement en inspectant les cas d'\'egalit\'e dans les diff\'erentes estimations.
\end{proof}

Pour la suite des op\'erations, notons que au cours de la preuve du Lemme~\ref{lem:wlsc} nous avons d\'emontr\'e le r\'esultat interm\'ediaire~\eqref{eq:proof wlsc}, sans utiliser l'hypoth\`ese que le potentiel $w$ n'a pas d'\'etats li\'es (qui n'est intervenu qu'\`a un stade ult\'erieur de la preuve). Nous formulons ceci sous forme d'un lemme que nous r\'eutiliserons au Chapitre~\ref{sec:Hartree}.

\begin{lemma}[\textbf{Localisation de l'\'energie}]\label{lem:loc ener}\mbox{}\\
Sous les hypoth\`eses~\eqref{eq:decrease w} et~\eqref{eq:sans confinement}, soit une suite $\Psi_N$ de quasi-minimiseurs de $E(N)$: 
$$ \langle \Psi_N, H_N \Psi_N \rangle = E(N) + o (N)$$ 
et $\gamma_N ^{(k)}$ les matrices de densit\'e r\'eduites associ\'ees. On a  
\begin{multline}\label{eq:loc ener}
\liminf_{N\to \infty} \frac{E(N)}{N} = \liminf_{N\to \infty} \Tr_\gH[T \gamma_N^{(1)}]+\frac{1}{2}\Tr_{\gH^2}[w \gamma_N^{(2)}] \geq 
\\\liminf_{R\to \infty} \liminf_{N\to \infty} \Tr_\gH[T \chi_R \gamma_N^{(1)} \chi_R] + \frac{1}{2}\Tr_{\gH^2}[w \:\chi_R ^{\otimes 2} \gamma_N^{(2)}\chi_R ^{\otimes 2}] 
\\+ \liminf_{R\to \infty} \liminf_{N\to \infty} \Tr_\gH[-\Delta \eta_R \gamma_N^{(1)} \eta_R]+\frac{1}{2} \Tr_{\gH^2}[w \:\eta_R ^{\otimes 2} \gamma_N^{(2)}\eta_R ^{\otimes 2}]
\end{multline}
o\`u $0\leq \chi_R \leq 1$ est $C^1$, \`a support dans $B(0,R)$, et $\eta_R = \sqrt{1-\chi_R ^2}$.
\end{lemma}

\subsection{Relation entre les diff\'erents r\'esultats de structure pour les \'etats bosoniques}\label{sec:rel deF}\mbox{}\\\vspace{-0.4cm}

Nous venons d'introduire deux th\'eor\`emes de structure pour les syst\`emes bosoniques \`a grand nombre de particules, qui indiquent que moralement, si $\Gamma_N$ est un \'etat \`a $N$ bosons sur un espace de Hilbert s\'eparable $\gH$, on a une mesure de probabilit\'e $\mu \in \PP (\gH)$ sur l'espace de Hilbert \`a un corps telle que 
\begin{equation}\label{eq:def quant formel}
\Gamma_N ^{(n)} \approx \int_{u\in \gH} |u ^{\otimes n} \rangle \langle u ^{\otimes n}| d\mu (u) 
\end{equation}
quand $N$ est grand et $n$ fixe. Les Chapitres~\ref{sec:deF finite dim} et~\ref{sec:locFock} de ces notes seront en grande partie consacr\'es \`a la preuve des ces th\'eor\`emes ``\`a la de Finetti''. Comme indiqu\'e \`a la Remarque~\ref{rem:weak deF}, le Th\'eor\`eme~\ref{thm:DeFinetti faible} est en fait plus g\'en\'eral que le Th\'eor\`eme~\ref{thm:DeFinetti fort}, et nous prouverons donc le premier. 

Pour bien voir l'apport du Th\'eor\`eme faible en dimension infinie, soulignons que la propri\'et\'e cl\'e qui permet de d\'emontrer le Th\'eor\`eme fort est la consistance~\eqref{eq:defi inf consistance}. Lorsqu'on part des matrices r\'eduites $\Gamma_N ^{(n)}$ d'un \'etat \`a $N$ particules $\Gamma_N$ et qu'on extrait des sous-suites au sens de la convergence faible-$\ast$ pour d\'efinir la hi\'erarchie $\left(\gamma ^{(n)}\right)_{n\in \N}$,
$$ \Gamma_N ^{(n)}\wto_\ast \gamma ^{(n)},$$
on a seulement
$$ \tr_{n+1}\left[\gamma ^{(n+1)}\right] \leq \lim_{N\to \infty} \tr_{n+1} \left[\gamma_N ^{(n+1)}\right] = \lim_{N\to \infty} \gamma_N ^{(n)} = \gamma ^{(n)}.$$
car \emph{la trace n'est pas continue}\footnote{Ceci est une caract\'erisation des espaces de dimension infinie} \emph{pour la topologie faible-$\ast$}, seulement semi-continue inf\'erieurement. Il est \'evident que la relation de ``sous-consistance''
\begin{equation}\label{eq:sous consist}
\tr_{n+1}\left[\gamma ^{(n+1)}\right] \leq  \gamma ^{(n)}
\end{equation}
est insuffisante pour avoir un th\'eor\`eme \`a la de Finetti. Un contre-exemple simple est donn\'e par la suite des matrices densit\'e d'un \'etat \`a un corps $v\in S\gH$: 
$$\gamma ^{(0)} = 1,\: \gamma ^{(1)} = |v\rangle \langle v|,\: \gamma ^{(n)} = 0 \mbox{ pour }n\geq 2$$ 

\medskip

Le parti-pris que nous suivrons dans ces notes sera de fournir une preuve la plus constructive possible du Th\'eor\`eme de de Finetti faible. Avant d'en annoncer le plan, disons quelques mots de l'approche beaucoup plus abstraite des r\'ef\'erences historiques~\cite{Stormer-69,HudMoo-75}. Ces travaux contiennent en fait une forme du th\'eor\`eme encore plus g\'en\'erale que la version faible du Th\'eor\`eme~\ref{thm:DeFinetti faible}. Ce r\'esultat s'applique \`a des \'etats ``abstraits'', c'est-\`a-dire non n\'ecessairement normaux et non n\'ecessairement bosoniques, qui peuvent \^etre d\'efinis comme suit:

\begin{definition}[\textbf{Etat abstrait \`a nombre infini de particules}]\label{def:etat infini abs}\mbox{}\\
Soit $\gH$ un espace de Hilbert s\'eparable et pour $n\in \N$, $\gH ^n = \bigotimes ^n \gH$ l'espace \`a $n$ corps correspondant. On appelle \emph{\'etat abstrait \`a nombre infini de particules} une suite $(\omega ^{(n)})_{n\in \N}$ avec
\begin{itemize}
\item $\omega ^{(n)}$ un \'etat abstrait \`a $n$ particules : $\omega ^{(n)}\in (\gB (\gH ^n)) ^*$, le dual des op\'erateurs born\'es sur $\gH ^n$, $\omega ^{(n)}\geq 0$ et 
\begin{equation}\label{eq:masse abstraite}
\omega ^{(n)} (\one_{\gH ^n}) = 1. 
\end{equation}
\item $\omega ^{(n)}$ est sym\'etrique au sens o\`u
\begin{equation}\label{eq:sym abstraite}
\omega ^{(n)} \left( B_1 \otimes \ldots B_n \right) = \omega ^{(n)} \left( B_{\sigma(1) }\otimes \ldots B_{\sigma(n)} \right) 
\end{equation}
pour tous $B_1,\ldots,B_n \in \gB (\gH)$ et toute permutation $\sigma \in \Sigma_n$. 
\item la suite $(\omega ^{(n)})_{n\in \N}$ est consistante:
\begin{equation}\label{eq:consistance abstraite}
\omega ^{(n+1)} \left(B_1 \otimes \ldots B_n \otimes \one_{\gH}\right)=  \omega ^{(n)} \left(B_1 \otimes \ldots B_n \right)
\end{equation}
pour tous $B_1,\ldots,B_n \in \gB (\gH)$.
\end{itemize}
\hfill \qed
\end{definition}

Ces \'etats abstraits ne sont en g\'en\'eral pas localement normaux, c'est-\`a-dire ne rentrent pas dans le cadre de la d\'efinition suivante:

\begin{definition}[\textbf{Etat normal, \'etat localement normal}]\label{def:etat infini normal}\mbox{}\\
Soit $\gH$ un espace de Hilbert s\'eparable et $(\omega ^{(n)})_{n\in \N}$  un \'etat abstrait \`a nombre infini de particules. On dit que  $(\omega ^{(n)})_{n\in \N}$ est localement normal si $\omega ^{(n)}$ est normal pour tout $n\in \N$, c'est-\`a-dire qu'il existe $\gamma ^{(n)}\in \gS ^1 (\gH ^n)$ un op\'erateur \`a trace tel que  
\begin{equation}\label{eq:etat normal}
\omega ^{(n)} (B_n) = \tr_{\gH ^n} [\gamma ^{(n)} B_n] 
\end{equation}
pour tout $B_n \in \gB (\gH ^n)$.
\hfill \qed
\end{definition}

En identifiant les op\'erateurs \`a trace et les \'etats normaux correspondants on voit bien que la D\'efinition~\ref{def:etat infini} est un cas particulier d'\'etat abstrait \`a nombre infini de particules. On notera que (par le th\'eor\`eme spectral) l'ensemble des combinaisons convexes d'\'etats purs (les projecteurs) co\"incide avec les \'etats \`a trace. Un \'etat abstrait n'est donc pas une combinaison convexe d'\'etats purs, donc pas un \'etat mixte (i.e. un m\'elange statistique d'\'etats purs). Son interpr\'etation physique n'est donc pas \'evidente.

La notion de consistance~\eqref{eq:consistance abstraite} est la g\'en\'eralisation naturelle de~\eqref{eq:defi inf consistance} mais il est important de noter que la sym\'etrie~\eqref{eq:sym abstraite} est de nature plus faible que la sym\'etrie bosonique. Il s'agit en fait de la notion de sym\'etrie correspondant \`a des particules indiscernables classiquement (le module de la fonction d'onde est sym\'etrique mais pas la fonction d'onde elle-m\^eme). On peut par exemple noter que si $\omega  ^{(n)}$ est normal au sens de~\eqref{eq:etat normal}, alors $\gamma ^{(n)}\in \gS ^1 (\gH ^n)$ satisfait
\[
U_{\sigma}  \gamma ^{(n)} U_{\sigma} ^* = \gamma ^{(n)}
\]
o\`u $U_{\sigma}$ est l'op\'erateur unitaire permutant les $n$ particules suivant $\sigma \in \Sigma_n$. La sym\'etrie bosonique correspond \`a la contrainte plus forte
\[
U_{\sigma}  \gamma ^{(n)} = \gamma ^{(n)} U_{\sigma} ^* = \gamma ^{(n)},
\]
cf Section~\ref{sec:forma quant}. On a \'egalement une notion d'\'etats produits plus g\'en\'erale que la D\'efinition~\ref{def:etat produit}.

\begin{definition}[\textbf{Etat produit abstrait \`a nombre infini de particules}]\mbox{}\label{def:etat produit abs}\\
On appelle \emph{\'etat produit abstrait} un \'etat abstrait \`a nombre infini de particules avec 
\begin{equation}\label{eq:etat produit abstrait}
\omega ^{(n)} = \omega ^{\otimes n} 
\end{equation}
pour tout $n\in \N$, o\`u $\omega \in (\gB(\gH))^*$ est un \'etat abstait \`a une particule (en particulier $\omega \geq 0$ et $\omega (\one_{\gH}) = 1$).\hfill\qed
\end{definition}

Le th\'eor\`eme de de Finetti quantique sous sa forme la plus g\'en\'erale \'enonce que \textbf{tout \'etat abstait \`a nombre infini de particules est une combinaison convexe d'\'etats produits abstraits}:

\begin{theorem}[\bf De Finetti quantique abstrait]\label{thm:DeFinetti abs}\mbox{}\\
Soit $\gH$ un espace de Hilbert s\'eparable et $(\omega^{(n)})_{n\in \N}$ un \'etat bosonique abstrait \`a nombre infini de particules sur $\gH$. Il existe une unique mesure de probabilit\'e $\mu \in \PP \left( (\gB (\gH)) ^* \right)$ sur le dual des op\'erateurs born\'es de $\gH$ telle que
\begin{equation}\label{eq:sup mes deF abs}
\mu \left( \left\{ \omega \in (\gB (\gH)) ^*, \; \omega \geq 0, \; \omega (\one_{\gH}) = 1 \right\} \right) = 1 
\end{equation}
et 
\begin{equation}
\omega^{(n)}=\int \omega ^{\otimes n} \, d\mu(\om)
\label{eq:melange abs}
\end{equation}
pour tout $n\geq0$.
\end{theorem}

\begin{remark}[Sur le th\'eor\`eme de de Finetti quantique abstrait]\label{rem:abs deF}\mbox{}\\\vspace{-0.4cm}
\begin{enumerate}
\item Ce th\'eor\`eme a \'et\'e d\'emontr\'e pour la premi\`ere fois par St\o{}rmer dans~\cite{Stormer-69}. Hudson et Moody~\cite{HudMoo-75} ont ensuite fourni une preuve plus simple en adaptant l'argument donn\'e par Hewitt-Savage pour la preuve du Th\'eor\`eme de de Finetti classique~\ref{thm:HS}: ils montrent que les \'etats produits sont les points extr\^emaux de l'ensemble convexe des \'etats abstraits \`a nombre infini de particules et obtiennent l'existence de la mesure par le th\'eor\`eme de Choquet. Cette preuve n\'ecessite la notion d'\'etats abstraits et ne fournit pas directement de preuve du Th\'eor\`eme~\ref{thm:DeFinetti fort}.
\item Hudson et Moody~\cite{HudMoo-75} d\'eduisent le th\'eor\`eme de de Finetti fort du th\'eor\`eme de de Finetti abstrait. Une adaptation de leur m\'ethode (voir~\cite[Appendice A]{LewNamRou-13}) montre que le th\'eor\`eme de de Finetti faible est aussi une cons\'equence du th\'eor\`eme abstrait.
\item Ce r\'esultat a \'et\'e utilis\'e pour d\'eriver des th\'eories de type Hartree pour des \'etats abstraits sans sym\'etrie bosonique dans~\cite{FanSpoVer-80,RagWer-89,PetRagVer-89}. Pour retrouver une th\'eorie de Hartree au sens usuel du terme il faut pouvoir montrer que l'\'etat limite est (localement) normal. En dimension finie $\gB (\gH)$ co\"incide bien s\^ur avec les op\'erateurs compacts, ce qui implique que tout \'etat abstait est normal, et on a pas cette difficult\'e.  
\end{enumerate}
\hfill\qed
\end{remark}

A ce stade nous avons donc le sch\'ema (``deF'' pour de Finetti)
\begin{equation}\label{eq:schema deF 1}
\boxed{\mbox{deF abstrait } \Rightarrow \mbox{ deF faible } \Rightarrow \mbox{ deF fort } },
\end{equation}
mais la preuve du th\'eor\`eme de de Finetti faible que nous allons pr\'esenter suivra une autre route, utilis\'ee dans~\cite{LewNamRou-13,LewNamRou-13b}. Elle part du th\'eor\`eme en dimension finie\footnote{Pour lequel il n'y a \'evidemment pas lieu de distinguer entre th\'eor\`eme faible et th\'eor\`eme fort.}:
\begin{equation}\label{eq:schema deF 2}
\boxed{\mbox{deF en dimension finie } \Rightarrow \mbox{ deF faible } \Rightarrow \mbox{ deF fort } }.
\end{equation}
Il se trouve que cette approche fournit une preuve plut\^ot plus longue que le sch\'ema~\eqref{eq:schema deF 1} partant de la preuve de Hudson-Moody du Th\'eor\`eme~\ref{thm:DeFinetti abs}. Ce d\'etour se motive par cinq raisons principales, d'ordre esth\'etique autant que pratique:
\begin{enumerate}
\item La preuve suivant~\eqref{eq:schema deF 2} est plus simple d'un point de vue conceptuel: elle ne n\'ecessite ni la notion d'\'etat abstrait ni le recours au Th\'eor\`eme de Choquet.
\item Gr\^ace \`a des progr\`es principalement d\^us \`a la communaut\'e d'information quantique (voir~\cite{ChrKonMitRen-07,Chiribella-11,Harrow-13,KonRen-05,FanVan-06,LewNamRou-13b}) on dispose maintenant d'une preuve enti\`erement constructive du th\'eor\`eme de de Finetti en dimension finie. Il s'agit d'abord de prouver des estimations explicites gr\^ace \`a une construction \`a $N$ fini, dans l'esprit de Diaconis-Freedman pour le cas classique, puis de passer \`a la limite comme dans la preuve du th\'eor\`eme de Hewitt-Savage que nous avons pr\'esent\'ee \`a la Section~\ref{sec:DF}.
\item La premi\`ere implication du sch\'ema~\eqref{eq:schema deF 2} est \'egalement essentiellement constructive, gr\^ace aux outils de localisation dans l'espace de Fock utilis\'es par exemple dans~\cite{Ammari-04,DerGer-99,Lewin-11}. Ces outils sont h\'erit\'es des m\'ethodes dites ``g\'eom\'etriques''~\cite{Enss-77,Enss-78,Sigal-82,Simon-77} qui adaptent au probl\`eme \`a $N$ corps des id\'ees de localisation naturelles pour le probl\`eme \`a un corps. Ils permettent entre autres une description fine du d\'efaut de compacit\'e par perte de masse \`a l'infini, dans l'esprit du principe de concentration-compacit\'e~\cite{Lions-84,Lions-84b}.
\item En particulier, la preuve de la premi\`ere implication dans~\eqref{eq:schema deF 2} fournit quelques corollaires importants qui permettent de prouver la validit\'e de la th\'eorie de Hartree dans le cas g\'en\'eral. Quand les hypoth\`eses faites \`a la Section~\ref{sec:Hartree deF faible} ne sont pas valables, le th\'eor\`eme de de Finetti faible et sa preuve fournie par le sch\'ema~\eqref{eq:schema deF 1} ne sont pas suffisants pour conclure: les particules s'\'echappant \`a l'infini peuvent former des \'etats li\'es d'\'energie n\'egative. Les m\'ethodes de localisation que nous pr\'esenterons permettent d'analyser ce ph\'enom\`ene.  
\item Au chapitre~\ref{sec:NLS} nous traiterons un cas o\`u le potentiel d'interaction d\'epend de $N$ pour d\'eriver une description de type Schr\"odinger non-lin\'eaire \`a la limite. Ceci revient \`a prendre une limite o\`u $w$ converge vers une masse de Dirac \emph{en m\^eme temps} que la limite $N\to \infty$. Dans ce cas, les arguments de compacit\'e ne sont pas suffisants et les estimations explicites que nous pr\'esenterons pour prouver le th\'eo\`r\`eme de de Finetti en dimension finie s'av\'ereront bien utiles.
\end{enumerate}

Un point de vue alternatif sur le sch\'ema~\eqref{eq:schema deF 2} est fourni par l'approche de Ammari-Nier~\cite{Ammari-hdr,AmmNie-08,AmmNie-09,AmmNie-11}, bas\'ee sur des m\'ethodes d'analyse semi-classique. Nous aurons l'occasion de pr\'eciser la relation entre les deux approches dans la suite de ces notes. 

Comme not\'e plus haut, la seconde implication dans~\eqref{eq:schema deF 2} est relativement \'evidente. Les chapitres suivants explorent le d\'ebut du sch\'ema et contiennent la preuve du th\'eor\`eme de de Finetti faible ainsi que plusieurs corollaires et r\'esultats interm\'ediaires utiles. Le th\'eor\`eme de de Finetti en dimension finie (o\`u il n'y a pas lieu de distinguer entre faible et fort) est discut\'e dans le Chapitre~\ref{sec:deF finite dim}. Les m\'ethodes de localisation permettant de montrer la premi\`ere implication dans~\eqref{eq:schema deF 2} sont l'objet du Chapitre~\ref{sec:locFock}.

\newpage

\section{\textbf{Th\'eor\`eme de de Finetti quantique en dimension finie}}\label{sec:deF finite dim}

Ce chapitre est consacr\'e au point de d\'epart du sch\'ema de preuve~\eqref{eq:schema deF 2}, \`a savoir une preuve du th\'eor\`eme de de Finetti fort dans le cas d'un espace de Hilbert s\'eparable $\gH$ \emph{de dimension finie},
\[
 \dim \gH = d.
\]
Dans ce cas il n'y a pas lieu de distinguer entre th\'eor\`eme fort et faible puisque les convergences forte et faible-$\ast$ sont \'equivalentes dans $\gS ^1 (\gH ^n)$. L'avantage de travailler en dimension finie est la possibilit\'e d'obtenir des estimations explicites dans l'esprit du th\'eor\`eme de Diaconis-Freedman (mais avec une m\'ethode compl\`etement diff\'erente). Nous allons en fait prouver le r\'esultat suivant, qui donne des bornes en norme de trace, la norme naturelle sur l'espace des \'etats quantiques.   

\begin{theorem}[\textbf{De Finetti quantique quantitatif}]\label{thm:DeFinetti quant} \mbox{}\\
Soit $\Gamma_N \in \gS ^1 (\gH_s ^N)$ un \'etat sur $\gH_s^N$ et $\gamma_N ^{(n)}$ ses matrices de densit\'e r\'eduites. Il existe une mesure de probabilit\'e $\mu_N \in \PP (S\gH)$ telle que, notant 
\begin{equation}\label{eq:finite deF etat}
\Gammat _N = \int_{u\in S \gH} |u ^{\otimes N}\rangle \langle u ^{\otimes N}| d\mu_N(u) 
\end{equation}
l'\'etat associ\'e et $\gammat_N ^{(n)}$ ses matrices de densit\'e r\'eduites, on ait
\begin{equation} \label{eq:error finite dim deF}
\Tr_{\gH^n} \Big| \gamma_N ^{(n)} - \widetilde \gamma _N^{(n)}  \Big| \leq \frac{2 n(d+2n)}{N}
\end{equation}
pour tout $n=1 \ldots N$.
\end{theorem}

\begin{remark}[Sur le th\'eor\`eme de de Finetti en dimension finie]\label{rem:quant deF}\mbox{}\\\vspace{-0.4cm}
\begin{enumerate}
\item Ce r\'esultat est essentiellement d\^u \`a Christandl, K\"onig, Mitchison et Renner~\cite{ChrKonMitRen-07}, avec des pr\'ecurseurs dans~\cite{KonRen-05} et surtout~\cite{FanVan-06}. On pourra trouver des d\'eveloppements dans~\cite{Chiribella-11,Harrow-13,LewNamRou-13b}. La communaut\'e d'information quantique s'est int\'eress\'ee \`a diverses variantes, voir par exemple~\cite{ChrKonMitRen-07,ChrTon-09,RenCir-09,Renner-07,BraHar-12}.
\item On peut ajouter une implication au sch\'ema~\eqref{eq:schema deF 2}:
\begin{equation}\label{eq:schema deF 3}
\boxed{\mbox{deF quantitatif } \Rightarrow \mbox{deF en dimension finie } \Rightarrow \mbox{ deF faible } \Rightarrow \mbox{ deF fort } }.
\end{equation}
En effet, en dimension finie on peut identifier la sph\`ere $S\gH$ avec une sph\`ere Euclidienne usuelle, compacte (de dimension $2d-1$, i.e. la sph\`ere unit\'e de $\R ^{2d}$). L'espace des mesures de probabilit\'e sur $S\gH$ est alors compact pour la topologie faible usuelle et on peut extraire de $\mu_N$ une sous-suite convergente pour d\'emontrer le Th\'eor\`eme~\ref{thm:DeFinetti fort} dans le cas $\dim \gH <\infty$ exactement de la m\^eme fa\c{c}on que nous avons d\'eduit le Th\'eor\`eme~\ref{thm:HS} du Th\'eor\`eme~\ref{thm:DF} \`a la Section~\ref{sec:DF}.
\item La borne~\eqref{eq:error finite dim deF} n'est pas optimale. On peut en fait obtenir l'estimation  
\begin{equation} \label{eq:error finite dim deF 2}
\Tr_{\gH^n} \Big| \gamma_N ^{(n)} - \widetilde \gamma _N^{(n)}  \Big| \leq \frac{2 nd}{N},
\end{equation}
avec la m\^eme construction, voir~\cite{Chiribella-11,ChrKonMitRen-07,LewNamRou-13b}. La preuve que nous pr\'esenterons ne donne que~\eqref{eq:error finite dim deF} mais me para\^it plus instructive. Pour les applications qui nous int\'eressent dans ces notes, $n$ sera toujours fixe (\'egal \`a $2$ la plupart du temps), et dans ce cas~\eqref{eq:error finite dim deF} et~\eqref{eq:error finite dim deF 2} donnent les m\^emes ordres de grandeur en fonction de $N$ et $d$.
\item La borne~\eqref{eq:error finite dim deF 2} dans la cas quantique est l'\'equivalent de l'estimation en $dn/N$ mentionn\'ee \`a la Remarque~\ref{rem:DiacFreLio}. On peut se demander si cet ordre de grandeur  est optimal. Il l'est clairement avec la construction que nous allons utiliser, mais il serait tr\`es int\'eressant de savoir si on peut faire mieux avec une autre construction. En particulier, peut-on prouver une borne ind\'ependante de $d$, rappelant la borne~\eqref{eq:DiacFreed} du th\'eor\`eme de Diaconis-Freedman ? 
\end{enumerate}

\hfill\qed
\end{remark}

La construction de $\Gammat_N$ vient de~\cite{ChrKonMitRen-07}. Elle est particuli\`erement simple mais utilise \`a fond le fait que l'espace $\gH$ est de dimension finie. L'approche que nous suivrons pour la preuve du Th\'eor\`eme~\ref{thm:DeFinetti quant} est due originellement \`a Chiribella~\cite{Chiribella-11}. Nous allons prouver une formule explicite donnant les matrices densit\'e de $\Gammat_N$ en fonction de celles de $\Gamma_N$, dans l'esprit de la Remarque~\ref{rem:marg DF}. Cette formule implique~\eqref{eq:error finite dim deF} de la m\^eme mani\`ere que~\eqref{eq:DF astuce} implique~\eqref{eq:DiacFreed}.

Dans la Section~\ref{sec:CKMR constr} nous pr\'esentons la construction, \'enon\c{c}ons la formule explicite de Chiribella et d\'emontrons le Th\'eor\`eme~\ref{thm:DeFinetti quant} comme corollaire. Avant de donner la preuve du r\'esultat de Chiribella, il est utile de fournir une motivation informelle et quelques heuristiques sur la construction de Christandl-K\"onig-Mitchison-Renner (CKMR), qui se trouve \^etre connect\'ee \`a des id\'ees bien connues d'analyse semi-classique. Ceci est l'objet de la Section~\ref{sec:CKMR heur}. Finalement nous prouverons la formule de Chiribella \`a la Section~\ref{sec:CKMR proof}, en suivant l'approche de~\cite{LewNamRou-13b}. Elle a \'et\'e trouv\'ee ind\'ependament par Lieb et Solovej~\cite{LieSol-13} (avec une motivation diff\'erente), et nous a \'et\'e inspir\'ee par les travaux de Ammari et Nier~\cite{AmmNie-08,AmmNie-09,AmmNie-11}. D'autres preuves sont possibles, cf~\cite{Chiribella-11} et~\cite{Harrow-13}. 

\subsection{La construction de CKMR et la formule de Chiribella}\label{sec:CKMR constr}\mbox{}\\\vspace{-0.4cm}

Notons tout d'abord que la construction de Diaconis-Freedman pr\'esent\'ee pr\'ec\'edement est purement classique, puisqu'elle est bas\'ee sur la notion de mesure empirique, qui n'a pas d'\'equivalent en m\'ecanique quantique. Une autre approche est donc clairement n\'ecessaire pour la preuve du Th\'eor\`eme~\ref{thm:DeFinetti quant}.

En dimension finie, on peut identifier la sph\`ere unit\'e $S\gH = \left\{ u \in \gH,\norm{u}=1 \right\}$ avec une sph\`ere euclidienne. On peut donc la munir de la mesure uniforme (mesure de Haar du groupe des rotations, simplement la mesure de Lebesgue sur la sph\`ere euclidienne), que nous noterons $du$. On a alors par le lemme de Schur une r\'esolution de l'identit\'e de la forme
\begin{equation}\label{eq:Schur}
 \dim \gH_s^N \int_{S\gH} | u^{\otimes N} \rangle  \langle u^{\otimes N} | \, du = \1_{\gH^N},  
\end{equation}
ce qui est une simple cons\'equence de l'invariance par rotation de $du$.

L'id\'ee de Christandl-K\"onig-Mitchison-Renner est alors de simplement poser 
\begin{align}\label{eq:def CKMR}
d\mu_N (u) &:=  \dim \gH_s^N \tr_{\gH_s ^N} \left[\Gamma_N | u^{\otimes N} \rangle  \langle u^{\otimes N} |\right] du \nonumber\\
&= \dim \gH_s^N \tr_{\gH_s ^N}  \left\langle u^{\otimes N}, \Gamma_N u^{\otimes N} \right\rangle  du 
\end{align}
c'est \`a dire de d\'efinir 
\begin{equation}\label{eq:def CKMR 2}
\Gammat_N =  \dim \gH_s^N \int_{S\gH} | u^{\otimes N} \rangle  \langle u^{\otimes N} | \left\langle u^{\otimes N}, \Gamma_N u^{\otimes N} \right\rangle du.
\end{equation}
L'observation de Chiribella\footnote{Formul\'ee dans le vocabulaire de l'information quantique comme une relation entre le ``clonage optimal'' et les ``canaux de mesure/pr\'eparation optimaux''.} est la suivante: 

\begin{theorem}[\textbf{Formule de Chiribella}] \label{thm:CKMR-identity} \mbox{}\\
Avec les d\'efinitions pr\'ec\'edentes, on a 
\begin{equation}\label{eq:CKMR exact}
\gammat _N^{(n)} = {{N+n+d-1}\choose n}^{-1}\sum_{\ell=0}^{n} {N \choose \ell}  \gamma_N^{(\ell)} \otimes _s \1_{\gH^{n-\ell}}
\end{equation}
avec la convention
$$ \gamma_N^{(\ell)} \otimes _s \1_{\gH^{n-\ell}}= \frac{1}{\ell!\,(n-\ell)!}\sum_{\sigma\in S_n} (\gamma_N^\ell)_{\sigma(1),...,\sigma(\ell)} \otimes (\1_{\gH^{n-\ell}})_{{\sigma(\ell+1)},...,{\sigma(n)}}$$
o\`u $(\gamma_N^\ell)_{\sigma(1),...,\sigma(\ell)}$ agit sur les variables $\sigma(1)\ldots,\sigma(\ell)$.
\end{theorem} 

Nous pouvons d\'eduire du r\'esultat pr\'ec\'edent une preuve simple du th\'eor\`eme de de Finetti quantitatif: 

\begin{proof}[Preuve du Th\'eor\`eme~\ref{thm:DeFinetti quant}]
On proc\`ede comme en~\eqref{eq:preuve DiacFreed}. Seul le premier terme de la somme~\eqref{eq:CKMR exact} compte:  
\begin{equation}
 \gammat _N^{(n)} - \gamma_N ^{(n)}  = ( C(d,n,N) - 1) \gamma_N ^{(n)} + B = -A + B \label{eq:estim CKMR}
\end{equation}
o\`u 
\[
C(d,n,N) = \frac{(N+d-1)!}{(N+n+d-1)!} \frac{N!}{(N-n)!} < 1,  
\]
et $A,B$ sont des op\'erateurs positifs. On a 
$$\tr_{\gH ^n} [-A+B] = \tr \left[\gammat _N^{(n)} - \gamma_N ^{(n)} \right]= 0,$$
et donc, par l'in\'egalit\'e triangulaire,
\[
\Tr \Big|\gammat _N^{(n)} - \gamma_N ^{(n)} \Big| \leq \tr A + \tr B = 2 \Tr A = 2 (1- C(d,n,N)).
\]
Ensuite, l'in\'egalit\'e \'el\'ementaire 
\begin{align*}
C(d,n,N) &=  \prod_{j=0} ^{n-1} \frac{N-j}{N+j+d}\ge  \left( 1 - \frac{2n + d -2}{N + d + n - 1}\right) ^k \geq 1 -n \frac{2n + d -2}{N + d + n - 1}
\end{align*}
donne
\begin{equation} \label{eq:error-NLR}
\Tr \Big| \gamma_N ^{(n)} - \gammat _N^{(n)} \Big| \le \frac{2 n(d+2n)}{N},
\end{equation}
ce qui est le r\'esultat souhait\'e.
\end{proof}

Tout repose d\'esormais sur la preuve du Th\'eor\`eme~\ref{thm:CKMR-identity}, qui est l'objet de la Section~\ref{sec:CKMR proof}. Avant de la pr\'esenter, nous donnons quelques indications heuristiques quant \`a la pertinence de la construction~\eqref{eq:def CKMR 2}.

\subsection{Heuristiques et motivations}\label{sec:CKMR heur}\mbox{}\\\vspace{-0.4cm}

La formule de Schur~\eqref{eq:Schur} exprime le fait que la famille $\left( u ^{\otimes N} \right)_{u\in S\gH}$ forme une base sur-compl\`ete de $\gH_s ^N$. Une telle base index\'ee par un param\`etre continu rappelle une d\'ecomposition en \'etats coh\'erents~\cite{KlaSka-85,ZhaFenGil-90}. Cette base est en fait ``de moins en moins sur-compl\`ete'' lorsque $N$ devient grand. En effet, on a clairement
\begin{equation}\label{eq:less over complete}
 \langle u ^{\otimes N}, v ^{\otimes N} \rangle_{\gH ^N}  = \langle u , v \rangle_{\gH} ^N \to 0 \mbox{ quand } N \to \infty 
\end{equation}
d\`es que $u$ et $v$ ne sont pas exactement colin\'eaires. La base $\left( u ^{\otimes N} \right)_{u\in S\gH}$ devient donc ``quasiment orthonormale'' quand $N$ tend vers l'infini, et il est alors tr\`es naturel de s'attendre \`a ce que tout op\'erateur ait une repr\'esentation approch\'ee de la forme~\eqref{eq:finite deF etat}.  

Dans le vocabulaire de l'analyse semi-classique~\cite{Lieb-73b,Simon-80,Berezin-72}, chercher \`a \'ecrire 
\[
 \Gamma_N = \int_{u\in S \gH} d\mu_N (u) |u ^{\otimes N} \rangle \langle u ^{\otimes N}| 
\]
revient \`a chercher un \emph{symbole sup\'erieur} $\mu_N$ repr\'esentant $\Gamma_N$. En fait, il se trouve~\cite{Simon-80} que l'on peut toujours trouver un tel symbole, seulement $\mu_N$ n'est en g\'en\'eral pas une mesure positive. Le probl\`eme que nous nous posons est donc de trouver une mani\`ere d'approcher le symbole sup\'erieur (dont on a par ailleurs pas d'expression explicite en fonction de l'\'etat lui-m\^eme) par une mesure positive. 

D'un autre c\^ot\'e, la mesure introduite en~\eqref{eq:def CKMR} est tr\`es exactement ce qu'on appelle le \emph{symbole inf\'erieur} de l'\'etat $\Gamma_N$. Une des raisons pour lesquelles les symboles inf\'erieurs et sup\'erieurs ont \'et\'es introduits est que ces deux objets a priori diff\'erents ont tendance \`a co\"incider dans les limites semi-classiques. Hors, la limite $N\to \infty$ peut \^etre vue comme une limite semi-classique, et il est donc tr\`es naturel de prendre le symbole inf\'erieur comme une approximation du symbole sup\'erieur dans cette limite.

On peut motiver ce choix de mani\`ere un peu plus pr\'ecise. Supposons que nous avons une suite d'\'etats  \`a $N$ particules d\'efinis \`a partir d'un m\^eme symbole sup\'erieur ind\'ependant de $N$,
\[
 \Gamma_N = \dim (\gH _s ^N) \int_{u\in S \gH}  \mu ^{\rm sup} (u) |u ^{\otimes N} \rangle \langle u ^{\otimes N}|du,  
\]
et calculons les symboles inf\'erieurs correspondants:
\begin{align*}
 \mu ^{\rm inf}_N (v) = \langle v ^{\otimes N}, \Gamma_N v ^{\otimes N} \rangle &= \dim (\gH _s ^N) \int_{u\in S \gH}  \mu ^{\rm sup} (u) \left|\langle u ^{\otimes N}, v ^{\otimes N}\rangle\right| ^2 du
 \\&=\dim (\gH _s ^N) \int_{u\in S \gH}  \mu ^{\rm sup} (u) \left|\langle u , v \rangle\right| ^{2N} du.
\end{align*}
Vue l'observation~\eqref{eq:less over complete} et la n\'ecessaire invariance par l'action de $S ^1$ de $\mu ^{\rm sup}$ il est clair qu'on a 
$$ \mu ^{\rm inf}_N (v) \to \mu ^{\rm sup} (v) \mbox{ quand } N\to\infty.$$ 
Autrement dit, le symbole inf\'erieur est, pour $N$ grand, une approximation du symbole sup\'erieur, qui a en outre l'avantage d'\^etre positif. Sans constituer une preuve rigoureuse du Th\'eor\`eme~\ref{thm:DeFinetti quant}, ce point de vue indique que la construction de CKMR est extr\^emement naturelle.

\subsection{La formule de Chiribella et la quantification d'anti-Wick}\label{sec:CKMR proof}\mbox{}\\\vspace{-0.4cm}

La preuve que nous allons donner du Th\'eor\`eme~\ref{thm:CKMR-identity} utilise le formalisme de la seconde quantification. Commen\c{c}ons par un lemme bien utile, qui exprime dans le vocabulaire esquiss\'e \`a la Section pr\'ec\'edente qu'un \'etat est enti\`erement caract\'eris\'e par son symbole inf\'erieur, un fait bien connu~\cite{Simon-80,KlaSka-85}.

\begin{lemma}[\textbf{Le symbole inf\'erieur d\'etermine l'\'etat}]\label{lem:uk-g-uk=0}\mbox{} \\
Si un op\'erateur $\gamma^{(k)}$ sur $\gH_s^k$ satisfait  
\bq \label{eq:uk-g-uk=0}
 \langle u^{\otimes k}, \gamma^{(k)} u^{\otimes k} \rangle =0\qquad \text{pour~tout}~u\in \gH,
 \eq
alors $\gamma^{(k)} \equiv 0$.
\end{lemma}

\begin{proof}[Preuve]
On utilise le produit tensoriel sym\'etrique 
$$\Psi_k\otimes_s\Psi_\ell(x_1,...,x_{k})=\frac{1}{\sqrt{\ell!(k-\ell)!k!}}\sum_{\sigma\in {S}_{k}}\Psi_\ell(x_{\sigma(1)},...,x_{\sigma(\ell)})\Psi_{k-\ell}(x_{\sigma(\ell+1)},...,x_{\sigma(k)})
$$
pour deux vecteurs $\Psi_\ell \in \gH^{\ell}$ et $\Psi_{k-\ell}\in\gH^{k-\ell}$.

En rempla\c{c}ant $u$ par $u+tv$ dans~\eqref{eq:uk-g-uk=0}  et en prenant la d\'eriv\'ee par rapport \`a $t$, on obtient
 \bqq
 \langle v \otimes_s u^{\otimes (k-1)}, \gamma^{(k)} v \otimes_s u^{\otimes (k-1)} \rangle =0
 \eqq
 pour tous $u,v\in \gH$. En prenant $v$ de la forme $v=v_1\pm \widetilde v_1$ puis $v = v_1 \pm i \widetilde v_1$ et en r\'eit\'erant l'argument, on conclut 
   \bqq
\langle v_1 \otimes_s v_2 \otimes_s \ldots \otimes_s v_k, \gamma^{(k)} \widetilde v_1 \otimes_s \widetilde v_2 \otimes_s \ldots \otimes_s \widetilde v_k \rangle =0
 \eqq
pour tous $v_j, \widetilde v_j \in \gH$. Les vecteurs de la forme  $v_1 \otimes_s v_2 \otimes_s \ldots \otimes_s v_k$ formant une base de $\gH_s ^k$, la preuve est compl\`ete.
\end{proof}
 
Nous allons utiliser les op\'erateurs de cr\'eation et d'annihiliation bosoniques standards. Pour tout $f_{k}\in \gH$, on d\'efinit l'op\'erateur de cr\'eation $a^*(f_{k}): \gH_s^{k-1} \to \gH_s^{k}$ par
$$
{a^*}({f_{k}})\left( {\sum\limits_{\sigma  \in {S_{k-1}}} {{f_{\sigma (1)}} } \otimes ... \otimes {f_{\sigma (k-1)}}} \right) = (k) ^{-1/2} \sum\limits_{\sigma  \in {S_{k}}} {{f_{\sigma (1)}} }  \otimes ... \otimes {f_{\sigma (k)}}.
$$
L'op\'erateur d'annihiliation $a(f): \gH^{k+1} \to \gH^{k}$ est l'adjoint formel de  $a^*(f)$ (d'o\`u la notation), d\'efini par
$$ a(f) \left( {\sum\limits_{\sigma  \in {S_{k+1}}} {{f_{\sigma (1)}} } \otimes ... \otimes {f_{\sigma (k+1)}}} \right) = (k+1) ^{1/2} \sum\limits_{\sigma  \in {S_{k+1}}} \left\langle f,f_{\sigma(1)} \right\rangle {{f_{\sigma (2)}} }  \otimes ... \otimes {f_{\sigma (k)}}$$
pour tous $f,f_1,...,f_{k}$ dans $\gH$. Ces op\'erateurs satisfont les {\it relations canoniques de commutation} (CCR en abr\'egeant le terme anglais)
\begin{equation}\label{eq:CCR}
[a(f),a(g)]=0,\quad[a^*(f),a^*(g)]=0,\quad [a(f),a^*(g)]= \langle f,g \rangle_{\gH}. 
\end{equation}
Un des int\'er\^ets de ces objets est que les matrices de densit\'e r\'eduites $\gamma_N ^{(n)}$ d'un \'etat bosonique $\Gamma_N$ sont d\'efinies via\footnote{On rappelle que notre convention est $\tr [\gamma_N ^{(n)}] = 1$.} 
\begin{equation} \label{eq:Wick}
 \langle v^{\otimes n}, \gamma_N^{(n)} v^{\otimes n} \rangle =\frac{(N-n)!}{N!} \tr_{\gH_s ^N} \left[ a^*(v)^n a(v)^n \Gamma_N \right]
 \end{equation}
ce qui caract\'erise compl\`etement $\gamma_N ^{(n)}$ vu le Lemme~\ref{lem:uk-g-uk=0}. La d\'efinition ci-dessus s'appelle une quantification de Wick, les op\'erateurs de cr\'eation et d'annihilation y apparaissent dans l'ordre normal, c'est-\`a-dire tous les cr\'eateurs \`a gauche et tous les annihiliateurs \`a droite.

L'observation cl\'e de la preuve du Th\'eor\`eme~\ref{thm:CKMR-identity} est alors le fait que les matrices de densit\'e de l'\'etat~\eqref{eq:def CKMR 2} sont alternativement d\'efinies \`a partir de $\Gamma_N$ par une quantification d'anti-Wick, o\`u les op\'erateurs de cr\'eation et d'annihilation apparaissent dans l'ordre anti-normal, c'est-\`a-dire tous les annihiliateurs \`a gauche et tous les cr\'eateurs \`a droite.

\begin{lemma}[\textbf{La construction de CKMR et la quantification d'anti-Wick}]\label{lem:A Wick}\mbox{}\\
Soit $\Gammat_N$ defini par~\eqref{eq:def CKMR 2} et $\gammat_N ^{(n)}$ ses matrices de densit\'e r\'eduites. On a 
\begin{equation}\label{eq:A Wick}
\langle v^{\otimes n}, \gammat _N^{(n)} v^{\otimes n} \rangle = \frac{(N+d-1)!}{(N+k+d-1)!} \tr_{\gH_s ^N} \left[ a(v)^n a ^*(v)^n \Gamma_N \right]
\end{equation}
pour tout $v\in \gH$.
\end{lemma}

\begin{proof}[Preuve]
Il suffit de consid\'erer le cas d'un \'etat pur $\Gamma_N = |\Psi_N \rangle\langle \Psi_N |$ et d'\'ecrire
\begin{align*}
\langle v^{\otimes k}, \widetilde \gamma _N^{(k)} v^{\otimes k} \rangle &= \dim \gH_s^N \int_{S\gH} du |\langle u^{\otimes N}, \Psi_N \rangle|^2 | \langle  u^{\otimes k},v^{\otimes k} \rangle|^2 \\
&= \dim \gH_s^N \int_{S\gH} du |\langle u^{\otimes (N +k)}, v^{\otimes k} \otimes \Psi_N \rangle|^2 \\
&= \frac{N!}{(N+k)!} \dim \gH_s^N \int_{S\gH} du |\langle u^{\otimes (N +k)}, a^*(v)^{k} \Psi_N \rangle|^2 \\
&=  \frac{N!}{(N+k)!}\frac{\dim \gH_s^N}{\dim \gH_s^{N+k}} \langle a(v)^{k} \Psi_N,a (v)^k \Psi_N \rangle \\
&= \frac{(N+d-1)!}{(N+k+d-1)!} \langle \Psi_N, a(v)^k a^*(v)^{k} \Psi_N \rangle
\end{align*}
en utilisant le lemme de Schur~\eqref{eq:Schur} dans $\gH_s ^{N+k}$ \`a la troisi\`eme ligne, le fait que $a^*(v)$ est l'adjoint de $a (v)$ \`a la quatri\`eme et en rappellant que 
\begin{equation}\label{eq:dim boson}
\dim \gH_s ^N = { N+d-1 \choose d-1}. 
\end{equation}
\end{proof}

La marche \`a suivre est maintenant claire: il suffit de comparer des polyn\^omes en $a ^*(v)$ et $a(v)$ \'ecrits dans les ordres normal et anti-normal. Cette op\'eration standard m\`ene au lemme final de la preuve:

\begin{lemma}[\textbf{Ordre normal et anti-normal}]\label{lem:Wick A Wick}\mbox{}\\
Soit $v\in S\gH$. On a  
\begin{equation}\label{eq:Wick A Wick}
a(v) ^n a ^*(v) ^n = \sum_{k=0} ^n \binom{n}{k} \frac{n!}{k!} a ^*(v) ^k a (v) ^k \mbox{ pour tout } n\in \N. 
\end{equation}
\end{lemma}

\begin{proof}[Preuve]
Pour faciliter le calcul il est utile de rappeller l'expression du $n$-i\`eme polyn\^ome de Laguerre   
\[
L_n (x) = \sum_{k=0} ^n \binom{n}{k} \frac{(-1) ^k}{k!} x ^k.
\]
Ces polyn\^omes satisfont la relation de r\'ecurrence 
\[
(n+1) L_{n+1} (x) = (2n +1) L_n (x) - x L_n (x) - n L_{n-1} (x)  
\]
et on peut voir que~\eqref{eq:Wick A Wick} se r\'e\'ecrit
\begin{equation}\label{eq:Wick A Wick proof}
a(v) ^n a ^*(v) ^n = \sum_{k=0} ^n c_{n,k} \, a^*(v) ^k a (v) ^k 
\end{equation}
o\`u les $c_{n,k}$ sont les coefficients du polyn\^ome
$$\tilde{L}_n (x) := n! \, L_n (-x).$$
Il suffit donc de montrer que 
\[
a(v) ^{n+1} a ^*(v) ^{n+1} = a^* (v)  a(v) ^n a ^*(v) ^n  a(v) + (2n+1) a(v) ^n a ^*(v) ^n - n ^2  a(v) ^{n-1} a ^*(v) ^{n-1}
\]
par une application r\'ep\'et\'ee de la CCR~\eqref{eq:CCR}.
\end{proof}

La formule~\eqref{eq:CKMR exact} se d\'eduit en combinant les Lemmes~\ref{lem:uk-g-uk=0},~\ref{lem:A Wick} et~\ref{lem:Wick A Wick} avec~\eqref{eq:Wick}.

\newpage

\section{\textbf{Localisation dans l'espace de Fock et applications}}\label{sec:locFock}

Nous nous tournons maintenant vers la premi\`ere implication du sch\'ema~\eqref{eq:schema deF 2}. Il va nous falloir convertir la convergence faible-$\ast$ des matrices densit\'e en convergence forte pour pouvoir appliquer le Th\'eor\`eme~\ref{thm:DeFinetti fort}. L'id\'ee est de localiser l'\'etat $\Gamma_N$ de d\'epart en utilisant soit des fonctions \`a support compact soit des projecteurs de rang fini. On peut alors travailler dans un cadre compact avec convergence $\gS^1$-forte, apppliquer le Th\'eor\`eme~\ref{thm:DeFinetti fort} et passer \`a la limite dans la localisation en derni\`ere \'etape. Plus pr\'ecis\'ement nous utiliserons une localisation en dimension finie, pour bien montrer que le th\'eor\`eme g\'en\'eral peut se d\'eduire de la preuve constructive en dimension finie pr\'esent\'ee au chapitre pr\'ec\'edent.

La difficult\'e ici est que la notion de localisation appropri\'ee pour un \'etat \`a $N$ corps (par exemple une fonction d'onde $\Psi_N \in L ^2 (\R ^{dN})$) est plus complexe que celle dont on a l'habitude pour des fonctions d'onde simples $\psi \in L ^2 (\R ^d)$. Il faut en fait travailler sur les matrices de densit\'e directement, et localiser toutes les matrices r\'eduites de telle fa\c{c}on que les matrices localis\'ees correspondent \`a un \'etat quantique. La proc\'edure de localisation peut faire perdre des particules et donc l'\'etat localis\'e ne sera en g\'en\'eral pas un \'etat \`a $N$ corps mais une superposition d'\'etats \`a $k$ corps, $0\leq k \leq N$, c'est \`a dire un \'etat sur l'espace de Fock. 

La proc\'edure de localisation que nous utiliserons est d\'ecrite Section~\ref{sec:loc rig}. Auparavant nous donnerons quelques consid\'erations heuristiques \`a la Section~\ref{sec:loc heur} pour pr\'eciser ce que nous disions plus haut, \`a savoir que la notion de localisation correcte dans $L ^2 (\R ^{dN})$ diff\`ere de la localisation habituelle dans $L ^2 (\R ^d)$. La Section~\ref{sec:proof deF faible} contient la preuve du th\'eor\`eme de de Finetti faible et d'un r\'esultat auxiliaire bien utile qui est une cons\'equence de la preuve par localisation. 

\subsection{Convergence faible et localisation pour un \'etat \`a deux corps}\label{sec:loc heur}\mbox{}\\\vspace{-0.4cm}

Les consid\'erations ci-dessous sont tir\'ees de~\cite{Lewin-11}. Prenons une suite d' \'etats  bosoniques \`a deux corps particuli\`erement simple 
\begin{equation}\label{eq:deux corps}
\Psi_n := \psi_n \otimes_s \phi_n = \frac{1}{\sqrt{2}} \left( \psi_n \otimes \phi_n + \phi_n \otimes \psi_n \right) \in L_s ^2 (\R ^{2d})
\end{equation}
avec $\psi_n$ et $\phi_n$ normalis\'ees dans $L ^2 (\R ^d)$. Ceci correspond \`a avoir une particule dans l'\'etat $\psi_n$ et une particule dans l'\'etat $\phi_n$. On supposera 
\[
 \bral \phi_n,\psi_n \ketr_{L ^2 (\R ^d)} = 0,
\]
ce qui assure $\norm{\Psi_n} = 1$. On pourra toujours supposer que 
\[
 \Psi_n \wto \Psi \mbox{ dans }  L ^2 (\R ^{2d})
\]
et la convergence est forte si et seulement si $\norm{\Psi} = \norm{\Psi_n}= 1$. Dans le cas o\`u de la masse se perd \`a la limite, i.e. $\norm{\Psi}<1$ la convergence reste faible. 

Nous travaillerons toujours dans un cadre localement compact et donc la seule source possible de perte de compacit\'e est une fuite de masse \`a l'infini~\cite{Lions-84,Lions-84b,Lions-85a,Lions-85b}. Une possibilit\'e est que les deux particules $\phi_n$ et $\psi_n$ disparaissent \`a l'infini
\[
\psi_n \wto 0, \: \phi_n \wto 0 \mbox{ dans } L^2 (\R ^{2d}) \mbox{ quand } n\to \infty, 
\]
auquel cas $\Psi_n \wto 0$ dans $L ^2 (\R ^{2d})$. C'est dans $L ^2 (\R ^{2d})$ la possibilit\'e la plus proche de la perte de masse simple dans $L ^2 (\R ^d)$, mais il existe d'autres possibilit\'es dans $L ^2 (\R ^{2d})$.
 
Un cas typique est celui o\`u seule une des deux particules part \`a l'infini, ce qu'on peut mat\'erialiser par 
\[
 \psi_n \wto 0 \mbox{ faiblement dans } L ^2 (\R ^d), \: \phi_n \to \phi \mbox{ fortement dans } L ^2 (\R ^{2d}).
\]
Pour la fuite de masse de $\psi_n$ on peut typiquement penser \`a l'exemple 
\begin{equation}\label{eq:fuite masse}
 \psi_n = \psi \left( . + x_n \right) 
\end{equation}
avec $|x_n| \to \infty$ quand $n\to \infty$ et $\psi$ disons r\'egulier \`a support compact. On a dans ce cas
\[
\Psi_n \wto 0 \mbox{ dans } L ^2 (\R ^{2d}) 
\]
mais pour des raisons physques \'evidentes on pr\'ef\'ererait avoir une notion de convergence faible assurant 
\begin{equation}\label{eq:geo conv}
\Psi_n \wto_{ g}  \half \phi.  
\end{equation}
En effet, puisque seule la particule dans l'\'etat $\psi_n$ dispara\^it \`a l'infini, il est naturel que l'\'etat limite soit un \'etat \`a une particule d\'ecrit par $\phi$. Nous notons cette convergence $\wto_{ g}$ car il s'agit pr\'ecis\'ement de la convergence dite ``g\'eom\'etrique'' discut\'ee par Mathieu Lewin dans~\cite{Lewin-11}. La difficult\'e est bien s\^ur que les membres de droite et de gauche de~\eqref{eq:geo conv} appartiennent \`a des espaces diff\'erents.

Pour introduire la notion de convergence correcte, il s'agit de regarder les matrices densit\'e de $\Psi_n$:
\begin{align}\label{eq:matr deux corps}
\gamma_{\Psi_n} ^{(2)} &= \ketl \phi_n \otimes_s \psi_n \ketr \bral \phi_n \otimes_s \psi _n \brar \wto_\ast 0 \mbox{ dans } \gS ^1\left (L ^2 (\R ^{2d}) \right) \nonumber\\
\gamma_{\Psi_n} ^{(1)} &= \half \ketl \phi_n \ketr \bral \phi_n  \brar + \half \ketl \psi_n \ketr \bral \psi_n  \brar \wto_\ast \frac{1}{2} \ketl \phi \ketr \bral \phi  \brar \mbox{ dans } \gS ^1\left (L ^2 (\R ^{d}) \right).
\end{align}
On voit alors bien que la paire $\left( \gamma_{\Psi_n} ^{(2)}, \gamma_{\Psi_n} ^{(1)}\right)$ converge vers la paire $\left( 0, \half \ketl \phi \ketr \bral \phi\brar\right)$ qui correspondent aux matrices de densit\'e de l'\'etat \`a un corps $\half \phi \in L ^2 (\R ^{d})$. Plus pr\'ecis\'ement, la notion de convergence g\'eom\'etrique est une notion de convergence dans l'espace de Fock (ici bosonique \`a deux particules)
\begin{equation}\label{eq:Fock deux corps}
\cF_s ^{\leq 2} (L ^2 (\R ^{d})) :=  \C \oplus L ^2 (\R ^d) \oplus L_s ^2 (\R^{2d}) 
\end{equation}
et on a au sens de la convergence g\'eom\'etrique sur $\gS ^1 \left( \cF_s ^{\leq 2} \right)$
\[
 0 \oplus 0 \oplus \ketl \Psi_n \ketr \bral \Psi_n \brar \wto_g  \half \oplus \half \ketl \phi \ketr \bral \phi  \brar \oplus 0,
\]
c'est \`a dire que toutes les matrices de densit\'e r\'eduites de l'\'etat de gauche convergent vers celles de l'\'etat de droite. On notera que la limite est bien de trace $1$ dans $\gS ^1 \left( \cF_s ^{\leq 2} \right)$, il n'y a donc pas de perte de masse dans $\cF_s ^{\leq 2}$. Plus pr\'ecis\'ement, dans $\cF_s ^{\leq N}$ \textbf{la perte de masse pour un \'etat \`a pur \`a $N$ particules est mat\'erialis\'ee par la convergence vers un \'etat mixte avec moins de particules.} 

\medskip

De la m\^eme fa\c{c}on que la convergence faible adapt\'ee aux probl\`emes \`a $N$ corps n'est pas la convergence faible $L^2 (\R^{dN})$ usuelle (a fortiori quand on prendra la limite $N\to \infty$), la bonne proc\'edure pour localiser un \'etat et ainsi transformer la convergence faible en convergence forte doit \^etre repens\'ee. Etant donn\'e un op\'erateur de localisation auto-adjoint positif $A$, disons $A= P$ un projecteur de rang fini ou $A = \chi$ la mutliplication par une fonction $\chi$ \`a support compact, on localise usuellement une fonction d'onde $\psi \in L ^2 (\R ^d)$ en introduisant 
\[
\psi_A = A \psi 
\]
ce qui correspond \`a faire l'association
\[
 \ketl \psi \ketr\bral \psi \brar \leftrightarrow \ketl A \psi \ketr\bral A \psi \brar.
\]
En imitant cette proc\'edure pour l'\'etat \`a deux particules~\eqref{eq:deux corps} on pourrait imaginer consid\'erer l'\'etat localis\'e d\'efini par sa matrice densit\'e \`a deux corps
\[
\gamma_{n,A}  ^{(2)}= \ketl A\otimes A \Psi_n \ketr \bral A\otimes A \Psi_n \brar.
\]
Il est alors clair que 
\[
 \gamma_{n,A} ^{(2)} \wto_\ast 0 \mbox{ dans } \gS^1 (\gH_s ^2),
\]
ce \`a quoi il fallait s'attendre, mais il est plus grave que l'on ait \'egalement pour la matrice \`a un corps correspondante
\[
 \gamma_{n,A} ^{(1)} \wto_\ast 0 \mbox{ dans } \gS ^1 (\gH)
\]
alors, que au vu de~\eqref{eq:matr deux corps} on voudrait avoir 
\[
 \gamma_{n,A} ^{(1)} \to  \frac{1}{2} \ketl A \phi \ketr \bral A \phi  \brar \mbox{ fortement dans } \gS ^1 (\gH).
\]
La solution \`a ce dilemme est de d\'efinir un \'etat localis\'e en demandant que ses matrices de densit\'e reduites soient $A\otimes A \gamma_{\Psi_n} ^{(2)} A\otimes A$ et $A\gamma_{\Psi_n} ^{(1)}A$. L'\'etat correspondant est alors uniquement d\'etermin\'e, et il se trouve qu'il s'agit d'un \'etat sur l'espace de Fock, comme nous l'expliquons \`a la section suivante.

\subsection{Localisation dans l'espace de Fock}\label{sec:loc rig}\mbox{}\\\vspace{-0.4cm}

Apr\`es les consid\'erations heuristiques pr\'ec\'edentes, nous introduisons maintenant la notion de localisation dans l'espace de Fock bosonique\footnote{La proc\'edure est la m\^eme pour des particules fermioniques.} 
\begin{align}\label{eq:Fock space}
\cF_s (\gH) &=  \C \oplus \gH \oplus \ldots \oplus \gH_s ^n \oplus \ldots\nonumber\\
\cF_s (L ^2 (\R ^d) ) &= \C \oplus L ^2 (\R ^d) \oplus \ldots \oplus L_s ^2 (\R ^{dn}) \oplus \ldots.
\end{align}
Dans le cadre de ce cours nous partirons toujours d'\'etats \`a $N$ corps, auquel cas il est suffisant de travailler dans l'espace de Fock tronqu\'e 
\begin{align}\label{eq:Fock space tronc}
\cF_s  ^{\leq N}(\gH) &=  \C \oplus \gH \oplus \ldots \oplus \gH_s ^n \oplus \ldots \oplus \gH_s ^N\nonumber\\
\cF_s ^{\leq N} (L ^2 (\R ^d) ) &= \C \oplus L ^2 (\R ^d) \oplus \ldots \oplus L_s ^2 (\R ^{dn}) \oplus \ldots \oplus L_s ^2 (\R ^{dN}).
\end{align}

\begin{definition}[\textbf{Etats bosoniques sur l'espace de Fock}]\label{def:Fock state}\mbox{}\\
Un \'etat sur l'espace de Fock est un op\'erateur auto-adjoint positif de trace $1$ sur $\cF_s$. On notera $\cS (\cF_s(\gH))$ l'ensemble des \'etats (bosoniques ici)
\begin{equation}\label{eq:Fock state}
\cS (\cF_s(\gH)) = \left\{ \Gamma \in \gS ^1 (\cF_s(\gH)), \Gamma = \Gamma ^* , \Gamma\geq 0, \tr_{\cF_s(\gH)} [\Gamma] = 1 \right\}.
\end{equation}
On dira qu'un \'etat est diagonal (stricto sensu, diagonal par blocs) si il peut s'\'ecrire
\begin{equation}\label{eq:Fock state diag}
\Gamma = G_{0}\oplus G_{1}\oplus \ldots \oplus G_n \oplus \ldots 
\end{equation}
avec $G_n \in \gS ^1 (\gH_s ^n)$. Un \'etat $\Gamma_N$ sur l'espace de Fock tronqu\'e $\cF_s ^{\leq N}(\gH)$, respectivement un \'etat diagonal sur l'espace de Fock tronqu\'e se d\'efinissent de la m\^eme mani\`ere. Pour un \'etat diagonal sur $\cF_s ^{\leq N}(\gH)$ de la forme 
\[
\Gamma = G_{0,N}\oplus G_{1,N}\oplus \ldots \oplus G_{N,N},  
\]
sa $n$-i\`eme matrice de densit\'e r\'eduite $\Gamma_N ^{(n)}$ est l'op\'erateur sur $\gH_s ^n$ donn\'e par 
\begin{equation}\label{eq:Fock mat red}
\Gamma_N ^{(n)} = {N\choose n}^{-1}\sum_{k=n}^N{k\choose n}\tr_{n+1\to k}G_{N,k}.
\end{equation}
\hfill\qed
\end{definition}

Le lecteur habitu\'e \`a ces concepts notera deux choses:
\begin{itemize}
\item Nous introduisons uniquement les concepts qui nous seront indispensables par la suite. On peut bien s\^ur d\'efinir les matrices r\'eduites d'\'etats g\'en\'eraux, mais nous ne nous en servirons pas par la suite. Pour un \'etat diagonal, la donn\'ee des matrices r\'eduites~\eqref{eq:Fock mat red} caract\'erisent compl\`etement l'\'etat. Pour un \'etat non diagonal il faut aussi sp\'ecifier ses matrices ``hors-diagonales'' $\Gamma ^{(p,q)}:\gH_s ^p \mapsto \gH_s ^q$ pour $p\neq q$.
\item La normalisation choisie dans~\eqref{eq:Fock mat red} n'est pas standard. Elle est choisie pour que, dans l'esprit du reste du cours, la $n$-i\`eme matrice r\'eduite d'un \'etat \`a $N$ particules (i.e. avec $G_{0,N} = \ldots = G_{N-1,N} = 0$ dans~\eqref{eq:Fock mat red}) soit de trace $1$. La convention standard serait plut\^ot qu'elle soit de trace ${N \choose n}$, ce qui est moins pratique pour appliquer le Th\'eor\`eme~\ref{thm:DeFinetti fort}.
\end{itemize}

Nous pouvons maintenant introduire le concept de localisation d'un \'etat. Nous nous limiterons aux \'etats \`a $N$ corps et aux op\'erateurs de localisation auto-adjoints, ce qui est suffisant pour nos besoins dans la suite. Le lemme/d\'efinition suivant est tir\'e de~\cite{Lewin-11}, on pourra en trouver d'autres versions par exemple dans~\cite{Ammari-04,DerGer-99,HaiLewSol_thermo-09}.

\begin{lemma}[\textbf{Localisation d'un \'etat \`a $N$ corps}]\label{lem:Fock loc}\mbox{}\\
Soit $\Gamma_N \in \cS (\gH_s ^N)$ un \'etat bosonique \`a $N$ corps et $A$ un op\'erateur auto-adjoint sur $\gH$ avec $0\leq A^2 \leq 1$. Il existe un unique \'etat diagonal $\Gamma_N ^A \in \cS (\cF_s ^{\leq N} (\gH))$ tel que 
\begin{equation}\label{eq:Fock loc mat}
\left(\Gamma_N ^A\right) ^{(n)} = A ^{\otimes n} \Gamma_N ^{(n)} A ^{\otimes n}
\end{equation}
pour tout $0 \leq n \leq N$. De plus, en \'ecrivant $\Gamma_N ^A$ sous la forme  
\[
\Gamma_N ^A = G_{0,N} ^A \oplus G_{1,N} ^A \oplus \ldots \oplus G_{N,N} ^A,  
\]
on a la relation fondamentale
\begin{equation}\label{eq:Fock funda rel}
\tr_{\gH_s ^n} \left[G_{N,n} ^A \right] =\tr_{\gH_s ^{N-n}} \left[G_{N,N-n} ^{\sqrt{1-A ^2}} \right].
\end{equation}
\end{lemma}

\begin{remark}[Localisation dans l'espace de Fock]\label{rem:Fock com loc}\mbox{}\\\vspace{-0,4cm}
\begin{enumerate}
\item La partie unicit\'e du lemme montre bien qu'on ne peut faire l'\'economie de travailler sur l'espace de Fock. L'\'etat localis\'e est en fait unique dans $\cS (\cF_s (\gH))$, mais pour le voir il faudrait introduire des d\'efinitions un peu plus g\'en\'erales, cf~\cite{Lewin-11}.
\item La relation~\eqref{eq:Fock funda rel} est une des pierres angulaires de la m\'ethode. En langage informel elle exprime le fait, que dans l'\'etat $\Gamma_N$, \textbf{la probabilit\'e d'avoir $n$ particules $A$-localis\'ees est \'egale \`a la probabilit\'e d'avoir $N-n$ particules $\sqrt{1-A^2}$ localis\'ees}. Pensons au cas d'une fonction de localisation particuli\`erement simple, $A= \one _{B(0,R)}$, la fonction indicatrice de la boule de rayon $R$. Nous disons alors simplement que la probabilit\'e d'avoir $n$ particules dans la boule \'egale la probabilit\'e d'avoir $N-n$ particules hors de la boule. \hfill\qed
\end{enumerate}

\begin{proof}[Preuve du Lemme~\ref{lem:Fock loc}.]
L'unicit\'e (en tous cas parmi les \'etats diagonaux de l'espace de Fock tronqu\'e) est une simple cons\'equence du fait que les matrices densit\'e d\'efinissent uniquement l'\'etat. On pourra voir les d\'etails dans~\cite{Lewin-11}.

Pour l'existence on peut utiliser l'identification usuelle 
$$\cF_s (A  \gH \oplus \sqrt{1-A ^2}\gH) \simeq \cF_s (A\gH) \otimes \cF_s (\sqrt{1-A ^2}\gH)$$
et d\'efinir l'\'etat localis\'e en prenant une trace partielle par rapport au deuxi\`eme espace dans le produit tensoriel du membre de droite. Nous allons suivre une route plus explicite mais \'equivalente. Pour simplifier les notations nous nous limitons au cas o\`u $A=P$ est un projecteur orthogonal et donc $\sqrt{1-A ^2} = \Pp$. 

Par d\'efinition et cyclicit\'e de la trace on obtient
\begin{align*}
P ^{\otimes n} \Gamma_N ^{(n)} P ^{\otimes n} &= \tr_{n+1\to N} \left[P ^{\otimes n} \otimes \one ^{\otimes (N-n)} \Gamma_N  P ^{\otimes n} \otimes \one ^{\otimes (N-n)} \right]\nonumber
\\&= \sum_{k=0} ^{N-n} {N-n \choose k} ^2 \tr_{n+1\to N} \left[P ^{\otimes n + k} \otimes \Pp ^{\otimes (N-n-k)} \Gamma_N  P ^{\otimes n+k} \otimes \Pp ^{\otimes (N-n-k)} \right]\nonumber
\\&= \sum_{k=n} ^{N} {N-n \choose k-n} ^2 \tr_{n+1\to N} \left[P ^{\otimes k } \otimes \Pp ^{\otimes (N-k)} \Gamma_N  P ^{\otimes k} \otimes \Pp ^{\otimes (N-k)} \right]\nonumber
\end{align*}
en \'ecrivant $\one = P + \Pp$ et en d\'eveloppant les termes $\one ^{\otimes (N-n)}$ par la formule du bin\^ome. Il suffit ensuite de noter que 
\begin{equation*}\label{eq:fact calcul}
{N-n \choose k-n} ^2 = { N \choose k} ^2 { N \choose n}{ k \choose n} ^{-1}
\end{equation*}
pour obtenir 
\begin{equation*}
P ^{\otimes n} \Gamma_N ^{(n)} P ^{\otimes n}  = \sum_{k=n} ^{N} { N \choose n} ^{-1} { k \choose n}  \tr_{n+1\to N} \left[G_{N,k} ^P \right]= \left(G_N  ^P \right)^{(n)}
\end{equation*}
avec (cf la d\'efinition~\eqref{eq:Fock mat red})
\begin{equation}\label{eq:Fock calcul loc}
G_{N,k} ^P =  { N \choose k} ^2 \tr_{k+1\to N} \left[P ^{\otimes k } \otimes \Pp ^{\otimes (N-k)} \Gamma_N  P ^{\otimes k} \otimes \Pp ^{\otimes (N-k)} \right]
\end{equation}
et
\[
G_N  ^P = G_{N,0} ^P \oplus \ldots \oplus G_{N,N} ^P.
\]
Il reste \`a montrer que $G_N ^P$ est bien un \'etat, c'est-\`a-dire que sa trace vaut $1$. Pour le voir il suffit d'\'ecrire
\begin{align*}
1 &= \tr_{\gH ^N} [\Gamma_N] = \tr_{\gH ^N} \left[\left(P+\Pp \right) ^{\otimes N} \Gamma_N \left(P+\Pp \right) ^{\otimes N}\right] \\
&= \sum_{k=0} ^N { N \choose k} ^2 \tr_{\gH ^N} \left[P ^{\otimes k } \otimes \Pp ^{\otimes (N-k)} \Gamma_N  P ^{\otimes k} \otimes \Pp ^{\otimes (N-k)}\right]
\\&= \sum_{k=0} ^N  \tr_{\gH ^k} \left[G_{N,k} ^P\right] = \tr_{\cF(\gH)} [G_{N} ^P].
\end{align*}
La relation~\eqref{eq:Fock funda rel} est une cons\'equence imm\'ediate de~\eqref{eq:Fock calcul loc} et de la sym\'etrie de $\Gamma_N$.
\end{proof}
\end{remark}

\subsection{Preuve du th\'eor\`eme de de Finetti faible et corollaires}\label{sec:proof deF faible}\mbox{}\\\vspace{-0.4cm}

Nous allons maintenant utiliser la proc\'edure de localisation que nous venons de d\'ecrire pour d\'emontrer la premi\`ere implication du sch\'ema~\eqref{eq:schema deF 2}. L'id\'ee est d'utiliser un projecteur de rang fini $P$ et d'utiliser~\eqref{eq:Fock mat red} avec~\eqref{eq:Fock loc mat}  pour \'ecrire (avec $n\in \N$ fixe)
\begin{align}\label{eq:wdeF formal}
P ^{\otimes n} \gamma_N ^{(n)} P ^{\otimes n} &= \sum_{k=n}^N {N\choose n}^{-1}{k\choose n}\tr_{n+1\to k} G_{N,k} ^P \nonumber
\\&\approx \sum_{k=n}^N \left(\frac{k}{N}\right) ^n \tr_{n+1\to k} G_{N,k} ^P
\end{align}
en utilisant l'estimation simple (voir le calcul dans~\cite{LewNamRou-13}, Equation (2.13))
\begin{equation}\label{eq:wdeF calcul}
{N\choose n}^{-1}{k\choose n} =  \left(\frac{k}{N}\right) ^n + O(N ^{-1}).
\end{equation}
Le raisonnement est alors le suivant: les termes o\`u $k$ est petit contribuent tr\`es peu \`a la somme~\eqref{eq:wdeF formal} \`a cause du facteur $\left(\frac{k}{N}\right) ^n$. Pour les termes o\`u $k$ est grand, on note que, \`a normalisation pr\`es, $G_{N,k} ^P$ est un \'etat bosonique \`a $k$ particules sur $P\gH$. On peut donc lui appliquer le th\'eor\`eme de de Finetti discut\'e au Chapitre~\ref{sec:deF finite dim}, sans se soucier d'une \'eventuelle perte de compacit\'e dans la limite $N\to \infty$ puisque $P\gH$ est de dimension finie. Comme $k$ est grand dans ces termes et $n$ est fixe, on obtient (toujours formellement)
$$\tr_{n+1\to k} G_{N,k} ^P \approx \tr_{\gH ^k} [G_{N,k} ^P] \int_{u \in SP\gH} d\nu_k (u) |u ^{\otimes n}\rangle \langle u ^{\otimes n}| $$
pour une certaine mesure $\nu_k$, et donc 
\[
P ^{\otimes n} \gamma_N ^{(n)} P ^{\otimes n} \approx \sum_{k \simeq N}^N \tr_{\gH ^k} [G_{N,k} ^P] \left(\frac{k}{N}\right) ^n \int_{u \in SP\gH} d\nu_k (u) |u ^{\otimes n}\rangle \langle u ^{\otimes n}|.
\]
Dans la limite $N\to \infty$, la somme discr\`ete devient une int\'egrale en $\lambda = k/N $, et en utilisant le fait que $G_{N} ^P$ est un \'etat pour g\'erer la normalisation il est naturel d'esp\'erer obtenir
\[
P ^{\otimes n} \gamma_N ^{(n)} P ^{\otimes n} \approx \int_{0} ^1 d\lambda \:\lambda ^n \int_{u \in SP\gH} d\nu_{\lambda} (u) |u ^{\otimes n}\rangle \langle u ^{\otimes n}|
\]
ce qu'on peut r\'e\'ecrire sous la forme~\eqref{eq:melange faible} en d\'efinissant (en coordonn\'ees ``radiales'' sur la boule unit\'e $BP\gH$)
\[
 d\mu (u) = d\mu \left( \norm{u},\frac{u}{\norm{u}} \right) := d \norm{u} ^2 \times d\nu_{\norm{u} ^2} \left(\frac{u}{\norm{u}}\right)
\]
Il restera en derni\`ere \'etape \`a appliquer cette proc\'edure pour une suite de projecteurs $P_{\ell} \to \one$ et \`a v\'erifier quelques relations de compatibilit\'e pour conclure. La mesure finale peut ne pas \^etre une probabilit\'e, ce qu'on peut compenser en rajoutant un delta \`a l'origine puisque cette op\'eration ne change aucune des formules pr\'ec\'edentes. 

On notera que (modulo quelques passages \`a la limite), cette preuve qui combine les m\'ethodes pr\'esent\'ees au Chapitre~\ref{sec:deF finite dim} et \`a la Section~\ref{sec:loc rig} donne une recette pour constuire la mesure de de Finetti ``\`a la main'', ce qui est bien utile en pratique, voir Chapitres~\ref{sec:Hartree} et~\ref{sec:NLS}. L'esprit de la preuve par localisation rappelle certains aspects de la m\'ethode de Ammari et Nier~\cite{Ammari-hdr,AmmNie-08}.

Passons maintenant \`a une preuve rigoureuse, suivant~\cite[Section 2]{LewNamRou-13}.

\begin{proof}[Preuve du Th\'eor\`eme~\ref{thm:DeFinetti faible}.]
On continue avec les notations ci-dessus. On d\'efinit $M_{P,N}^{(n)}$ une mesure sur $[0,1]$, \`a valeurs dans les matrices positives hermitiennes de taille $\dim\,\left(\otimes_s^n (P\gH)\right)$ par la formule suivante:
$$\dM_{P,N}^{(n)}(\lambda):=\sum_{k=n}^N\;\delta_{k/N}(\lambda)\;\tr_{n+1\to k}G^P_{N,k}.$$
On a alors, en utilisant~\eqref{eq:wdeF calcul} 
\begin{equation}\label{eq:wdeF preuve 1}
\tr\left|P^{\otimes n}\gamma^{(n)}_{\Gamma_N}P^{\otimes n}-\int_0^1 \lambda^n\,\dM_{P,N}^{(n)}(\lambda)\right| \leq \frac{C}{N}\sum_{k=n}^N\tr G^P_{N,k} \to 0 \mbox{ quand } N\to \infty.
\end{equation}
Comme $P$ est un projecteur de rang fini et $\gamma_N ^{(n)}$ converge faible-$\ast$ par hypoth\`ese, 
\begin{equation}\label{eq:wdeF preuve 2}
P^{\otimes n}\gamma^{(n)}_{N}P^{\otimes n}\to P^{\otimes n}\gamma^{(n)}P^{\otimes n} 
\end{equation}
fortement en norme de trace. D'autre part, 
$$\tr_{\gH^n}\left[\int_0^1\dM_{P,N}^{(n)}(\lambda)\right]=\sum_{k=n+1}^N\;\tr_{\gH^k}G^P_{N,k}\leq\sum_{k=0}^N\;\tr_{\gH^k}G^P_{N,k}=1,$$
donc $M_{P,N}^{(n)}$ est une suite de mesures positives born\'ees sur un espace de dimension fini compact (les matrices hermitiennes positives de taille $\dim P \gH$ ayant une trace plus petite que $1$). On peut donc en extraire une sous-suite convergeant faiblement au sens des mesures vers $M_P^{(n)}$. En combinant avec~\eqref{eq:wdeF preuve 1} et~\eqref{eq:wdeF preuve 2} on a 
\begin{equation}\label{eq:wdeF preuve 3}
P^{\otimes n}\gamma^{(n)}P^{\otimes n}=\int_0^1 \lambda^n\,\dM_P^{(n)}(\lambda).
\end{equation}

Il s'agit maintenant de voir que la suite de mesures $\left(M_P^{(n)}\right)_{n\in \N}$ que l'on vient d'obtenir est consistante au sens o\`u
,pour $n\geq0$, 
\begin{equation}
\int_{0} ^1 f(\lambda) \tr_{n+1} dM_P^{(n+1)} (\lambda)= \int_{0} ^1 f(\lambda) d M_P^{(n)} (\lambda)
\label{eq:A_P_consistent 2}
\end{equation}
pour tout fonction continue $f$ sur $[0,1]$ s'annulant en $0$, avec $\tr_{n+1}$ la trace partielle par rapport \`a la derni\`ere variable. On a 
\begin{align*}
\tr_{n+1} d M_{P,N}^{(n+1)}(\lambda)&=\sum_{k=n+1}^N\;\delta_{k/N}(\lambda)\;\tr_{n \to k} G_{N,k} \\
&=d M_{P,N}^{(n)}(\lambda)- \delta_{n/N}(\lambda) G^P_{N,n}
\end{align*}
et donc
\begin{align*}
\int_{0} ^1 f(\lambda) \tr_{\gH^n}\left|\tr_{n+1}M_{P,N}^{(n+1)}(\lambda)-M_{P,N}^{(n)}(\lambda)\right|&\leq  \int_{0} ^1 f(\lambda) \delta_{n/N}(\lambda) \tr_{\gH ^n} G^P_{N,n} 
\\&\leq  f\left(\frac{n}{N}\right)
\end{align*}
puisque $\tr_{\gH ^n} G^P_{N,n}\leq 1$. En passant \`a la limite on obtient bien~\eqref{eq:A_P_consistent 2} pour toute fonction $f$ telle que $f(0) = 0$. 

Nous appliquons maintenant le Th\'eor\`eme~\ref{thm:DeFinetti fort} en dimension finie. Stricto sensu, il ne s'applique que \`a $\lambda$ fixe, mais en approchant $dM_P ^{(n)}$ par des fonctions \'etag\'ees puis en passant \`a la limite, on obtient une mesure $\nu_P$ sur $[0,1]\times S\gH\cap (P\gH)$ telle que 
$$\int_0 ^1 f(\lambda) \dM_P^{(n)}(\lambda)=\int_{S\gH}\int_0 ^1 f(\lambda) d\nu_P(\lambda,u)|u^{\otimes n}\rangle\langle u^{\otimes n}|$$
pour toute fonction continue $f$ s'annulant en $0$. On a donc maintenant
\begin{align*}
P^{\otimes n}\gamma^{(n)}P^{\otimes n}&=\int_{0}^1\int_{S\gH}d\nu_{P}(\lambda,u)\,\lambda^n |u^{\otimes n}\rangle\langle u^{\otimes n}|\\
&=\int_{0}^1\int_{S\gH}d\nu_{P}(\lambda,u)\,|(\sqrt\lambda u)^{\otimes n}\rangle\langle (\sqrt\lambda u)^{\otimes n}|\\
&= \int_{B \gH} d\mu_P (u) |u^{\otimes n}\rangle\langle u^{\otimes n}| 
\end{align*}
en d\'efinissant la mesure $\mu_P$ en coordonn\'ees radiales. On est libre de rajouter une masse de Dirac \`a l'origine pour faire de $\mu_P$ une mesure de probabilit\'e.

L'argument peut-\^etre appliqu\'e \`a une suite de projecteurs de rang fini convergeant vers l'identit\'e. On a alors une suite de probabilit\'es $\mu_k$ sur $B\gH$ telles que 
$$P_k^{\otimes n}\gamma^{(n)}P_k^{\otimes n}=\int_{B\gH}d\mu_k(u)\,|u^{\otimes n}\rangle\langle u^{\otimes n}|.$$
En prenant une suite croissante de projecteurs (i.e. satisfaisant $P_k \gH \subset P_{k+1} \gH$) il est clair que $\mu_k$ co\"incide avec $\mu_\ell$ sur $P_{\ell} \gH$ pour $\ell\leq k$. Comme toutes ces mesures ont leur support dans un ensemble born\'e, il existe (voir par exemple~\cite[Lemme 1]{Skorokhod-74}) une unique mesure de probabilit\'e\footnote{Pour la construire, noter que la $\sigma$-fermeture de l'union pour $k\geq 0$ des bor\'eliens de $P_k\gH$ co\"incide avec les bor\'eliens de $\gH$.} $\mu$ sur $B\gH$ qui co\"incide avec $\mu_k$ sur $P_k\gH$ au sens o\`u:  
$$\int_{B\gH}d\mu_k(u)\,|u^{\otimes n}\rangle\langle u^{\otimes n}|=\int_{B\gH}d\mu(u)\,|(P_ku)^{\otimes n}\rangle\langle (P_ku)^{\otimes n}|.$$
On conclut donc 
$$P_k^{\otimes n}\gamma^{(n)}P_k^{\otimes n}=P_k^{\otimes n}\left(\int_{B\gH}d\mu(u)\,|u^{\otimes n}\rangle\langle u^{\otimes n}|\right)P_k^{\otimes n}$$
et il ne reste plus qu'\`a prendre la limite $k\to \infty$ pour obtenir l'existence de la mesure satisfaisant~\eqref{eq:melange faible}. 

Prouvons maintenant l'unicit\'e de la mesure. Soient $\mu$ et $\mu'$ satisfaisant
\begin{equation}
\int_{B\gH}|u^{\otimes k}\rangle\langle u^{\otimes k}|d\mu(u) -\int_{B\gH}|u^{\otimes k}\rangle\langle u^{\otimes k}|d\mu'(u) = 0
\label{eq:uniqueness}
\end{equation}
pour tout $k\geq1$. Soient $V = \mathrm{vect} (e_1,\ldots,e_d)$ un sous-espace de $\gH$ de dimension finie, $P_i=|e_i\rangle\langle e_i|$ les projections associ\'ees et $\mu_V,\mu'_V$ les projections cylindriques de $\mu$ et $\mu'$ sur~$V$. En applicant $P_{i_1}\otimes\cdots\otimes P_{i_k}$ \`a gauche et  $P_{j_1}\otimes\cdots\otimes P_{j_k}$ \`a droite de~\eqref{eq:uniqueness}, on obtient 
$$\int_{BV}u_{i_1}\cdots u_{i_k}\overline{u_{j_1}}\cdots \overline{u_{j_k}}\,d(\mu_V-\mu_V')(u)=0$$
pour tous les multi-indices $i_1,...,i_k$ et $j_1,...,j_k$ (on a not\'e $u=\sum_{j=1}^du_je_j$). D'autre part, par l'invariance $S^1$ des deux mesures, il est clair que si $k\neq \ell$,
$$\int_{BV}u_{i_1}\cdots u_{i_k}\overline{u_{j_1}}\cdots \overline{u_{j_\ell}}\,d(\mu_V-\mu_V')(u)=0.$$
Comme les polyn\^omes en les $u_i$ et $\overline{u_j}$ sont denses dans $C^0(BV,\C)$ (fonctions continues de la boule unit\'e de $V$), on d\'eduit que les projections cylindriques de $\mu$ et $\mu'$ sur $V$ co\"incident. Ceci \'etant vrai pour tout sous-espace $V$ de dimension finie, on conclut que les deux mesures doivent co\"incider partout.
\end{proof}

Nous donnons maintenant un cas particulier, adapt\'e aux besoins de la suite des notes, d'un corollaire bien utile de la m\'ethode de preuve ci-dessus (voir~\cite[Th\'eor\`eme 2.6]{LewNamRou-13} pour l'\'enonc\'e g\'en\'eral):

\begin{corollary}[\textbf{Localisation et mesure de de Finetti}]\label{cor:wdeF loc}\mbox{}\\
Soit $(\Gamma_N)_{N\in \N}$ une suite d'\'etats \`a $N$ corps sur $\gH = L ^2 (\R ^d)$ satisfaisant les hypoth\`eses du Th\'eor\`eme~\ref{thm:DeFinetti faible} et $\mu$ la mesure de de Finetti associ\'ee. Supposons que 
\begin{equation}\label{eq:wdeF cor assum}
\tr \left[ -\Delta \gamma_N^{(1)} \right] = \tr \left[ |\nabla| \gamma_N^{(1)} |\nabla| \right] \leq C 
\end{equation}
pour une constante ind\'ependante de $N$. Soit $\chi$ une fonction de localisation \`a support compact avec $0\leq \chi \leq 1$ et $G_N^{\chi}$ l'\'etat localis\'e d\'efini par le Lemme~\ref{lem:Fock loc}. Alors
\begin{equation}
\lim_{N\to\ii}\sum_{k=0}^Nf\left(\frac{k}{N}\right)\tr_{\gH^k} G_{N,k}^\chi = \int_{B\gH} d\mu(u)\;f(\|\chi u\|^2)
\label{eq:wdeF cor}
\end{equation}
pour toute fonction continue $f$ sur $[0,1]$.
\end{corollary}
 
\begin{proof}[Peuve du Corollaire~\ref{cor:wdeF loc}.]
Par densit\'e des polyn\^omes dans les fonctions continues sur $[0,1]$ il suffit de consid\'erer le cas $f(\lambda)=\lambda^n$ avec $n=0,1,...$. On utilise alors~\eqref{eq:Fock mat red},~\eqref{eq:Fock loc mat} et~\eqref{eq:wdeF calcul} \`a nouveau pour \'ecrire
\[
\left|\sum_{k=0}^N \left(\frac{k}{N}\right) ^n \tr_{\gH^k} G_{N,k}^\chi  - \tr_{\gH ^n} \left[\chi ^{\otimes n} \gamma_N ^{(n)} \chi ^{\otimes n} \right]\right| \to 0 \mbox{ quand } N\to \infty.
\]
L'hypoth\`ese~\eqref{eq:wdeF cor assum} assure que 
\[
 \tr \left[ \left( \sum_{j=1}^{n} -\Delta_j \right) \gamma_N^{(n)} \right] \leq C n
\]
et on peut donc supposer que 
\[
\left( \sum_{j=1}^{n} |\nabla|_j \right) \gamma_N^{(n)} \left( \sum_{j=1}^{n} |\nabla|_j \right) \wto_{\ast} \left( \sum_{j=1}^{n} |\nabla|_j \right) \gamma ^{(n)} \left( \sum_{j=1}^{n} |\nabla|_j \right)\mbox{ quand } N\to \infty. 
\]
L'op\'erateur de multiplication par $\chi$ est relativement compact par rapport au Laplacien, donc 
$$D_{n} ^\chi:= \chi ^{\otimes n} \left( \sum_{j=1}^{n} |\nabla|_j \right) ^{-1}$$
est un op\'erateur compact. On a alors 
\begin{align*}
 \tr_{\gH ^n} \left[\chi ^{\otimes n} \gamma_N ^{(n)} \chi ^{\otimes n} \right] & = \tr_{\gH ^n} \left[ D_{n} ^\chi \left( \sum_{j=1}^{n} |\nabla|_j \right) \gamma_N ^{(n)} \left( \sum_{j=1}^{n} |\nabla|_j \right)  D_{n} ^\chi \right] \\
&\to \tr_{\gH ^n} \left[D_{n} ^\chi  \left( \sum_{j=1}^{n} |\nabla|_j \right) \gamma ^{(n)}  \left( \sum_{j=1}^{n} |\nabla|_j \right)  D_{n} ^\chi \right]\\ 
 &= \tr_{\gH ^n} \left[\chi ^{\otimes n} \gamma ^{(n)} \chi ^{\otimes n} \right].
\end{align*}
On conclut alors en utilisant~\eqref{eq:melange faible} que
\begin{align*}
\sum_{k=0}^N \left(\frac{k}{N}\right) ^n \tr_{\gH^k} G_{N,k}^\chi &\to  \int_{B\gH} d\mu(u) \tr_{\gH ^n} \left[ |(\chi u)  ^{\otimes n}\rangle \langle (\chi u)  ^{\otimes n}|\right]\\
&= \int_{B\gH} d\mu(u) \norm{\chi u} ^{2n} = \int_{B\gH} d\mu(u) f(\norm{\chi u} ^{2}).
\end{align*}
\end{proof}

\begin{remark}[Mesure de de Finetti faible et perte de masse.]\mbox{}\\
\vspace{-0.4cm}
\begin{enumerate}
\item L'hypoth\`ese~\eqref{eq:wdeF cor assum} est l\`a pour assurer la compacit\'e forte locale des matrices de densit\'e, cf la preuve. Elle est bien s\^ur tr\`es naturelle pour les applications que nous visons (\'energie cin\'etique uniform\'ement born\'ee). La convergence~\eqref{eq:wdeF cor} signifie que la masse de la mesure de de Finetti $\mu$ sur la sph\`ere $\{\|u\|^2=\lambda\}$ correspond \`a la probabilit\'e qu'une fraction $\lambda$ des particules ne s'\'echappe pas \`a l'infini.
\item Nous utiliserons ce corollaire pour d\'eduire des informations sur les particules qui s'\'echappent \`a l'infini de la fa\c{c}on suivante. On d\'efinit la fonction  
$$ \eta = \sqrt{1-\chi ^2}$$ 
qui ``localise pr\`es de l'infini''. On ne peut bien s\^ur pas appliquer le r\'esultat directement \`a l'\'etat localis\'e $G_N^{\eta}$ mais on peut utiliser la relation~\eqref{eq:Fock funda rel} pour obtenir
\begin{align}
\lim_{N\to\ii}\sum_{k=0}^N f\left(\frac{k}{N}\right)\tr_{\gH^k}G_{N,k}^{\eta}&=\lim_{N\to\ii}\sum_{k=0}^N f\left(1-\frac{k}{N}\right)\tr_{\gH^k}G_{N,k}^{\chi}\nn\\
&=\int_{B\gH}d\mu(u)\;f\left(1- \|\chi u \|^2\right), \label{eq:wdeF cor 2}
\end{align}
ce qui fournit un certain contr\^ole sur la perte de masse \`a l'infini, encod\'e par la mesure de de Finetti qui pourtant d\'ecrit par d\'efinition les particules qui restent pi\'eg\'ees. Typiquement, on appliquera~\eqref{eq:wdeF cor 2} avec $f (\lambda) \simeq \eH (\lambda)$, l'\'energie de Hartree \`a masse $\lambda$, voir le chapitre suivant.  
\end{enumerate}
\hfill\qed
\end{remark}

\newpage

\section{\textbf{D\'erivation de la th\'eorie de Hartree: cas g\'en\'eral}}\label{sec:Hartree}

On se tourne maintenant vers le cas g\'en\'eral de l'obtention de la fonctionnelle de Hartree comme limite du probl\`eme \`a $N$ corps dans le r\'egime de champ moyen. Au Chapitre~\ref{sec:quant} nous avons vu que sous des hypoth\`eses physiques simplificatrices, le r\'esultat \'etait une cons\'equence assez directe des th\'eor\`emes de de Finetti faible et fort. Le cas g\'en\'eral demande une discussion plus pouss\'ee et nous serons amen\'es \`a utiliser pleinement les outils de localisation dans l'espace de Fock introduits au Chapitre~\ref{sec:locFock}.

Le cadre est maintenant celui de particules dans un potentiel $V$ non pi\'egeant interagissant via un potentiel d'interaction qui peut potentiellement avoir des \'etats li\'es. Rappellons les notations: le Hamiltonien du syt\`eme complet est 
\begin{equation}\label{eq:bis quant hamil}
H_N ^V = \sum_{j=1} ^N T_j + \frac{1}{N-1}\sum_{1\leq i < j \leq N} w(x_i-x_j), 
\end{equation}
agissant sur l'espace de Hilbert $\gH_s ^N = \bigotimes_s ^N \gH$, avec $\gH = L ^2 (\R ^d)$. Le Hamiltonien \`a un corps est  
\begin{equation}\label{eq:bis op Schro}
T = - \Delta + V, 
\end{equation}
que l'on suppose auto-adjoint et born\'e inf\'erieurement. Nous avons sp\'ecifi\'e la pr\'esence du potentiel $V$ dans la notation~\eqref{eq:bis quant hamil} car nous serons amen\'es \`a consid\'erer le syst\`eme pour les particules perdues \`a l'infini d\'ecrit par le Hamiltonien $H_N ^0$ o\`u on prend $V\equiv 0$. On peut faire les m\^emes g\'en\'eralisations que celles mentionn\'ees \`a la Remarque~\ref{rem:energie cin} mais nous ne traiterons explicitement que le cas mod\`ele ci-dessus.

Le potentiel d'interaction $w:\R \mapsto \R$ sera relativement born\'e par rapport \`a $T$:
\begin{equation}\label{eq:bis T controls w}
 -\beta_-(T_1+ T_2)-C\leq w(x_1-x_2) \leq \beta_+ (T_1+T_2) + C,
\end{equation}
sym\'etrique
\[
w(-x) = w(x), 
\]
et d\'ecroissant \`a l'infini
\begin{equation}\label{eq:bis decrease w}
w\in L^p (\R ^d) + L ^{\infty} (\R ^d), \max(1,d/2) < p < \infty \to 0, w(x) \to 0 \mbox{ quand }  |x|\to\infty.
\end{equation}
Ceci assure que $H_N$ est auto-adjoint et born\'e inf\'erieurement. A nouveau, il est assez vain de consid\'erer des potentiels \`a un corps partiellement confinants et donc on supposera que $V$ n'est pi\'egeant dans aucune direction:
\begin{equation}\label{eq:bis sans confinement}
V \in L ^p (\R ^d) + L ^{\infty} (\R ^d), \max{1,d/2}\leq p <\infty,  V(x) \to 0 \mbox{ quand } |x|\to \infty.
\end{equation}
L'\'energie fondamentale de~\eqref{eq:quant hamil} est toujours donn\'ee par
\begin{equation}\label{eq:bis quant ground state}
E ^V (N) = \inf \sigma_{\gH ^N} H_N ^V  = \inf_{\Psi \in \gH ^N, \norm{\Psi} = 1} \left\langle \Psi, H_N ^V \Psi \right\rangle_{\gH ^N}.
\end{equation}
Enfin, rappellons les objets limites. La fonctionnelle de Hartree avec le potentiel $V$ est donn\'ee par 
\begin{equation}\label{eq:bis hartree func}
\EH  ^V [u] := \int_{\R ^d } |\nabla u| ^2 + V |u| ^2 + \frac12 \iint_{\R ^d \times \R ^d } |u(x)| ^2 w(x-y) |u(y)| ^2 dx dy
\end{equation}
et on utilisera la notation $\EH ^0$ pour la fonctionnelle invariante par translation o\`u $V\equiv 0$. L'\'energie Hartree \`a masse $\lambda$ est donn\'ee par 
\begin{equation}\label{eq:bis Hartree perte masse}
\eH  ^V(\lambda) := \inf_{\norm{u} ^2 = \lambda} \EH ^V [u],0\leq \lambda \leq 1.
\end{equation}
Sous les hypoth\`eses ci-dessus on aura toujours l'in\'egalit\'e de liaison large 
\begin{equation}\label{eq:hartree liaison}
\eH ^V (1) \leq \eH ^V (\lambda) + \eH ^0 (1-\lambda) 
\end{equation}
qui se prouve ais\'ement en \'evaluant l'\'energie d'une suite de fonctions avec une masse $\lambda$ dans le puits de potentiel de $V$ et une masse $1-\lambda$ s'\'echappant \`a l'infini. Nous allons d\'emontrer le th\'eor\`eme suivant, issu de~\cite{LewNamRou-13} (cas particulier du Th\'eor\`eme 1.1).

\begin{theorem}[\textbf{D\'erivation de la th\'eorie de Hartree, cas g\'en\'eral}]\label{thm:hartree general}\mbox{}\\
Sous les hypoth\`eses pr\'ec\'edentes on a les r\'esultats suivants:

\medskip

\noindent$(i)$ \underline{Convergence de l'\'energie.}
\begin{equation}
\boxed{\lim_{N\to\ii}\frac{E^V(N)}{N}=\eH^V(1).}
\label{eq:general ener}
\end{equation}

\medskip

\noindent $(ii)$ \underline{Convergence des \'etats.} Soit une suite $\Psi_N$ normalis\'ee dans $L ^2 (\R ^{dN})$ de quasi-minimiseurs pour $H_N ^V$:  
\begin{equation}\label{eq:quasi min}
\pscal{\Psi_N,H_N^V\Psi_N}= E^V(N)+o(N) \mbox{ quand } N \to \infty, 
\end{equation}
et $\gamma^{(k)}_N$ les matrices de densit\'e r\'eduites correspondantes. Il existe $\mu \in \PP (B\gH)$ une probabilit\'e sur la boule unit\'e de $\gH$ avec $\mu (\cM ^V) = 1$ o\`u
\begin{equation}
\cM^V =\Big\{u\in B\gH\ :\ \cEH^V [u]=e^V_H(\norm{u}^2)=e^V_H(1)-e^0_H(1-\norm{u}^2)\Big\},
\label{eq:def_M_V} 
\end{equation}
telle que, le long d'une sous-suite 
\begin{equation}
\boxed{\gamma^{(k)}_{N_j}\wto_*\int_{\cM^V}|u^{\otimes k}\rangle\langle u^{\otimes k}|\,d\mu(u)}
\label{eq:general state}
\end{equation}
faiblement-$\ast$ dans $\gS^1(\gH^k)$, pour tout $k\geq1$.

\medskip

\noindent $(iii)$ Si de plus on a l'in\'egalit\'e de liaison stricte 
\begin{equation}
\eH^V(1)<\eH^V(\lambda)+\eH^0(1-\lambda)
\label{eq:hartree liaison stricte}
\end{equation}
pour tout $0\leq\lambda<1$, la mesure $\mu$ est \`a support dans la sph\`ere $S\gH$ et la limite~\eqref{eq:general state} a lieu en norme de trace. En particulier, si $\eH^V(1)$ a un minimiseur $u_H$ unique \`a une phase pr\`es, alors on a pour toute la suite
\begin{equation}
\boxed{\gamma^{(k)}_{N}\to|u_H^{\otimes k}\rangle\langle u_H^{\otimes k}| \mbox{  fortement dans } \gS ^1(\gH ^k)} 
\label{eq:general BEC}
\end{equation} 
pour tout $k\geq1$ fixe.
\end{theorem}

La preuve se fait en deux \'etapes. Dans un premier temps (Section~\ref{sec:trans inv}) on consid\`ere le cas compl\`etement invariant par translations o\`u $V\equiv 0$, qui d\'ecrira les particules s'\'echappant loin du puits de potentiel de $V$ dans le cas g\'en\'eral. On montre que l'\'energie $\eH ^0 (1)$ est bien la limite de $ N ^{-1} E ^0 (N)$. Dans ce cas on ne peut esp\'erer beaucoup plus puisqu'il existe toujours des suites minimisantes dont toutes les matrices de densit\'e convergent vers $0$ en norme de trace. En plus des outils d\'ej\`a d\'ecrits on utilisera une id\'ee de Lieb, Thirring et Yau~\cite{LieThi-84,LieYau-87} pour r\'ecup\'erer un peu de compacit\'e. 

On utilise ensuite pleinement les m\'ethodes de localisation du Chapitre~\ref{sec:locFock} pour traiter le cas g\'en\'eral (Section~\ref{sec:general general}). Dans l'esprit du principe de concentration-compacit\'e on localisera les suites minimisantes \`a l'int\'erieur et \`a l'ext\'erieur d'une boule. La partie localis\'ee \`a l'int\'erieur de la boule est d\'ecrite par la limite faible-$\ast$ des matrices de densit\'e pour lesquelles on peut utiliser le th\'eor\`eme de de Finetti faible. La partie localis\'ee \`a l'ext\'erieur de la boule ne verra plus le potentiel $V$ et on pourra donc lui appliquer les r\'esultats de la Section~\ref{sec:trans inv}.

\subsection{Le probl\`eme invariant par translation}\label{sec:trans inv}\mbox{}\\\vspace{-0.4cm}

On traite ici le cas particulier o\`u $V \equiv 0$. Dans ce cas il est possible de construire des suites de quasi-minimiseurs (au sens de~\eqref{eq:quasi min}) pour $E ^0 (N)$ avec $\gamma_N ^{(1)} (.-y_N) \wto_{\ast} 0$ pour toute suite de translations $x_N$. On peut donc avoir \emph{\'evanescence} dans la terminologie de Lions~\cite{Lions-84,Lions-84b} et sans travail sp\'ecifique on ne peut esp\'erer mieux que la convergence de l'\'energie. En effet, il est possible de construire un \'etat o\`u le mouvement relatif des particules est li\'e par le potentiel d'interaction $w$ mais o\`u il y a \'evanescence pour le centre de masse, ce qui entra\^ine l'\'evanescence de la suite. On se contentera donc de prouver la convergence de l'\'energie:

\begin{theorem}[\textbf{Syst\`emes invariants par translation}]\label{thm:trans inv}\mbox{}\\
Sous les hypoth\`es pr\'ec\'edentes on a 
\begin{equation}\label{eq:triv ener}
\boxed{ \lim_{N\to \infty} \frac{E^0(N)}{N} =e^0_{\rm H}(1).}
\end{equation}
\end{theorem}

La premi\`ere \'etape consiste \`a utiliser une partie du potentiel d'interaction pour cr\'eer un potentiel \`a un corps attractif (grosso-modo cela revient \`a retirer le degr\'e de libert\'e du centre de masse). Ceci d\'efinira un probl\`eme auxilaire dont l'\'energie est proche du syst\`eme de d\'epart, mais pour lequel les particules resteront toujours pi\'eg\'ees. Ceci est l'objet du lemme suivant qui est inpir\'e par les travaux~\cite{LieThi-84,LieYau-87}:

\begin{lemma}[\textbf{Probl\`eme auxilaire avec liaison}]\label{lem:triv aux}\mbox{}\\
On s\'epare $w = w ^+ - w ^-$ en partie positive et n\'egative et on d\'efinit pour un certain $\ep >0$
$$ \wep (x) = w (x) + \ep w ^- (x).$$ 
Soit le Hamiltonien auxiliaire
\begin{equation}\label{eq:aux hamil}
H_N ^{\eps}= \sum_{i=1}^{N}\Big(K_i -\eps w_-(x_i) \Big) + \frac{1}{N-1} \sum_{1\leq i<j \leq } w_\eps (x_i-x_j) 
\end{equation}
et $\Eep (N)$ l'\'energie fondamentale associ\'ee. On a alors
\begin{equation}\label{eq:aux low bound}
\aep:= \lim_{N\to \infty} \frac{\Eep(N)}{N} \leq \lim_{N\to \infty} \frac{E ^0 (N)}{N} =:a.
\end{equation}
\end{lemma}

\begin{proof}[Preuve.]
Pour tout $\Psi \in L_s ^2 (\R ^{dN})$ on peut \'ecrire en utilisant la sym\'etrie 
\[
\left\langle \Psi, \left( \sum_{i=1}^{N} -\Delta_i \right) \Psi \right\rangle = \frac{N}{N-1} \left\langle \Psi, \left( \sum_{i=1}^{N-1} -\Delta_i \right) \Psi \right\rangle.
\]
De la m\^eme fa\c{c}on
\[
\frac{2}{N(N-1)} \left\langle \Psi, \sum_{1\leq i<j \leq N}  w (x_i-x_j)\Psi \right\rangle = \tr_{\gH ^2} \left[w \gammaP ^{(2)}\right] =  \tr_{\gH ^2} \left[\wep \gammaP ^{(2)}\right] - \eps \tr_{\gH ^2} \left[w^- \gammaP ^{(2)}\right]
\]
avec $\gammaP = |\Psi \rangle \langle \Psi |$, puis 
\[
 \tr_{\gH ^2} \left[\wep \gammaP ^{(2)}\right] = \frac{2}{(N-1)(N-2)} \tr \left[\sum_{1\leq i<j \leq N-1} w_\eps(x_i-x_j) \gammaP \right]
\]
et 
\[
 \eps \tr_{\gH ^2} \left[w^- \gammaP ^{(2)}\right]  = \frac{\eps}{N-1} \tr \left[\sum_{i=1} ^{N-1} w ^- (x_i-x_N) \gammaP \right].
\]
Tout ceci implique 
\begin{multline}
\frac{N-1}{N}\langle \Psi , H_N^0  \Psi \rangle\\ =\left\langle \Psi, \left( \sum_{i=1}^{N-1}\Big(-\Delta_i -\eps w ^- (x_i-x_N) \Big) + \frac{1}{N-2} \sum_{1\leq i<j \leq N-1} w_\eps(x_i-x_j) \right) \Psi \right\rangle.
\label{eq:trick-Nam}
\end{multline}
Dans l'\'equation pr\'ec\'edente le Hamiltonien entre parenth\`eses depend de $x_N$ via le potentiel \`a un corps $\eps w ^- (x_i-x_N)$  mais comme les autres termes sont invariants par translation, le bas du spectre est en fait ind\'ependant de $x_N$. On a donc 
\begin{equation}\label{eq:aux low bound pre}
\frac{N-1}{N}\langle \Psi , H_N^0  \Psi \rangle \geq E^\eps(N-1) \langle \Psi,\Psi\rangle  
\end{equation}
pour tout $\Psi\in L ^2_s (\R ^{dN})$, ce qui implique
$$\frac{E^0(N)}{N}\ge \frac{E ^\eps(N-1)}{N-1}.$$
Les suites $E^0(N)/N$ et $E^\eps(N)/N$ sont croissantes puisque d\'efinies par des probl\`emes de minimisation sur des espaces variationnels d\'ecroissants (cf~\eqref{eq:MF quant form 1}). L'utilisation d'\'etats test simples montre par ailleurs que 
$$ \frac{E^0(N)}{N} \leq 0,\quad \frac{E^\eps(N)}{N} \leq 0$$
et donc les que limites $a$ et $\aep$ existent; ~\eqref{eq:aux low bound pre} implique alors clairement~\eqref{eq:aux low bound}.
\end{proof}

On va maintenant s'employer \`a prouver une borne inf\'erieure sur $\Eep(N)$ pour $\epsilon$ suffisament petit. La raison pour laquelle il est plus ais\'e de travailler sur $\Eep(N)$ est que toutes les suites de matrices densit\'e r\'eduites correspondantes seront compactes dans $\gS ^1$. En effet, la compacit\'e ou son absence (en langage physique, l'existence ou l'absence de liaisons entre les particules) r\'esulte d'une comparaison entre les effets attractifs et r\'epulsifs des potentiels \`a un et deux corps. Ici les potentiels \`a un corps $\ep w ^-$ et \`a deux corps $\wep$ sont bien \'equilibr\'es puisque ils ont \'et\'es construits pr\'ecis\'ement en recherchant cet effet \`a partir du m\^eme potentiel \`a deux corps originel $w$. C'est cette observation qui nous permet de conclure maintenant la preuve du Th\'eor\`eme~\ref{thm:trans inv}.

\begin{proof}[Preuve du Th\'eor\`eme~\ref{thm:trans inv}.]
Comme d'habitude, seule la borne inf\'erieure est non-triviale. Comme $\eH ^0 (1) \leq 0$ on peut supposer $a<0$ car sinon $a=e^0_{\rm H}(1)=0$ et il n'y a rien \`a prouver. On va d\'emontrer la borne inf\'erieure 
\begin{equation}\label{eq:aux prob bound}
 \aep \geq e_{\rm H} ^{\eps}(1) := \inf_{\norm{u}^2=1} \left\{ \langle u,(-\Delta-\eps w_-) u\rangle + \frac{1}{2} \iint |u(x)|^2 w_\eps (x-y) |u(y)|^2 dxdy \right\},
\end{equation}
et comme il est ais\'e de montrer que $e_{\rm H} ^{\eps}(1) \to \eH  ^{0} (1)$ quand $\ep \to$ on obtiendra~\eqref{eq:triv ener} en combinant avec~\eqref{eq:aux low bound}.  

Soit une suite $\Psi_N$ de fonctions d'onde telles que 
$$\langle \Psi_N, H_N ^{\eps} \Psi_N \rangle = \Eep(N)+o(N).$$ 
et $\gamma_N^{(k)}$ les matrices de densit\'e r\'eduites correspondantes. Alors 
$$
a_\eps = \lim_{N\to \infty}\frac{\langle \Psi_N, H_N ^{\eps} \Psi_N \rangle}{N}=  \lim_{N\to \infty} \left( \Tr_\gH\left[(-\Delta-\eps w_-)\gamma_N^{(1)}\right] + \frac{1}{2} \Tr_{\gH^2}\left[w_\eps \gamma_N^{(2)}\right] \right).
$$ 
Modulo l'extraction diagonale habituelle, on suppose que 
$$\gamma^{(k)}_N\wto_\ast\gamma^{(k)}$$
et nous allons montrer que la convergence est en r\'ealit\'e forte. On utilisera le th\'eor\`eme de de Finetti fort pour obtenir la th\'eorie de Hartree dans un second temps.

En prenant une partition de l'unit\'e $\chi_R^2 +\eta_R^2=1$ et en utilisant le Lemme~\ref{lem:loc ener} on a 
\begin{multline}
\label{eq:split-energy-translation}
a_\eps \ge  \liminf_{R\to \infty}\liminf_{N\to \infty} \left\{ \Tr_\gH\left[(-\Delta-\eps w_-) \chi_R \gamma_N^{(1)}\chi_R\right] + \frac{1}{2} \Tr_{\gH^2}\left[w_\eps \chi_R^{\otimes 2}\gamma_N^{(2)}\chi_R^{\otimes 2}\right] \right. \\
\left. + \Tr_\gH\left[ -\Delta \eta_R \gamma_N^{(1)}\eta_R\right] + \frac{1}{2} \Tr_{\gH^2}\left[w_\eps \eta_R^{\otimes 2}\gamma_N^{(2)}\eta_R^{\otimes 2}\right] \right \} 
\end{multline}
et nous allons utiliser les \'etats $\chi_R$-- et $\eta_R$--localis\'es $G_N^\chi$ et $G_N^\eta$ d\'efinis \`a partir de $\Psi_N$ par le Lemme~\ref{lem:Fock loc} pour estimer s\'epar\'ement les deux termes du membre de droite de~\eqref{eq:split-energy-translation}.

\medskip

\paragraph*{\bf Le terme $\chi_R$-localis\'e.} En utilisant~\eqref{eq:Fock mat red} et~\eqref{eq:Fock loc mat} on a 
\begin{multline}
\label{eq:localization-chi-eps}
\Tr_\gH\left[(-\Delta-\eps w_-) \chi_R \gamma_N^{(1)}\chi_R\right] + \frac{1}{2} \Tr_{\gH^2}\left[w_\eps \chi_R^{\otimes 2}\gamma_N^{(2)}\chi_R^{\otimes 2}\right]\\
= \frac{1}{N} \sum_{k=1}^N\Tr_{\gH^k}\left[ \left( \sum_{i=1}^k (-\Delta-\eps w_-)_i + \frac{1}{N-1} \sum_{i<j}^k w_\eps (x_i-x_j)\right) G^\chi_{N,k} \right]
\end{multline}
On applique l'in\'egalit\'e
\begin{equation}\label{eq:A-tB}
A+tB = (1-t)A+ t(A+B) \ge (1-t) \inf \sigma(A) + t \inf\sigma  (A+B) 
\end{equation}
avec 
\begin{equation*}
A=\sum_{\ell =1}^k (-\Delta-\eps w_-)_\ell,\quad  A+B =H_{\eps,k},\quad t=(k-1)/(N-1). 
\end{equation*}
On a 
$$ \lim_{\eps \to 0} \inf \sigma(-\Delta-\eps w_-) = \inf \sigma(-\Delta) = 0$$
et comme on suppose que $a<0$, on a, pour $\eps$ assez petit
$$\inf \sigma(-\Delta-\eps w_-) > a \ge a_\eps \ge k^{-1}\inf \sigma(H_{\eps,k}),$$
soit
$$\inf\sigma(A)\ge \inf \sigma(A+B).$$
On peut alors \'ecrire   
$$\Tr_\gH \left[(-\Delta-\eps w_-) \chi_R \gamma_N^{(1)}\chi_R\right] + \frac{1}{2} \Tr_{\gH^2}\left[w_\eps \chi_R^{\otimes 2}\gamma_N^{(2)}\chi_R^{\otimes 2}\right]  \ge \sum_{k=1}^N  \frac{k\Tr G^{\chi}_{N,k}}{N} \; \frac{\Eep (k)}{k},$$
mais comme
$$\sum_{k=0}^N \frac{k\Tr G^\chi_{N,k}}{N} = \Tr \left[\chi_R^2 \gamma_N^{(1)}\right]\underset{N\to\ii}{\longrightarrow}\Tr \left[\chi_R^2 \gamma^{(1)}\right]\quad \text{et}\quad \lim_{k\to \infty} \frac{\Eep (k)}{k} =a_\eps,$$
et que $k\mapsto \frac{\Eep (k)}{k}$ est croissante on conclut que
\begin{equation}
\label{eq:energy-chiR-translation}
\liminf_{N\to \infty} \left( \Tr_\gH\left[(-\Delta-\eps w_-) \chi_R \gamma_N^{(1)}\chi_R\right] + \frac{1}{2} \Tr_{\gH^2}\left[w_\eps \chi_R^{\otimes 2}\gamma_N^{(2)}\chi_R^{\otimes 2}\right] \right)
\geq a_\eps \Tr \left[\chi_R^2 \gamma^{(1)}\right]
\end{equation}
par convergence monotone. 
 
\medskip

\paragraph*{\bf Le terme $\eta_R$-localis\'e.} En utilisant les re\'sultats de la Section~\ref{sec:loc rig} comme ci-dessus on a 
\begin{multline}
\label{eq:localization-translation-eta}
\Tr_\gH\left[ T \eta_R \gamma_N^{(1)}\eta_R\right] + \frac{1}{2} \Tr_{\gH^2}\left[w_\eps \eta_R^{\otimes 2}\gamma_N^{(2)}\eta_R^{\otimes 2}\right]\\
= \frac{1}{N} \sum_{k=1}^N \Tr_{\gH^k}\left[ \left( \sum_{i=1}^k -\Delta_i + \frac{1}{N-1} \sum_{i<j}^k w_\eps (x_i-x_j)\right) {G}^{\eta}_{N,k} \right].
\end{multline}
Ici on utilise 
$$-\Delta \ge 0,\quad w_\eps = w+ 2\eps w ^- \ge (1-2\eps)w \mbox{ et } E^0(k)\le ak<0$$ 
pour obtenir 
\bqq
\sum_{i=1}^k -\Delta_i + \frac{1}{N-1} \sum_{i<j}^k w_\eps (x_i-x_j) &\ge& \frac{(1-2\eps)(k-1)}{N-1} H^0_k \hfill\\
&\ge &  \frac{(1-2\eps)(k-1)}{N-1} E^0 (k) \ge E^0(k) - 2\eps a k
\eqq
pour tout $k\ge 1$. En combinant avec~\eqref{eq:localization-translation-eta} on obtient
\begin{equation*}
\Tr_\gH\left[ -\Delta \eta_R \gamma_N^{(1)}\eta_R\right] + \frac{1}{2} \Tr_{\gH^2}\left[w_\eps \eta_R^{\otimes 2}\gamma_N^{(2)}\eta_R^{\otimes 2}\right] \ge \sum_{k=1}^N  \frac{k\Tr G^{\eta}_{N,k}}{N} \cdot \left(\frac{E^0(k)}{k} - 2\eps a \right).
\end{equation*}
On d\'eduit enfin
\begin{multline}
\label{eq:energy-etaR-translation}
\liminf_{N\to \infty} \left( \Tr_\gH\left[ T \eta_R \gamma_N^{(1)}\eta_R\right] + \frac{1}{2} \Tr_{\gH^2}\left[w_\eps \eta_R^{\otimes 2}\gamma_N^{(2)}\eta_R^{\otimes 2}\right]\right)\geq (1-2 \eps) a \big(1-\Tr[\chi_R^2 \gamma^{(1)}]\big) 
\end{multline}
en utilisant
$$\sum_{k=0}^N \frac{k}{N} \Tr G^\eta_{N,k} = \Tr \left[\eta_R^2 \gamma_N^{(1)} \right]\underset{N\to\ii}{\longrightarrow}1-\Tr \left[\chi_R^2 \gamma^{(1)}\right]\quad \text{et}\quad \lim_{k\to \infty} \frac{E^0(k)}{k}  =a$$
et le fait que $k\mapsto \frac{E^0(k)}{k}$ est une suite croissante.

\medskip

\paragraph*{\bf Conclusion.}

En ins\'erant~\eqref{eq:energy-chiR-translation} et~\eqref{eq:energy-etaR-translation} dans~\eqref{eq:split-energy-translation} on trouve
\begin{align*}
a_\eps &\ge  \liminf_{R\to \infty}\left( a_\eps \Tr \left[\chi_R^2 \gamma^{(1)}  \right] + (1-2\eps) a \big(1-\Tr \left[\chi_R^2\gamma^{(1)}\right]\big) \right)\\
&=  a_\eps \Tr [\gamma^{(1)}] + (1-2\eps) a \big(1-\Tr [\gamma^{(1)}]\big)
\end{align*}
Comme on a suppos\'e que $a_\eps \le a <0$, on obtient
\bq \label{eq:concentration-case}
\tr[\gamma^{(1)}]=1.
\eq
Il n'y a donc pas de perte de masse pour le probl\`eme auxiliaire d\'efini au Lemme~\ref{lem:triv aux}. Ceci suffit pour conclure \`a  la convergence forte des matrices r\'eduites $\gamma_N ^{(k)}$. En effet, on peut appliquer le th\'eor\`eme de de Finetti faible \`a la suite des limites $\gamma ^{(k)}$ pour obtenir une mesure $\mu$ sur la boule de $\gH$. Hors,~\eqref{eq:concentration-case} combin\'e avec~\eqref{eq:melange faible} implique que la mesure doit en fait vivre sur la sph\`ere unit\'e. On a donc pour tout $k$ 
$$ \tr[\gamma^{(k)}] = 1,$$
ce qui implique que la convergence de $\gamma_N ^{(k)}$ vers $\gamma ^{(k)}$ est en r\'ealit\'e forte en norme de trace. On peut alors revenir \`a~\eqref{eq:split-energy-translation} pour obtenir 
\begin{multline*}
\liminf_{N\to\ii}\left( \Tr_\gH\left[(-\Delta-\eps w_-)\gamma_N^{(1)}\right] + \frac{1}{2} \Tr_{\gH^2}\left[w_\eps \gamma_N^{(2)}\right] \right)\\
\geq \Tr_\gH\left[(-\Delta-\eps w_-)\gamma^{(1)}\right] + \frac{1}{2} \Tr_{\gH^2}\left[w_\eps \gamma^{(2)}\right].
\end{multline*}
On applique ensuite le th\'eor\`eme de de Finetti fort aux limites des matrices de densit\'e pour conclure que le membre de droite est n\'ecessairement plus grand que $e_{\rm H} ^{\eps}(1)$ ce qui donne bien~\eqref{eq:aux prob bound} et donc conclut la preuve.
\end{proof}

\subsection{Conclusion de la preuve dans le cas g\'en\'eral}\label{sec:general general}\mbox{}\\\vspace{-0.4cm}

Nous avons maintenant presque tous les ingr\'edients n\'ecessaires pour la preuve du Th\'eor\`eme~\ref{thm:hartree general}. Il ne nous manque en fait qu'un tout petit peu d'information compl\'ementaire sur le probl\`eme invariant par translation, comme nous l'expliquons maintenant.

Consid\'erons pour $k\geq 2$ l'\'energie 
\begin{equation}\label{eq:b_j}
b_k(\lambda):= \frac{1}{k}\inf \sigma_{\gH^k} \left( \sum_{i=1}^k -\Delta_i + \frac{\lambda}{k-1} \sum_{i<j}^k w (x_i-x_j) \right)
\end{equation}
c'est-\`a-dire une \'energie pour $k$ particules avec une interaction d'intensit\'e $\frac{\lambda}{k-1}$. Les r\'esultats de la section pr\'ec\'edente montrent d\'ej\`a que la limite $k\to \infty$ d'une telle \'energie est donn\'ee par $\lambda \eH ^0 (\lambda)$ lorsque $\lambda$ est un param\`etre fixe. En utilisant la m\'ethode de localisation dans l'espace de Fock, l'\'energie des particules partant \`a l'infini lors de la minimisation d'une \'energie  pour $N\geq k$ particules sera naturellement d\'ecrite comme une superposition d'\'energies de syst\`emes \`a $k$ particules avec une interaction d'intensit\'e $1/(N-1)$ h\'erit\'ee du probl\`eme originel. Autrement dit, il s'agira d'\'evaluer une superposition des \'energies $b_k (\lambda)$ avec  
$$\lambda = \frac{k-1}{N-1},$$
un peu comme en~\eqref{eq:localization-chi-eps} et~\eqref{eq:localization-translation-eta}. Comme ce $\lambda$ d\'epend de $k$ il sera utile de savoir que $b_k(\lambda)$ est \'equicontinue en $\lambda$:

\begin{lemma}[\textbf{Equi-continuit\'e de l'\'energie en fonction de l'interaction}]\label{lem:equi conti}\mbox{}\\
On prend les conventions
\begin{equation*}
b_0 (\lambda) = b_1 (\lambda) = 0.
\end{equation*}
On a alors pour tout $\lambda \in [0,1]$ 
\begin{equation}\label{eq:lim b lambda}
\lim_{k\to\ii} \lambda\, b_k(\lambda) =\eH^0(\lambda).
\end{equation}
De plus, pour tout $0\leq \lambda \leq \lambda ' \leq 1$ 
\begin{equation}\label{eq:equi conti b lambda}
0\leq b_k (\lambda) -b_k (\lambda') \leq C |\lambda - \lambda'|
\end{equation}
o\`u $C$ ne d\'epend pas de $k$.
\end{lemma}

\begin{proof}[Preuve.]
On commence par confimer notre affirmation que~\eqref{eq:lim b lambda} est une cons\'equence directe de l'analyse de la section pr\'ec\'edente. Pour $\lambda = 0$ il n'y a rien \`a prouver. Pour $\lambda > 0$ on utilise le Th\'eor\`eme~\ref{thm:trans inv} pour obtenir  
\begin{align*}
\lim_{k\to\ii} \lambda\, b_k(\lambda)&=\lambda \inf_{\norm{u}^2=1}\left(\pscal{u,Ku}+\frac{\lambda}{2}\iint w(x-y)|u(x)|^2|u(y)|^2\right)\\
&=\inf_{\norm{u}^2=\lambda}\left(\pscal{u,Ku}+\frac{1}{2}\iint w(x-y)|u(x)|^2|u(y)|^2\right)=\eH^0(\lambda).
\end{align*} 
Le fait que 
$$b_k(\lambda)\ge b_k(\lambda')~~\text{pour tout}~0\le \lambda<\lambda'\le 1$$
est une cons\'equence de~\eqref{eq:A-tB}. Pour l'\'equicontinuit\'e~\eqref{eq:equi conti b lambda} on fixe un $ 0 <\alpha \leq 1$ et on remarque que, avec 
$$\delta := (\lambda'-\lambda)(\alpha^{-1}-\lambda)^{-1},$$
on a 
\begin{align*}
\frac{1}{k}\left( \sum_{i=1}^k K_i + \frac{\lambda'}{k-1} \sum_{i<j}^k w_{ij} \right) &= \frac{1-\delta}{k} \left( \sum_{i=1}^k K_i + \frac{\lambda}{k-1} \sum_{i<j}^k w_{ij} \right) \\
&+ \frac{\delta}{k\alpha} \left( \alpha \sum_{i=1}^k K_i + \frac{1}{k-1}\sum_{i<j}^k w_{ij} \right)\\
&\ge (1-\delta) b_k(\lambda) - \frac{C\delta}{\alpha}
\end{align*}
en utilisant le fait que le spectre de l'op\'erateur appraraissant \`a la deuxi\`eme ligne est born\'e inf\'erieurement. On d\'eduit donc 
$$ 0\le b_k(\lambda)- b_k(\lambda') \le \delta (b_k(\lambda)+C\alpha^{-1})\le C|\lambda'-\lambda|$$
puisque $b_k(\lambda)$ est uniform\'ement born\'e et $|\delta| \leq C |\lambda- \lambda '|$.
\end{proof}

Nous pouvons maintenant conclure:

\begin{proof}[Preuve du Th\'eor\`eme~\ref{thm:hartree general}.] Soit $\Psi_N$ une suite de fonctions d'onde \`a $N$ corps telle que 
$$\langle \Psi_N, H_N^V \Psi_N \rangle = E^V(N)+o(N)$$
avec $\gamma_N^{(k)}$ les matrices de densit\'e r\'eduites associ\'ees. Modulo une extraction diagonale on suppose que 
$$\gamma_N^{(k)}\wto \gamma^{(k)}$$ 
faiblement-$\ast$. Le Th\'eor\`eme~\ref{thm:DeFinetti faible} fournit alors une mesure de probabilit\'e $\mu$ telle que 
\begin{equation}\label{eq:gen melange}
 \gamma ^{(k)} = \int_{u\in B\gH} d\mu (u) |u ^{\otimes k} \rangle \langle u ^{\otimes k} |.
\end{equation}
Soit une partition de l'unit\'e r\'eguli\`ere $\chi_R^2 +\eta_R^2=1$ comme pr\'ec\'edement et les \'etats localis\'es $G_N^\chi$ et $G_N^\eta$ construits \`a partir de $|\Psi_N \rangle \langle \Psi_N |$. En utilisant \`a nouveau le Lemme~\ref{lem:loc ener} on obtient
\begin{align} \label{eq:split-energy-general}
\lim_{N\to \infty}\frac{E^V(N)}{N} &= \lim_{N\to \infty} \left( \Tr_\gH\left[T\gamma_N^{(1)}\right] + \frac{1}{2} \Tr_{\gH^2}\left[w \gamma_N^{(2)}\right] \right) \nn\\
&\ge  \liminf_{R\to \infty}\liminf_{N\to \infty} \left\{ \Tr_\gH\left[T \chi_R \gamma_N^{(1)}\chi_R\right] + \frac{1}{2} \Tr_{\gH^2}\left[w \chi_R^{\otimes 2}\gamma_N^{(2)}\chi_R^{\otimes 2}\right] \right. \nn\\
&\qquad\left. +\Tr_\gH\left[ -\Delta \eta_R \gamma_N^{(1)}\eta_R\right] + \frac{1}{2} \Tr_{\gH^2}\left[w\eta_R^{\otimes 2}\gamma_N^{(2)}\eta_R^{\otimes 2}\right] \right\}.
\end{align}
On utilise d'abord la compacit\'e locale forte et~\eqref{eq:gen melange} pour le terme $\chi_R$-localis\'e:
\begin{multline}
\liminf_{N\to \infty} \left\{ \Tr_\gH\left[ T \chi_R \gamma_N^{(1)}\chi_R\right] + \frac{1}{2} \Tr_{\gH^2}\left[w \chi_R^{\otimes 2}\gamma_N^{(2)}\chi_R^{\otimes 2}\right]\right\}\\
\geq  \Tr_\gH\left[T \chi_R \gamma^{(1)}\chi_R\right] + \frac{1}{2} \Tr_{\gH^2}\left[w \chi_R^{\otimes 2}\gamma^{(2)}\chi_R^{\otimes 2}\right]= \int_{B\gH} \E^V_{\rm H}[\chi_R u] d\mu(u).
\label{eq:localization-chi-general}
\end{multline}
Notre t\^ache principale est de contr\^oler le second terme du membre de droite de~\eqref{eq:split-energy-general}. Nous pr\'etendons que 
\begin{equation}
\label{eq:localization-eta-general}
\liminf_{N\to \infty} \left( \Tr\left[T \eta_R \gamma_N^{(1)}\eta_R\right] + \frac{1}{2} \Tr_{\gH^2}\left[w \eta_R^{\otimes 2}\gamma_N^{(2)}\eta_R^{\otimes 2}\right]\right) \ge \int_{B\gH} e^{0}_{\rm H}(1-\norm{\chi_R u}^2) d\mu(u). 
\end{equation}
En effet, en utilisant l'\'etat $\eta_R$-localis\'e $G_N^\eta$  on a  
\begin{align*}
\Tr_\gH\left[ -\Delta \eta_R \gamma_N^{(1)}\eta_R\right] &+ \frac{1}{2} \Tr_{\gH^2}\left[w \eta_R^{\otimes 2}\gamma_N^{(2)}\eta_R^{\otimes 2}\right]
=\frac{1}{N} \sum_{k=1}^N \Tr_{\gH^k}\left[ \left( \sum_{i=1}^k -\Delta_i + \frac{1}{N-1} \sum_{i<j}^k w_{ij}\right) {G}^{\eta}_{N,k} \right] \\
&\ge\sum_{k=1}^N \frac{1}{N} \Tr {G}^{\eta}_{N,k} \inf \sigma_{\gH^k} \left( \sum_{i=1}^k -\Delta_i + \frac{1}{N-1} \sum_{i<j}^k w_{ij}\right) \\
&\geq \sum_{k=0}^N \tr G_{N,k}^\eta \frac{k}{N} b_k \left( \frac{k-1}{N-1}\right)
\end{align*}
o\`u $b_k$ est la fonction d\'efinie en~\eqref{eq:b_j}. D'un autre c\^ot\'e
\begin{equation} \label{eq:localization-eta-bk}
\lim_{N\to\ii}\sum_{k=0}^N \tr G_{N,k}^\eta \left( \frac{k}{N} b_k \left( \frac{k-1}{N-1}\right) - e_{\rm H}^0 \left(\frac{k}{N}\right) \right) = 0,
\end{equation}
puisque en utilisant l'\'equicontinuit\'e de $\{b_k\}_{k=1}^\infty$ et la convergence $\lim_{k\to \infty} \lambda b_k(\lambda)=e_{\rm H}^{0}(\lambda)$ fournies par le Lemme~\ref{lem:equi conti} on obtient 
$$ \lim_{N \to \infty} \sup_{k=1,2,...,N} \left|\frac{k}{N} b_k\left(\frac{k-1}{N-1}\right)-e_{\rm H}^{0}\left(\frac{k}{N}\right)\right| =0,$$
qu'il suffit de combiner avec
$$\sum_{k=0}^N \tr G_{N,k}^\eta = 1$$
pour obtenir~\eqref{eq:localization-eta-bk}. A ce stade on a donc 
\[
\liminf_{N\to \infty} \Tr_\gH\left[ -\Delta \eta_R \gamma_N^{(1)}\eta_R\right] + \frac{1}{2} \Tr_{\gH^2}\left[w \eta_R^{\otimes 2}\gamma_N^{(2)}\eta_R^{\otimes 2}\right] \geq \lim_{N\to\ii}\sum_{k=0}^N \tr G_{N,k}^\eta e_{\rm H}^0 \left(\frac{k}{N}\right) .
\]
On utilise maintenant la relation~\eqref{eq:Fock funda rel} et le Corollaire~\ref{cor:wdeF loc} comme indiqu\'e Section~\ref{sec:proof deF faible} pour obtenir 
\begin{multline*}
\lim_{N\to\ii}\sum_{k=0}^N \tr G_{N,k}^\eta\; e_{\rm H}^0 \left(\frac{k}{N}\right) =\lim_{N\to\ii}\sum_{k=0}^N \tr G_{N,N-k}^\chi \;e_{\rm H}^0 \left(\frac{k}{N}\right)\\
=\lim_{N\to\ii}\sum_{k=0}^N \tr G_{N,k}^\chi \;e_{\rm H}^0 \left(1-\frac{k}{N}\right)=\int_{B\gH}\;e_{\rm H}^0 (1-\|\chi_R u\|^2)d\mu(u),
\end{multline*}
ce qui conclut la preuve de~\eqref{eq:localization-eta-general}.

Il reste maintenant \`a ins\'erer~\eqref{eq:localization-chi-general} et~\eqref{eq:localization-eta-general} dans~\eqref{eq:split-energy-general} et \`a utiliser le lemme de Fatou, ce qui donne 
\begin{align}\label{eq:conclu preuve general}
\lim_{N\to \infty}\frac{E^V(N)}{N} &\ge \liminf_{R\to \infty} \left( \int_{B\gH} \left[\E^V_{\rm H}[\chi_R u]+ e^0_{\rm H}(1-\norm{\chi_R u}^2) \right] d\mu(u)\right) \nn\\
&\ge \int_{B\gH} \liminf_{R\to \infty} \left[\E^V_{\rm H}[\chi_R u]+ e^0_{\rm H}(1-\norm{\chi_R u}^2)\right] d\mu(u) \nn\\
&= \int_{B\gH} \left[\E^V_{\rm H}[u]+ e^0_{\rm H}(1-\norm{u}^2) \right] d\mu(u) \nn\\
&\geq \int_{B\gH} \left[\eH^V(\norm{u}^2)+ e^0_{\rm H}(1-\norm{u}^2) \right] d\mu(u)\ge e_{\rm H}(1), 
\end{align}
en utilisant la continuit\'e de $\lambda \mapsto e_{\rm H}^0(\lambda)$ et l'in\'egalit\'e de liaison large~\eqref{eq:hartree liaison}. Ceci conclut la preuve de~\eqref{eq:general ener}, et les autres r\'esultats du th\'eor\`eme suivent en analysant les cas d'\'egalit\'e dans~\eqref{eq:conclu preuve general}.
\end{proof}

\newpage

\section{\textbf{D\'erivation de fonctionnelles de type Gross-Pitaevskii}}\label{sec:NLS}

Nous allons maintenant nous tourner vers l'obtention de fonctionnelles de type Schr\"odinger non lin\'eaire (NLS) avec non-lin\'earit\'es locales:
$$\ENLS [\psi] := \int_{\R ^d } \left| \nabla \psi \right| ^2 + V  |\psi | ^2 + \frac{a}{2} |\psi| ^4.$$ 
On peut obtenir cette description en partant d'une fonctionnelle de Hartree comme~\eqref{eq:Hartree func} et en prenant un potentiel d'interaction convergeant (au sens des distributions) vers une masse de Dirac
$$ w_L (x) = L ^{-d} w \left( \frac{x}{L}\right) \wto \left(\int_{\R ^d} w \right) \delta_0 \mbox{ quand } L\to 0.$$ 
Puisqu'on a d\'ej\`a obtenu~\eqref{eq:Hartree func} comme limite d'un probl\`eme \`a $N$ corps, on peut \^etre tent\'e de voir l'obtention de la description NLS comme un simple passage \`a la limite dans un probl\`eme \`a un corps. L'inconv\'enient de cette approche est l'absence totale de contr\^ole sur les relations entre les param\`etres physiques $N$ et $L$. On pourrait voir ceci comme un probl\`eme d'\'echange de limites: il n'est pas du tout clair que les limites $N\to \infty$ et $L\to 0$ commutent\footnote{La limite $L\to 0$ du probl\`eme \`a $N$ corps est d'ailleurs tr\`es difficile \`a d\'efinir proprement.}. 

La bonne question d'un point de vue physique est ``Quelle relation doivent satisfaire $N$ et $L$ pour que l'on puisse obtenir l'\'energie NLS en prenant \emph{simultan\'ement} $N\to \infty$ et $L\to 0$ dans le probl\`eme \`a $N$ corps ?'' Le processus permettant d'obtenir la fonctionnelle NLS est donc plus subtil que pour l'obtention de la fonctionnelle de Hartree, et nous allons dans un premier temps donner quelques explications sur ce point.

\subsection{Remarques pr\'eliminaires}\label{sec:GP rem}\mbox{}\\\vspace{-0.4cm}

Jusqu'\`a pr\'esent nous avons travaill\'e avec seulement deux param\`etres physiques: le nombre de particules $N$ et la force des interactions $\lambda$. Pour avoir un probl\`eme limite bien d\'efini nous avons \'et\'e amen\'e \`a consid\'erer la limite de champ moyen o\`u $\lambda \propto N ^{-1}$. Dans ce cas, la port\'ee des interactions est fixe et chaque particule interagit avec toutes les autres, ce qui donne une force d'interaction typique par particule d'ordre $\lambda N \propto 1$, comparable avec son \'energie propre (cin\'etique + potentielle). On a vu que cet \'equilibrage des forces agissant sur chaque particule, combin\'e \`a des r\'esultats de strucutre ``\`a la de Finetti'' pour l'ensemble des \'etats bosoniques, m\`ene naturellement \`a la conclusion que les particules se comportent approximativement de mani\`ere ind\'ependante les unes des autres et donc que des descriptions de type Hartree sont valables dans la limite $N\to \infty$.

Il existe d'autres fa\c{c}ons de justifier de tels mod\`eles: l'\'equilibrage des forces qui permet \`a un probl\`eme limite d'\'emerger peut r\'esulter d'un m\'ecanisme plus subtil. Par exemple, dans les gaz alcalins ultra-froids  qui ont permis l'observation en laboratoire de condensats de Bose-Einstein, il a plus \`a voir avec la \emph{dilution} du syst\`eme qu'avec la faible force des interactions. Pour une description th\'eorique de ce ph\'enom\`ene on peut introduire dans notre mod\`ele une longueur caract\'eristique $L$ pour la port\'ee des interactions. En prenant la taille du syst\`eme total comme r\'ef\'erence, un syst\`eme dilu\'e se mat\'erialise par la limite $L\to 0$. La force d'interaction typique agissant sur chaque particule sera d'ordre $ \lambda N L ^{d}$ (force des interactions $\times$ nombre de particules typique dans une boule de rayon $L$ autour d'une particule) et c'est ce param\`etre que l'on peut vouloir fixer pour obtenir un probl\`eme limite. Plusieurs r\'egimes sont alors 
possibles en fonction du ratio entre $\lambda$ et $L$.  

On peut discuter ces diff\'erentes possiblit\'es en partant du Hamiltonien \`a $N$ corps\footnote{Encore une fois, il est possible de rajouter Laplaciens fractionnaires et/ou champs magn\'etiques, cf Remarque~\ref{rem:energie cin}, ce que nous n\'egligerons pour simplifier.} 
\begin{equation}\label{eq:GP start hamil}
H_N = \sum_{j=1} ^N  - \Delta_j  + V(x_j)  + \frac{1}{N-1} \sum_{1\leq i<j \leq N}N ^{d\beta} w( N ^{\beta} (x_i-x_j))
\end{equation}
qui correspond \`a faire les choix 
$$L=N ^{-\beta},\quad  \lambda \propto N ^{d\beta - 1},$$ 
le param\`etre fixe $0\leq \beta$ servant \`a ajuster le rapport entre $L$ et $\lambda$. On consid\'erera le potentiel d'interaction de r\'ef\'erence $w$ comme \'etant fixe, et on notera 
\begin{equation}\label{eq:w_N}
 w_N (x) := N ^{d\beta} w( N ^{\beta} x). 
\end{equation}
Pour $\beta >0$, $w_N$ converge vers une masse de Dirac au sens des distributions 
\begin{equation}\label{eq:GP conv delta}
 w_N \to \left(\int_{\R ^d} w \right) \delta_0,  
\end{equation}
mat\'erialisant la courte port\'ee des interactions/la dilution du syt\`eme. En raisonnant formellement, on peut vouloir remplacer directement $w_N$ par $\left(\int_{\R ^d} w \right) \delta_0$ dans~\eqref{eq:GP start hamil} auquel cas on se retrouve avec une limite de type champ moyen, avec une masse de Dirac comme potentiel d'interaction. Ceci est bien s\^ur purement formel (sauf en 1D o\`u l'injection de Sobolev $H ^1 \hookrightarrow C^0$ permet de d\'efinir proprement les interactions de contact).  En admettant que ces manipulations aient un sens et que l'on puisse approximer le fondamental de~\eqref{eq:GP start hamil} sous la forme 
\begin{equation}\label{eq:GP ansatz}
\Psi_N = \psi ^{\otimes N} 
\end{equation}
on obtient une fonctionnelle de Hartree
\begin{equation}\label{eq:GP intro Hartree}
\EH [\psi] := \int_{\R ^d } \left| \nabla \psi \right| ^2 + V  |\psi | ^2 + \frac{1}{2} |\psi| ^2 (w \ast |\psi| ^2 ) .
\end{equation}
o\`u le potentiel d'interaction est \`a remplacer par une masse de Dirac \`a l'origine, ce qui m\`ene \`a la fonctionnelle de Gross-Pitaevskii
\begin{equation}\label{eq:GP intro nls}
\ENLS [\psi] := \int_{\R ^d } \left| \nabla \psi \right| ^2 + V  |\psi | ^2 + \frac{a}{2} |\psi| ^4.
\end{equation}
Au vu de~\eqref{eq:GP conv delta}, le choix logique semble d'imaginer que lorsque $\beta >0$, on obtienne \`a partir de la minimisation du Hamiltonien~\eqref{eq:GP start hamil} la fonctionnelle ci-dessus avec 
\begin{equation}\label{eq:GP defi a}
a = \int_{\R ^d} w. 
\end{equation}
En fait, on peut obtenir la gamme de r\'esultats suivants (d\'ecrits dans le cas $d=3$, le cas $d\leq 2$ menant \`a des subtilit\'es et le cas $d\geq 4$ n'ayant pas grand int\'er\^et physiquement parlant): 
\begin{itemize}
\item Si $\beta = 0$ nous retrouvons le r\'egime de champ moyen (MF, pour mean-field) \'etudi\'e pr\'ec\'edement. La port\'ee du potentiel d'interaction est fixe et son intensit\'e d\'ecro\^it comme $N ^{-1}$. Le probl\`eme limite  est alors~\eqref{eq:GP intro Hartree}, comme d\'emontr\'e pr\'ec\'edement. 
\item Si $ 0 < \beta < 1$, le probl\`eme limite est bien comme on pouvait s'y attendre~\eqref{eq:GP intro nls} avec le choix de param\`etre~\eqref{eq:GP defi a}. Nous parlerons de limite de Schr\"odinger non-lin\'eaire (NLS). Ce cas n'a pas \'et\'e consid\'e\'e sp\'ecifiquement dans la lit\'erature mais, au moins lorsque $w\geq 0$ et lorsque $V$ est confinant, on peut adapter l'analyse du cas plus difficile $\beta = 1$. 
\item Si $\beta = 1$, la fonctionnelle limite est maintenant~\eqref{eq:GP intro nls} avec 
$$a = 4\pi \times \mbox{ longueur de diffusion de } w$$
(voir~\cite[Appendice C]{LieSeiSolYng-05} pour une d\'efinition). Dans ce cas, le fondamental de~\eqref{eq:GP start hamil} inclut une correction \`a l'ansatz~\eqref{eq:GP ansatz}, sous la forme de corr\'elations \`a courte port\'ee. En fait, il faut s'attendre \`a avoir plus pr\'ecis\'ement 
\begin{equation}\label{eq:GP ansatz 2}
\Psi_N (x_1,\ldots,x_N) \approx \prod_{j=1} ^N \psi (x_j) \prod_{1\leq i<j\leq N} f \left( N ^{\beta} (x_i - x_j)\right)  
\end{equation}
o\`u $f$ est li\'e au probl\`eme \`a deux corps d\'efini par $w$ (solution de diffusion d'\'energie minimale). Il se trouve que lorsque $\beta = 1$, la correction a un effet sur le premier ordre de l'\'energie, celui de remplacer $\int_{\R ^d}w$ par la longueur de diffusion correspondante comme remarqu\'e pour la premi\`ere fois dans~\cite{Dyson-57}. On parle alors de limite de Gross-Pitaevskii (GP), qui a fait l'objet d'une longue et remarquable s\'erie de travaux d\^us \`a Lieb, Seiringer et Yngvason  (voir par exemple~\cite{LieYng-98,LieYng-01,LieSeiYng-01,LieSeiYng-00,LieSeiYng-05,LieSei-06,LieSeiSolYng-05}).   
\end{itemize}
Comme nous l'avons d\'ej\`a mentionn\'e, les probl\`emes d'\'evolution correspondants ont \'egalement \'et\'es abondament \'etudi\'es. L\`a aussi il convient de distinguer les limites de champ moyen  (voir par exemple~\cite{BarGolMau-00,AmmNie-08,FroKnoSch-09,RodSch-09,Pickl-11})  de Schr\"odinger non-lin\'eaire~\cite{ErdSchYau-07,Pickl-10} et de Gross-Pitaevskii~\cite{ErdSchYau-09,Pickl-10b}. 

Il y a une diff\'erence physique fondamentale entre les r\'egimes MF et GP:  Dans les deux cas le param\`etre d'interaction effectif $\lambda N L ^3$ est d'ordre $1$ mais on passe (toujours en 3D) d'un cas avec des collisions nombreuses mais faibles $\lambda = N ^{-1}$, $L=1$ \`a un cas avec des collisions rares mais tr\`es fortes, $\lambda = N ^2$, $L = N ^{-1}$. Les diff\'erents cas NLS interpolent en quelque sorte entre ces deux extr\^emes, le passage entre ``collisions faibles mais fr\'equentes'' et ``collisions fortes mais rares'' se faisant \`a $\lambda = 1, L = N ^{-1/3}$, i.e. $\beta = 1/3$.

\medskip

La difficult\'e pour obtenir~\eqref{eq:GP intro nls} en passant \`a la limite $N\to \infty$ est qu'il y a en fait deux limites distinctes, $N\to \infty$ et $w_N \to a \delta_0$ \`a contr\^oler en m\^eme temps. Un simple argument de compacit\'e/passage \`a la limite sera insuffisant dans ce cas et il faudra donc travailler avec des estimations quantitatives. Le but de ce chapitre est de montrer comment on peut proc\'eder \`a partir du th\'eor\`eme de de Finetti quantitatif pr\'esent\'e au Chapitre~\ref{sec:deF finite dim}. Puisque ce th\'eor\`eme n'est valable qu'en dimension finie, il nous faudra disposer d'une mani\`ere naturelle de projeter le probl\`eme sur des espaces de dimension finie. Nous nous placerons donc dans le cadre des gaz de bosons pi\'eg\'es en supposant que, pour deux constantes~$c,C~>~0$   
\begin{equation}\label{eq:GP asum V}
c |x| ^s - C\leq V(x)
\end{equation}
Dans ce cas, le Hamiltonien \`a un corps $-\Delta + V$ a un spectre discret sur lequel on a bon contr\^ole via des in\'egalit\'es \`a la Lieb-Thirring. 

Les r\'esultats que nous allons obtenir sont valables pour $0 < \beta < \beta_0$ o\`u $\beta_0 = \beta_0 (d,s)$ ne d\'epend que de la dimension de l'espace physique et du potentiel $V$. Nous fournirons des estimations explicites de $\beta_0$, mais la m\'ethode que nous pr\'esentons, issue de~\cite{LewNamRou-14}, est pour l'instant limit\'ee \`a des $\beta$ relativement petits. En particulier on obtiendra toujours~\eqref{eq:GP defi a} comme param\`etre d'interaction. L'avantage principal par rapport \`a la m\'ethode de Lieb-Seiringer-Yngvason~\cite{LieYng-98,LieSeiYng-00,LieSeiYng-05,LieSei-06} est que nous pouvons dans certains cas relaxer l'hypoth\`ese $w\geq 0$ qui est toujours faite dans les travaux sus-mentionn\'es (voir aussi~\cite{LieSeiSolYng-05}). En particulier nous pr\'esenterons la premi\`ere d\'erivation des fonctionnelles NLS attractives\footnote{On utilise souvent le vocabulaire de l'optique non-lin\'eaire pour distinguer les cas attractif et r\'epulsif: r\'epulsif 
= d\'efocusant, attractif = focusant.} en 1D et 2D.

\subsection{Enonc\'es et discussion}\label{sec:GP statement}\mbox{}\\\vspace{-0.4cm}

Nous prendrons des hypoth\`eses confortables sur $w$: 
\begin{equation}\label{eq:GP asum w}
w\in L^\infty(\R^d,\R) \mbox{ et } w(x) = w(-x) .
\end{equation}
Sans perte de g\'en\'eralit\'e, on posera 
$$ \sup_{\R ^d} |w| = 1$$
pour simplifier certaines expressions. On supposera \'egalement que 
\begin{equation}\label{eq:replace delta}
x \mapsto (1+|x|) w (x) \in L ^1 (\R ^d) 
\end{equation}
ce qui facilite le remplacement de $w_N$ par une masse de Dirac. Comme d'habitude nous utiliserons la m\^eme notation $w_N$ pour le potentiel d'interaction~\eqref{eq:w_N} et l'op\'erateur de multiplication par $w_N (x-y)$ agissant sur $L ^2 (\R ^{2d})$.

\medskip

Pour $\beta = 0$ on a montr\'e pr\'ec\'edemment que 
\begin{equation}\label{eq:GP Hartree lim}
\lim_{N\to \infty} \frac{E(N)}{N} = \eH. 
\end{equation}
Nous traitons maintenant le cas $0 < \beta < \beta_0 (d,s) < 1$ o\`u on obtient l'\'energie fondamentale de~\eqref{eq:GP intro nls} \`a la limite:
\begin{equation}\label{eq:GP eNLS}
\eNLS := \inf_{\norm{\psi}_{L ^2 (\R ^d)} = 1} \ENLS [\psi]
\end{equation}
avec $a$ donn\'e par~\eqref{eq:GP defi a}. A cause de la non-lin\'earit\'e locale, la th\'eorie NLS est plus d\'elicate que la th\'eorie de Hartree et il nous faut quelques hypoth\`eses de structure sur le potentiel d'interaction (voir~\cite{LewNamRou-14} pour une discussion plus pouss\'ee):  
\begin{itemize}
\item Quand $d=3$, il est bien connu qu'un fondamental pour~\eqref{eq:GP intro nls} existe si et seulement si $a \geq 0$. Ceci est d\^u au fait que la non-lin\'earit\'e cubique est sur-critique\footnote{On pourra consulter par exemple~\cite{KilVis-08} pour une classification des non-lin\'earit\'es dans l'\'equation de Schr\"odinger non-lin\'eaire.} dans ce cas. De plus, il est facile de voir que $N ^{-1} E(N) \to -\infty$ si $w$ est n\'egatif \`a l'origine. L'hypoth\`ese optimale se trouve \^etre une hypoth\`ese de stabilit\'e classique pour le potentiel d'interaction:
\begin{equation}\label{eq:GP hyp 3}
\iint_{\R ^d \times \R ^d} \rho (x) w(x-y) \rho(y) dx dy \geq 0, \mbox{ pour tout } \rho \in L^1 (\R ^d), \: \rho \geq 0. 
\end{equation}
Cette hypoth\`ese est v\'erifi\'ee d\`es que $w= w_1 + w_2$, $w_1 \geq 0$ et $\hat{w_2}\geq 0$ avec $\hat{w_2}$ la transform\'ee de Fourier de $w_2$. Elle implique clairement $\int_{\R ^d} w \geq 0$, et on peut voir ais\'ement par un changement d'\'echelle que si elle viol\'ee pour un certain $\rho \geq 0$ alors $E(N)/N \to -\infty$.
\item Quand $d=2$, la non-lin\'earit\'e cubique est critique. Un minimiseur pour~\eqref{eq:GP intro nls} existe si et seulement si $a > - a ^*$ o\`u 
\begin{equation}\label{eq:GP a star}
a ^* = \norm{Q}_{L ^2} ^2  
\end{equation}
avec $Q\in H ^1 (\R^2)$ l'unique~\cite{Kwong-89} (modulo translations) solution de 
\begin{equation}\label{eq:GP defi Q}
-\Delta Q + Q - Q ^3 = 0.  
\end{equation}
Le param\`etre d'interaction critique $a ^*$ est la meilleure constante possible dans l'in\'egalit\'e d'interpolation
\begin{equation}\label{eq:GP interpolation}
\int_{\R ^2} |u| ^4 \leq C \left(\int_{\R^2} |\nabla u | ^2\right) \left(\int_{\R ^2} |u | ^2\right). 
\end{equation}
Voir~\cite{GuoSei-13,Maeda-10} pour l'existence d'un fondamental et~\cite{Weinstein-83} pour l'in\'egalit\'e~\eqref{eq:GP interpolation}. Une discussion p\'edagogique de ce genre de sujets se trouve dans~\cite{Frank-14}.

Vues les consid\'erations ci-dessus, il est clair qu'en 2D il nous faudra supposer $\int w \geq - a ^*$, mais ce n'est en fait pas suffisant: comme en 3D, si le potentiel d'interaction est suffisament n\'egatif \`a l'origine, on peut voir que $N ^{-1} E(N) \to -\infty$. L'hypoth\`ese ad\'equate est cette fois 
\begin{equation}\label{eq:GP hyp 2}
\norm{u}_{L^2}^2\norm{\nabla u}_{L^2}^2+\frac12\iint_{\R^2\times \R^2}|u(x)|^2|u(y)|^2 w(x-y)\,dx\,dy > 0
\end{equation}
pour tout $u\in H^1(\R^2)$. En rempla\c{c}ant $u$ par $\lambda u(\lambda x)$ et en prenant la limite $\lambda\to 0$
on obtient
$$\norm{u}_{L^2}^2\norm{\nabla u}_{L^2}^2+\frac12\left(\int_{\R^2}w\right)\int_{\R^2}|u(x)|^4\,dx\geq0,\qquad\forall u\in H^1(\R^2),$$
ce qui implique que 
$$\int_{\R^2}w(x)\,dx\geq-a^*.$$
Un argument de changement d'\'echelle montre que si l'in\'egalit\'e stricte dans~\eqref{eq:GP hyp 2} est invers\'ee pour un certain $u$, alors $E(N)/N \to - \infty$. La cas ou il peut il y a voir \'egalit\'e dans~\eqref{eq:GP hyp 2} est laiss\'e de c\^ot\'e car il requiert une analyse plus pouss\'ee, cf l'analyse du cas correspondant au niveau de la fonctionnelle NLS dans~\cite{GuoSei-13}.
\item Quand $d=1$, la non-lin\'earit\'e cubique est sous-critique et il y a toujours un minimiseur pour la fonctionnelle~\eqref{eq:GP intro nls}. Dans ce cas nous n'aurons pas besoin d'hypoth\`ese suppl\'ementaire. 
\end{itemize}

Nous pouvons maintenant \'enoncer le 

\begin{theorem}[\textbf{D\'erivation de la fonctionnelle NLS}]\label{thm:deriv nls}\mbox{}\\
On se place dans l'un des cas $d=1$, $d=2$ avec~\eqref{eq:GP hyp 2}, $d= 3$ avec~\eqref{eq:GP hyp 3} et on suppose 
\begin{equation}\label{eq:GP beta 0}
0 < \beta \leq \beta_0 (d,s) :=\frac{s}{2ds + s d^2 + 2d ^2} < 1.
\end{equation}
o\`u $s$ est l'exposant apparaissant dans~\eqref{eq:GP asum V}. On a alors:  
\begin{enumerate}
\item \underline{Convergence de l'\'energie:}
\begin{equation}\label{eq:GP ener convergence}
\frac{E(N)}{N} \to \eNLS. 
\end{equation}
\item \underline{Convergence des \'etats:} Soit $\Psi_N$ un fondamental de~\eqref{eq:GP start hamil} et 
$$ \gamma_N ^{(n)}:= \tr_{n+1\to N} \left[ |\Psi_N\rangle \langle \Psi_N | \right] $$
ses matrices de densit\'e r\'eduites. Modulo une sous-suite, on a pour tout $n\in \N$  
\begin{equation}\label{eq:GP state convergence}
\lim_{N\to \infty}\gamma_N ^{(n)} = \int_{u\in \MNLS} d\mu (u) |u ^{\otimes n} \rangle \langle u ^{\otimes n}| 
\end{equation}
fortement dans $\gS ^1 (L ^2 (\R ^{dn}))$. Ici $\mu$ est une mesure de probabilit\'e support\'ee sur  
\begin{equation}\label{eq:GP nls set}
\MNLS = \left\{ u \in L ^2 (\R ^d), \norm{u}_{L ^2} = 1, \ENLS [u] = \eNLS \right\}. 
\end{equation}
En particulier, quand~\eqref{eq:GP intro nls} a un minimiseur $\uNLS$ unique \`a une phase pr\`es, on a pour toute la suite
\begin{equation}\label{eq:GP BEC}
\lim_{N\to \infty} \gamma_N ^{(n)} =|\uNLS ^{\otimes n} \rangle \langle \uNLS ^{\otimes n}|.
\end{equation}
\end{enumerate}
\end{theorem}

L'unicit\'e de $\uNLS$ est assur\'ee si $a \geq 0$ ou $|a|$ est petit. Si ces conditions ne sont pas satisfaites, on peut montrer l'absence d'unicit\'e pour certains potentiels de pi\'egeage ayant plusieurs minima~\cite{AshFroGraSchTro-02,GuoSei-13}. 

\begin{remark}[Sur la d\'erivation de la fonctionnelle NLS.]\label{rem:GP NLS}\mbox{}\\
\vspace{-0.4cm}
\begin{enumerate}
\item L'hypoth\`ese $\beta < \beta_0 (d,s)$ est dict\'ee par la m\'ethode de preuve mais n'est certainement pas optimale, on pourra voir dans cette direction les travaux~\cite{LieSei-06,LieSeiYng-00,LieYng-98,LieYng-01}. On peut relaxer un peu la condition sur $\beta$, au prix de calculs un peu plus lourds dont nous pr\'ef\'erons faire l'\'economie dans ces notes, voir~\cite{LewNamRou-14}. En 1D, on peut obtenir le r\'esultat pour tout $\beta >0$.
\item Regardons un peu plus en d\'etail les conditions sur $\beta_0 (d,s)$. Pour le cas d'un potentiel de pi\'egeage quadratique $V(x) = |x| ^2$ par exemple, nous pouvons traiter $\beta < 1/24$ en 3D, $\beta < 1/12$ en 2D et $\beta < 1/4$ en 1D. La m\'ethode de preuve s'adapte sans difficult\'es au cas de particules dans un domaine born\'e, ce qui correspond \`a prendre formellement $s=\infty$. On obtient alors $\beta_0(d,s) = 1/15$ en 3D, $1/8$ en 2D et $1/3$ en 1D. L'am\'elioration de ces seuils pour le cas de potentiels comportant une partie attractive reste un probl\`eme ouvert.
\item Lorsque $\beta$ s'\'eloigne du seuil critique $\beta_0 (d,s)$, la m\'ethode de preuve fournit des estimations quantitatives pour la convergence~\eqref{eq:GP ener convergence}, voir ci-dessous. Voir~\cite[Remarque 4.2]{LewNamRou-14} pour une discussion des cas o\`u un taux de convergence pour le minimiseur peut \^etre d\'eduit en se basant sur des outils de~\cite{CarFraLie-14,Frank-14} et des hypoth\`eses sur le comportement de la fonctionnelle NLS.\hfill\qed 
\end{enumerate}
\end{remark}

La preuve de ce r\'esultat occupe le reste de ce chapitre. On proc\`ede en deux temps. Le gros du travail consiste en l'obtention d'une estimation quantitative entre l'\'energie $N$ corps par particules $N ^{-1} E(N)$ et l'\'energie de Hartree  
\begin{equation}\label{eq:GP hartree e}
\eH := \inf_{\norm{u}_{L ^2 (\R ^d)} = 1} \EH [u] 
\end{equation}
donn\'ee par la minimisation de la fonctionnelle
\begin{equation}\label{eq:GP hartree f}
\EH [u] := \int_{\R ^d} \left(\left|\nabla u \right| ^2 + V |u| ^2 \right) dx + \frac{1}{2} \iint_{\R ^d \times \R ^d} |u(x)| ^2 w_N (x-y) |u(y)| ^2 dx dy.  
\end{equation}
Ces objets d\'ependent de $N$ quand $\beta >0$, d'o\`u la n\'ecessit\'e d'\'eviter les arguments de compacit\'e et d'obtenir de vraies estimations. Une fois que le lien entre $N ^{-1} E(N)$ et~\eqref{eq:GP hartree e}  est \'etabli, il reste \`a estimer la diff\'erence $|\eNLS - \eH|$, ce qui est un probl\`eme beaucoup plus simple. La plupart des hypoth\`eses contraignantes que nous avons faites sur $w$ ne servent que lors de cette deuxi\`eme \'etape. Les estimations sur la diff\'erence $|\eH-N ^{-1} E(N)|$ sont valides sans supposer~\eqref{eq:replace delta} et~\eqref{eq:GP hyp 3} ou~\eqref{eq:GP hyp 2}. Ils fournissent donc une information sur la divergence de $N ^{-1} E(N)$ dans le cas o\`u $\eH$ ne converge pas vers~$\eNLS$:

\begin{theorem}[\textbf{D\'erivation quantitative de la th\'eorie de Hartree}]\label{thm:error Hartree}\mbox{}\\
On fait les hypoth\`eses~\eqref{eq:GP asum V},~\eqref{eq:GP asum w}. Soit 
\begin{equation}\label{eq:GP rate intro}
t :=  \frac{1+2d\beta}{2 + d/2 + d/s}.
\end{equation}
Si on a 
\begin{equation}\label{eq:GP asum t}
t > 2d \beta ,
\end{equation}
alors, pour tout $d\geq 1$ il existe une constante $C_{d}$ telle que 
\begin{equation}\label{eq:GP energy estimate}
\eH \geq \frac{E(N)}{N} \geq \eH - C_{d} N ^{-t + 2 d\beta}.  
\end{equation} 
\end{theorem}

\begin{remark}[Estimations explicites dans la limite de champ moyen]\label{rem:GP Hartree quant}\mbox{}\\
\vspace{-0.4cm}
\begin{enumerate}
\item La condition~\eqref{eq:GP asum t} est satisfaite si $0 \leq \beta < \beta_0 (d,s)$. Pour la preuve du Th\'eor\`eme~\ref{thm:deriv nls} on ne s'int\'eresse qu'\`a des cas o\`u $|\eH|$ est born\'e ind\'ependament de $N$, et~\eqref{eq:GP energy estimate} donne alors une information non triviale seulement si~\eqref{eq:GP asum t} est satisfait.
\item Le r\'esultat est valable dans le cas de la limite de champ moyen o\`u $\beta = 0$ et donc $\eH$ ne d\'epend pas de $N$. On obtient alors des estimations explicites pr\'ecisant le Th\'eor\`eme~\ref{thm:confined}. Ces estimations pr\'esentent une nouveaut\'e dans le cas o\`u la fonctionnelle de Hartree a plusieurs minimiseurs, ou un seul minimiseur d\'eg\'en\'er\'e. Dans le cas contraire\footnote{L'exemple le plus simple assurant unicit\'e et non d\'eg\'en\'erescence est celui o\`u $\hat{w}> 0$ avec $\hat{w}$ la transform\'ee de Fourier de $w$.} de meilleures estimations sont connues, avec une erreur d'ordre $N ^{-1}$ donn\'ee par la th\'eorie de Bogoliubov~\cite{LewNamSerSol-13,Seiringer-11,GreSei-12,DerNap-13}. Voir~\cite{NamSei-14} pour des extensions de la th\'eorie de Bogoliubov au cas de minimiseurs mutliples et/ou d\'eg\'en\'er\'es.
\end{enumerate}
 \hfill\qed
\end{remark}

La preuve du Th\'eor\`eme~\ref{thm:error Hartree} occupera la Section~\ref{sec:GP Hartree bounds}. On compl\`etera ensuite la preuve du Th\'eor\`eme~\ref{thm:deriv nls} \`a la Section~\ref{sec:GP Hartee to NLS}.

\subsection{Estimations quantitatives pour la th\'eorie de Hartree}\label{sec:GP Hartree bounds}\mbox{}\\\vspace{-0.4cm}

L'id\'ee principale de la preuve est d'appliquer le Th\'eor\`eme~\ref{thm:DeFinetti quant} sur un sous-espace propre de basse \'energie de l'op\'erateur \`a un corps 
$$T = -\Delta + V $$
agissant sur $\gH = L ^2 (\R ^d)$. L'hypoth\`ese~\eqref{eq:GP asum V} assure que la r\'esolvante de cet op\'erateur est compacte et donc que son spectre est constitu\'e d'une suite de valeurs propres tendant vers l'infini. On note $P_-$ et $P_+$ les projecteurs spectraux correspondant aux \'energies respectivement plus grandes et plus petites qu'une certaine troncature $\Lambda$:
\begin{equation}\label{eq:GP projectors}
P_- = \1_{(-\infty, \Lambda)} \left( T \right), \: P_+ = \1_{\gH} - P_-=P_- ^{\perp}. 
\end{equation}
On notera 
\begin{equation}\label{eq:NT}
N_\Lambda := \dim (P_- \gH ) = \mbox{ nombre de valeurs propres de } T \mbox{ inf\'erieures \`a } \Lambda.
\end{equation}
Puisque la pr\'ecision du Th\'eor\`eme de de Finetti quantitatif d\'epend de la dimension de l'espace sur lequel on l'applique, il est clair qu'il nous faudra un contr\^ole convenable sur~$N_\Lambda$. Les outils pour obtenir ce contr\^ole sont bien connus sous le nom d'in\'egalit\'es de type Lieb-Thirring, ou plus pr\'ecis\'ement de Cwikel-Lieb-Rosenblum dans ce cas. Nous nous servirons du lemme suivant: 

\begin{lemma}[\textbf{Nombre d'\'etats li\'es d'un op\'erateur de Schr\"odinger}] \label{lem:GP nb bound states}\mbox{}\\
Soit $V$ satisfaisant~\eqref{eq:GP asum V}. Pour tout $d\geq 1$, il existe une constante $C_{d}>0$ telle que, pour $\Lambda$ assez grand
\begin{equation}\label{eq:GP bound NT}
N_\Lambda \leq C_{d} \Lambda ^{d/s + d/2}. 
\end{equation}
\end{lemma}

\begin{proof}[Preuve]
Quand $d\geq 3$, ceci est une application de~\cite[Th\'eor\`eme 4.1]{LieSei-09}. Pour $d\leq 2$, le r\'esultat suit ais\'ement en appliquant~\cite[Th\'eor\`eme 2.1]{ComSchSei-78} ou~\cite[Th\'eor\`eme 15.8]{Simon-05}, voir~\cite{LewNamRou-14}. Les lecteurs familiaris\'es peuvent se convaincre que le membre de droite de~\eqref{eq:GP bound NT} est proportionnel au nombre de niveaux d'\'energies attendu dans la limite semi-classique. On renvoie \`a~\cite[Chapitre 4]{LieSei-09} pour une discussion plus pouss\'ee de ce genre d'in\'egalit\'es.
\end{proof}

Le raisonnement que nous allons impl\'ementer est le suivant:
\begin{enumerate}
\item Les vecteurs propres de $T$ forment une base de $L ^2 (\R ^d)$ sur laquelle les $N$ bosons doivent se r\'epartir. Les m\'ethodes du Chapitre~\ref{sec:locFock} fournissent la bonne fa\c{c}on d'analyser la r\'epartition des particules entre $P_- \gH$ et $P_+ \gH$.
\item Si la troncature $\Lambda$ est choisie assez grande, les particules vivant sur les niveaux d'\'energie \'elev\'es ont une \'energie par unit\'e de masse bien plus grande que l'\'energie de Hartree que l'on cherche \`a obtenir. Il ne peut donc il y avoir que peu de particules sur les niveaux d'\'energie \'elev\'es.
\item Les particules vivant sur $P_- \gH$ forment un \'etat de $\cF_s  ^{\leq N}(P_- \gH)$ (espace de Fock bosonique tronqu\'e). Puisque $P_- \gH$ est de dimension finie, on peut se servir du Th\'eor\`eme~\ref{thm:DeFinetti quant} pour d\'ecrire ces particules. Ceci donnera l'\'energie de Hartree \`a une erreur pr\`es, qui d\'epend de $\Lambda$ et du nombre attendu de particules $P_-$-localis\'ees. Plus pr\'ecis\'ement, au vu de l'estimation~\eqref{eq:error finite dim deF}, il faut s'attendre \`a ce que l'erreur soit du genre 
\begin{equation}\label{eq:GP error heurist}
 \frac{(\Lambda + N ^{d\beta})\times N_\Lambda}{N_-}, 
\end{equation}
i.e. dimension de l'espace localis\'e $\times$ norme d'op\'erateur du Hamiltonien restreint \`a cet espace localis\'e $/$ nombre de particules localis\'ees.
\item Il s'agit ensuite d'optimiser la valeur de $\Lambda$ en gardant l'heuristique suivante en t\^ete: si $\Lambda$ est grand, il y aura beaucoup de particules $P_-$-localis\'ees, ce qui favorise le d\'enominateur de~\eqref{eq:GP error heurist}. En revanche, prendre $\Lambda$ petit r\'eduit le num\'erateur de~\eqref{eq:GP error heurist}. La valeur optimale obtenue en consid\'erant ces deux effets donne les taux d'erreur du Th\'eor\`eme~\ref{thm:error Hartree}.
\end{enumerate}

\begin{proof}[Preuve du Th\'eor\`eme~\ref{thm:error Hartree}.] La borne sup\'erieure dans~\eqref{eq:GP energy estimate} se prouve comme d'habitude en prenant un \'etat test de forme $\psi ^{\otimes N}$. Seule la borne inf\'erieure est non triviale. On proc\`ede en plusieurs \'etapes.

\medskip

\noindent\textbf{Etape 1, troncature du Hamiltonien.} Il s'agit d'abord de se convaincre qu'il est l\'egitime de raisonner uniquement en termes de particules $P_+$ et $P_-$-localis\'ees comme nous l'avons fait ci-dessus. C'est l'objet du lemme suivant, qui estime le Hamiltonien \`a $2$ corps
\begin{equation}\label{eq:GP two body hamil}
H_2 = T\otimes \one + \one \otimes T + w_N 
\end{equation}
en fonction de sa restriction aux syst\`emes $P_-$-localis\'es $P_- \otimes P_- \, H_2 \, P_-\otimes P_-$ et d'une borne grossi\`ere sur l'\'energie des particules $P_+$-localis\'ees.

\begin{lemma}[\textbf{Hamiltonien tronqu\'e}]\label{lem:GP localize-energy}\mbox{}\\
On suppose que $\Lambda \ge C N^{2d\beta}$ pour une constante $C$ assez grande. On a alors   
\begin{align} \label{eq:GP H2-localized-error}
H_2 \ge  P_-\otimes P_- H_2 P_-\otimes P_- + \frac{\Lambda}{2} P_+\otimes P_+ - \frac{36 N^{2d\beta}}{\Lambda}
\end{align}
\end{lemma}

\begin{proof}[Preuve] 
On note 
\begin{equation}\label{eq:GP two body hamil non int}
H_2 ^0 =  H_1 \otimes \1 + \1 \otimes H_1 
\end{equation}
le Hamiltonien \`a deux corps sans interaction. On peut alors \'ecrire 
\begin{align}\label{eq:GP loc 1 body}
H_2 ^0 = \left( P_- + P_+ \right) ^{\otimes 2} H_2 ^0 \left( P_- + P_+ \right) ^{\otimes 2} &= \sum_{i,j,k,\ell \in \{ -,+\}} P_i \otimes P_j \, H_2^{0} \, P_k \otimes P_\ell \nonumber \\
&= \sum_{i,j \in \{ -,+\}} P_i \otimes P_j \, H_2^{0} \, P_i \otimes P_j. 
\end{align}
En effet, 
$$P_i \otimes P_j \, H_2 ^0 \, P_k \otimes P_\ell =0$$
si $i\ne k$ or $j \neq \ell$, puisque $T$ commute avec $P_{\pm}$ et que $P_- P_+ =0$. On note ensuite que 
$$P_+ H_1 P_+\ge \Lambda P_+ \mbox{ et } P_- H_1 P_- \ge -C P_-,$$ 
ce qui donne
\begin{align*}
P_+ \otimes P_+ \, H_2 ^0\, P_+ \otimes P_+ &\ge 2 \Lambda \, P_+ \otimes P_+ \\
P_+ \otimes P_- \, H_2 ^0\, P_+ \otimes P_- &\ge (\Lambda - C) P_+ \otimes P_-\\
P_- \otimes P_+ \,H_2 ^0\, P_- \otimes P_+ &\ge (\Lambda - C) P_+ \otimes P_-.
\end{align*}
Ainsi,  
\begin{equation}\label{eq:GP loc 1 body lower bound}
H_2 ^0  \ge P_- \otimes P_- \, H_2 ^0 \, P_- \otimes P_- + (\Lambda -C) ( P_+ \otimes P_+ + P_- \otimes P_+ + P_+ \otimes P_- ).
\end{equation}

On se tourne vers les interactions:
\begin{equation}\label{eq:GP split interaction}
w_N = \left( P_- + P_+ \right) ^{\otimes 2} w_N \left( P_- + P_+ \right) ^{\otimes 2}= \sum_{i,j,k,\ell \in \{ -,+\}} P_i \otimes P_j \,w_N \,P_k \otimes P_\ell,
\end{equation}
et il s'agit de borner la diff\'erence 
$$w_N- P_- \otimes P_- \, w_N \, P_- \otimes P_-$$
en contr\^olant les termes non diagonaux $(i,j)\neq (k,\ell)$ de~\eqref{eq:GP split interaction} en fonction des termes diagonaux. 
A cette fin, on \'ecrit 
$$w_N=(w_N)^+ - (w_N)^- \mbox{ avec } (w_N)^{\pm} \ge 0$$
et on va utiliser le fait bien connu que les \'el\'ements diagonaux d'un op\'erateur auto-adjoint positif contr\^olent les \'el\'ements hors-diagonaux\footnote{Cf pour une matrice hermitienne positive $(m_{i,j})_{1\leq i, j \leq n}$ l'in\'egalit\'e $2 |m_{i,j}| \leq m_{i,i} + m_{j,j}$.}.

Comme $w_N ^{\pm}$ vus comme op\'erateurs de multiplication sur $L ^2 (\R ^{2d})$  sont positifs, on a pour tout $b>0$ et tout  $(i,j) \neq (k,\ell)$
$$
\left(b^{1/2} P_i \otimes P_j \pm b^{-1/2} P_k\otimes P_\ell \right) \, w_N^{\pm} \, \left(b^{1/2} P_i \otimes P_j \pm b^{-1/2} P_k\otimes P_\ell  \right) \ge 0.
$$
En combinant ces in\'egalit\'es de mani\`ere appropri\'ee on obtient 
$$
P_i \otimes P_j \, w_N \, P_k \otimes P_\ell + P_k \otimes P_\ell \, w_N \, P_i \otimes P_j \ge - b P_i \otimes P_j |w_N| P_i \otimes P_j - b^{-1} P_k \otimes P_\ell |w_N| P_k \otimes P_\ell
$$
pour tout $b>0$. On rappelle alors qu'en tant qu'op\'erateur 
$$|w_N| \le \|w_N\|_{L^\infty} \le N^{d\beta}$$
et on choisit $b=12 N^{d\beta}/\Lambda$, ce qui donne
$$
P_i \otimes P_j \, w_N \, P_k \otimes P_\ell + P_k \otimes P_\ell \, w_N  \, P_i \otimes P_j \ge - \frac{12 N^{2d\beta}}{\Lambda} P_i \otimes P_j - \frac{\Lambda}{12} P_k \otimes P_\ell .
$$
On applique cette borne \`a tous les termes $(i,j,k,\ell)$ de~\eqref{eq:GP split interaction} o\`u au moins un indice diff\`ere de $-$  pour obtenir
\begin{multline} \label{eq:GP loc interaction lower bound}
w_N \ge  P_- \otimes P_- \, w_N \, P_- \otimes P_-  -\frac{36 N^{2d\beta}}{\Lambda} P_- \otimes P_- \\
- \left( \frac{\Lambda}{3} + \frac{48 N^{2d\beta}}{\Lambda}\right) \left(P_- \otimes P_+ + P_+ \otimes P_- + P_+ \otimes P_+\right).
\end{multline}
En combinant~\eqref{eq:GP loc 1 body lower bound} et~\eqref{eq:GP loc interaction lower bound} on obtient pour tout $\Lambda\ge 1$ la borne inf\'erieure
\begin{multline*}
H_2 \ge  P_- \otimes P_- \, H_2 \, P_- \otimes P_- - \frac{36 N^{2d\beta}}{\Lambda} P_- \otimes P_- \\
+ \left( \frac{2\Lambda}{3}- \frac{48N^{2d\beta}}{\Lambda}-C\right) \left(P_- \otimes P_+ + P_+ \otimes P_- + P_+ \otimes P_+\right)
\end{multline*}
Puisqu'on suppose  $\Lambda \ge C N^{2d\beta}$ pour une grande constante $C$, on peut utiliser $P_- \otimes P_- \le \1$ et $P_- \otimes P_+,\, P_+ \otimes P_-\ge 0$ pour d\'eduire
$$
H_2 \ge  P_- \otimes P_- \, H_2 \, P_- \otimes P_- - \frac{36 N^{2d\beta}}{\Lambda} + \frac{\Lambda}{2} P_+ \otimes P_+,
$$
ce qui conclut la preuve.
\end{proof}

\medskip

\noindent\textbf{Etape 2, estimation de l'\'energie localis\'ee.} 
Soit $\Psi_N$ un minimiseur pour l'\'energie \`a $N$ corps, $\Gamma_N = |\Psi_N \rangle \langle \Psi_N |$ et 
$$ \gamma_N ^{(n)} = \tr_{n+1 \to N} [\Gamma_N] $$
les matrices r\'eduites correspondantes. Nous pouvons maintenant raisonner uniquement en termes des \'etats $P_-$ et $P_+$-localis\'es d\'efinis comme au Lemme~\ref{lem:Fock loc} par les relations
\begin{equation}\label{eq:GP loc pm}
\left(G_N ^{\pm}\right) ^{(n)} = P_{\pm} ^{\otimes n} \gamma_N ^{(n)} P_\pm ^{\otimes n}. 
\end{equation}
On rappelle que $G_N ^{\pm}$ sont des \'etats sur l'espace de Fock tronqu\'e, i.e. 
\begin{equation}\label{eq:GP nomalization-localized-state}
\sum_{k=0} ^N \tr_{\gH ^k }[G_{N,k} ^{\pm}] = 1. 
\end{equation}
Nous comparons maintenant $\eH$ et l'\'energie localis\'ee de $\Gamma_N$ d\'efinie par le Hamiltonien tronqu\'e du Lemme~\ref{lem:GP localize-energy}:

\begin{lemma}[\textbf{Borne inf\'erieure pour l'\'energie localis\'ee}]\label{lem:GP main terms}\mbox{}\\
Si $\Lambda\geq C N ^{2d\beta}$ pour une constante $C$ assez grande, on a  
\begin{equation}\label{eq:low bound main}
  \frac{1}{2} \tr \left[ P_- \otimes P_- H_2 P_- \otimes P_- \gamma_{N}^{(2)} \right] + \frac{\Lambda}{4} \tr
  \left[P_+ \otimes P_+ \gamma_N^{(2)} \right] \ge \eH  - C \frac{\Lambda N_\Lambda}{N}- C\frac{\Lambda}{N ^2} .
\end{equation}
\end{lemma}

La preuve de ce lemme consiste en une application combin\'ee du Th\'eor\`eme~\ref{thm:DeFinetti quant} et des m\'ethodes du Chapitre~\ref{sec:locFock}. On va d\'efinir une mesure de de Finetti approch\'ee en partant de $G_N ^-$. L'id\'ee est proche de celle que nous avons utilis\'e pour la preuve du th\'eor\`eme de de Finetti faible \`a la Section~\ref{sec:proof deF faible}:
\begin{equation}\label{eq:GP def-mu-N-localized}
d\mu_N(u) = \sum_{k=2}^N {N \choose 2}^{-1} {k \choose 2} \dim \left( (P_-\gH)_s^k \right) \pscal{u^{\otimes k},G_{N,k}^- u^{\otimes k}} du
\end{equation}
o\`u $du$ est la mesure uniforme sur la sph\`ere $SP_-\gH$. Le choix des poids dans la somme ci-dessus vient du fait que l'on cherche \`a approximer la matrice de densit\'e \`a deux corps localis\'ee $P_- \otimes P_- \gamma_N ^{(2)} P_- \otimes P_-$, ce qui est l'objet du 

\begin{lemma} [\textbf{De Finetti quantitatif pour un \'etat localis\'e.}] \label{lem:GP deF-localized-state}\mbox{}\\
Pour tout  $\Lambda>0$,  on a 
$$
\Tr_{\gH^2} \left| P_-^{\otimes 2} \gamma_{N}^{(2)} P_-^{\otimes 2} - \int_{SP_-\gH} |u^{\otimes 2}\rangle \langle u^{\otimes 2}| d\mu_N(u)\right| \le \frac{8 N_\Lambda}{N}.
$$
\end{lemma}

\begin{proof}[Preuve.] A normalisation pr\`es, $G_{N,k}^-$ est un \'etat sur $(P_-\gH)_s^{\otimes k}$. En appliquant le Th\'eor\`eme~\ref{thm:DeFinetti quant} avec la construction explicite~\eqref{eq:def CKMR} on a donc 
\begin{align*}
\tr_{\gH ^2}\left| \Tr_{3\to k}\left[G_{N,k}  ^-\right]  - \int_{SP_- \gH} |u ^{\otimes 2}\rangle \langle u ^{\otimes 2} | d\mu_{N,k}(u) \right| \leq 8 \frac{N_\Lambda}{k} \tr_{\gH ^k} \left[ G_{N,k} ^-\right]
\end{align*}
avec 
$$d\mu_{N,k}(u)= \dim (P_-\gH)_s^k \pscal{u^{\otimes k},G_{N,k}^- u^{\otimes k}} du.$$
Au vu de~\eqref{eq:GP loc pm} et~\eqref{eq:GP def-mu-N-localized} on d\'eduit 
\begin{align*}
\Tr_{\gH^2} \left| P_-^{\otimes 2} \gamma_{N}^{(2)} P_-^{\otimes 2} - \int_{SP_-\gH} |u^{\otimes 2}\rangle \langle u^{\otimes 2}| d\mu_N(u)\right|  \leq 
\sum_{k=2}^N {N \choose 2}^{-1} {k \choose 2} 
\frac{8N_\Lambda}{k} \tr_{\gH ^k} \left[ G_{N,k} ^-\right].
\end{align*}
Il reste alors \`a utiliser la normalisation~\eqref{eq:GP nomalization-localized-state} et 
$$
{N \choose 2}^{-1} {k \choose 2}  = \frac{k(k-1)}{N(N-1)} \le \frac{k}{N}
$$
pour conclure la preuve.
\end{proof}

On peut maintenant passer \`a la 

\begin{proof}[Preuve du Lemme~\ref{lem:GP main terms}.] 
On commence par le terme $P_-$-localis\'e. Par cyclicit\'e de la trace on a 
\begin{align*}
\tr \left[ P_- \otimes P_- H_2 P_- \otimes P_- \gamma_{N}^{(2)} \right] &= \tr \left[ P_- \otimes P_- H_2 P_- \otimes P_- \big( P_- \otimes P_- \gamma_{N}^{(2)} P_- \otimes P_-\big) \right].
\end{align*}
On applique alors le Lemme~\ref{lem:GP deF-localized-state}, ce qui donne 
$$
\Tr_{\gH^2} \left| P_-^{\otimes 2} \gamma_{N}^{(2)} P_-^{\otimes 2} - \int_{SP_-\gH} |u^{\otimes 2}\rangle \langle u^{\otimes 2}| d\mu_N(u)\right| \le \frac{8 N_\Lambda}{N}.
$$
D'autre part on a bien s\^ur
\begin{equation}\label{eq:GP bound H-}
\norm{P_- \otimes P_- H_2 P_- \otimes P_- } \leq 2\Lambda + \|w_N\|_{L^\infty} \le 3\Lambda
\end{equation}
en norme d'op\'erateur, et donc 
\begin{align} \label{eq:GP P-term-0}
\frac{1}{2}\tr \left[ P_- \otimes P_- H_2 P_- \otimes P_- \gamma_{N}^{(2)} \right] &\ge \frac{1}{2}\int_{SP_- \gH} \tr_{\gH ^2} \left[ H_2 |u ^{\otimes 2}\rangle \langle u ^{\otimes 2}| \right]  d\mu_N -\frac{C\Lambda N_\Lambda}{N} \nn\\
& = \int_{SP_- \gH} \EH [u] d\mu_N -\frac{C\Lambda N_\Lambda}{N}. 
\end{align}
Par le principe variationnel $\EH [u] \geq \eH$, on d\'eduit
\begin{align} \label{eq:GP P-term}
\frac{1}{2}\tr \left[ P_- \otimes P_- H_2 P_- \otimes P_- \gamma_{N}^{(2)} \right] \ge 
 \eH \sum_{k=2} ^N \binom{N}{2} ^{-1} \binom{k}{2} \tr_{\gH ^k} \left( G_{N,k} ^- \right) -\frac{C\Lambda N_\Lambda}{N}
\end{align}
o\`u le calcul de $\int d\mu_N$ est imm\'ediat par la formule de Schur~\eqref{eq:Schur}. 

Pour les termes $P_+$-localis\'es on utilise~\eqref{eq:GP loc pm},~\eqref{eq:Fock mat red}  et~\eqref{eq:Fock funda rel} pour obtenir
\begin{align} \label{eq:GP P+term}
\frac{\Lambda}{4}\tr[ P_+ \otimes P_+ \gamma_N^{(2)} P_+\otimes P_+] &= \frac{\Lambda}{4} \sum_{k=2} ^N \binom{N}{2} ^{-1} \binom{k}{2} \tr \left[ G_{N,k}^+ \right] \nn\\
&= \frac{\Lambda}{4} \sum_{k=0} ^{N-2} \binom{N}{2} ^{-1} \binom{N-k}{2} \tr \left[ G_{N,k}^- \right]. 
\end{align}
En rassemblant~\eqref{eq:GP P-term},~\eqref{eq:GP P+term} et en rappellant que 
$$
\binom{N}{2} ^{-1} \binom{k}{2} = \frac{k^2}{N^2} + O(N^{-1}), \quad \binom{N}{2} ^{-1} \binom{N-k}{2} = \frac{(N-k)^2}{N^2} + O(N^{-1}),
$$
on trouve
\begin{align} \label{eq:GP P-P+-together}
&\quad \quad \frac{1}{2}\tr \left[ P_- \otimes P_- H_2 P_- \otimes P_- \gamma_{N}^{(2)} \right] + \frac{\Lambda}{4} \tr[P_+ \otimes P_+ \gamma_N^{(2)}] \nn\\
&\ge \sum_{k=0}^N \tr \left[ G_{N,k}^- \right] \left( \frac{k^2}{N^2}  \eH + \frac{(N-k)^2}{N^2}\frac{\Lambda}{4}\right) - \frac{C(|\eH|+\Lambda)}{N ^2}- \frac{C\Lambda N_\Lambda}{N}. 
\end{align}
Le premier terme d'erreur vient du fait que les sommes dans~\eqref{eq:GP P-term} et~\eqref{eq:GP P+term} ne vont pas exactement de $0$ \`a $N$, et nous avons utilis\'e la normalisation des \'etat localis\'es~\eqref{eq:GP nomalization-localized-state} pour contr\^oler le terme manquant. 

Il est facile de voir que pour tout $p,q$, $0 \leq \lambda \leq 1$
$$ p \lambda ^2 + q(1-\lambda) ^2 \geq   p - \frac{p ^2}{q}.$$
On prend alors $p=\eH$, $q=\Lambda/4$, $\lambda=k/N$ et on utilise~\eqref{eq:GP nomalization-localized-state} \`a nouveau pour d\'eduire 
\begin{align*} 
 &  \tr \left[ P_- \otimes P_- H_2 P_- \otimes P_- \gamma_{N}^{(2)} \right] + \frac{\Lambda}{2} \tr(P_+ \otimes P_+ \gamma_N^{(2)}) \\
& \quad\quad \quad\quad\quad\quad\quad\quad\ge \eH - \frac{\eH^2}{\Lambda} - \frac{C(|\eH|+\Lambda)}{N ^2}- \frac{C\Lambda N_\Lambda}{N} .
\end{align*}
de~\eqref{eq:GP P-P+-together}. Il reste \`a appliquer l'estimation simple
\begin{equation}\label{eq:GP naive bound}
|\eH| \le C+ \|w_N\|_{L^\infty} \le C + N^{d\beta} \leq C + \Lambda
\end{equation}
pour obtenir la borne inf\'erieure d\'esir\'ee.
\end{proof}

\medskip 

\noindent\textbf{Etape 3, optimisation finale.}. Il ne nous reste qu'\`a optimiser la valeur de $\Lambda$. En effet, on rappelle que par d\'efinition
$$ \frac{E(N)}{N} = \frac{1}{2} \tr_{\gH ^2}[ H_2 \gamma_N ^{(2)}] $$ 
avec le Hamiltonien \`a deux corps~\eqref{eq:GP two body hamil}. En combinant les Lemmes~\ref{lem:GP localize-energy} et~\ref{lem:GP main terms} on a la borne inf\'erieure 
\begin{equation*}
\frac{E(N)}{N} \geq \eH - \frac{C N^{2d\beta}}{\Lambda} -  C \frac{\Lambda N_\Lambda}{N}- C\frac{\Lambda}{N ^2}
\end{equation*}
pour tout $\Lambda \ge CN^{2d\beta}$ avec $C$ assez grand. En utilisant le Lemme~\ref{lem:GP nb bound states}, ceci se r\'eduit \`a 
\begin{equation*}
\frac{E(N)}{N} \geq \eH - \frac{C N^{2d\beta}}{\Lambda} - C_{d} \frac{\Lambda^{1+ d/s + d/2}}{N}.
\end{equation*}
En optimisant par rapport \`a $\Lambda$ on obtient
\begin{equation}\label{eq:choice e}
\Lambda= N ^{t}
\end{equation}
avec 
$$ t = 2s \frac{1+2d\beta}{4s + ds + 2d}$$
et la condition $t>2d\beta$ dans~\eqref{eq:GP asum t} assure que $\Lambda \gg N^{2d\beta}$ pour $N$ grand. On conclut donc
\[
\eH \ge \frac{E(N)}{N} \geq \eH - C_{d} N ^{-t + 2d \beta}, 
\]
ce qui est le r\'esultat d\'esir\'e.
\end{proof}

\begin{remark}[Note pour plus tard.]\label{rem:GP by product}\mbox{}\\
En suivant les \'etapes de la preuve plus pr\'ecis\'ement on obtient des informations sur le comportement asymptotique des minimiseurs. Plus sp\'ecifiquement, revenant \`a~\eqref{eq:GP P-term-0}, utilisant le Lemme~\ref{lem:GP localize-energy} et ignorant les termes $P_+$ localis\'es qui sont positifs, on a
$$ \eH \geq \frac{E(N)}{N} \ge \frac{1}{2}\Tr [P_-^{\otimes 2} H_2 P_-^{\otimes 2} \gamma_N^{(2)}] + o(1) \ge \int_{S\gH } \EH [u] d\mu_N (u) + o(1)$$
et donc
\begin{equation} \label{eq:GP mu-N-localized-cv-energy}
o(1) \geq \int_{S\gH } \left( \EH [u] - \eH \right) d\mu_N (u)
\end{equation}
quand $N\to\infty$. Nous ne sp\'ecifions pas ici (cf le point (3) de la Remarque~\ref{rem:GP NLS}) l'ordre de grandeur exact du $o(1)$ obtenu par l'optimisation finale de l'\'etape 3 ci-dessus. L'estimation~\eqref{eq:GP mu-N-localized-cv-energy} dit moralement que $\mu_N$ doit \^etre concentr\'ee sur les minimiseurs de $\EH$, ce dont nous nous servirons pour prouver le Th\'eor\`eme~\ref{thm:deriv nls}. 
\hfill\qed
\end{remark}

\subsection{De Hartree \`a NLS}\label{sec:GP Hartee to NLS} \mbox{}\\\vspace{-0.4cm}

Il nous reste \`a d\'eduire le Th\'eor\`eme~\ref{thm:deriv nls} comme corollaire de l'analyse ci-dessus. On commence par le lemme suivant

\begin{lemma}[\textbf{Stabilit\'e des fonctionnelles \`a un corps}]\label{lem:GP Hartree NLS}\mbox{}\\
On consid\`ere les fonctionnelles~\eqref{eq:GP hartree f} et~\eqref{eq:GP intro nls}. Sous les hypoth\`eses du Th\'eor\`eme~\ref{thm:deriv nls}, il existe un minimiseur pour $\ENLS$. De plus, pour toute fonction normalis\'ee $u\in L^2(\R^d)$ on a 
\bq \label{eq:GP hartree NLS 0}
\norm{|u|}_{H^1}^2 \le C (\EH [u]+C)
\eq
et 
\bq \label{eq:GP hartree NLS 1}
\left| \EH [u] - \ENLS[u] \right| \le C N^{-\beta} \left( 1 + \int_{\R ^d} |\nabla u| ^2 \right) ^2. 
\eq
Par cons\'equent 
\begin{equation}\label{eq:GP hartree NLS 2}
|\eH - \eNLS | \le C N ^{-\beta}.
\end{equation}
\end{lemma}

\begin{proof}[Preuve]
Les hypoth\`eses de stabilit\'e que nous avons faites garantissent que les suites minimisantes pour la minimisation de $\ENLS$ sont born\'ees dans $H^1$. L'hypoth\`ese~\eqref{eq:GP conv delta} permet d'estimer ais\'ement la diff\'erence entre les fonctionnelles de Hartree et NLS pour des fonctions born\'ees dans $H^1$. Les d\'etails sont omis, on peut les trouver dans~\cite{LewNamRou-14}. 
\end{proof}

Nous somme maintenant arm\'es pour compl\'eter la d\'erivation de la fonctionnelle NLS.

\begin{proof}[Preuve du Th\'eor\`eme~\ref{thm:deriv nls}]
Combiner~\eqref{eq:GP hartree NLS 2} avec~\eqref{eq:GP energy estimate} conclut la preuve de~\eqref{eq:GP ener convergence}. Il reste donc \`a prouver la convergence des \'etats, deuxi\`eme point du Th\'eor\`eme. On proc\`ede en quatre \'etapes:

\noindent\textbf{Etape 1, compacit\'e forte des matrices densit\'e.} On extrait une sous-suite diagonale le long de laquelle
\begin{equation}\label{eq:weak CV}
\gamma_N ^{(n)} \wto_* \gamma ^{(n)} 
\end{equation}s
quand $N\to \infty$, pour tout $n\in\N$. On a par ailleurs 
\begin{equation}\label{eq:unif bound kinetic}
\tr \Big[H_1\gamma_N^{(1)}\Big]=\Tr \Big[ \left(- \Delta + V \right) \gamma_N^{(1)} \Big] \leq C, 
\end{equation}
ind\'ependament de $N$. Pour le voir, on choisit un $\alpha >0$ et on d\'efinit
\begin{equation*}\label{eq:start hamil eta}
H_{N,\alpha} = \sum_{j=1} ^N  \left( - \Delta_j + V(x_j) \right) + \frac{1+\alpha}{N-1} \sum_{1\leq i<j \leq N}N ^{d\beta} w( N ^{\beta} (x_i-x_j))
\end{equation*}
auquel on applique le Th\'eor\`eme~\ref{thm:error Hartree}. On trouve en particulier $H_{N,\alpha}\geq -CN$ et on d\'eduit
$$\eNLS+o(1)\geq \frac{\pscal{\Psi_N,H_N\Psi_N}}N\geq -C(1+\alpha)^{-1}+\frac{\alpha}{1+\alpha}\tr \big[H_1\gamma_N^{(1)}\big],$$
ce qui donne~\eqref{eq:unif bound kinetic}. Comme $H_1=- \Delta + V$ est \`a r\'esolvante compacte,~\eqref{eq:weak CV} et~\eqref{eq:unif bound kinetic} impliquent que, \`a une sous-suite pr\`es, $\gamma_{N}^{(1)}$ converge fortement en norme de trace. Comme not\'e pr\'ec\'edement, le Th\'eor\`eme~\ref{thm:DeFinetti fort} implique que $\gamma_N^{(n)}$ aussi converge fortement pour $n\geq1$.

\medskip

\noindent\textbf{Etape 2, d\'efinition de la mesure limite.} On all\'egera les notations en notant $r_N$ la meilleure borne sur $\left|E(N)/N - \eNLS\right|$ obtenue pr\'ec\'edemment. Soit $d\mu_N$ d\'efinie comme au Lemme~\ref{lem:GP deF-localized-state}, satisfaisant 
$$\mu_N(SP_-\gH)=\tr \left[P_-^{\otimes 2}\gamma_N^{(2)}P_-^{\otimes 2}\right]$$
On a
$$
\Tr \left| P_-^{\otimes 2} \gamma_N^{(2)} P_-^{\otimes 2} - \int_{SP_-\gH} |u^{\otimes 2}\rangle \langle u^{\otimes 2}| d\mu_N (u) \right| \leq \frac{8N_\Lambda}{N}\leq C\frac{\Lambda^{1+d/s+d/2}}{N}\to0.
$$
On peut d'autre part d\'eduire des estimations d'\'energie de la Section~\ref{sec:GP Hartree bounds} un contr\^ole sur le nombre de particules excit\'ees:
\begin{equation}\label{eq:bound excited particles}
1-\mu_N(SP_-\gH)=\Tr\left[(1-P_-^{\otimes 2}) \gamma_N^{(2)}\right] \leq \frac{r_N}{\Lambda}. 
\end{equation}
Par l'in\'egalit\'e triangulaire et l'in\'egalit\'e de Cauchy-Schwarz on d\'eduit
\bq \label{eq:mu-N-localized-cv-gamma2}
\Tr \left|\gamma_N^{(2)} - \int_{SP_-\gH} |u^{\otimes 2}\rangle \langle u^{\otimes 2}| d\mu_N (u) \right| \leq C\frac{\Lambda^{1+d/s+d/2}}{N}+C\sqrt{\frac{r_N}{\Lambda}}.
\eq
On note maintenant $P_K$ le projecteur spectral de $H_1$ sur les \'energies sous une troncature $K$, d\'efini comme en~\eqref{eq:GP projectors}. Comme $\gamma_N ^{(2)} \to \gamma ^{(2)}$ et $P_K \to \1$  
$$ 
\lim_{K \to \infty} \lim_{N\to \infty} \mu_N (S P_K\gH) = 1.
$$
Cette condition permet d'utiliser le Th\'eor\`eme de Prokhorov's et~\cite[Lemme~1]{Skorokhod-74} pour assurer que, apr\`es \'eventuellement extraction d'une sous-suite, $\mu_N$ converge vers une mesure $\mu$ sur la boule $B\gH$. En passant \`a la limite, on trouve
$$\gamma^{(2)}=\int_{B\gH} |u^{\otimes 2}\rangle \langle u^{\otimes 2}| d\mu(u)$$
et il s'ensuit que $\mu$ est \`a support sur la sph\`ere $S\gH$ puisque $\tr [\gamma ^{(2)}] = 1$ par convergence forte de la sous-suite.

\medskip

\noindent\textbf{Etape 3, la mesure limite ne charge que les minimiseurs NLS.} En utilisant \eqref{eq:GP mu-N-localized-cv-energy} et  
$$\mu_N(SP_-\gH)=1+O\left(\frac{r_N}{\Lambda}\right),$$
on d\'eduit que
$$\int_{SP_-\gH } \big(\EH  [u] -\eH \big)d\mu_N (u)\leq o(1)$$
dans la limite $N\to \infty$. Par les estimations du Lemme~\ref{lem:GP Hartree NLS}, il s'ensuit que, pour une constante $B$ assez grande (ind\'ependament de $N$),
$$\frac{B^2}{C}\int_{\|\nabla u\|_{L ^2}\geq B} d\mu_N (u)\leq \int_{\|\nabla u \|_{L ^2}\geq B} \big(\EH [u] -\eH \big)d\mu_N (u)\leq o(1),$$
et
$$\int_{\|\nabla u\|_{L^2}\leq B} \big(\ENLS [u] -\eNLS\big)d\mu_N\leq C(1+B^4) N^{-\beta} +\int_{\|\nabla u\|_{L^2}\leq B} \big(\EH [u] -\eH\big)d\mu_N (u)\leq o(1).$$
En passant \`a la limite $N\to\ii$, on voit maintenant que $\mu$ a son support dans $\MNLS$.

A ce stade, en utilisant~\eqref{eq:mu-N-localized-cv-gamma2} et la convergence de $\mu_N$ on a, fortement en norme de trace,  
$$ \gamma_N^{(2)} \to \int_{\MNLS} |u^{\otimes 2}\rangle \langle u^{\otimes 2}| d\mu(u),$$
o\`u $\mu$ est une probabilit\'e \`a support dans $\MNLS$. En prenant une trace partielle on obtient
$$ \gamma_N^{(1)} \to \int_{\MNLS} |u\rangle \langle u| d\mu(u)$$
et il ne reste donc plus qu'\`a obtenir la convergence des matrices densit\'e r\'eduites d'ordre $n >2$. 

\medskip

\noindent\textbf{Etape 4, matrices densit\'e d'ordre \'elev\'e.} On veut obtenir
$$ \gamma_N^{(n)} \to \int_{\MNLS} |u^{\otimes n}\rangle \langle u^{\otimes n}| d\mu(u),$$
en norme de trace quand $N\to \infty$. Vue la d\'efinition de $\mu$, il suffit de montrer
\begin{equation}\label{eq:claim higher DM}
\Tr \left|\gamma_N^{(n)} - \int_{SP_-\gH} |u^{\otimes n}\rangle \langle u^{\otimes n}| d\mu_N (u) \right| \to 0
\end{equation}
o\`u $\mu_N$ est la mesure d\'efinie en appliquant le Lemme~\ref{lem:GP deF-localized-state} \`a $\gamma_N ^{(2)}$. Pour ce faire on commence par approximer $\gamma_N ^{(n)}$ en utilisant une nouvelle mesure, a priori diff\'erente de $\mu_N = \mu_N^2$ 
\bq \label{eq:def-mu-N-localized bis}
d\mu_N ^n (u) = \sum_{k=n}^N  { N \choose n} ^{-1} {k\choose n} \dim (P_-\gH)_s^k \pscal{u^{\otimes k},G_{N,k}^- u^{\otimes k}} du.
\eq
En proc\'edant comme \`a la preuve du Lemme~\ref{lem:GP deF-localized-state} on obtient 
\begin{equation}\label{eq:deF higher DM}
\Tr_{\gH^n} \left| P_-^{\otimes n} \gamma_{N}^{(n)} P_-^{\otimes n} - \int_{SP_-\gH} |u^{\otimes n}\rangle \langle u^{\otimes n}| d\mu_N ^n (u)\right| \le C\frac{ n N_\Lambda}{N}
\end{equation}
Une estimation similaire \`a~\eqref{eq:bound excited particles} montre ensuite que 
$$\Tr_{\gH^n} \left| \gamma_{N}^{(n)} - \int_{SP_-\gH} |u^{\otimes n}\rangle \langle u^{\otimes n}| d\mu_N ^n (u)\right| \to 0.$$
En utilisant \`a nouveau la borne
$$ { N \choose n} ^{-1} {k\choose n} = \left(\frac{k}{N}\right) ^n + O (N ^{-1})$$
ainsi que l'in\'egalit\'e triangulaire et la formule de Schur~\eqref{eq:Schur} on d\'eduit de~\eqref{eq:deF higher DM} que 
\begin{align}\label{eq:higher DM final}
\Tr \left|\gamma_N^{(n)} - \int_{SP_-\gH} |u^{\otimes n}\rangle \langle u^{\otimes n}| d\mu_N (u) \right| &\leq  \sum_{k=0} ^N \left( \left(\frac{k}{N}\right) ^2 - \left(\frac{k}{N}\right) ^n \right) \tr_{\gH ^k} \left[ G_{N,k} ^- \right]  \nonumber
\\&+ \sum_{k=0} ^{n-1} \left(\frac{k}{N}\right) ^n \tr_{\gH ^k} \left[ G_{N,k} ^- \right] \nonumber
\\&+ \sum_{k=0} ^{2} \left(\frac{k}{N}\right) ^2 \tr_{\gH ^k} \left[ G_{N,k} ^- \right] + o(1).
\end{align}
Finalement, en combinant les diff\'erentes bornes obtenues on a
$$ \sum_{k=2} ^N  \left(\frac{k}{N}\right) ^2 \tr_{\gH ^k} \left[ G_{N,k} ^- \right] \to 1$$
mais, par~\eqref{eq:GP nomalization-localized-state} il suit que
$$\sum_{k=0} ^N  \left(\frac{k}{N}\right) ^2 \tr_{\gH ^k} \left[ G_{N,k} ^- \right] \to 1.$$
On peut donc appliquer l'in\'egalit\'e de Jensen pour obtenir
$$1\geq \sum_{k=0} ^N  \left(\frac{k}{N}\right) ^n \tr_{\gH ^k} \left[ G_{N,k} ^- \right]\geq \left( \sum_{k=0} ^N  \left(\frac{k}{N}\right) ^2 \tr_{\gH ^k} \left[ G_{N,k} ^- \right]\right) ^{n/2} \to 1.$$
Il ne reste qu'\`a ins\'erer ceci et~\eqref{eq:GP nomalization-localized-state} dans~\eqref{eq:higher DM final} pour conclure la preuve de~\eqref{eq:claim higher DM} et donc celle du th\'eor\`eme. 
\end{proof}

\appendix

\section{\textbf{Usage quantique du th\'eor\`eme classique}}\label{sec:class quant}

Cet appendice est consacr\'e \`a une preuve alternative d'une version plus faible du Th\'eor\`eme~\ref{thm:confined}. La m\'ethode, introduite dans~\cite{Kiessling-12} est de port\'ee moins g\'en\'erale que celles d\'ecrites pr\'ec\'edemment, ce qui est in\'evitable puisqu'elle consiste en l'application du th\'eor\`eme de Hewitt-Savage (de Finetti classique) \`a un probl\`eme quantique. Nous suivrons ici une note non publi\'ee de Mathieu Lewin et Nicolas Rougerie~\cite{LewRou-unpu12}.

Dans certains cas (absence de champ magn\'etique essentiellement), la fonction d'onde $\Psi_N$ minimisant une \'energie \`a $N$ corps peut-\^etre choisie strictement positive. Le fondamental du probl\`eme quantique peut alors \^etre enti\`erement analys\'e en termes de la densit\'e \`a $N$ corps $\rho_{\Psi_N} = |\Psi_N| ^2$ qui est un objet purement classique (mesure de probabilit\'e sym\'etrique) dont la limite peut-\^etre d\'ecrite \`a l'aide du Th\'eor\`eme~\ref{thm:HS}. Cette approche ne fonctionne toutefois que sous certaines hypoth\`eses sur le Hamiltonien \`a un corps du probl\`eme, beaucoup plus restrictives que celles \'evoqu\'ees \`a la Remarque~\ref{rem:energie cin}. 

\subsection{Formulation classique du probl\`eme quantique}\label{sec:hyp class quant}\mbox{}\\\vspace{-0.4cm}

On consid\`ere ici un Hamiltonien quantique \`a $N$ corps agissant sur $L ^2 (\R ^{dN})$
\begin{equation}\label{eq:app A hamil N}
H_N = \sum_{j=1} ^N \left(T_j + V (x_j)\right) + \frac{1}{N-1} \sum_{1\leq i < j \leq N} w(x_i-x_j)
\end{equation}
o\`u les potentiels de pi\'egeage $V$ et d'interaction $w$ sont choisis comme \`a la Section~\ref{sec:Hartree deF fort}. En particulier on suppose que $V$ est confinant. Nous aurons besoin d'hypoth\`eses assez fortes sur l'op\'erateur $T$ d\'ecrivant l'\'energie cin\'etique des particules. L'approche pr\'esent\'ee dans cet appendice est en effet bas\'ee sur la notion suivante:

\begin{definition}[\textbf{Energie cin\'etique \`a noyau positif}]\label{def:hyp class quant}\mbox{}\\
On dit que $T$ est \`a noyau positif si il existe $T(x,y):\R ^d \times \R ^d \to \R ^+$ tel que  
\begin{equation}\label{eq:hyp class quant}
\bral \psi, T \psi \ketr = \iint_{\R ^d\times \R ^d} T(x,y) \left| \psi (x) - \psi (y) \right| ^2 dxdy 
\end{equation}
pour toute fonction $\psi \in L ^2 (\R^d)$.\hfill\qed
\end{definition}

Il est bien connu que l'\'energie cin\'etique pseudo-relativiste est de cette forme. On a en effet
\begin{equation}\label{eq:app A cin}
\bral \psi, \sqrt{-\Delta} \, \psi \ketr =  \frac{\Gamma(\frac{d+1}{2})}{2\pi ^{(n+1)/2}} \iint_{\R ^d \times \R ^d} \frac{\left| \psi (x) - \psi (y) \right| ^2}{|x-y| ^{d+1}} dxdy,
\end{equation}
voir~\cite[Th\'eor\`eme 7.12]{LieLos-01}. Plus g\'en\'eralement on peut consid\'erer $T = |p| ^s$, $0< s < 2$:
$$ \bral \psi, |p| ^s \ketr = \bral \psi, (-\Delta) ^{s/2} \psi \ketr =  C_{d,s} \iint_{\R ^d \times \R ^d} \frac{\left| \psi (x) - \psi (y) \right| ^2}{|x-y| ^{d+s}} dxdy$$
en rappellant la correspondance~\eqref{eq:quant mom}. 

L'\'energie cin\'etique non-relativiste ne rentre pas dans ce cadre, mais on peut toutefois lui appliquer les consid\'erations de cet appendice car 
\begin{equation}\label{eq:app A cin bis}
\bral \psi, -\Delta \, \psi \ketr = \int_{\R ^d} |\nabla \psi| ^2 = C_d \: \underset{s\uparrow 1}{\lim} \:(1-s) \iint_{\R ^d \times \R ^d} \frac{\left| \psi (x) - \psi (y) \right| ^2}{|x-y| ^{d+2s}} dxdy 
\end{equation}
avec 
$$ C_d = \left( \int_{S^{d-1}} \cos \theta \, d\sigma \right) ^{-1}$$
o\`u $S^{d-1}$ est la sph\`ere euclidienne munie de sa mesure de Lebesgue $\sigma$ et $\theta$ repr\'esente l'angle par rapport \`a la verticale. On peut donc voir $-\Delta$ comme un cas limite de la D\'efinition~\ref{def:hyp class quant}. La formule~\eqref{eq:app A cin bis} est d\'emontr\'ee dans~\cite[Corollaire 2]{BouBreMir-01} et \cite{MasNag-78}, voir aussi~\cite{BouBreMir-02,MazSha-02}.

Les cas notoirement exclus du cadre ci-dessus sont ceux avec champ magn\'etique $T= \left( p + A \right) ^2$ et $T= \left| p + A \right|$, que l'on peut traiter avec les m\'ethodes du corps de ces notes, mais pas celles de cet appendice. 

\medskip

Une cons\'equence importante du choix d'\'energie cin\'etique fait ici est qu'on a par l'in\'egalit\'e triangulaire
$$ \bral \psi, T \psi \ketr \geq \bral |\psi|, T |\psi| \ketr $$
et donc l'\'energie totale \`a $N$ corps 
$$ \E_N [\Psi_N] = \bral \Psi_N, H_N \Psi_N \ketr$$
satisfait
$$ \E_N [\Psi_N] \geq \E_N \left[| \Psi_N |\right].$$
L'\'energie fondamentale peut donc se calculer en prenant uniquement des fonctions test positives
\begin{equation}\label{eq:app A formul class pre}
E(N) = \inf \left\{ \E_N [\Psi_N], \Psi_N \in L_s ^2 (\R ^{dN})\right\} = \inf \left\{ \E_N [\Psi_N], \Psi_N \in L_s ^2 (\R ^{dN}), \Psi_N \geq 0 \right\}.
\end{equation}
Cette remarque permet de d\'emontrer un fait \'evoqu\'e pr\'ec\'edemment: le fondamental bosonique est identique au fondamental absolu dans le cas d'une \'energie cin\'etique de forme~\eqref{eq:hyp class quant} ou~\eqref{eq:app A cin bis}, voir~\cite[Chapitre 3]{LieSei-09}. 

On rappelle la d\'efinition de la fonctionnelle de Hartree: 
$$ \EH [u] = \bral u,T u \ketr + \int_{\R ^d } V |u| ^2 + \frac{1}{2} \iint_{\R ^d \times \R ^d } |u(x)| ^2 w(x-y) |u(y)| ^2 dx dy,$$
on infimum \'etant not\'e $\eH$. Nous allons dans la suite de cet appendice prouver l'\'enonc\'e suivant, qui est une variante un peu affaiblie du Th\'eor\`eme~\ref{thm:confined}:

\begin{theorem}[\textbf{D\'erivation de la th\'eorie de Hartree, \'enonc\'e alternatif}]\label{thm:class quant confined}\mbox{}\\
On fait pour $V$ et $w$ les hypoth\`eses de la Section~\ref{sec:Hartree deF fort}, en particulier~\eqref{eq:hartree confine 2}. On suppose de plus que $T$ est soit \`a noyau positif au sens de la D\'efinition~\ref{def:hyp class quant}, soit un cas limite de cette d\'efinition, comme en~\eqref{eq:app A cin bis}. On a alors 
$$\lim_{N\to\ii}\frac{E(N)}{N}=\eH.$$
Soit $\Psi_N \geq 0$ un fondamental de $H_N$ r\'ealisant l'infimum~\eqref{eq:app A formul class pre} et 
$$ \rho_N ^{(n)} (x_1,\ldots,x_n) := \int_{\R ^{d(N-n)}} \left|\Psi_N\left( x_1,\ldots,x_N \right)\right| ^2 dx_{n+1}\ldots dx_N $$
sa $n$-i\`eme densit\'e r\'eduite. Il existe une mesure de probabilit\'e $\mu$ sur $\cM_{\rm H}$ l'ensemble des minimiseurs de $\EH$ (modulo une phase), telle que, le long d'une sous-suite 
\begin{equation}\label{eq:app A def fort result}
\lim_{N\to\ii} \rho^{(n)}_{N}= \int_{\cM_{\rm H}} \left|u ^{\otimes n}\right| ^2 d\mu(u) \mbox{ pour tout } n\in \N,
\end{equation}
fortement dans $L ^1 \left(\R ^{dn}\right).$ En particulier, si $\eH$ a un minimiseur unique (modulo une phase constante), alors pour toute la suite
\begin{equation}
\lim_{N\to\ii} \rho^{(n)}_{N} = \left|\uH^{\otimes n}\right| ^2 \mbox{ pour tout } n\in \N,.
\label{eq:app A BEC-confined} 
\end{equation}
\end{theorem}

\begin{remark}[Unicit\'e pour la th\'eorie de Hartree]\label{rem:app A unic Hartree}\mbox{}\\
Dans le cas d'une \'en\'ergie cin\'etique \`a noyau positif, l'unicit\'e pour la th\'eorie de Hartree est imm\'ediate si $w\geq 0$. En effet, l'\'energie cin\'etique $\bral \psi,T \psi \ketr$ est dans ce cas une fonction strictement convexe de $\rho = \sqrt{|\psi| ^2}$, voir~\cite[Chapitre 7]{LieLos-01}. Si $w$ est positif l'\'energie $\EH$ est donc elle-m\^eme strictement convexe en $\rho$. \hfill\qed 
\end{remark}

Le cas particulier o\`u $T= -\Delta$ a \'et\'e obtenu par Kiessling dans~\cite{Kiessling-12}, et nous suivrons sa m\'ethode de preuve pour le cas g\'en\'eral. Elle consiste \`a traiter le probl\`eme quantique comme un probl\`eme purement classique, ce qui explique que nous n'obtiendrons que la convergence des densit\'es r\'eduites~\eqref{eq:app A def fort result} au lieu de celle des matrices de densit\'e r\'eduites~\eqref{eq:def fort result}. On continue dans la direction de~\eqref{eq:app A formul class pre} en \'ecrivant
\begin{equation}\label{eq:app A formul class}
E(N) =  \inf \left\{ \E_N \left[\sqrt{\mubf_N}\right], \mubf_N \in \PP_s  (\R ^{dN})\right\}
\end{equation}
o\`u $\mubf_N$ joue le r\^ole de $|\Psi_N| ^2$ et on s'est servi du fait que l'on peut supposer $\Psi_N \geq 0$. L'objet que nous avons \`a traiter est une probabilit\'e sym\'etrique de $N$ variables, et notre strat\'egie sera similaire \`a celle employ\'ee pour la preuve du Th\'eor\`eme~\eqref{thm:aHS}:  
\begin{itemize}
\item Puisque le probl\`eme est confin\'e, on pourra passer ais\'ement \`a la limite et obtenir un probl\`eme pour un \'etat classique \`a nombre infini de particules $\mubf \in \PP_s (\R ^{d\N})$. On utilisera alors le Th\'eor\`eme~\ref{thm:HS} pour d\'ecrire la limite $\mubf ^{(n)}$ de $\mubf_N ^{(n)}$, pour tout $n$, par une unique mesure de probabilit\'e $P_{\mubf} \in \PP (\PP (\R ^d))$.  
\item Le point subtil est de d\'emontrer que l'\'energie limite est bien une fonction affine de $\mubf$, ce qui utilise de mani\`ere essentielle le fait que l'\'energie cin\'etique est \`a noyau positif (ou un cas limite de telles \'energies), ainsi que le Th\'eor\`eme de Hewitt-Savage. 
\end{itemize}

Ces deux \'etapes sont contenues dans les deux sections suivantes. On concluera bri\`evement la preuve du Th\'eor\`eme~\ref{thm:class quant confined} dans une troisi\`eme section.

\subsection{Passage \`a la limite}\label{sec:lim class quant}\mbox{}\\\vspace{-0.4cm}

Le probl\`eme limite que nous allons obtenir est d\'ecrit par la fonctionnelle (comparer avec~\eqref{eq:aHS somme})
\begin{multline}\label{eq:app A prob lim}
\E [\mubf]:= \limsup_{n\to \infty} \frac{1}{n} T \left(\sqrt{\mubf ^{(n)}} \right) 
\\ + \int_{\R ^d} V(x) d\mubf ^{(1)}(x) + \frac{1}{2}\iint_{\R ^d\times \R ^d } w(x-y) d\mubf ^{(2)}(x,y), 
\end{multline}
o\`u $\mubf \in \PP_s(\R ^{d\N})$ et on a pos\'e 
\begin{equation}\label{eq:app A kin class}
 T \left(\sqrt{\mubf_n} \right) := \bral \sqrt{\mubf_n}, \left(\sum_{j=1} ^n T_j \right) \sqrt{\mubf_n} \ketr 
\end{equation}
pour toute mesure de probabilit\'e $\mubf_n \in \PP (\R ^{dn})$.

\begin{lemma}[\textbf{Passage \`a la limite}]\label{lem:app A lim}\mbox{}\\
Soit $\mubf_N \in \PP_s(\R ^{dN})$ r\'ealisant l'infimum dans~\eqref{eq:app A formul class}. Le long d'une sous-suite on a 
$$ \mubf_N ^{(n)} \wto_* \mubf ^{(n)} \in \PP_s (\R ^{dn})$$
pour tout $n\in \N$, au sens des mesures. La suite $\left( \mubf ^{(n)}\right)_{n\in \N}$ d\'efinit une mesure de probabilit\'e $\mubf\in\PP_s (\R^{d\N})$ et on a 
\begin{equation}\label{eq:app A lim inf}
\liminf_{N\to \infty} \frac{E(N)}{N} \geq \E[\mubf]. 
\end{equation}
\end{lemma}

\begin{proof}[Preuve]
L'extraction des limites $\mubf ^{(n)}$ fonctionne comme en Section~\ref{sec:appli HS}. L'existence de la mesure $\mubf \in \PP_s (\R ^{d\N})$ \'egalement, en utilisant le th\'eor\`eme de Kolmogorov.

Passer \`a la liminf dans les termes 
$$ \frac{1}{N} \sum_{j=1} ^N \int_{\R^{dN}} V(x_j) \: d\mubf_N (x_1,\ldots,x_N) = \int_{\R ^d} V(x) \: d\mubf_N ^{(1)} (x)$$ 
et 
$$ \frac{1}{N} \frac{1}{N-1} \sum_{1\leq i<j\leq N} ^N \int_{\R^{dN}} w(x_i-x_j) \:d \mubf_N (x_1,\ldots,x_N) = \frac{1}{2}\iint_{\R ^d\times \R ^d } w(x-y) \: d\mubf_N ^{(2)}(x,y)$$
utilise les m\^emes id\'ees qu'aux Chapitres~\ref{sec:class} et~\ref{sec:quant}, nous n'\'elaborerons pas plus sur ce point.  

Le point nouveau est le traitement de l'\'energie cin\'etique en vue d'obtenir   
\begin{equation}\label{eq:app A claim 1}
 \liminf_{N\to \infty} \frac{1}{N} T \left(\sqrt{\mubf_N }\right) \geq \limsup_{n\to \infty} \frac{1}{n} T \left(\sqrt{\mubf ^{(n)}} \right). 
\end{equation}
Pour ce faire, on note $\Psi_N = \mubf_N ^2$ un minimiseur dans~\eqref{eq:app A formul class pre} et on a alors  
\begin{align*}
\frac{1}{N} T \left(\sqrt{\mubf_N }\right) &= \frac{1}{N} \tr \left[ \sum_{j=1} ^N T_j \ketl \Psi_N \ketr \bral \Psi_N \brar \right]\\
&=\frac{1}{n} \tr \left[ \sum_{j=1} ^n T_j \gamma_N ^{(n)} \right]
\end{align*}
o\`u $\gamma_N ^{(n)}$ est la $n$-i\`eme matrice de densit\'e reduite de $\Psi_N$. Il s'agit d'un op\'erateur \`a trace, que l'on d\'ecompose sous la forme  
$$ \gamma_N ^{(n)} = \sum_{k=1} ^{+\infty} \lambda_n ^k  | u_n ^k \rangle \langle u_n ^k |$$
avec $u_n ^k$ normalis\'e dans $L_s ^2(\R ^{dk})$ et $\sum_{k=1} ^{\infty} \lambda_k ^n = 1$. En ins\'erant cette d\'ecomposition dans l'\'equation pr\'ec\'edente et en rappellant~\eqref{eq:app A kin class} on obtient par lin\'earit\'e de la trace
\begin{align*}
\frac{1}{N} T \left(\sqrt{\mubf_N }\right) &= \frac{1}{n}  \sum_{k=1} ^{+\infty} \lambda_n ^k \: T\left( \sqrt{|u_n ^k| ^2} \right) \\
&\geq \frac{1}{n}  T\left( \sqrt{ \sum_{k=1} ^{+\infty} \lambda_n ^k |u_n ^k| ^2} \right) \\
&= \frac{1}{n}  T\left( \sqrt{ \rho_{\gamma_N ^{(n)}}} \right)
\end{align*}
o\`u l'in\'egalit\'e utilise la convexit\'e de l'\'energie cin\'etique en la densit\'e $\rho$, d\'ej\`a rappell\'ee \`a la Remarque~\ref{rem:app A unic Hartree}, cf~\cite[Chapitre 7]{LieLos-01}. Dans la derni\`ere \'egalit\'e 
$$ \rho_{\gamma_N ^{(n)}} = \sum_{k=1} ^\infty \lambda_n ^k |u_n ^k| ^2$$
est la densit\'e\footnote{Formellement, la partie diagonale du noyau.} de $\gamma_N ^{(n)}$ et il est ais\'e de voir que 
$$ \rho_{\gamma_N ^{(n)}} = \rho_{N} ^{(n)} = \int_{\R^{d(N-n)}} |\Psi_N(x_1,\ldots,x_N)| ^2 dx_{n+1}\ldots dx_N,$$
ce qui donne
$$ \frac{1}{N} T \left(\sqrt{\mubf_N }\right) \geq \frac{1}{n} T \left( \sqrt{\mubf_N ^{(n)}}\right).$$
Pour obtenir~\eqref{eq:app A claim 1}, il ne reste plus qu'\`a passer d'abord \`a la liminf en $N$ (en utilisant le lemme de Fatou), puis \`a la limsup en $n$.
\end{proof}

La notion d'\'energie cin\'etique \`a noyau positif est d\'ej\`a cruciale \`a ce niveau. C'est elle qui garantit la propri\'et\'e de convexit\'e que nous venons d'utiliser. Elle va jouer un r\^ole encore plus important \`a la section suivante.

\subsection{Le probl\`eme limite}\label{sec:linear class quant}

Il s'agit maintenant de montrer que la fonctionnelle~\eqref{eq:app A prob lim} est affine sur $\PP_s (\R ^{d\N})$. Les deux derniers termes le sont \'evidemment, ce qui n'est pas surprenant puisque ce sont des termes de nature classique. Il suffit donc de montrer que le premier terme, qui encode l'aspect quantique du probl\`eme, est \'egalement lin\'eaire en la densit\'e:    

\begin{lemma}[\textbf{Lin\'earit\'e de l'\'energie cin\'etique limite}]\label{lem:app A kin lim}\mbox{}\\
La fonctionnelle  
$$ T \left(\sqrt{\mubf}\right) := \limsup_{n\to \infty} \frac{1}{n} T \left(\sqrt{\mubf ^{(n)}} \right)$$
est affine sur $\PP_s (\R ^{d\N})$.
\end{lemma}

Kiessling~\cite{Kiessling-12} a donn\'e une preuve tr\`es \'el\'egante de ce lemme dans le cas de l'\'energie cin\'etique non relativitse. Il note que dans ce cas 
$$\frac{1}{n} T \left(\sqrt{\mubf ^{(n)}} \right) = \frac{1}{n} \sum_{j=1} ^n \int_{\R^{dn}} \left|\nabla_j \sqrt{\mubf ^{(n)}}\right| ^2 = \frac{1}{4n} \sum_{j=1} ^n \int_{\R^{dn}} \left|\nabla_j \log \mubf ^{(n)}\right| ^2 \mubf ^{(n)}$$
et que la derni\`ere expression est identique \`a l'information de Fisher de la mesure de probabilit\'e $\mubf ^{(n)}$. La quantit\'e que nous \'etudions peut donc s'interpr\'eter comme une ``information de Fisher moyenne'' de la mesure $\mubf \in \PP_s(\R ^{d\N})$, en analogie avec l'entropie moyenne introduite en~\eqref{eq:aHS somme}.

Cette quantit\'e a une connexion int\'eressante avec l'entropie classique d'une mesure de probabilit\'e. En faisant \'evoluer $\mubf ^{(n)}$ suivant le flot de la chaleur, on peut montrer qu'\`a chaque instant le long du flot, l'information de Fisher est la d\'eriv\'ee de l'entropie. Puisque le flot de la chaleur est lin\'eaire et que l'entropie moyenne est affine (cf le calcul simple pr\'esent\'e lors de la preuve du Th\'eor\`eme~\ref{thm:aHS}, issu de~\cite{RobRue-67}), Kiessling d\'eduit que l'information de Fisher moyenne est affine. Un autre point de vue sur cette question est pr\'esent\'e dans~\cite[Section~5]{HauMis-14}.

On suivra ici une approche plus p\'edestre qui a l'avantage de s'adapter aux \'energies cin\'etiques g\'en\'erales d\'ecrites \`a la D\'efinition~\ref{def:hyp class quant}, entre autres l'\'energie cin\'etique pseudo-relativiste. 

\begin{proof}[Preuve]
Le Th\'eor\`eme~\ref{thm:HS} implique que $\PP_s(\R ^{d\N})$ est l'enveloppe convexe des mesures de probabilit\'e sym\'etriques caract\'eris\'ees par $\mubf ^{(n)} = \rho ^{\otimes n}$, $\rho \in \PP (\R ^d)$. Pour prouver le lemme, il suffit donc de prendre 
$$ \mubf_1 = \rho_1 ^{\otimes n}, \quad \mubf_2 = \rho_2 ^{\otimes n}, \quad \mubf = \half\mubf_1 + \half\mubf_2 $$ 
avec $\rho_1,\rho_2 \in \PP (\R ^d)$, $\rho_1 \neq \rho_2$ et de d\'emontrer que 
\begin{equation}\label{eq:app A claim 2}
\left| T\left( \sqrt{ \mubf} \right) - \half T \left( \sqrt{ \mubf_1} \right)- \half T\left( \sqrt{\mubf_2} \right)\right| \leq o(n)\to 0
\end{equation}
quand $n\to \infty$. Par sym\'etrie de $\mubf ^{(n)}$ et au vu de la D\'efinition~\ref{def:hyp class quant}, il s'agit de calculer 
\begin{equation}\label{eq:app A claim 2 calcul} 
T\left( \sqrt{ \mubf } \right) = n \int_{\R ^{d(n-1)}} d\hat{X} \iint_{\R ^d \times \R ^d} T(x,y) \left| \sqrt{\mubf (X)} - \sqrt{\mubf (Y)}\right| ^2, 
\end{equation}
o\`u on a not\'e
$$X = (x_1,\ldots,x_n), \quad Y = (y_1,x_2\ldots,x_n), \quad \hat{X} = (x_2,\ldots,x_n).$$

Nous commen\c{c}ons par pr\'etendre que pour tout $X,Y$
\begin{align}\label{eq:app A claim 3}
| I (X,Y) | &=\left| \left| \sqrt{\mubf (X)} - \sqrt{\mubf (Y)}\right| ^2 - \half \left| \sqrt{\mubf_1 (X)} - \sqrt{\mubf_1 (Y)}\right| ^2 -\half \left| \sqrt{\mubf_2 (X)} - \sqrt{\mubf_2 (Y)}\right| ^2\right| \nonumber
\\& \leq C \left(\prod_{j=2} ^n  \rho_1 (x_j) \rho_2 (x_j)  \right) ^{1/2} \left( \left| \sqrt{\rho_1 (x_1)} - \sqrt{\rho_1 (y_1)} \right| ^2 + \left|\sqrt{\rho_2 (x_1)} - \sqrt{\rho_2 (y_1)}\right| ^2\right).
\end{align}
Nous allons prouver~\eqref{eq:app A claim 3} dans le cas 
\begin{equation}\label{eq:app A claim 3 choice}
\rho_2(y_1) \leq \rho_2 (x_1) \mbox{ et } \rho_1(y_1) \leq \rho_1 (x_1) 
\end{equation}
en laissant au lecteur les adaptation n\'ecessaires pour les autres cas. On all\`ege la notation en posant $u_i = \sqrt{\rho_i},\quad U_i = \sqrt{\rho_i ^{\otimes n}}$. 

Apr\`es avoir d\'evelopp\'e les carr\'es on obtient 
\begin{multline*}
 2 I (X,Y) = U_1 (X) U_1 (Y) + U_2 (X) U_2 (Y) \\
 - \sqrt{U_1 ^2 (X) U_1 ^2 (Y) + U_2 ^2 (X) U_2 ^2 (Y) + U_1 ^2 (X) U_2 ^2 (Y) + U_2 ^2 (X) U_1 ^2 (Y)}. 
\end{multline*}
Ensuite
\begin{multline*} 
U_1 ^2 (X) U_1 ^2 (Y) + U_2 ^2 (X) U_2 ^2 (Y) + U_1 ^2 (X) U_2 ^2 (Y) + U_2 ^2 (X) U_1 ^2 (Y) \\
=  \left(U_1 (X) U_1 (Y) + U_2 (X) U_2 (Y) \right) ^2 +  U_1 ^2 (X) U_2 ^2 (Y) + U_2 ^2 (X) U_1 ^2 (Y) - 2  U_1 (X) U_2  (Y) U_1  (Y) U_2 (X)
\end{multline*}
et donc,  
\begin{align*}
 2 |I(X,Y)| &= \big| \left(U_1 (X) U_1 (Y) + U_2 (X) U_2 (Y)\right)  
 \\ &\times \left. \left( 1 - \sqrt{1+ \frac{U_1 ^2 (X) U_2 ^2 (Y) + U_2 ^2 (X) U_1 ^2 (Y) - 2  U_1 (X) U_2 ^2 (Y) U_1  (Y) U_2 (X)}{\left(U_1 (X) U_1 (Y) + U_2 (X) U_2 (Y) \right) ^2}} \right)\right|
 \\&\leq \frac{U_1 ^2 (X) U_2 ^2 (Y) + U_2 ^2 (X) U_1 ^2 (Y) - 2  U_1 (X) U_2 ^2 (Y) U_1  (Y) U_2 (X)}{U_1 (X) U_1 (Y) + U_2 (X) U_2 (Y) }
 \\&= \left(\prod_{j=2} ^n u_1 (x_j) u_2(x_j) \right)\frac{\left| u_1 (x_1) u_2(y_1) - u_1(y_1) u_2(x_1) \right| ^2}{u_1 (x_1) u_1 (y_1) + u_2 (x_1) u_2 (y_1)}
 \\&= \left(\prod_{j=2} ^n u_1 (x_j) u_2(x_j) \right) \frac{\left|  u_2(y_1)\left(u_1 (x_1) - u_1(y_1)\right) + u_1(y_1) \left( u_2(y_1) -u_2(x_1)\right) \right| ^2}{u_1 (x_1) u_1 (y_1) + u_2 (x_1) u_2 (y_1)}
\\&\leq 2 \left(\prod_{j=2} ^n u_1 (x_j) u_2(x_j) \right) \left( \frac{u_2 (y_1) ^2}{u_1 (x_1) u_1 (y_1) + u_2 (x_1) u_2 (y_1)} \left(u_1 (x_1) - u_1(y_1)\right) ^2 \right. 
\\& \left. + \frac{u_1(y_1) ^2}{u_1 (x_1) u_1 (y_1) + u_2 (x_1) u_2 (y_1)} \left( u_2(y_1) -u_2(x_1)\right) ^2 \right).
\end{align*}
L'estimation~\eqref{eq:app A claim 3} suit imm\'ediatement dans le cas~\eqref{eq:app A claim 3 choice} et par des consid\'erations similaires dans les autres cas.

En ins\'erant~\eqref{eq:app A claim 3} dans~\eqref{eq:app A claim 2 calcul} et en rappellant~\eqref{eq:app A kin class} on obtient 
\begin{equation}\label{eq:app A calcul}
\left| T\left( \sqrt{ \mubf} \right) - \half T \left( \sqrt{ \mubf_1} \right)- \half T\left( \sqrt{\mubf_2} \right)\right| \leq C n \left( \int_{\R ^d} \sqrt{\rho_1} \sqrt{\rho_2} \right) ^{n-1} \big( T \left(\sqrt{\rho_1}\right) + T \left(\sqrt{\rho_2}\right)\big) 
\end{equation}
par Fubini. On peut supposer que $T\left(\sqrt{\rho_1}\right)$ et $T\left(\sqrt{\rho_2}\right)$ sont finis, sinon toutes les quantit\'es que nous \'evaluons sont \'egales \`a $+\infty$ et il n'y rien \`a prouver. Il reste \`a noter que 
$$ \delta := \int_{\R ^d} \sqrt{\rho_1} \sqrt{\rho_2} < \frac{1}{2} \left( \int_{\R ^d} \rho_1 + \int_{\R ^d} \rho_2 \right) < 1$$
puisque $\rho_1 \neq \rho_2$ par hypoth\`ese. On conclut que 
$$\left| T\left( \sqrt{ \mubf} \right) - \half T \left( \sqrt{ \mubf_1} \right) - \half T\left( \sqrt{\mubf_2} \right)\right| \leq C n \delta ^{n-1}$$
et $\delta ^{n-1}\to 0$ quand $n\to +\infty,$ comme d\'esir\'e.

Pour traiter le cas de l'\'energie cin\'etique usuelle on applique d'abord le raisonnement ci-dessus \`a l'\'energie $T_s$ d\'efinie par le noyau positif 
$$ T_s (x,y) = C_d |x-y| ^{-(d+2s)}.$$
avec $0<s<1$ fixe, pour obtenir l'analogue de~\eqref{eq:app A calcul} avec $T=T_s$. On multiplie alors cette in\'egalit\'e par $(1-s)$ puis on en prend la limite $s\to 1$. Utilisant~\eqref{eq:app A cin bis} ceci donne~\eqref{eq:app A calcul} avec $T$ l'\'energie cin\'etique non relativiste. On peut ensuite passer \`a la limite $n\to \infty$ pour obtenir le r\'esultat souhait\'e. 

\end{proof}

\subsection{Conclusion}\label{sec:concl class quant}\mbox{}\\\vspace{-0.4cm}


La borne sup\'erieure est comme d'habitude triviale en prenant un \'etat test factoris\'e. En combinant les Lemmes~\ref{lem:app A lim} et~\ref{lem:app A kin lim} ainsi que la repr\'esentation de $\mubf$ donn\'ee par le Th\'eor\`eme~\ref{thm:HS} on d\'eduit  
\begin{align*}
\liminf_{N\to \infty} \frac{E(N)}{N} &= \E \left[  \int_{\PP (\R ^d)} \rho ^{\otimes \infty} \d\mu (\rho) \right] \\
&=  \int_{\PP (\R ^d)}  \E \left[  \rho ^{\otimes \infty}  \right] \d\mu (\rho) \\
&= \int_{\PP (\R ^d)}  \EH \left[  \sqrt{\rho}  \right] \d\mu (\rho) \geq \int_{\PP (\R ^d)}  \eH  \; \d\mu (\rho) = \eH,
\end{align*}
ce qui fournit la convergence de l'\'energie. La convergence des densit\'es r\'eduites suit en notant qu'il doit y avoir \'egalit\'e dans toutes les in\'egalit\'es pr\'ec\'edentes. 

\newpage

\section{\textbf{Bosons en dimension finie \`a grande temp\'erature}}\label{sec:large T}

Nous avons jusqu'\`a pr\'esent consid\'er\'e des syst\`emes quantiques de champ moyen uniquement \`a temp\'erature nulle, et obtenu quand $N\to \infty$ des mesures de de Finetti concentr\'ees sur les minimiseurs de la fonctionnelle d'\'energie limite. Il est possible, en prenant une limite de grande temp\'erature en m\^eme temps que la limite de champ moyen, d'obtenir une mesure de Gibbs \`a la limite. Dans cet appendice nous expliquerons ceci pour le cas de bosons dans un espace de dimension finie, en suivant~\cite{Gottlieb-05,LewRou-unpu13}. 

En dimension inifinie, des probl\`emes importants se posent, notamment pour la d\'efinition du probl\`eme limite. Les mesures de Gibbs non lin\'eaires qu'on obtient jouent un r\^ole important en th\'eorie quantique des champs~\cite{Derezinski-13,Simon-74,Summers-12,Gallavotti-85,GliJaf-87} et dans la construction de solutions peu r\'eguli\`eres \`a l'\'equation de Schr\"odinger non lin\'eaire, voir par exemple~\cite{LebRosSpe-88,Bourgain-94,Bourgain-96,Tzvetkov-08,BurTzv-08,BurTzv-08b,BurThoTzv-10,ThoTzv-10,Suzzoni-11}. On renvoie \`a l'article~\cite{LewNamRou-14b} pour des r\'esultats sur la limite ``champ moyen/grande temp\'erature en dimension infinie'' et une discussion plus pouss\'ee de ces sujets.

\subsection{Cadre et r\'esultat}\label{sec:app B cadre}\mbox{}\\\vspace{-0.4cm}

Dans cet appendice l'espace \`a une particule sera un espace de Hilbert $\gH$ de dimension finie
$$\dim \gH = d.$$
On consid\`ere le Hamiltonien de type champ moyen
\begin{equation}\label{eq:app B hamil}
H_N = \sum_{j=1} ^N h_j + \frac{1}{N-1} \sum_{1\leq i<j \leq N} w_{ij} 
\end{equation}
o\`u $h$ est un op\'erateur auto-adjoint sur $\gH$ et $w$ un op\'erateur auto-adjoint sur $\gH \otimes \gH$ sym\'etrique au sens o\`u
$$ w (u\otimes v) = w (v\otimes u), \quad \forall u,v\in \gH.$$
La fonctionnelle d'\'energie est comme d'habitude d\'efinie par 
$$ \E_N [\Psi_N] = \bral \Psi_N, H_N \Psi_N \ketr$$ 
pour $\Psi_N \in \bigotimes_s ^N \gH $ et s'\'etend aux \'etats mixtes $\Gamma_N$ de $\gH ^N = \bigotimes_s ^N \gH$ par la formule 
$$ \E_N [\Gamma_N] = \tr_{\gH ^N} \left[ H_N \Gamma_N \right].$$ 
L'\'etat d'\'equilibre du syt\`eme \`a temp\'erature $T$ est obtenu en minimisant la fonctionnelle d'\'energie libre 
\begin{equation}\label{eq:app B free ener func}
\F_N [\Gamma_N] := \E_N [\Gamma_N] + T \tr\left[ \Gamma_N \log \Gamma_N \right] 
\end{equation}
parmi les \'etats mixtes, ce qui donne pour minimiseur l'\'etat de Gibbs 
\begin{equation}\label{eq:app B quant Gibbs}
\Gamma_N = \displaystyle \frac{\exp \left( - T ^{-1} H_N \right)}{\tr\left[ \exp \left( - T ^{-1} H_N \right)\right]}. 
\end{equation}
L'\'energie libre minimum est obtenue \`a partir de la fonction de partition (facteur de normalisation dans~\eqref{eq:app B quant Gibbs}) comme suit:
\begin{equation}\label{eq:app B free ener} 
F_N = \inf \left\{ \F_N [\Gamma_N], \Gamma_N \in \cS (\gH ^N) \right\} = -T \log \tr\left[ \exp \left( - \frac{1}{T} H_N \right)\right]. 
\end{equation}
On va s'int\'eresser au comportement de ces objets dans la limite 
\begin{equation}\label{eq:app B lim}
N\to \infty, \quad T = tN, \quad t \mbox{ fixe } 
\end{equation}
qui se trouve \^etre le bon r\'egime pour obtenir un probl\`eme limite int\'eressant. On va en fait obtenir \`a la limite une fonctionnelle d'\'energie libre classique, que nous d\'efinissons maintenant. 

\medskip 

Puisque $\gH$ est de dimension finie, on peut d\'efinir $du$ la mesure de Lebesgue normalis\'ee sur sa sph\`ere unit\'e $S\gH$. Les objets limites seront des mesures de de Finetti, donc des mesures de probabilit\'e $\mu$ sur $S\gH$, et plus pr\'ecis\'ement des fonctions $L ^1 (S\gH,du)$. On introduit pour ces objets une fonctionnelle d'\'energie libre classique 
\begin{equation}\label{eq:app B free ener func class}
\Fcl [\mu] = \int_{S\gH} \EH [u] \mu(u) du + t  \int_{S\gH} \mu(u) \log \left(\mu (u)\right) du
\end{equation}
dont on notera $\Fcle$ l'infimum parmi les fonctions $L ^1$ positives et normalis\'ees. Il est atteint par la mesure de Gibbs classique
\begin{equation}\label{eq:app B Gibbs class}
\mucl = \frac{\exp\left( -t^{-1} \EH [u] \right)}{\int_{S\gH} \exp\left( -t^{-1} \EH [u] \right) du} 
\end{equation}
et on a 
$$ \Fcle = -t \log \left( \int_{S\gH} \exp\left( -\frac{1}{t} \EH [u] \right) du \right)$$
Ici $\EH[u]$ est la fonctionnelle d'\'energie de Hartree 
\begin{equation}\label{eq:app B Hartree}
 \EH [u] = \frac{1}{N} \bral u ^{\otimes N}, H_N u ^{\otimes N} \ketr_{\gH ^N} = \bral u,h u \ketr_{\gH} + \half \bral u\otimes u, w u \otimes u \ketr_{\gH ^2}.  
\end{equation}
Le th\'eor\`eme que nous allons d\'emontrer, d\^u \`a Gottlieb~\cite{Gottlieb-05} (voir aussi~\cite{GotSch-09,JulGotMarPol-13}) est de nature semi-classique puisqu'il fait le lien entre les mesures de Gibbs quantique~\eqref{eq:app B quant Gibbs} et classique~\eqref{eq:app B Gibbs class}:

\begin{theorem}[\textbf{Limite champ moyen/grande temp\'erature en dimension finie}]\label{thm:app B}
Dans la limite~\eqref{eq:app B lim}, on a
\begin{equation}\label{eq:app B lim ener}
F_N = - T \log \dim\left( \gH_s ^N \right)  + N \Fcle + O(d).
\end{equation}
De plus, en notant $\gamma_N ^{(n)}$ la $n$-i\`eme matrice de densit\'e r\'eduite de l'\'etat de Gibbs~\eqref{eq:app B quant Gibbs}, 
\begin{equation}\label{eq:app B lim mat}
\gamma_N ^{(n)} \to \int_{S\gH} | u ^{\otimes n} \rangle \langle u ^{\otimes n} |\mucl(u) du 
\end{equation}
fortement en norme de trace sur $\gH ^n$.
\end{theorem}

\begin{remark}[Limite de champ moyen/grande temp\'erature]\label{rem:app B thm}\mbox{}\\
Quelques commentaires:
\begin{enumerate}
\item Il faut comprendre ce th\'eor\`eme comme disant que essentiellemnt, dans la limite qui nous int\'eresse
$$ \Gamma_N \approx \int_{S\gH} |u ^{\otimes N} \rangle \langle u ^{\otimes N} | \, \mucl(u) du.$$
L'\'etat de Gibbs est donc essentiellement une superposition d'\'etats de Hartree. Les notions de matrices de densit\'es r\'eduites et de mesures de de Finetti donnent la bonne mani\`ere de rendre ceci rigoureux. On verra que la mesure de de Finetti (symbole inf\'erieur) associ\'e \`a $\Gamma_N$ par les m\'ethodes du Chapitre~\ref{sec:deF finite dim} converge vers $\mucl (u) du$.
\item Notons que le premier terme dans l'expansion d'\'energie~\eqref{eq:app B lim ener} diverge tr\`es rapidement, voir~\eqref{eq:dim boson}. L'\'energie libre classique appara\^it seulement comme une correction. Vue la d\'ependence en la dimension $d$ de ce premier terme, il est clair que l'approche de cet appendice ne peut s'adapter telle quelle dans un espace de dimension infinie. 
\item Notre m\'ethode de preuve diff\`erera de celle de~\cite{Gottlieb-05}. On exploitera plus \`a fond le caract\`ere semi-classique du probl\`eme en utilisant les in\'egalit\'es de Berezin-Lieb introduites dans~\cite{Berezin-72,Lieb-73b,Simon-80}. La m\'ethode pr\'esent\'ee~\cite{LewRou-unpu13} doit beaucoup \`a l'article fondateur~\cite{Lieb-73b} et rappelle certains aspects de~\cite{LieSeiYng-05}.  
\item Il sera crucial pour cette preuve que le symbole inf\'erieur de $\Gamma_N$ constitue une mesure de de Finetti approch\'ee pour $\Gamma_N$. Cela nous permettra d'appliquer la premi\`ere in\'egalit\'e de Berezin-Lieb pour obtenir une borne inf\'erieure sur l'entropie. Un nouvel int\'er\^et des constructions du Chapitre~\ref{sec:deF finite dim} se r\'ev\`ele donc dans cet appendice o\`u on utilise non seulement l'estimation fournie par le Th\'eor\`eme~\ref{thm:DeFinetti quant} mais \'egalement la forme particuli\`ere de la mesure construite.  
\end{enumerate}\hfill\qed

\end{remark}

\subsection{In\'egalit\'es de Berezin-Lieb}\label{sec:app B Berezin}\mbox{}\\\vspace{-0.4cm}

On rappelle la d\'ecomposition de l'identit\'e~\eqref{eq:Schur} sur $\gH ^N$ fournie par le lemme de Schur. On a donc pour chaque \'etat $\Gamma_N \in \cS (\gH_s ^N)$ un symbole inf\'erieur d\'efini comme 
$$ \mu_N = \dim \left( \gH ^N_s\right)\tr \left[\Gamma_N |u ^ {\otimes N}\rangle \langle u ^ {\otimes N} | \right].$$
La premi\`ere in\'egalit\'e de Berezin-Lieb est l'\'enonc\'e suivant

\begin{lemma}[\textbf{Premi\`ere in\'egalit\'e de Berezin-Lieb}]\label{lem:Ber Lie 1}\mbox{}\\
Soit $\Gamma_N \in \cS (\gH ^N_s)$ de symbole inf\'erieur $\mu_N$ et $f:\R^+ \to \R$ une fonction convexe. On a 
\begin{equation}\label{eq:app B Ber Lie 1}
\tr \left[ f(\Gamma_N) \right] \geq \dim\left( \gH_s ^N \right)\int_{S\gH} f\left(\frac{\mu_N}{\dim\left( \gH_s ^N \right)} \right) du. 
\end{equation}
\end{lemma}

La seconde in\'egalit\'e de Berezin-Lieb s'applique \`a des \'etats ayant un symbole sup\'erieur positif (voir Section~\ref{sec:CKMR heur}). On peut en fait montrer que tout \'etat a un symbole sup\'erieur, mais il n'est en g\'en\'eral pas donn\'e par une mesure positive.

\begin{lemma}[\textbf{Seconde in\'egalit\'e de Berezin-Lieb}]\label{lem:Ber Lie 2}\mbox{}\\
Soit $\Gamma_N \in \cS (\gH ^N_s)$ de symbole sup\'erieur $\mu_N \geq 0$,
\begin{equation}\label{eq:app B symb sup}
\Gamma_N = \int_{u\in S \gH}  |u ^{\otimes N} \rangle \langle u ^{\otimes N}| \mu_N (u) du 
\end{equation}
et $f:\R ^+ \to \R$ une fonction convexe. On a 
\begin{equation}\label{eq:app B Ber Lie 2}
\tr \left[ f(\Gamma_N) \right] \leq \dim\left( \gH_s ^N \right) \int_{S\gH} f \left(\frac{\mu_N}{\dim\left( \gH_s ^N \right)} \right) du. 
\end{equation}
\end{lemma}

\begin{proof}[Preuve des Lemmes~\ref{lem:Ber Lie 1} et~\ref{lem:Ber Lie 2}]
On suit~\cite{Simon-80}. Comme $\Gamma_N$ est un \'etat, on le d\'ecompose sous la forme 
$$ \Gamma_N = \sum_{k=1} ^{\infty} \lambda_N ^k |V_N ^k \rangle \langle V_N ^k|$$
avec $V_N ^k \in \gH ^N_s$ normalis\'e et $\sum_k \lambda_N ^k = 1$. On note 
$$ \mu_N ^k (u)=  \left| \bral V_N ^k , u ^{\otimes N} \ketr \right| ^2$$
et par~\eqref{eq:Schur} on a 
\begin{equation}\label{eq:app B sum 1}
 \dim\left( \gH_s ^N \right) \int_{S\gH} \mu_N ^k (u) du  = \bral V_N ^k , V_N ^k \ketr =  1.
\end{equation}
D'autre part, comme $(V_N ^k)_{k}$ est une base de $\gH ^N _s$, pour tout $u\in S\gH$
\begin{equation}\label{eq:app B sum 2}
 \sum_k \mu_N ^k (u) = \sum_k \left| \bral V_N ^k , u ^{\otimes N} \ketr \right| ^2 =1.
\end{equation}

\medskip

\noindent \emph{Premi\`ere in\'egalit\'e.} Ici $\mu_N$ est le symbole inf\'erieur de $\Gamma_N$ et on a 
$$ \mu_N (u) = \dim\left( \gH_s ^N \right)\sum_k \lambda_N ^k \mu_N ^k (u)$$
et donc 
$$ \dim\left( \gH_s ^N \right) \int_{S\gH} f\left(\frac{\mu_N}{\dim\left( \gH_s ^N \right)} \right) du \leq \dim\left( \gH_s ^N \right) \int_{S\gH}  \sum_k f \left(\lambda_N ^k \right) \mu_N ^k (u) du$$
par l'in\'egalit\'e de Jensen et~\eqref{eq:app B sum 2}, puis 
$$ \dim\left( \gH_s ^N \right) \int_{S\gH}  \sum_k f \left(\lambda_N ^k \right) \mu_N ^k (u) du = \sum_k f \left(\lambda_N ^k \right) = \tr \left[f (\Gamma_N)\right]$$
par~\eqref{eq:app B sum 1}.

\medskip

\noindent \emph{Seconde in\'egalit\'e.} Ici $\Gamma_N$ et $\mu_N$ sont reli\'es par~\eqref{eq:app B symb sup}. On \'ecrit
\begin{align*}
 \tr\left[ f(\Gamma_N) \right] = \sum_k f\left(\lambda_N ^k\right) &= \sum_k f \left( \langle V_N ^k, \Gamma_N V_N ^k \rangle \right) \\
 &= \sum_k f \left( \int_{S\gH} \mu_N (u) \mu_N ^k (u) du\right) 
 \\&\leq \sum_k \dim\left( \gH_s ^N \right) \int_{S\gH} f\left(\frac{\mu_N(u)}{\dim\left( \gH_s ^N \right)} \right) \mu_N ^k (u) du\\
 &= \dim\left( \gH_s ^N \right) \int_{S\gH} f\left(\frac{\mu_N (u)}{\dim\left( \gH_s ^N \right)} \right)  du
\end{align*}
en utilisant l'in\'egalit\'e de Jensen et~\eqref{eq:app B sum 1} pour montrer l'in\'egalit\'e puis~\eqref{eq:app B sum 2} pour conclure.
\end{proof}

Nous avons pr\'esent\'e ici une version sp\'ecifique de ces fameuses in\'egalit\'es. Il est clair que la preuve s'applique plus g\'en\'eralement \`a tout op\'erateur auto-adjoint sur un espace de Hilbert disposant d'une d\'ecomposition en \'etats coh\'erents de forme~\eqref{eq:Schur}. 
Dans la section suivante ces in\'egalit\'es serviront \`a traiter le terme d'entropie en prenant $f(x) = x \log x$. Ceci compl\'etera le traitement de l'\'energie utilisant le Th\'eor\`eme~\ref{thm:DeFinetti quant} et fera le lien avec les consid\'erations pr\'esent\'ees au Chapitre~\ref{sec:deF finite dim}.

\subsection{Preuve du Th\'eor\`eme~\ref{thm:app B}}\label{sec:app B preuve}\mbox{}\\\vspace{-0.4cm}

\noindent\emph{Borne sup\'erieure.} On prend comme \'etat test
$$ \Gamma_N ^{\rm test} := \int_{S\gH} |u ^{\otimes N} \rangle \langle u ^{\otimes N} |\: \mucl(u) du. $$
L'\'energie \'etant lin\'eaire en la matrice densit\'e 
$$ \E_N \left[ \Gamma_N ^{\rm test} \right] = \int_{S\gH} \E_N \left [|u ^{\otimes N} \rangle \langle u ^{\otimes N} |\right]\mucl(u) du = N \int_{S\gH} \EH [u] \mucl(u) du.$$
Pour le terme d'entropie on utilise la seconde in\'egalit\'e de Berezin-Lieb, Lemme~\ref{lem:Ber Lie 2}, avec $f(x) = x \log x$, ce qui donne 
\begin{align*}
\tr\left[ \Gamma_N ^{\rm test} \log \Gamma_N  ^{\rm test}\right] &\leq \dim\left( \gH_s ^N \right) \int_{S\gH} \frac{\mucl (u)}{\dim\left( \gH_s ^N \right)} \log\left( \frac{\mucl (u)}{\dim\left( \gH_s ^N \right)}\right) du 
\\&= - \log \dim\left( \gH_s ^N \right) + \int_{S\gH} \mucl (u) \log\left( \mucl (u) \right) du.
 \end{align*}

En sommant ces estimations on obtient 
$$ F_N \leq \F_N \left[ \Gamma_N ^{\rm test} \right] \leq -T \log \dim\left( \gH_s ^N \right) + N \Fcle$$
puisque $\mucl$ minimise $\Fcl$.

\medskip

\noindent\emph{Borne inf\'erieure.} Pour l'\'energie on utilise les matrices de densit\'e r\'eduites comme d'habitude pour \'ecrire 
$$ \E_N [\Gamma_N] = N \tr_{\gH} \left[ h \gamma_N ^{(1)}\right] + \frac{N}{2}\tr_{\gH ^2} \left[ w \gamma_N ^{(2)}\right].$$
En notant 
$$ \mu_N (u) = \dim \left(\gH ^N_s \right) \bral u ^{\otimes N}, \Gamma_N u ^{\otimes N}\ketr$$
le symbole inf\'erieur de $\Gamma_N$, on rappelle qu'il a \'et\'e d\'emontr\'e au Chapitre~\ref{sec:deF finite dim} que 
\begin{align*}
\tr_{\gH} \left| \gamma_N ^{(1)} - \int_{S\gH} |u\rangle \langle u| \mu_N (u) du \right| &\leq C_1 \frac{d}{N}\\ 
\tr_{\gH^2} \left| \gamma_N ^{(2)} - \int_{S\gH} |u ^{\otimes 2}\rangle \langle u ^{\otimes 2}| \mu_N (u) du \right| &\leq C_2 \frac{d}{N}.
\end{align*}
Puisque nous travaillons en dimension finie, $h$ et $w$ sont des op\'erateurs born\'es et il s'ensuit que 
\begin{align*}
 \E_N [\Gamma_N] &\geq N \int_{S\gH} \tr_{\gH} \left[ h |u\rangle \langle u|\right] \mu_N (u) du + \frac{N}{2} \int_{S\gH} \tr_{\gH ^2} \left[ w |u ^{\otimes 2}\rangle \langle u ^{\otimes 2}|\right] \mu_N (u) du - C d\\
 &= N \int_{S\gH} \EH[u] \mu_N (u) du - Cd.
\end{align*}
Pour estimer l'entropie on utilise la premi\`ere in\'egalit\'e de Berezin-Lieb, Lemme~\ref{lem:Ber Lie 1}, avec $f(x) = x\log x$, ce qui donne
\begin{align*}
\tr\left[ \Gamma_N \log \Gamma_N \right] &\geq \dim\left( \gH_s ^N \right) \int_{S\gH} \frac{\mu_N (u)}{\dim\left( \gH_s ^N \right)} \log\left( \frac{\mu_N (u)}{\dim\left( \gH_s ^N \right)}\right) du \nonumber
\\&= - \log \dim\left( \gH_s ^N \right) + \int_{S\gH} \mu_N (u) \log\left( \mu_N (u) \right) du.
\end{align*}
Il ne reste qu'\`a grouper ces estimations pour d\'eduire   
\begin{align*}
 F_N &= \F_N [\Gamma_N] \geq - T \log \dim\left( \gH_s ^N \right) + N \Fcl[\mu_N] -Cd \\
 &\geq - T \log \dim\left( \gH_s ^N \right) + N \Fcle -Cd 
\end{align*}
puisque 
$$\int_{S\gH} \mu_N (u) du =1$$ 
par d\'efinition.
 
\medskip

\noindent\emph{Convergence des matrices de densit\'e r\'eduites.} Le symbole inf\'erieur $\mu_N (u) du$ est une mesure de probailit\'e sur l'espace compact $S\gH$. On en extrait une sous-suite convergente
$$ \mu_N (u) du \to \mu (du) \in \PP(S\gH)$$
et il s'ensuit des r\'esultats du Chapitre~\ref{sec:deF finite dim} que, pout tout $n\geq 0$ et le long d'une sous-suite, 
\begin{equation}\label{eq:app B preuve 1}
 \gamma_N ^{(n)} \to \int_{S\gH} | u ^{\otimes n} \rangle \langle u ^{\otimes n}| d\mu(u). 
\end{equation}
Les estimations d'\'energie pr\'ec\'edentes une fois combin\'ees donnent
\begin{equation}\label{eq:app B preuve}
 \Fcle \geq \Fcl[\mu_N] - C\frac{d}{N}. 
\end{equation}
Pour passer \`a la liminf $N\to \infty$ dans cette estimation, le terme d'\'energie est trait\'ee comme dans les chapitres pr\'ec\'edents. Pour le terme d'entropie on remarque que comme $du$ est normalis\'ee
$$ \int_{S\gH} \mu_N \log \mu_N du \geq 0$$
puisque cette quantit\'e peut s'interpr\'eter comme l'entropie relative de $\mu_N$ par rapport \`a la fonction constante $1$. Par le lemme de Fatou on d\'eduit donc de~\eqref{eq:app B preuve} que 
$$ \Fcle \geq \Fcl [\mu]$$
et donc que $d\mu (u)= \mucl(u) du$ par unicit\'e du minimiseur de $\Fcl$. L'unicit\'e de la limite garantit \'egalement que toute la suite converge, et il ne reste qu'\`a revenir \`a~\eqref{eq:app B preuve 1} pour conclure. \hfill\qed

%

\begin{thebibliography}{100}

\bibitem{Aftalion-06}
{\sc A.~Aftalion}, {\em Vortices in {B}ose--{E}instein Condensates}, vol.~67 of
  Progress in nonlinear differential equations and their applications,
  Springer, 2006.

\bibitem{Ammari-04}
{\sc Z.~Ammari}, {\em Scattering theory for a class of fermionic
  {P}auli-{F}ierz models}, J. Funct. Anal., 208 (2004), pp.~302--359.

\bibitem{Ammari-hdr}
\leavevmode\vrule height 2pt depth -1.6pt width 23pt, {\em Syst\`emes
  hamiltoniens en th\'eorie quantique des champs : dynamique asymptotique et
  limite classique.}
\newblock Habilitation {\`a} Diriger des Recherches, University of Rennes I,
  February 2013.

\bibitem{AmmNie-08}
{\sc Z.~Ammari and F.~Nier}, {\em Mean field limit for bosons and infinite
  dimensional phase-space analysis}, Annales Henri Poincar\'e, 9 (2008),
  pp.~1503--1574.
\newblock 10.1007/s00023-008-0393-5.

\bibitem{AmmNie-09}
\leavevmode\vrule height 2pt depth -1.6pt width 23pt, {\em Mean field limit for bosons and propagation
  of {W}igner measures}, J. Math. Phys., 50 (2009).

\bibitem{AmmNie-11}
\leavevmode\vrule height 2pt depth -1.6pt width 23pt, {\em {Mean field propagation of Wigner measures
  and BBGKY hierarchies for general bosonic states}}, J. Math. Pures Appl., 95
  (2011), pp.~585--626.

\bibitem{AndGuiZei-10}
{\sc G.~Anderson, A.~Guionnet, and O.~Zeitouni}, {\em An introduction to random
  matrices}, Cambridge University Press, 2010.

\bibitem{AshFroGraSchTro-02}
{\sc W.~Aschbacher, J.~Fr{\"o}hlich, G.~Graf, K.~Schnee, and M.~Troyer}, {\em
  Symmetry breaking regime in the nonlinear hartree equation}, J. Math. Phys.,
  43 (2002), pp.~3879--3891.

\bibitem{BarGolMau-00}
{\sc C.~Bardos, F.~Golse, and N.~J. Mauser}, {\em Weak coupling limit of the
  {$N$}-particle {S}chr\"odinger equation}, Methods Appl. Anal., 7 (2000),
  pp.~275--293.

\bibitem{BenGui-97}
{\sc G.~Ben~Arous and A.~Guionnet}, {\em Large deviations for {W}igner's law
  and {V}oiculescu's noncommutative entropy}, Probab. Theory Related Fields,
  108 (1997), pp.~517--542.

\bibitem{BenZei-98}
{\sc G.~Ben~Arous and O.~Zeitouni}, {\em Large deviations from the circular
  law}, ESAIM: Probability and Statistics, 2 (1998), pp.~123--134.

\bibitem{BenLie-83}
{\sc R.~{Benguria} and E.~H. {Lieb}}, {\em {Proof of the Stability of Highly
  Negative Ions in the Absence of the Pauli Principle}}, Physical Review
  Letters, 50 (1983), pp.~1771--1774.

\bibitem{Berezin-72}
{\sc F.~A. Berezin}, {\em Convex functions of operators}, Mat. Sb. (N.S.),
  88(130) (1972), pp.~268--276.

\bibitem{BetUel-10}
{\sc V.~Betz and D.~Ueltschi}, {\em Critical temperature of dilute bose gases},
  Phys. Rev. A, 81 (2010), p.~023611.

\bibitem{BloDalZwe-08}
{\sc I.~Bloch, J.~Dalibard, and W.~Zwerger}, {\em Many-body physics with
  ultracold gases}, Rev. Mod. Phys., 80 (2008), p.~885.

\bibitem{Bose-24}
{\sc {Bose}}, {\em {Plancks Gesetz und Lichtquantenhypothese}}, Zeitschrift fur
  Physik, 26 (1924), pp.~178--181.

\bibitem{BouErdYau-12}
{\sc P.~Bourgade, L.~Erd\"os, and H.-T. Yau}, {\em Bulk universality of general
  $\beta$-ensembles with non-convex potential}, J. Math. Phys., 53 (2012),
  p.~095221.

\bibitem{BouErdYau-14}
\leavevmode\vrule height 2pt depth -1.6pt width 23pt, {\em Universality of
  general $\beta$-ensembles}, Duke Mathematical Journal, 163 (2014),
  pp.~1127--1190.

\bibitem{Bourgain-94}
{\sc J.~Bourgain}, {\em Periodic nonlinear {S}chr\"odinger equation and
  invariant measures}, Comm. Math. Phys., 166 (1994), pp.~1--26.

\bibitem{Bourgain-96}
\leavevmode\vrule height 2pt depth -1.6pt width 23pt, {\em Invariant measures
  for the 2d-defocusing nonlinear {S}chr\"odinger equation}, Comm. Math. Phys.,
  176 (1996), pp.~421--445.

\bibitem{BouBreMir-01}
{\sc J.~Bourgain, H.~Br{\'e}zis, and P.~Mironescu}, {\em {Another look at
  Sobolev spaces}}, in Optimal control and Partial Differential equations, IOS
  Press, 2001, pp.~439--455.  

\bibitem{BouBreMir-02}
\leavevmode\vrule height 2pt depth -1.6pt width 23pt, {\em {Limiting embedding theorems for $W^{s,p}$ when $s\uparrow 1$ and applications}}, J. Anal. Math. 87 (2001), pp.~77--101.  
    
\bibitem{BouPasShc-95}
{\sc A.~Boutet~de Monvel, L.~Pastur, and M.~Shcherbina}, {\em On the
  statistical mechanics approach in the random matrix theory: integrated
  density of states}, J. Stat. Phys., 79 (1995), pp.~585--611.

\bibitem{BraHar-12}
{\sc F.~Brand\~{a}o and A.~Harrow}, {\em Quantum de {F}inetti theorems under
  local measurements with applications}, Proc. of the 45th ACM Symposium on
  theory of computing (STOC 2013), pp. 861-870,  (2013), pp.~861--870.

\bibitem{BurThoTzv-10}
{\sc N.~{Burq}, L.~{Thomann}, and N.~{Tzvetkov}}, {\em {Long time dynamics for
  the one dimensional non linear Schr\"odinger equation}}, Ann. Inst. Fourier.,
  63 (2013), pp.~2137--2198.

\bibitem{BurTzv-08}
{\sc N.~Burq and N.~Tzvetkov}, {\em Random data {C}auchy theory for
  supercritical wave equations. {I}. {L}ocal theory}, Invent. Math., 173
  (2008), pp.~449--475.

\bibitem{BurTzv-08b}
\leavevmode\vrule height 2pt depth -1.6pt width 23pt, {\em Random data {C}auchy
  theory for supercritical wave equations. {II}. {A} global existence result},
  Invent. Math., 173 (2008), pp.~477--496.

\bibitem{CagLioMarPul-92}
{\sc E.~Caglioti, P.-L. Lions, C.~Marchioro, and M.~Pulvirenti}, {\em A special
  class of stationary flows for two-dimensional {E}uler equations: a
  statistical mechanics description}, Comm. Math. Phys., 143 (1992),
  pp.~501--525.

\bibitem{CarFraLie-14}
{\sc E.~Carlen, R.~Frank, and E.~Lieb}, {\em Stability estimates for the lowest
  eigenvalue of a {S}chr\"{o}dinger operator}, Geom. Func. Anal.,  (2014).

\bibitem{ChaGozZit-13}
{\sc D.~Chafa\"i, N.~Gozlan, and P.-A. Zitt}, {\em First order asymptotics for
  confined particles with singular pair repulsions}, Annals of applied
  Probability, 24 (2014), pp.~2371--2413.

\bibitem{CheHaiPavSei-14}
{\sc T.~Chen, C.~Hainzl, N.~Pavlov\'ic, and R.~Seiringer}, {\em On the
  well-posedness and scattering for the {G}ross-{P}itaevskii hierarchy via
  quantum de {F}inetti}, Lett. Math. Phys., 104 (2014), pp.~871--891.

\bibitem{CheHaiPavSei-13}
\leavevmode\vrule height 2pt depth -1.6pt width 23pt, {\em Unconditional
  uniqueness for the cubic {G}ross-{P}itaevskii hierarchy via quantum de
  {F}inetti}, Comm. Pure. Appl. Math.,  (2014).

\bibitem{CheDal-03}
{\sc F.~Chevy and J.~Dalibard}, {\em Les condensats de {B}ose-{E}instein},
  Bulletin de la Soci\'et\'e Fran\c{c}aise de Physique, 142 (2003).

\bibitem{Chiribella-11}
{\sc G.~Chiribella}, {\em On quantum estimation, quantum cloning and finite
  quantum de {F}inetti theorems}, in Theory of Quantum Computation,
  Communication, and Cryptography, vol.~6519 of Lecture Notes in Computer
  Science, Springer, 2011.

\bibitem{ChrKonMitRen-07}
{\sc M.~Christandl, R.~K{\"o}nig, G.~Mitchison, and R.~Renner}, {\em
  One-and-a-half quantum de {F}inetti theorems}, Comm. Math. Phys., 273 (2007),
  pp.~473--498.

\bibitem{ChrTon-09}
{\sc M.~Christiandl and B.~Toner}, {\em Finite de Finetti theorem for
  conditional probability distributions describing physical theories}, J. Math.
  Phys., 50 (2009), p.~042104.

\bibitem{RenCir-09}
{\sc J.~Cirac and R.~Renner}, {\em {de Finetti Representation Theorem for
  Infinite-Dimensional Quantum Systems and Applications to Quantum
  Cryptography}}, Phys. Rev. Lett., 102 (2009), p.~110504.

\bibitem{CohDalLal-05}
{\sc C.~Cohen-Tannoudji, J.~Dalibard, and F.~Lal\"oe}, {\em {La condensation de
  Bose-Einstein dans les gaz}}, Einstein aujourd'hui, CNRS Editions et EDP
  Sciences,  (2005).

\bibitem{ColYuk-00}
{\sc A.~Coleman and V.~Yukalov}, {\em Reduced Density Matrices: {C}oulson's
  Challenge}, Springer Verlag, 2000.

\bibitem{ComSchSei-78}
{\sc J.~Combes, R.~Schrader, and R.~Seiler}, {\em Classical bounds and limits
  for energy distributions of {H}amilton operators in electromagnetic fields},
  Annals of Physics, 111 (1978), pp.~1 -- 18.

\bibitem{Cooper-08}
{\sc N.~R. {Cooper}}, {\em {Rapidly rotating atomic gases}}, Advances in
  Physics, 57 (2008), pp.~539--616.

\bibitem{CorDerZin-09}
{\sc H.~D. Cornean, J.~Derezinski, and P.~Zin}, {\em On the infimum of the
  energy-momentum spectrum of a homogeneous bose gas}, J. Math. Phys., 50
  (2009), p.~062103.

\bibitem{DalGioPitStr-99}
{\sc F.~Dalfovo, S.~Giorgini, L.~P. Pitaevskii, and S.~Stringari}, {\em Theory
  of {B}ose-{E}instein condensation in trapped gases}, Rev. Mod. Phys., 71
  (1999), pp.~463--512.

\bibitem{Dalibard-01}
{\sc J.~Dalibard}, {\em La condensation de {B}ose-{E}instein en phase gazeuse},
  Images de la Physique,  (2001).

\bibitem{DalGerJuzOhb-11}
{\sc J.~Dalibard, F.~Gerbier, G.~Juzeli\={u}nas, and P.~\"{O}hberg}, {\em
  Artificial gauge potentials for neutral atoms}, Rev. Mod. Phys., 83 (2011),
  p.~1523.

\bibitem{DeFinetti-31}
{\sc B.~de~Finetti}, {\em Funzione caratteristica di un fenomeno aleatorio}.
\newblock Atti della R. Accademia Nazionale dei Lincei, 1931.
\newblock Ser. 6, Memorie, Classe di Scienze Fisiche, Matematiche e Naturali.

\bibitem{DeFinetti-37}
\leavevmode\vrule height 2pt depth -1.6pt width 23pt, {\em La pr\'evision, ses
  lois logiques, ses sources subjectives}, Annales de l'IHP, 7 (1937),
  pp.~1--68.

\bibitem{Suzzoni-11}
{\sc A.-S. {de Suzzoni}}, {\em {Invariant measure for the cubic wave equation
  on the unit ball of $\R ^3$}}, Dyn. Partial Differ. Equ., 8 (2011),
  pp.~127--147.

\bibitem{dellAntonio-67}
{\sc G.~{dell'Antonio}}, {\em On the limits of sequences of normal states},
  Comm. Pure Appl. Math., 20 (1967), p.~413.

\bibitem{Derezinski-13}
{\sc J.~Derezi{\'n}ski}, {\em Quantum fields with classical perturbations}, J.
  Math. Phys., 55 (2014), p.~075201.

\bibitem{DerGer-99}
{\sc J.~Derezi{\'n}ski and C.~G{\'e}rard}, {\em Asymptotic completeness in
  quantum field theory. {M}assive {P}auli-{F}ierz {H}amiltonians}, Rev. Math.
  Phys., 11 (1999), pp.~383--450.

\bibitem{DerNap-13}
{\sc J.~{Derezi{\'n}ski} and M.~{Napi{\'o}rkowski}}, {\em {Excitation spectrum
  of interacting bosons in the mean-field infinite-volume limit}}, Annales
  Henri Poincaré, 15 (2014), pp.~2409--2439.

\bibitem{DiaFre-80}
{\sc P.~Diaconis and D.~Freedman}, {\em Finite exchangeable sequences}, Ann.
  Probab., 8 (1980), pp.~745--764.

\bibitem{Dynkin-53}
{\sc E.~B. Dynkin}, {\em Classes of equivalent random quantities}, Uspehi
  Matem. Nauk (N.S.), 8 (1953), pp.~125--130.

\bibitem{Dyson-57}
{\sc F.~J. Dyson}, {\em Ground-state energy of a hard-sphere gas}, Phys. Rev.,
  106 (1957), pp.~20--26.

\bibitem{Dyson-62a}
\leavevmode\vrule height 2pt depth -1.6pt width 23pt, {\em Statistical theory
  of the energy levels of a complex system. part {I}}, J. Math. Phys., 3
  (1962), pp.~140--156.

\bibitem{Dyson-62b}
\leavevmode\vrule height 2pt depth -1.6pt width 23pt, {\em Statistical theory
  of the energy levels of a complex system. part {II}}, J. Math. Phys., 3
  (1962), pp.~157--165.

\bibitem{Dyson-62c}
\leavevmode\vrule height 2pt depth -1.6pt width 23pt, {\em Statistical theory
  of the energy levels of a complex system. part {III}}, J. Math. Phys., 3
  (1962), pp.~166--175.

\bibitem{Einstein-24}
{\sc A.~Einstein}, {\em Quantentheorie des einatomigen idealen {G}ases},
  Sitzber. Kgl. Preuss. Akad. Wiss., 1924, pp.~261--267.

\bibitem{ElgErdSchYau-06}
{\sc A.~Elgart, L.~Erd{\H{o}}s, B.~Schlein, and H.-T. Yau}, {\em
  Gross-{P}itaevskii equation as the mean field limit of weakly coupled
  bosons}, Arch. Ration. Mech. Anal., 179 (2006), pp.~265--283.

\bibitem{ElgSch-07}
{\sc A.~Elgart and B.~Schlein}, {\em Mean field dynamics of boson stars}, Comm.
  Pure Appl. Math., 60 (2007), pp.~500--545.

\bibitem{Enss-77}
{\sc V.~Enss}, {\em A note on {H}unziker's theorem}, Commun. Math. Phys., 52
  (1977), pp.~233--238.

\bibitem{Enss-78}
\leavevmode\vrule height 2pt depth -1.6pt width 23pt, {\em Asymptotic completeness for quantum mechanical potential
  scattering. {I}. {S}hort range potentials}, Commun. Math. Phys., 61 (1978),
  pp.~285--291.

\bibitem{ErdSchYau-07}
{\sc L.~Erd{\"{o}}s, B.~Schlein, and H.-T. Yau}, {\em Derivation of the cubic
  non-linear {S}chr\"odinger equation from quantum dynamics of many-body
  systems}, Invent. Math., 167 (2007), pp.~515--614.

\bibitem{ErdSchYau-09}
\leavevmode\vrule height 2pt depth -1.6pt width 23pt, {\em Rigorous derivation of
  the {G}ross-{P}itaevskii equation with a large interaction potential}, J.
  Amer. Math. Soc., 22 (2009), pp.~1099--1156.

\bibitem{FanSpoVer-80}
{\sc M.~Fannes, H.~Spohn, and A.~Verbeure}, {\em Equilibrium states for mean
  field models}, J. Math. Phys., 21 (1980), pp.~355--358.

\bibitem{FanVan-06}
{\sc M.~Fannes and C.~Vandenplas}, {\em Finite size mean-field models}, J.
  Phys. A, 39 (2006), pp.~13843--13860.

\bibitem{Fetter-09}
{\sc A.~Fetter}, {\em Rotating trapped {B}ose-{E}instein condensates}, Rev.
  Mod. Phys., 81 (2009), p.~647.

\bibitem{Forrester-10}
{\sc P.~Forrester}, {\em Log-gases and random matrices}, London Mathematical
  Society Monographs Series, Princeton University Press, 2004.

\bibitem{Frank-14}
{\sc R.~L. Frank}, {\em Ground states of semi-linear {PDE}s}.
\newblock Lecture notes, 2014.

\bibitem{Freedman-77}
{\sc D.~Freedman}, {\em A remark on the difference between sampling with and
  without replacement}, Journal of the American Statistical Association, 73
  (1977), p.~681.

\bibitem{FroKnoSch-09}
{\sc J.~Fr{\"o}hlich, A.~Knowles, and S.~Schwarz}, {\em On the mean-field limit
  of bosons with {C}oulomb two-body interaction}, Commun. Math. Phys., 288
  (2009), pp.~1023--1059.

\bibitem{Gallavotti-85}
{\sc G.~Gallavotti}, {\em {Renormalization theory and ultraviolet stability for
  scalar fields via renormalization group methods}}, Reviews of Modern Physics,
  57 (1985), pp.~471--562.

\bibitem{Ginibre-65}
{\sc J.~Ginibre}, {\em Statistical ensembles of complex, quaternion, and real
  matrices}, J. Math. Phys., 6 (1965), pp.~440--449.

\bibitem{GinVel-79}
{\sc J.~Ginibre and G.~Velo}, {\em The classical field limit of scattering
  theory for nonrelativistic many-boson systems. {I}}, Commun. Math. Phys., 66
  (1979), pp.~37--76.

\bibitem{GliJaf-87}
{\sc J.~Glimm and A.~Jaffe}, {\em Quantum Physics: A Functional Integral Point
  of View}, Springer-Verlag, 1987.

\bibitem{Golse-13}
{\sc F.~{Golse}}, {\em {On the Dynamics of Large Particle Systems in the Mean
  Field Limit}}, ArXiv e-prints,  (2013).

\bibitem{GolMouRic-13}
{\sc F.~{Golse}, C.~{Mouhot}, and V.~Ricci}, {\em {Empirical measures and
  Vlasov hierarchies}}, Kinetic Theory and Related Fields, 6 (2013),
  pp.~919--943.

\bibitem{Gottlieb-05}
{\sc A.~D. Gottlieb}, {\em Examples of bosonic de {F}inetti states over finite
  dimensional {H}ilbert spaces}, J. Stat. Phys., 121 (2005), pp.~497--509.

  
\bibitem{GotSch-09}
{\sc A.~Gottlieb and T.~Schumm}, {\em Quantum noise thermometry for bosonic
  {J}osephson junctions in the mean-field regime}, Phys. Rev. A, 79 (2009),
  p.~063601.

\bibitem{GreSei-12}
{\sc P.~{Grech} and R.~{Seiringer}}, {\em {The excitation spectrum for weakly
  interacting bosons in a trap}}, Comm. Math. Phys., 322 (2013), pp.~559--591.

\bibitem{Grunbaum-71}
{\sc S.~A. Gr\"unbaum}, {\em {Propagation of chaos for the Boltzmann
  equation}}, Arch. Rational Mech. Anal., 42 (1971), pp.~323--345.  
  
\bibitem{GuoSei-13}
{\sc Y.~Guo and R.~Seiringer}, {\em Symmetry breaking and collapse in
  {B}ose-{E}instein condensates with attractive interactions}, Lett. Math.
  Phys., 104 (2014), pp.~141--156.

\bibitem{HaiLewSol_thermo-09}
{\sc C.~Hainzl, M.~Lewin, and J.~P. Solovej}, {\em The thermodynamic limit of
  quantum {C}oulomb systems. {P}art {I} ({G}eneral {T}heory) and {II}
  ({A}pplications)}, Advances in Math., 221 (2009), pp.~454--487 and 488--546.

\bibitem{Harrow-13}
{\sc A.~Harrow}, {\em The church of the symmetric subspace}, preprint arXiv,
  (2013).

\bibitem{HauMis-14}
{\sc M.~Hauray and S.~Mischler}, {\em On {K}ac's chaos and related problems},
  J. Func. Anal., 266 (2014), pp.~6055--6157.

\bibitem{Hepp-74}
{\sc K.~Hepp}, {\em The classical limit for quantum mechanical correlation
  functions}, Comm. Math. Phys., 35 (1974), pp.~265--277.

\bibitem{HewSav-55}
{\sc E.~Hewitt and L.~J. Savage}, {\em Symmetric measures on {C}artesian
  products}, Trans. Amer. Math. Soc., 80 (1955), pp.~470--501.

\bibitem{HudMoo-75}
{\sc R.~L. Hudson and G.~R. Moody}, {\em Locally normal symmetric states and an
  analogue of de {F}inetti's theorem}, Z. Wahrscheinlichkeitstheorie und Verw.
  Gebiete, 33 (1975/76), pp.~343--351.

\bibitem{JulGotMarPol-13}
{\sc B.~Juli\'a-D\'iaz, A.~Gottlieb, J.~Martorell, and A.~Polls}, {\em Quantum
  and thermal fluctuations in bosonic {J}osephson junctions}, Phys. Rev. A, 88
  (2013), p.~033601.

\bibitem{Khintchine-32}
{\sc A.~Y. Khintchine}, {\em Sur les classes d'\'ev\`enements \'equivalents},
  Mat. Sbornik, 39 (1932), pp.~40--43.

\bibitem{Kiessling-89}
{\sc M.~K.-H. Kiessling}, {\em On the equilibrium statistical mechanics of
  isothermal classical self-gravitating matter}, Jour. Stat. Phys., 55 (1989),
  pp.~203--257.

\bibitem{Kiessling-93}
\leavevmode\vrule height 2pt depth -1.6pt width 23pt, {\em Statistical
  mechanics of classical particles with logarithmic interactions}, Comm. Pure.
  Appl. Math., 46 (1993), pp.~27--56.

\bibitem{Kiessling-12}
\leavevmode\vrule height 2pt depth -1.6pt width 23pt, {\em The {H}artree limit
  of {B}orn's ensemble for the ground state of a bosonic atom or ion}, J. Math.
  Phys., 53 (2012), p.~095223.

\bibitem{KieSpo-99}
{\sc M.~K.-H. Kiessling and H.~Spohn}, {\em A note on the eigenvalue density of
  random matrices}, Comm. Math. Phys., 199 (1999), pp.~683--695.

\bibitem{KilVis-08}
{\sc R.~Killip and M.~Visan}, {\em Nonlinear schr{\"o}dinger equations at
  critical regularity}.
\newblock Lecture notes for the summer school of Clay Mathematics Institute,
  2008.

\bibitem{KlaSka-85}
{\sc J.~Klauder and B.~Skagerstam}, {\em Coherent States, Applications in
  Physics and Mathematical Physics}, World Scientific, Singapore, 1985.

\bibitem{KnoPic-10}
{\sc A.~Knowles and P.~Pickl}, {\em Mean-field dynamics: singular potentials
  and rate of convergence}, Commun. Math. Phys., 298 (2010), pp.~101--138.

\bibitem{KonRen-05}
{\sc R.~K\"{o}nig and R.~Renner}, {\em A de {F}inetti representation for finite
  symmetric quantum states}, J. Math. Phys., 46 (2005), p.~122108.

\bibitem{Kwong-89}
{\sc M.~K. Kwong}, {\em Uniqueness of positive solutions of {${\Delta} u - u +
  u ^{p} = 0 \mbox{ in } {\R} ^{n}$}}, Arch. Rational Mech. Anal., 105 (1989),
  pp.~243--266.

\bibitem{LebRosSpe-88}
{\sc J.~L. Lebowitz, H.~A. Rose, and E.~R. Speer}, {\em Statistical mechanics
  of the nonlinear {S}chr\"odinger equation}, J. Statist. Phys., 50 (1988),
  pp.~657--687.

\bibitem{Lewin-11}
{\sc M.~Lewin}, {\em Geometric methods for nonlinear many-body quantum
  systems}, J. Funct. Anal., 260 (2011), pp.~3535--3595.

\bibitem{LewNamRou-13}
{\sc M.~Lewin, P.~T. Nam, and N.~Rougerie}, {\em Derivation of {H}artree's
  theory for generic mean-field {B}ose systems}, Advances in Mathematics, 254
  (2014).

\bibitem{LewNamRou-14b}
\leavevmode\vrule height 2pt depth -1.6pt width 23pt, {\em Derivation of
  nonlinear {G}ibbs measures from many-body quantum mechanics}, preprint arXiv,
   (2014).

\bibitem{LewNamRou-13b}
\leavevmode\vrule height 2pt depth -1.6pt width 23pt, {\em Remarks on the
  quantum de {F}inetti theorem for bosonic systems}, Appl. Math. Res. Express,
  (2014).
   
\bibitem{LewNamRou-14}
\leavevmode\vrule height 2pt depth -1.6pt width 23pt, {\em The mean-field
  approximation and the non-linear {S}chr\"odinger functional for trapped bose
  gases}, preprint arXiv,  (2014).

\bibitem{LewNamSerSol-13}
{\sc M.~Lewin, P.~T. Nam, S.~Serfaty, and J.~P. Solovej}, {\em Bogoliubov
  spectrum of interacting {B}ose gases}, Comm. Pure Appl. Math., in press
  (2013).

\bibitem{LewRou-unpu12}
{\sc M.~Lewin and N.~Rougerie}.
\newblock unpublished, 2012.

\bibitem{LewRou-unpu13}
\leavevmode\vrule height 2pt depth -1.6pt width 23pt.
\newblock unpublished, 2013.

\bibitem{Lieb-73b}
{\sc E.~H. Lieb}, {\em The classical limit of quantum spin systems}, Comm.
  Math. Phys., 31 (1973), pp.~327--340.

\bibitem{Lieb-76}
\leavevmode\vrule height 2pt depth -1.6pt width 23pt, {\em The stability of matter}, Rev. Mod. Phys., 48 (1976),
  pp.~553--569.

\bibitem{LieLos-01}
{\sc E.~H. Lieb and M.~Loss}, {\em Analysis}, vol.~14 of Graduate Studies in
  Mathematics, American Mathematical Society, Providence, RI, second~ed., 2001.

\bibitem{LieSei-06}
{\sc E.~H. Lieb and R.~Seiringer}, {\em Derivation of the {G}ross-{P}itaevskii
  equation for rotating {B}ose gases}, Commun. Math. Phys., 264 (2006),
  pp.~505--537.

\bibitem{LieSei-09}
\leavevmode\vrule height 2pt depth -1.6pt width 23pt, {\em The {S}tability of
  {M}atter in {Q}uantum {M}echanics}, Cambridge Univ. Press, 2010.

\bibitem{LieSeiSolYng-05}
{\sc E.~H. Lieb, R.~Seiringer, J.~P. Solovej, and J.~Yngvason}, {\em The
  mathematics of the {B}ose gas and its condensation}, Oberwolfach {S}eminars,
  Birkh{\"a}user, 2005.

\bibitem{LieSeiYng-00}
{\sc E.~H. Lieb, R.~Seiringer, and J.~Yngvason}, {\em Bosons in a trap: A
  rigorous derivation of the {G}ross-{P}itaevskii energy functional}, Phys.
  Rev. A, 61 (2000), p.~043602.

\bibitem{LieSeiYng-01}
\leavevmode\vrule height 2pt depth -1.6pt width 23pt, {\em A rigorous
  derivation of the {G}ross-{P}itaevskii energy functional for a
  two-dimensional {B}ose gas}, Comm. Math. Phys., 224 (2001), p.~17.

\bibitem{LieSeiYng-05}
\leavevmode\vrule height 2pt depth -1.6pt width 23pt, {\em {Justification of
  $c$-Number Substitutions in Bosonic Hamiltonians}}, Phys. Rev. Lett., 94
  (2005), p.~080401.

\bibitem{LieSol-13}
{\sc E.~H. Lieb and J.~P. Solovej}, {\em in preparation}.

\bibitem{LieSol-01}
{\sc E.~H. Lieb and J.~P. Solovej}, {\em Ground state energy of the
  one-component charged {B}ose gas}, Commun. Math. Phys., 217 (2001),
  pp.~127--163.

\bibitem{LieSol-04}
\leavevmode\vrule height 2pt depth -1.6pt width 23pt, {\em {Ground state energy of the
  two-component charged Bose gas.}}, Commun. Math. Phys., 252 (2004),
  pp.~485--534.

\bibitem{LieThi-84}
{\sc E.~H. Lieb and W.~E. Thirring}, {\em Gravitational collapse in quantum
  mechanics with relativistic kinetic energy}, Ann. Physics, 155 (1984),
  pp.~494--512.

\bibitem{LieYau-87}
{\sc E.~H. Lieb and H.-T. Yau}, {\em The {C}handrasekhar theory of stellar
  collapse as the limit of quantum mechanics}, Commun. Math. Phys., 112 (1987),
  pp.~147--174.

\bibitem{LieYng-98}
{\sc E.~H. Lieb and J.~Yngvason}, {\em Ground state energy of the low density
  {B}ose gas}, Phys. Rev. Lett., 80 (1998), pp.~2504--2507.

\bibitem{LieYng-01}
\leavevmode\vrule height 2pt depth -1.6pt width 23pt, {\em The ground state
  energy of a dilute two-dimensional {B}ose gas}, J. Stat. Phys., 103 (2001),
  p.~509.

\bibitem{Lions-84}
{\sc P.-L. Lions}, {\em The concentration-compactness principle in the calculus
  of variations. {T}he locally compact case, {P}art {I}}, Ann. Inst. H.
  Poincar{\'e} Anal. Non Lin{\'e}aire, 1 (1984), pp.~109--149.

\bibitem{Lions-84b}
\leavevmode\vrule height 2pt depth -1.6pt width 23pt, {\em The
  concentration-compactness principle in the calculus of variations. {T}he
  locally compact case, {P}art {II}}, Ann. Inst. H. Poincar{\'e} Anal. Non
  Lin{\'e}aire, 1 (1984), pp.~223--283.

\bibitem{Lions-85a}
\leavevmode\vrule height 2pt depth -1.6pt width 23pt, {\em The
  concentration-compactness principle in the calculus of variations. {T}he
  limit case. {I}}, Rev. Mat. Iberoamericana, 1 (1985), pp.~145--201.

\bibitem{Lions-85b}
\leavevmode\vrule height 2pt depth -1.6pt width 23pt, {\em The
  concentration-compactness principle in the calculus of variations. {T}he
  limit case. {II}}, Rev. Mat. Iberoamericana, 1 (1985), pp.~45--121.

\bibitem{Lions-CdF}
\leavevmode\vrule height 2pt depth -1.6pt width 23pt, {\em Mean-field games and applications}.
\newblock Lectures at the {Coll\`ege de France}, Nov 2007.

\bibitem{Maeda-10}
{\sc M.~Maeda}, {\em On the symmetry of the ground states of nonlinear
  {S}chr\"odinger equation with potential}, Adv. Nonlinear Stud., 10 (2010),
  pp.~895--925.


\bibitem{MasNag-78}
{\sc W.~Masja, J.~Nagel}, {\em \"Uber \"aquivalente normierung der anisotropen funktionalra\"ume $H ^{\mu} (\R ^n)$}, Beitr\"age zur Analysis, 12 (1978), pp.~7--17.


\bibitem{MazSha-02}
{\sc V.~Maz'ya, T.~Shaposhnikova}, {\em On the Bourgain, Brezis, and Mironescu theorem concerning limiting embeddings of fractional Sobolev spaces}, J. Func. Anal., 195
  (2002), pp.~230--238.
  
  
\bibitem{Mehta-04}
{\sc M.~Mehta}, {\em Random matrices. Third edition}, Elsevier/Academic Press,
  2004.

\bibitem{MesSpo-82}
{\sc J.~Messer and H.~Spohn}, {\em Statistical mechanics of the isothermal
  {L}ane-{E}mden equation}, J. Statist. Phys., 29 (1982), pp.~561--578.


\bibitem{Mischler-11}
{\sc S.~Mischler}, {\em {Estimation quantitative et uniforme en temps de la
  propagation du chaos et introduction aux limites de champ moyen pour des
  syst\`emes de particules}}.
\newblock Cours de l'Ecole doctorale {EDDIMO}, 2011.

\bibitem{MisMou-13}
{\sc S.~Mischler and M.~C.}, {\em {Kac's Program in Kinetic Theory}},
  Inventiones mathematicae, 193 (2013), pp.~1--147.
  
\bibitem{NamSei-14}
{\sc P.-T. Nam and R.~Seiringer}, {\em Collective excitations of {B}ose gases
  in the mean-field regime}, Arch. Rat. Mech. Anal,  (2014).

\bibitem{PenOns-56}
{\sc O.~Penrose and L.~Onsager}, {\em {Bose-Einstein Condensation and Liquid
  Helium}}, Phys. Rev., 104 (1956), pp.~576--584.

\bibitem{PetSmi-01}
{\sc C.~Pethick and H.~Smith}, {\em Bose-Einstein Condensation of Dilute
  Gases}, Cambridge University Press, 2001.

\bibitem{PetRagVer-89}
{\sc D.~Petz, G.~A. Raggio, and A.~Verbeure}, {\em Asymptotics of
  {V}aradhan-type and the {G}ibbs variational principle}, Comm. Math. Phys.,
  121 (1989), pp.~271--282.

\bibitem{Pickl-10b}
{\sc P.~Pickl}, {\em Derivation of the time dependent {G}ross {P}itaevskii
  equation with external fields}, ArXiv e-prints,  (2010).

\bibitem{Pickl-10}
\leavevmode\vrule height 2pt depth -1.6pt width 23pt, {\em Derivation of the
  time dependent {G}ross-{P}itaevskii equation without positivity condition on
  the interaction}, J. Stat. Phys., 140 (2010), pp.~76--89.

\bibitem{Pickl-11}
\leavevmode\vrule height 2pt depth -1.6pt width 23pt, {\em A simple derivation
  of mean-field limits for quantum systems}, Lett. Math. Phys., 97 (2011),
  pp.~151--164.

\bibitem{PitStr-03}
{\sc L.~Pitaevskii and S.~Stringari}, {\em Bose-Einstein Condensation}, Oxford
  Science Publications, Oxford, 2003.

\bibitem{RagWer-89}
{\sc G.~A. Raggio and R.~F. Werner}, {\em Quantum statistical mechanics of
  general mean field systems}, Helv. Phys. Acta, 62 (1989), pp.~980--1003.

\bibitem{ReeSim4}
{\sc M.~Reed and B.~Simon}, {\em Methods of {M}odern {M}athematical {P}hysics.
  {IV}. {A}nalysis of operators}, Academic Press, New York, 1978.

\bibitem{Renner-07}
{\sc R.~Renner}, {\em Symmetry of large physical systems implies independence
  of subsystems}, Nature Physics, 3 (2007), pp.~645--649.

\bibitem{RobRue-67}
{\sc D.~Robinson and D.~Ruelle}, {\em {Mean entropy of states in classical
  statistical mechanics}}, Comm. Math. Phys., 5 (1967), pp.~288--300.

\bibitem{Robinson-70}
{\sc D.~W. Robinson}, {\em Normal and locally normal states}, Commun. Math.
  Phys., 19 (1970), pp.~219--234.

\bibitem{RodSch-09}
{\sc I.~Rodnianski and B.~Schlein}, {\em Quantum fluctuations and rate of
  convergence towards mean field dynamics}, Commun. Math. Phys., 291 (2009),
  pp.~31--61.

\bibitem{Rougerie-XEDP}
{\sc N.~Rougerie}, {\em Sur la mod\'elisation de l'interaction entre polarons
  et cristaux quantiques}, S\'eminaire Laurent Schwartz,  (2012-2013).

\bibitem{RouSer-13}
{\sc N.~Rougerie and S.~Serfaty}, {\em Higher dimensional {C}oulomb gases and
  renormalized energy functionals}, Comm. Pure Appl. Math.,  (2014).

\bibitem{RouSerYng-13b}
{\sc N.~Rougerie, S.~Serfaty, and J.~Yngvason}, {\em Quantum {H}all phases and
  plasma analogy in rotating trapped bose gases}, J. Stat. Phys.,  (2013).

\bibitem{RouSerYng-13}
\leavevmode\vrule height 2pt depth -1.6pt width 23pt, {\em Quantum {H}all
  states of bosons in rotating anharmonic traps}, Phys. Rev. A, 87 (2013),
  p.~023618.

\bibitem{RouYng-14}
{\sc N.~Rougerie and J.~Yngvason}, {\em Incompressibility estimates for the
  {L}aughlin phase}, Comm. Math. Phys.,  (2014).

\bibitem{SanSer-13}
{\sc E.~Sandier and S.~Serfaty}, {\em 1D {L}og {G}ases and the {R}enormalized
  {E}nergy: Crystallization at vanishing temperature}, Proba. Theor. Rel.
  Fields,  (2014).

\bibitem{SanSer-12}
\leavevmode\vrule height 2pt depth -1.6pt width 23pt, {\em 2D coulomb gases and
  the renormalized energy}, Annals of Probability,  (2014).

\bibitem{Seiringer-03}
{\sc R.~Seiringer}, {\em Ground state asymptotics of a dilute, rotating gas},
  J. Phys. A, 36 (2003), pp.~9755--9778.

\bibitem{Seiringer-11}
\leavevmode\vrule height 2pt depth -1.6pt width 23pt, {\em The excitation spectrum for weakly interacting
  bosons}, Commun. Math. Phys., 306 (2011), pp.~565--578.

\bibitem{SeiUel-09}
{\sc R.~Seiringer and D.~Ueltschi}, {\em Rigorous upper bound on the critical
  temperature of dilute {B}ose gases}, Phys. Rev. B, 80 (2009), p.~014502.

\bibitem{SeiYngZag-12}
{\sc R.~Seiringer, J.~Yngvason, and V.~A. Zagrebnov}, {\em {Disordered
  Bose-Einstein condensates with interaction in one dimension}}, J. Stat.
  Mech., 2012 (2012), p.~P11007.

\bibitem{Serfaty-14}
{\sc S.~Serfaty}, {\em Coulomb gases and {G}inzburg-{L}andau vortices}, Zurich
  Lectures in Advanced Mathematics, 2014.

\bibitem{Sigal-82}
{\sc I.~M. Sigal}, {\em Geometric methods in the quantum many-body problem.
  {N}on existence of very negative ions}, Commun. Math. Phys., 85 (1982),
  pp.~309--324.

\bibitem{Simon-74}
{\sc B.~Simon}, {\em The {$P(\phi )_{2}$} {E}uclidean (quantum) field theory},
  Princeton University Press, Princeton, N.J., 1974.
\newblock Princeton Series in Physics.

\bibitem{Simon-77}
\leavevmode\vrule height 2pt depth -1.6pt width 23pt, {\em Geometric methods in multiparticle quantum systems},
  Commun. Math. Phys., 55 (1977), pp.~259--274.

\bibitem{Simon-79}
\leavevmode\vrule height 2pt depth -1.6pt width 23pt, {\em Trace ideals and
  their applications}, vol.~35 of London Mathematical Society Lecture Note
  Series, Cambridge University Press, Cambridge, 1979.

\bibitem{Simon-80}
\leavevmode\vrule height 2pt depth -1.6pt width 23pt, {\em The classical limit
  of quantum partition functions}, Comm. Math. Phys., 71 (1980), pp.~247--276.

\bibitem{Simon-05}
\leavevmode\vrule height 2pt depth -1.6pt width 23pt, {\em Functional integration and quantum physics}, AMS Chelsea
  Publishing, Providence, RI, second~ed., 2005.

\bibitem{Skorokhod-74}
{\sc A.~Skorokhod}, {\em Integration in {H}ilbert space}, Ergebnisse der
  Mathematik und ihrer Grenzgebiete, Springer-Verlag, 1974.

\bibitem{Solovej-06}
{\sc J.~P. Solovej}, {\em {Upper bounds to the ground state energies of the
  one- and two-component charged Bose gases}}, Commun. Math. Phys., 266 (2006),
  pp.~797--818.

\bibitem{Solovej-notes}
\leavevmode\vrule height 2pt depth -1.6pt width 23pt, {\em Many body quantum
  mechanics}.
\newblock LMU, 2007.
\newblock Lecture notes.

\bibitem{Spohn-80}
{\sc H.~Spohn}, {\em Kinetic equations from {H}amiltonian dynamics: {M}arkovian
  limits}, Rev. Modern Phys., 52 (1980), pp.~569--615.

\bibitem{Spohn-81}
\leavevmode\vrule height 2pt depth -1.6pt width 23pt, {\em On the {V}lasov
  hierarchy}, Math. Methods Appl. Sci., 3 (1981), pp.~445--455.

\bibitem{Spohn-12}
\leavevmode\vrule height 2pt depth -1.6pt width 23pt, {\em Large scale dynamics
  of interacting particles}, Springer London, 2012.

\bibitem{Stormer-69}
{\sc E.~St{\o}rmer}, {\em Symmetric states of infinite tensor products of
  {$C^{\ast} $}-algebras}, J. Functional Analysis, 3 (1969), pp.~48--68.

\bibitem{Summers-12}
{\sc S.~J. {Summers}}, {\em {A Perspective on Constructive Quantum Field
  Theory}}, ArXiv e-prints,  (2012).

\bibitem{ThoTzv-10}
{\sc L.~Thomann and N.~Tzvetkov}, {\em Gibbs measure for the periodic
  derivative nonlinear {S}chr\"odinger equation}, Nonlinearity, 23 (2010),
  p.~2771.

\bibitem{Tzvetkov-08}
{\sc N.~Tzvetkov}, {\em Invariant measures for the defocusing nonlinear
  {S}chr\"odinger equation}, Ann. Inst. Fourier (Grenoble), 58 (2008),
  pp.~2543--2604.

\bibitem{Weinstein-83}
{\sc M.~I. Weinstein}, {\em Nonlinear {S}chr\"odinger equations and sharp
  interpolation estimates}, Comm. Math. Phys., 87 (1983), pp.~567--576.

\bibitem{ZhaFenGil-90}
{\sc W.~Zhang, D.~Feng, and R.~Gilmore}, {\em Coherent states : theory and some
  applications}, Rev. Mod. Phys., 62 (1990), p.~867.

\end{thebibliography}

\end{document}